%% file: main.tex
\documentclass[a4paper,UKenglish,cleveref, autoref, thm-restate]{lipics-v2021}
\nolinenumbers
\hideLIPIcs  %uncomment to remove references to LIPIcs series (logo, DOI, ...), e.g. when preparing a pre-final version to be uploaded to arXiv or another public repository

\title{The DAG Visit approach for Pebbling\\ and I/O Lower Bounds} %TODO Please add

\titlerunning{The DAG Visit approach for Pebbling and I/O Lower Bound} %TODO optional, please use if title is longer than one line

\author{Gianfranco Bilardi}{University of Padova, Department of Information Engineering, Italy}{bilardi\@dei.unipd.it}{}{This work was supported in part by the Italian National Center for
HPC, Big Data, and Quantum Computing; by MIUR, the Italian Ministry of
Education, University and Research, under PRIN Project n. 20174LF3T8
AHeAD (Efficient Algorithms for HArnessing Networked Data); and by the
University of Padova, under Project CPGA$^3$ (Parallel and
Hierarchical Computing: Architectures, Algorithms, and Applications).}%TODO mandatory, please use full name; only 1 author per \author macro; first two parameters are mandatory, other parameters can be empty. Please provide at least the name of the affiliation and the country. The full address is optional. Use additional curly braces to indicate the correct name splitting when the last name consists of multiple name parts.

\author{Lorenzo De Stefani\footnote{Corresponding author}}{Department of Computer Science, Brown University, United States of America}{lorenzo\_destefani\@brown.edu}{https://orcid.org/0000-0001-9569-2086}{}

\authorrunning{G. Bilardi and L. De Stefani} %TODO mandatory. First: Use abbreviated first/middle names. Second (only in severe cases): Use first author plus 'et al.'

\Copyright{G. Bilardi and L. De Stefani} %TODO mandatory, please use full first names. LIPIcs license is "CC-BY";  http://creativecommons.org/licenses/by/3.0/

\ccsdesc[500]{Theory of computation~Design and analysis of algorithms}

\keywords{Pebbling, Directed Acyclic Graph, Pebbling number, I/O complexity} %TODO mandatory; please add comma-separated list of keywords

\category{} %optional, e.g. invited paper

\usepackage{amsthm}
\usepackage{amssymb}

\usepackage{amsmath}
\usepackage{tikz}
\usepackage{mathdots}
\usepackage{yhmath}
\usepackage{cancel}
\usepackage{color}
\usepackage{siunitx}
\usepackage{array}
\usepackage{multirow}
\usepackage{amssymb}
\usepackage{gensymb}
\usepackage{tabularx}
\usepackage{booktabs}
\usetikzlibrary{fadings}
\usepackage{xifthen}
\usepackage{enumitem}

\usepackage{pgf, tikz}
\usetikzlibrary{arrows, automata, petri}

\usepackage{longtable}
\usepackage{todonotes}

\usepackage{algorithm}
\usepackage[noend]{algpseudocode}
\usepackage[utf8]{inputenc}

\usepackage{subcaption}
\usepackage{amsmath}
\usepackage{amsfonts}

\newcommand{\BO}[1]{\mathcal{O}\left(#1\right)}
\newcommand{\BOme}[1]{\Omega\left(#1\right)}

\newcommand{\io}{I/O}

\newcommand{\enset}{Q}
\newcommand{\bs}{b}

\newcommand{\mymax}[1]{\max\{#1\}}

\newcommand{\myie}{\emph{i.e.,}}

\newcommand{\predc}[1]{\mathrm{pre}\left(#1\right)}
\newcommand{\sucx}[1]{\mathrm{suc}\left(#1\right)}
\newcommand{\anc}[1]{\mathrm{anc}\left(#1\right)}
\newcommand{\des}[1]{\mathrm{des}\left(#1\right)}

\newcommand{\dout}[1]{d_{out}\left(#1\right)}
\newcommand{\deout}{d_{out}}
\newcommand{\dein}{d_{in}}

\newcommand{\gra}{G=\left(V,E\right)}

\newcommand{\visrule}{r}
\newcommand{\ruleof}[2]{#1\left(#2\right)}
\newcommand{\sinrule}{r^{(sin)}}
\newcommand{\toporule}{r^{(top)}}
\newcommand{\visset}[1]{\mathcal{R}\left(#1\right)}

\newcommand{\visitset}[2]{\Psi_{#1}\left(#2\right)}

\newcommand{\visit}[1][]{
    \ifthenelse{\isempty{#1}}{\psi}{\psi\left[#1\right]}%
}

\newcommand{\infix}[3]{#1[#2..#3]}

\EventEditors{Anuj Dawar and Venkatesan Guruswami}
\EventNoEds{2}
\EventLongTitle{42nd IARCS Annual Conference on Foundations of Software Technology and Theoretical Computer Science (FSTTCS 2022)}
\EventShortTitle{FSTTCS 2022}
\EventAcronym{FSTTCS}
\EventYear{2022}
\EventDate{December 18--20, 2022}
\EventLocation{IIT Madras, Chennai, India}
\EventLogo{}
\SeriesVolume{250}
\ArticleNo{21}
\begin{document}

\maketitle

\begin{abstract}
We introduce the notion of an $\visrule{}$-visit of a Directed Acyclic 
Graph DAG
$\gra$, a sequence of the vertices of the DAG complying with a given
rule $\visrule{}$. A rule $\visrule{}$ specifies for each vertex 
$v\in V$ a family of $\visrule{}$-enabling
sets of (immediate) predecessors: before visiting $v$, at least one of
its enabling sets must have been visited.  Special cases are the
$\toporule{}$-rule (or, topological rule), for which the only enabling set
is the set of all predecessors and the $\sinrule{}$-rule (or, singleton
rule), for which the enabling sets are the singletons containing
exactly one predecessor. The $\visrule{}$-boundary complexity of a DAG 
$G$, $b_{\visrule{}}\left(G\right)$,
is the minimum integer $b$ such that there is an $\visrule{}$-visit where, at
each stage, for at most $b$ of the vertices yet to be visited an enabling
set has already been visited. By a reformulation of known results, it
is shown that the boundary complexity of a DAG $G$ is a lower bound to
the pebbling number of the reverse DAG, $G^R$. Several known pebbling
lower bounds can be cast in terms of the  $\sinrule$-boundary complexity. The main contributions of this paper are as follows:
\begin{itemize}
     \item An existentially tight $\BO{\sqrt{d_{out} n}}$ upper bound to the
$\sinrule{}$-boundary complexity of any DAG of $n$ vertices and out-degree
$d_{out}$.
\item An existentially tight $\BO{\frac{\deout{}}{\log_2 \deout{}}   
\log_2 n}$ \
upper bound
to the $\toporule{}$-boundary complexity of any DAG. (There are DAGs for
which $\toporule{}$ provides a tight pebbling lower bound, whereas
$\sinrule{}$ does not.)
\item A visit partition technique for I/O lower bounds, which
   generalizes the $S$-partition I/O technique introduced by Hong and
   Kung in their classic paper ``I/O complexity: The Red-Blue pebble
   game''.  The visit partition approach yields tight I/O bounds for
   some DAGs for which the $S$-partition technique can only yield an
   $\BOme{1}$ lower bound.
\end{itemize}
\end{abstract}

\newpage
\section{Introduction}
\label{sec:intro}
A \emph{visit} of a Directed Acyclic Graph (DAG) is a sequence of all
its vertices.  We consider different types of visits, where a type is
specified by a \emph {visit rule} $\visrule$, a prescription that a
vertex $v$ can be visited only after all the vertices in one of a
given family of \emph{enabling sets} of predecessors of $v$ have been
visited.  One example is the \emph {singleton visit rule}, $\sinrule$,
where each vertex is enabled by each singleton containing one of its
predecessors. \emph{Breadth First Search} (BFS) and \emph{Depth First
Search} (DFS) visits are special cases of $\sinrule$-visits. Another
example is the \emph {topological visit rule}, $\toporule$, where a
vertex $v$ is enabled only by the set of all its predecessors. The
$\toporule$-visits are exactly the topological orderings of the
DAG. Many other rules are possible; for example, the enabling sets of
a vertex could be those with a majority of its predecessors.

In this work, we investigate the $r$-\emph {boundary complexity} of
DAGs. The boundary complexity of $G$, $b_r(G)$, is the minimum integer
$b$ such that there exists an $r$-visit where, at each stage, for at
most $b$ of the vertices yet to be visited an enabling set has
already been visited. By a reformulation of the results of Bilardi,
Pietracaprina, and D'Alberto~\cite{bilardi2000space}, in terms of the
familiar concept of visit, we show that the boundary complexity of a
DAG $G$ is a lower bound to the \emph{pebbling number} $p(G_R)$ of its
reverse DAG, \myie{} $p(G_R) \geq b_r(G)$.  The pebbling number of
a DAG provides a measure of the space required by a computation with
data dependences described by that DAG, in the \emph {pebble game}
framework, introduced by Friedman~\cite{friedman1971algorithmic},
Paterson and Hewitt~\cite{paterson1970comparative}, Hopcroft, Paul and
Valiant~\cite{hopcroft1977time}.  While a pebbling game resembles a
visit where each vertex can be visited multiple times, the relation
between pebbling number and boundary complexity is rather subtle, as
indicated by the fact that it involves graph reversal. Several
pebbling lower bounds arguments in the literature (examples are mentioned in Section~\ref{sec:visits and space complexity}) can indeed be recast
in terms of the $r^{(sin)}$-boundary complexity, thus achieving some
unification in the derivation of these results. In this context, it is
natural to explore the potential of the visit approach to yield
significant pebbling lower bounds for arbitrary DAGs.\\

\noindent\textbf{Main contributions:}
We begin our study with the singleton rule and show that, for any DAG
$G$ with $n$ nodes and out-degree at most $\deout$, $b_{\sinrule}(G)
\leq 4\sqrt{\deout n}$. As a universal bound, this result cannot be
improved, as shown by matching existential lower bounds.  With respect
to pebbling, there are DAGs such that $b_{\sinrule}(G)
=\Omega(p(G_R))$, for which the singleton rule provides asymptotically
tight pebbling lower bounds. But there are also DAGs with very low
$\sinrule$-boundary complexity, where the reverse DAG has high
pebbling number. For example, Paul, Tarjan, and
Celoni~\cite{paul1976space} introduced a DAG, which we will
denote as $PTC$, of $n$ vertices and in-degree $\dein{} =O(1)$, and
proved that $p(PTC)=\Theta(\frac{n}{\log n})$.  This DAG can be easily
modified to yield a DAG $PTC^{+}$ with $p(PTC^+)=\Theta(\frac{n}{\log
n})$ and $b_{\sinrule}((PTC^+)_R)=2$, thus exhibiting a large gap
between boundary and pebbling complexity.

It is natural to wonder whether other visit rules can lead to better
bounds, whereas the singleton rule does not. We have then turned our
attention to the topological rule showing that for any DAG $G$ with
$n$ nodes and out-degree at most $\deout \geq 2$ the boundary $b_{\toporule}(G)
=\frac{\deout-1}{\log_2 \deout}\log_2 n$. This bound is existentially
tight. It indicates that the potential of the topological rule for
pebbling lower bounds is limited.  However, the topological technique
is not subsumed by the singleton one, as we exhibit DAGs for which
topological visits yield a tight pebbling lower bound, whereas
singleton visits yield a trivial lower bound.

For an arbitrary visit rule, $\visrule$, we show that $b_r(G) \leq
(\deout-1)\ell+1$, where $\ell$ is the length of the longest paths of
$G$. This result is also existentially tight and is consistent with
the known pebbling upper bound, $p(G_R) \leq (\deout-1)\ell+1$.  It
remains an open question whether a tight boundary-complexity lower bound to the pebbling number can always be found by tailoring the choice of $r$ to
the DAG, or there are DAGs for which the two
metrics exhibit a gap for any rule.

We also exploit visits to analyze the \emph {I/O complexity} of a DAG
$G$, $IO\left(M,G\right)$, pioneered by Hong and Kung
in~\cite{jia1981complexity}. This quantity is the minimum number of
accesses to the second level of a two-level memory, with the first level (\myie{}, the cache) of size $M$, required to compute $G$.  Such computation can be modeled
by a game with pebbles of two colors. Let $k(G,M)$ be the smallest
integer $k$ such that the vertices of $G$ can be topologically
partitioned into a sequence of $k$ subsets, each with
a \emph{dominator set} and \emph{minimum set} no larger than $M$.
($D$ is a dominator of $U$ if the vertices of $U$ can be computed from
those of $D$. The minimum set of $U$ contains those vertices of $U$
with no successor in $U$.)  Then, $IO\left(G,M\right) \geq M
(k(G,2M)-1)$~\cite{jia1981complexity}.  Dominators play a role in the
red-blue game (where pebbles are initially placed on input vertices
which, if unpebbled, cannot be replebbled), but not in standard
pebbling (where a pebble can be placed on an input vertex at any time,
hence a dominator of what is yet to be computed needs not be currently
in memory). Intuitively, each segment of a computation must read a
dominator set $D$ of the vertices being computed and at least $|D|-M$
of these reads must be to the second level of the memory. It is also
shown in~\cite{jia1981complexity} that the minimum set, say $Y$, of a
segment of the computation must be present in memory at the end of
such segment, so that at least $|Y|-M$ of its elements must have been
written to the second level of the memory.  In the visit perspective,
the minimum set emerges as the boundary of topological visits,
capturing a space requirement at various points of the computation. In
addition to providing some intuition on minimum sets, this insight
suggests a generalization of the partitioning technique to any type of
visit. In fact, the universal upper bounds on visit boundaries
mentioned above do indicate that the singleton rule has the potential
to yield better lower bounds than the topological one. Following
this insight, we have developed the visit partition technique. For
some DAGs for which $S$ partitions can only lead to a trivial,
$\Omega(1)$, lower bound, visit partitions yield a much higher and
tight lower bound.

\subparagraph*{Further related work:} 
Since the work of Hong and Kung~\cite{jia1981complexity}, \io{} complexity has attracted
considerable attention, thanks also to the increasing impact of the
memory hierarchy on the performance of all computing systems, from
general purpose processors, to accelerators such as GPUs, FPGAs, and
Tensor engines. Their \emph{$S$-partition} technique has been the
foundation to lower bounds for a number of important computational
problems, such as the Fast Fourier Transform~\cite{jia1981complexity},
the definition-based matrix
multiplication~\cite{ballard2012brief,irony2004communication,
scquizzato2013communication}, sparse matrix
multiplication~\cite{pagh2014input}, Strassen's matrix
multiplication~\cite{bilardi2017complexity} (this work also introduces
the ``\emph{G-flow}'' technique, based on the Grigoriev flow of
functions~\cite{grigor1976application}, to lower bound the size of
dominator sets), and various integer multiplication
algorithms~\cite{bds19,ldshimspaa}. Ballard et
al.~\cite{ballard2011minimizing, ballard2010communication} generalized
the results on matrix multiplication of~\cite{jia1981complexity}, by
means of the approach proposed by Irony, Toledo, and Tiskin in~\cite{irony2004communication} based
on the Loomis-Whitney geometric theorem~\cite{loomis1949}, which captures a trade-off between dominator size and minimum set size. The
same papers present tight \io{} complexity bounds for various linear
algebra algorithms for LU/Cholesky/LDLT/QR factorization and
eigenvalues and singular values computation.

After four decades from its introduction, the \emph{S-partition}
technique~\cite{jia1981complexity} is still the state of the art
for \io{} lower bounds that do hold when recomputation (the repeated
evaluation of the same DAG vertex) is
allowed. Savage~\cite{savage1995extending} has proposed the
\emph{S-span} technique, as `` a slightly weaker but simpler
version of the Hong-Kung lower bound on I/O
time''\cite{savage97models}. The \emph{S-covering}
technique~\cite{bilardi2000space}, which merges and extends aspects
from both~\cite{jia1981complexity} and~\cite{savage1995extending}, is
in principle more general than the $S$-partition technique and leads
to interesting resources-augmentation considerations; however, we are
not aware of its application to specific DAGs.

A number of \io{} lower bound techniques have been proposed and
applied to specific DAG algorithms for executions \emph{without}
recomputations.  These include the \emph{edge expansion} technique
of~\cite{ballard2012graph}, the \emph{path routing} technique
of~\cite{scott2015matrix}, and the \emph{closed dichotomy width}
technique of~\cite{bilardi1999processor}.  While the emphasis in this
paper is on models with recomputation, Section~\ref{sec:otherModels} does show how the visit
partition technique specializes when recomputation is not allowed.
%The composition of \io lower bounds for nr-computations was studied
%in~\cite{elango2015characterizing}

Automatic techniques to derive \io{} lower bounds - with and without
recomputation - have been developed, in part with the goal of
automatic performance evaluation and code restructuring for improving
temporal locality in programs by Elango \emph{et al.}~\cite{elango2015characterizing}, Carpenter \emph{et al.}~\cite{sada16}, 
Olivry \emph{et al.}~\cite{sada2021}.

\subparagraph*{Paper organization:} The visit framework is formulated in
Section~\ref{sec:Visits}.  The relationship between boundary
complexity and pebbling number is discussed in Section~\ref{sec:visits
and space complexity}.  Section~\ref{sec:upperbounds} presents
universal upper bounds to the boundary complexity.
Section~\ref{sec:iocomp} develops the visit partition technique for
I/O lower bounds.  Conclusions are offered in
Section~\ref{sec:conclusion}.  
% Appendices A-D contain proofs and technical material not included in the main body due to space limitations.

\section{Visits of a DAG}
\label{sec:Visits}
A \emph{Directed Acyclic Graph} (DAG) $\gra$ consists of a finite set
of \emph{vertices} $V$ and of a set of directed \emph{edges}
$E\subseteq V \times V$, which form no directed cycle. We say that
edge $\left(u,v\right)\in E$ is directed from $u$ to $v$. We let $\predc{v}=\{u~|~(u,v)\in E\}$ denote the set of predecessors of $v$ and $\sucx{v}=\{u~|~(v,e)\in E\}$ denote  the set of its successors. The maximum in-degree (resp. out-degree) of $G$ is defined as $\dein{}=\max_{v\in V} |\predc{v}|$ (resp., $\deout{}=\max_{v\in V}|\sucx{v}|$). Further, we denote as  $\des{v}$ (resp., $\anc{v}$) of $v$'s descendants (resp., ancestors), that is, the vertices that can be reached from (resp., can reach) $v$ with a directed path. 
Given $V'\subseteq V$, we say that
$G'=\left(V', E \cap\left(V'\times V'\right)\right)$ is the sub-DAG of
$G$ \emph{induced by} $V'$.

Let $\phi=\left(v_1,\dots,v_i,\ldots,v_j,\ldots, v_k\right)$ be a sequence of vertices with  $|\phi|=k$. Let $\phi[i]=v_i$, for $1\leq i\leq j\leq k$, we denote as $\phi(i..j]=\left(v_{i+1},\ldots,v_j\right)$ the infix from the $i$-th element excluded to the $j$-th included. If $j<i$, $\phi(i..j]$ is the empty sequence.  Depending on the context, we sometimes interpret a sequence as the set of items appearing in the sequence.  
  
A \emph{visit} of a DAG $\gra$ is a sequence of all its vertices,
without repetitions, complying with a \emph{visit rule}:
\begin{definition}[Visit rule]	
A \emph{visit rule} for a DAG $\gra$ is a function $\visrule:V\rightarrow2^{2^{V}}$ where $\ruleof{\visrule{}}{v}\subseteq 2^{\predc{v}}$ is a non-empty family of sets of predecessors of $v$ called \emph{enablers of $v$}. The set of visit rules of $G$ is denoted as $\visset{G}$. 
\end{definition}
Intuitively, a rule $\visrule\in \visset{G}$ permits a vertex $v$ to
be visited only after at least one of its enablers $\enset{}\in
\ruleof{\visrule}{v}$ has been entirely visited.
\begin{definition}[$\visrule$-sequence and $\visrule{}$-visit]\label{def:visit}
Given a DAG $\gra$ and a visit rule $\visrule\in\visset{G}$, a sequence $\visit{}$ of \emph{distinct} vertices is an $\visrule$\emph{-sequence} of $G$ if, for every $1\leq i\leq |\visit{}|$,
the prefix $\visit{}[1..i-1]$ includes an enabler $\enset\in
\visrule\left(\visit{}[i]\right)$. The $\visrule$-sequences with
$|\visit{}|=n$ are called $\visrule{}$-\emph{visits} and their
set is denoted as $\visitset{\visrule}{G}$.
\end{definition}

Clearly, any prefix of an $\visrule{}$-sequence is an
$\visrule{}$-sequence. Of particular interest are the ``\emph{
topological visit rule}'' defined as $\ruleof{\toporule}{v}=
\{\predc{v}\}$ and the ``\emph{singleton visit rule}'' defined as
$\ruleof{\sinrule}{v}=\{\{u\}~|~u\in \predc{v}\}$ if $|\predc{v}|>0$
and $\ruleof{\sinrule}{v}=\{\emptyset\}$ otherwise.

A vertex $v$ not contained in $\psi$, but enabled by some non-empty set $Q$ included in $\psi$, is considered to be a ``boundary'' vertex.
\begin{definition}[Boundary of an $\visrule{}$-sequence]\label{def:bundary}
Given a DAG $\gra$, $r\in \visset{G}$, and an $\visrule{}$-sequence $\visit{}$ of $G$, the $\visrule$\emph{-boundary} of $\visit{}$ is defined as the set:
\begin{equation*}
B_{\visrule}\left(\visit{}\right)=\left\{ v\in V\setminus\visit{}~|~\exists \enset{}\neq \emptyset\ s.t.\ \enset\in \ruleof{\visrule}{v}\wedge Q\subseteq \visit{}\right\} .
\end{equation*}
\end{definition}
 Input vertices are never contained in the boundary of any sequence since their only enabler is the empty set.
\begin{definition}[Boundary complexity]~\label{def:boundarycomplexity}
The $\visrule$-boundary complexity of an $\visrule$-sequence $\visit{}$ is defined as:
\begin{equation*}
	\bs{}_{\visrule}\left(\visit{}\right)=\max_{i\in\left\{ 1,\ldots,|\visit{}|\right\} }\left|B_{\visrule}\left(\visit{}[1..i]\right)\right|.
\end{equation*}
The $\visrule$\emph{-boundary complexity} of $G$ is defined as the minimum $\visrule$-boundary complexity among all $\visrule$-visits of $G$:
\begin{equation*}
	b_{\visrule}\left(G\right)=\min_{\visit{}\in\visitset{\visrule}{G}} b_{\visrule}(\visit).
\end{equation*}
\end{definition}
By definition, for any $\visrule\in \visset{G}$, any $\psi\in\visitset{\toporule}{G}$ and, for any $i=1,2,\ldots,n$ we have $B_{\toporule}\left(\visit{}[1..i]\right)\subseteq B_{\visrule}\left(\visit{}[1..i]\right)$; thus, $\bs_{\toporule}\left(\visit{}\right)\leq \bs_{\visrule}\left(\visit{}\right)$. Similarly, if $\emptyset \in \visrule(v)$ only when $v$ has no predecessors, then $\bs_{\visrule}\left(\visit{}\right)\leq \bs_{\sinrule}\left(\visit{}\right)$, for any $\psi\in\visitset{\visrule}{G}$.

\section{Boundary complexity and pebbling number}\label{sec:visits and space complexity}
In this section, we discuss an interesting relationship between the pebbling number of a DAG $\gra$ and the boundary complexity of its \emph{reverse} DAG $G_{R}=\left(V, E_{R}\right)$, where $E_{R} = \{(u,v)| (v,u) \in E \}$, which can prove useful in deriving pebbling lower bounds.
\begin{theorem}[Pebbling lower bound]\label{thm:visitlwb}
Let $G_R$ be the reverse of $\gra$. Then, for any $\visrule \in \visset{G_R}$, the pebbling number of $G$ satisfies:
\begin{equation*}\label{eq:visitlwb}
p\left(G\right)\geq b_r(G_R)=\min_{\visit{}\in\visitset{\visrule}{G_R} }b_r\left(\visit\right).\end{equation*}
\end{theorem}
\begin{proof}
Consider a pebbling schedule $\phi$ of $G$ which uses $s$ pebbles, and any visit rule $\visrule\in \visset{G_R}$. 
The proof proceeds by constructing an $\visrule$-visit $\psi$ of $G_R$ whose $\visrule$-boundary complexity is itself bounded from above by $s$. 
Let $T$ denote the total number of steps of the pebbling $\phi$. The construction of $\psi$ proceeds iteratively starting from the end of the pebbling schedule $\phi$ to its beginning: Let $v$ denote the vertex which is being pebbled at the $i$-th step of $\phi$ for $1\leq i\leq T$, then the same vertex $v$ is visited in $G_R$ if and only if all the vertices in at least one of the subsets in the enabling family $\ruleof{h}{v}$ have already been visited. By construction, $\psi$ is indeed a valid $\visrule$-visit of $G_R$. 

By the construction of the reverse DAG $G_R$, the set of the successors of any vertex $v\in V$ in $G$ corresponds to the set of predecessors of $v$ in $G_R$. By the rules of the pebble game, when considering a complete pebbling schedule for $G$, each vertex $v\in V$  is pebbled in for the first time before any of its successors. As each enabler in $\ruleof{\visrule{}}{v}$ is a subset of the predecessors of $v$ in $G_R$ and, hence, a subset of the successors of $v$ in $G$, each vertex will surely be visited in $\visit{}$ at the step corresponding to its first pebbling in $\phi$ unless it has already been visited. This allows us to conclude that all vertices of $G_R$ are indeed visited by $\psi$.

Let $\phi[t]$ denote the vertex being pebbled at the $t$-th step of the pebbling schedule, for $1\leq t\leq T$. In order to prove that the statement holds, it must be shown that, fixed an index $i$, $1\leq i\leq n$ with $\psi[i]=\phi[t]$,  for some $t\in \{1,2,\ldots,T\}$, the vertices in $B_{\visrule{}}\left(\visit[1..i]\right)$ must be pebbled (\myie{} held in the memory) at the end of $t$-th step of  $\phi$. 

Let $v\in B_{\visrule{}}\left(\visit[1..i]\right)$. By the construction of $\psi$ there must exist two indices $t_{1}$ and $t_{2}$, with $1\leq t_{1}\leq t< t_{2}\leq T$, such that $\phi[t_{1}]=v$, $\phi[t_{2}]\in \enset{} \in \ruleof{\visrule{}}{v}$, and $\phi[t']\neq v$ for every $t_{1}<t'< t_{2}$ (if that was not the case, then $v$ would have been visited then). As a consequence, the value of $v$ computed at step $t_{1}$ of $\phi$ is used to compute $v_{t^{2}}$ and therefore it must reside in memory at the end of step $t$ (\myie{} it has to be pebbled). Since $i$ was chosen arbitrarily, the same reasoning applies for all indices $i=1,2,\ldots, n$. We can thus conclude that the maximum number of pebbles used by $\phi$ (and thus, the memory space used by $\psi$) is no less than $\max_{1\leq i\leq n}\left|B_{\visrule{}}\left(\visit[1..i]\right)\right|=b_\visrule{}\left(\psi\right)$.

The theorem follows by minimizing over all possible $\visrule{}$-visits $\psi\in\visitset{\visrule}{G_R}$.
\end{proof}
In general, the analysis of the boundary complexity is simpler than the analysis of the pebbling number, in part because, in a visit, a vertex can occur only once, whereas, in a pebbling schedule, a vertex can occur any number of times.

The proof of the preceding theorem
%, given in Appendix,
is a reformulation of a result obtained by Bilardi \emph{et al.} in\cite{bilardi2000space}. They introduce the \emph{Marking Rule technique}, which is applied to DAG $G$ rather than to its reverse. The advantage of visits over markings lies in a more direct leverage of intuition, given the widespread utilization of various kinds of visits (\emph{e.g.}, breadth-first search, depth-first search, topological ordering) in the theory and applications of graphs.

We will explore the potential of the visit approach to yield
interesting lower bounds for specific DAGs in the next section.  Here,
we investigate whether Theorem~\ref{thm:visitlwb} could be strengthened
by restricting the set of $\visrule$-visits $\psi$ among which
$b_r\left(\visit\right)$ is minimized. The answer turns out to be
negative.  To clarify in what sense, we need to consider that the
proof of the theorem is based on mapping each pebbling schedule $\pi$
of $G$ to an $\visrule$-visit $f_r(\pi)$ of $G_R$, such that the boundary of each
prefix of $f_r(\pi)$ is completely covered with pebbles at some stage
of $\pi$. In terms of such mapping, we have:
\begin{lemma}[Visit from pebbling schedule]\label{lem:existence}
  For any $\visrule\in \visset{G_R}$ and any
  $\visit{}\in \visitset{\visrule}{G_R}$,  there exists a pebbling
  schedule $\pi$ of $G$ such that $f_r(\pi)=\visit{}$. 
\end{lemma}
% Comment for CR
\begin{proof}
To prove the lemma, we construct a pebbling schedule $\phi$ for $G$, which corresponds to the $\visrule$-visit $\psi$ of $G_R$: the construction proceeds iteratively starting from the end of $\psi$. For any index $i=1,2,\ldots, n$ let $\psi[i]$ denote the vertex visited at the $i$-th step of $\psi$.

Let $G_n$ denote the sub-DAG of $G$ induced by the subset of $V$ composed of $\psi[n]$ and the set of its ancestors in $G$. The schedule $\phi$ starts following the steps of any pebbling schedule of $G_n$. Once $\psi[n]$ has been pebbled, all the pebbles on vertices of $G_n$, except for $\psi[n]$, are removed. The schedule $\phi$ then proceeds according to the same procedure up to $\psi[1]$. During the $i$-th step, the pebbling schedule for $G_i$ never re-pebbles any of the vertices in $\psi[i+1..n]$. As they were previously pebbled, we can assume that they maintain a pebble (\myie{} they are kept in memory) until the end of the computation. Hence it is possible to pebble $G_i$ without re-pebbling the vertices in $\psi[i+1..n]$. By construction, $\phi$ is a complete pebbling of $G$.

Crucially, the order according to which the vertices are pebbled for the \emph{last time} in $\phi$ corresponds to the reverse order of appearance of the vertices in the given visit $\psi$. When applying the conversion between pebbling schedules and visits discussed in the proof of Theorem~\ref{thm:visitlwb}, we have that the vertices are visited in $G_R$ according to the order of their \emph{last pebbling} in $\phi$. The lemma follows.
\end{proof}
Theorem~\ref{thm:visitlwb} provides a general approach for obtaining pebbling lower bounds, which encompasses a number of arguments developed in the literature to analyze DAGs such as directed trees~\cite{paterson1970comparative}, pyramids~\cite{ranjan2012upper} and stacks of superconcentrator~\cite{hopcroft1977time}. 
% CR Version
% A reformulation of these arguments within the visit framework can be found in~\cite{thesis} for trees and stacks of superconcentrators and in Appendix~\ref{sec:examplespace} for pyramid DAGs. 
% Extended version
\input{examplespace}

\section{Upper bounds on boundary complexity}\label{sec:upperbounds}
It is natural to wonder whether the pebbling lower bound of
Theorem~\ref{thm:visitlwb} is tight. As we will see in this section,
both the singleton and the topological rules, while providing tight
bounds for some DAGs, yield weak lower bounds for others.  Whether a
tight lower bound could be obtained for any DAG $G$, by tailoring the
visit rule to $G$, does remain an open question.

In particular, we will establish universal upper bounds on the
boundary complexity of any DAG, with respect to any rule, in terms of
outdegree and depth.  We will also establish (different) universal
upper bounds for both the singleton and the topological rule in terms
of outdegree and the number of vertices. Before presenting these results,
we introduce the notion of enabled reach, a particular set of vertices
associated with a vertex $v$, in the context of a partial visit that
includes $v$.  This concept will play a role in the derivation of each
of the three universal upper bounds.

\subsection{The \emph{enabled reach} of a vertex}\label{sec:enabledrea}
In the construction of a visit sequence, we will use a divide and conquer approach whereby, having constructed a prefix $\psi$ of the sequence, the next segment, $\phi$, of the sequence is obtained by visiting a suitably chosen sub-DAG, $G'=(V',E')$, according to an appropriate rule $\visrule'$. It is useful for the boundary of $G'$ to be ``self-contained'' in the sense that its visit does not generate any boundary outside $V'$. If this is the case, the boundary of $\psi\phi$ will be a subset of the boundary of $\psi$, so that the visit of $G'$ contributes to the reduction of both the set of vertices yet to be visited and the current boundary. The enabled reach, a set of vertices introduced next, induces a sub-DAG $G'$ with the desired properties.

\begin{definition}[Enabled reach]\label{defEnabledReach}
  Let $\gra$ and $\visrule \in \visset{G}$. Given an
  $\visrule{}$-sequence $\visit{}$ and a vertex $v \in \visit{}$, the
  $\visrule$-enabled reach of $v$ given $\visit{}$ is the set:
\[
\mathrm{reach}_r\left(v|\visit{}\right)=\left\{ u~|~\exists~ \visit{}'u\subseteq \mathrm{des}(v) \ s.t.\ \visit{}\visit{}'u\textrm{ is an }r\textrm{-sequence}\right\}.
\]
\end{definition}

Intuitively, we can think of the $\visrule$-enabled reach of a vertex $v$ given an $\visrule{}$-sequence $\psi$ as the set of all the descendants of $v$ which can be visited by extending $\psi$ only with descendants of $v$. As an example, for $\sinrule$, we have that the $\sinrule$-enabled reach of a vertex $v$ given a $\psi$ corresponds to the set of the descendants of $v$ not in $\psi$. 
The enabled reach exhibits the following crucial property:
\begin{lemma}
%  Uncomment for CR
% [Proof in Appendix~\ref{appc}]
\label{lem:enabledreachproperty}
	Given $\gra$, $v\in V$ and $\visrule\in\visset{G}$, let $\visit{}$ be an $\visrule{}$-sequence including $v$.  Let $G'=\left(V',E'\right)$ be the sub-DAG induced by $V'=\mathrm{reach}_r\left(v|\visit{}\right)$. Let $\visrule'\in \visset{G'}$ be such that $\ruleof{\visrule'}{v}= \{\enset \setminus \visit{}|\enset\in \ruleof{\visrule}{v}\wedge \enset \setminus \visit{}\subseteq V'\}$, for all $v\in V'$. If $\visit{}'\in \visitset{r'}{G'}$ then (a) $\visit{}\visit{}'$ is a $\visrule{}$-sequence of $G$; (b) for any $i=1,\ldots|\psi'|$, $B_{\visrule}\left(\visit{}\visit'[1..i]\right)\subseteq B_{\visrule}\left(\visit{}\right)\cup B_{\visrule{}'}\left(\psi'[1..i]\right)$; and (c) $b_{\visrule}\left(\psi\psi'\right)\leq b_{\visrule{}}\left(\psi\right)+b_{\visrule{}'}\left(\psi'\right)$.
\end{lemma}
% Comment for CR
\begin{proof}
By definition, $\enset{}'\in \ruleof{\visrule{}}{v}$ if an only is there exists $\enset{}\in \ruleof{\visrule{}}{v}$ such that $\enset{}'=\enset{}\setminus \visit{}$.
By construction, for any $1\leq j\leq |\visit{}'|$, a vertex $v$ appears in $\visit{}'[1..j]$ if there exists $\enset{}'\in \ruleof{\visrule{}'}{v}$ such that $\enset{}'\in\visit{}'[1..j]$ which implies there exists $\enset{}\in \visit{}\visit{}'[1..j]$, and, thus, $\visit{}\visit{}'[1..j]$ is an $\visrule{}$-sequence.

Recall that a vertex appears in the boundary of a $\visrule{}$ sequence if it is enabled but not visited. 
By construction, $\visit{}'= \mathrm{reach}\left(v|\psi\right)$, thus, by definition, any vertex which is enabled by a subset of $\visit{}\visit''$ must either be included in $B_{\visrule}\left(\visit{}\right)$ or must be included among the vertices of $\mathrm{reach}\left(v|\psi\right)$ not yet visited in $\visit{}\visit{}'[1..j]$ and enabled by $\psi'$, that is $B_{\visrule{}'}\left(\psi'[1..i]\right)$. Hence, we have that $b_{\visrule}\left(\psi\psi'\right)\leq b_{\visrule{}}\left(\psi\right)+b_{\visrule{}'}\left(\psi'\right)$.
\end{proof}
Lemma~\ref{lem:enabledreachproperty} states that visiting the $\visrule{}$-enabled reach of a vertex $v$ given a $\visrule{}$-sequence $\visit{}$ of $G$ does not enable any vertex outside $\mathrm{reach}_\visrule(v|\visit{})$ which was not enabled by $\visit{}$ alone. 
Therefore, once the sub-DAG induced by $\mathrm{reach}_\visrule(v|\visit{})$ is visited, the only vertices left in the $\visrule$-boundary are those enabled by $\visit{}$ that have not been visited thus far.

The enabled reach will be a key ingredient in the construction of
visits in the next three subsections.  The choice of both
$\visrule$-sequence $\psi$ and of vertex $v$ has to be tailored to the
particular $\visrule$. Also, highly influenced by $\visrule$ are the
size of the boundary of $\psi$ and the reduction achieved by $G'$ in
the parameters (\emph{e.g.}, depth or number of vertices) governing
the boundary complexity, hence the shape of the resulting bound.

\subsection{General rules}\label{sec:depth}
The \emph{topological depth} of a DAG is the length (\myie{} number of edges) of its longest directed paths. The boundary complexity, according to any visit rule, can be bounded in terms of the depth and the out-degree.  The basic property that is exploited is that if $b$ is a successor of $a$, then the depth of the sub-DAG induced by the descendants of $b$ is smaller than the depth of the sub-DAG induced by the descendants of $a$.
\begin{theorem}
% Uncomment for CR
% [Proof in Appendix~\ref{app:genupp}]
\label{thm:depthbound}
	Consider $\gra$ with maximum out-degree $\deout$ and topological depth $\ell$. For any visit rule $\visrule\in \visset{G}$, there exists an $\visrule$-visit $\psi\in\visitset{\visrule{}}{G}$ such that $b_{\visrule}(\visit)\leq \left(\deout{}-1\right)\ell+1$.
\end{theorem}
\begin{proof}
	 We construct inductively an $\visrule{}$-visit $\psi$ of $G$ such that  $b_{\visrule}(\visit)\leq \ell\left(\deout{}-1\right)$. 

In the base case $\ell=0$, that is, all the vertices of $G$ are input vertices without predecessors.  Any permutation of the input vertices is an $\visrule{}$-visit. As input vertices are enabled by the empty set, by Definition~\ref{def:bundary}, they do not appear in the boundary and, thus, the $\visrule{}$-boundary complexity of any such visit is zero. 

For the general case $\ell\geq 1$,
the visit begins by visiting any input vertex of $G$.
If none of its direct successors are enabled according to $\visrule{}$, by definition, $B_{\visrule{}}\left(\visit{}[1]\right)=\emptyset$.  If that is the case, the visit proceeds by selecting another input vertex of $G$.

Without loss of generality, let $v$ denote the first input vertex of $G$  whose $\visrule{}$-enabled reach given the $\visrule{}$-sequence $\visit{}_{pre}$ constructed so far is not empty. As visiting $v$ can only enable its at most $\deout{}$ successors, we have $|B_r\left(\psi_{pre}\right)|\leq\deout$. Let $G'=\left(V',E'\right)$ denote the sub-DAG of $G$ induced by the $\visrule$-enabled reach of $v$ given $\psi_{pre}$ and let  $\visrule'(v)=\{\enset\setminus \visit{}_{pre}~|~\enset\in\visrule (v)\}$ for all $v\in \mathrm{reach}_r\left(v|\visit{}_{pre}\right)$. By construction, $G'$ is a sub-DAG of the DAG induced by the set of descendants of $v$, whose topological depth must be at most $\ell-1$. Thus, $G'$ has topological depth at most $\ell-1$ as well. Hence, by the inductive hypothesis,  for any visit rule of $G'$, and, in particular, for  $\visrule'$ there exists an $\visrule'$-visit of $G'$, denoted as $\psi'$ such that $b_{\visrule{}'}\left(\visit{}'\right)\leq \left(\ell-1\right)\left(\deout{}-1\right)+1$. 

By Lemma~\ref{lem:enabledreachproperty},  $\visit{}_{pre}\visit{}'$ is an $\visrule{}$-sequence of $G$ and  vertices in $\psi'$ do not enable any vertex in $V\setminus \visit{}'$. 
Further, as at the first step of $\psi'$ one successor of $v$ is visited, from that step onward, at most $\deout{}-1$ successors of $v$ which are yet to be visited may be in the boundary. Thus:
\begin{align*}
    |B_r\left(\visit{}_{pre}[1..i]\right)|&=0 &\mathrm{for}\ i=1,\ldots,|\visit{}_{pre}|\\
    |B_r\left(\visit{}_{pre}\right)|&\leq\deout \\
    |B_r\left(\visit{}_{pre}\visit{}'[1..j]\right)|&\leq (\deout - 1)+\left(\ell-1\right)\left(\deout{}-1\right)+1 &\mathrm{for}\ j=1,\ldots,|\visit{}'|\\
    |B_r\left(\visit{}_{pre}\visit{}'\right)|&\leq (\deout - 1)
\end{align*}
The visit then proceeds by visiting any input vertex of $G$ which is yet to be visited and its enabled reach given the $\visrule{}$-sequence constructed so far and by repeating the operations previously described. This ensures that $G$ is entirely visited by $\psi$. By repeating the considerations on the boundary size previously discussed, we can conclude that the maximum boundary size of $\psi$ is at most $\left(\deout{}-1\right)\ell+1$. The theorem follows.
\end{proof}
This upper bound is \emph{existentially tight}: For some visit rules $\visrule$, there exist some DAGs
for which $b_{\visrule}(G)=\Theta\left(\deout{}\ell\right)$. An example is given by the reverse $q$-pyramid DAG, discussed in 
%  Uncomment for CR
% Appendix
% Comment for CR
Section
~\ref{sec:rpyramid}, for which $b_{\sinrule}(G)=\Theta\left(\deout{}\ell\right)=\Theta\left(\sqrt{q n}\right)$, since $\deout{}=q$ and $\ell=\sqrt{n/q}$. 

Below, we derive universal upper bounds for the singleton and the
topological rule, which are expressed in terms of $n$ and $\deout{}$.
These bounds are tighter than that of Theorem~\ref{thm:depthbound} for
DAGs with $\ell$ suitably large (as a function of $n$ and $\deout$).

\subsection{Singleton rule $\sinrule{}$}\label{sec:visitssing}
The $\sinrule$ rule has the interesting property that the enabled reach of $v$, given $\psi$, contains all the descendants of $v$ not in $\psi$. One can easily find a $v$ with suitably few descendants, say, less than $n/2$. A $\sinrule$-sequence $\psi$ that contains $v$ can be obtained as the sequence of vertices on a path from an input to $v$. If this path has length $k$, then its boundary could be of a size as big as $(\deout -1)k+1$; therefore, a small $k$ is a prerequisite to guaranteeing small boundary complexity.  In general, a good enough upper bound to $k$ cannot be guaranteed for the entire DAG.  However, it is possible to partition the DAG into a sequence of ``blocks'' such that (i) blocks can be visited one at a time in the order they appear in the sequence; (ii) there is a reasonably small upper bound ($O(\sqrt{\deout n})$) on the number of nodes that are enabled by the nodes in a block, but lie outside the block, and they lie all in the next block; and (iii) in each block, each node is reachable from one of the block inputs by a path of reasonably small length ($O(\sqrt{\deout n})$). When the details are filled in, the outlined approach yields the following results.
\begin{theorem}
% Uncomment for CR
% [Proof in Appendix~\ref{app:sing}]
\label{thm:gensin}
	Given an $\gra$ with $|V|=n$ and maximum out-degree at most $\deout{}$ there exists a visit $\visit\in\visitset{\sinrule}{G}$ s.t. $b_{\sinrule}(\visit{})\leq 4\left(\sqrt{2}+1\right)\sqrt{\deout{}n}$.
\end{theorem}
% Uncomment for CR
\subsection{Proof of upper bound to $\sinrule{}$-boundary complexity}
\noindent\textbf{Theorem~\ref{thm:gensin}}\emph{ Given an $\gra$ with
$|V|=n$ and maximum out-degree $\deout{}$ there exists a visit
$\visit\in\visitset{\sinrule}{G}$ s.t. $b_{\sinrule}(\visit{})\leq
c\sqrt{\deout{}n}$}, where $c= 4\left(\sqrt{2}+1\right)$.
\begin{proof}
For $\deout{} = 0$, all vertices in $G$ are isolated and, having no
predecessors, are enabled only by the empty set. Thus, no vertex
belongs to the $\sinrule$-boundary of any $\sinrule$-visit $\visit$,
hence $b_{\sinrule}(\visit{})=0$, and the stated bound holds.  In the
sequel, we assume $\deout{}\geq 1$, and proceed by induction on $n$.
% This concludes the proof for the case $\deout = 0$. 
% For $\deout \geq n/c^2$, the statement is trivially verified as the boundary complexity of any visit, including all $\sinrule$-visits, is bounded from above by $n$.
% In the following, we assume $1\leq \deout{}< n/c^2$.
% \begin{theorem}[Block lemma for $h^{(sing)}$ visit]\label{thm:gensin}
% Given a DAG $\gra$ with $|V|=n$ and $0<\beta\leq n$, consider a $\beta$-block partition of the $\sinrule$-schedule of $G$ with $k$ total blocks. 
%     Let $G^{(j)}=\left(V^{(j)},E^{(j)}\right)$ denote the sub-DAGs of $G$ induced by the $j$-th block of the $\beta$-block partition. For $1\leq j\leq k$, there exist a $\sinrule$-visit $\psi^{(j)}$ of $G^{(j)}$ with boundary complexity $B_{\psi^{(j)}}^{\sinrule}\leq \deout \frac{n}{\beta}$    
% % 	Given an $n$-vertex DAG $G(V,E)$ there exists a visit $\psi\in\Psi_{h^{(sing)}}\left(G\right)$ s.t. $B_{\psi}^{h^{(sing)}}\leq c\sqrt{d^+n}$ for $c = 3\sqrt{2}(\sqrt{2}+1)$.
% \end{theorem}
% \begin{proof}
% For  $j\in \{1,2,\ldots, k\}$ let us denote $|V^{(j)}|=n'$.
% The proof is by induction on $n$:

\noindent \emph{Base:} For $n=1$, that is, $G=(\{u\},\emptyset)$, the
statement is trivially verified as $v$ is an input vertex and the only
$\sinrule$-visit, $\psi=u$, has $\sinrule$-boundary complexity zero.

\noindent \emph{Inductive step $(n \geq 2)$:} 
% Let us assume that the statement is verified for $n\geq 1$. In the
% following, we show that the statement holds also for $n\geq i+1$.
\noindent\textbf{Case 1: $|I|>1$.} Here, no vertex $v$ is an ancestor
of all vertices. Let $v \in I$ and let $\psi'$ be a $\sinrule{}$-visit
of the DAG induced by $\des{v}$, with boundary complexity at most at
most $c\sqrt{\deout{}n}$, which does exist by the inductive
hypothesis, since $|\des{v}|<n$.  Similarly, let $\psi''$ be a
$\sinrule{}$-visit of the DAG induced by $V \setminus \des{v}$, with
boundary complexity at most $c\sqrt{\deout{}n}$.  Clearly,
$\psi=\psi'\psi'' \in\visitset{\sinrule{}}{G}$.  By the definition of
enabled reach, for $\sinrule{}$, we have that
$\mathrm{reach}_{\sinrule{}}\left(v|v\right)=\des{v} \setminus
\{v\}$. Hence, by Lemma~\ref{lem:enabledreachproperty},
$B_{\sinrule{}}(\psi')=\emptyset$. The stated bound follows.

\noindent\textbf{Case 2: $|I|=1$.} Let $I=\{u\}$.  We partition $V$
into non-empy ``\emph{levels}'' $L\left(1\right),
\ldots,L\left(\ell_s\right)$, such that $v\in L(i)$ if and only if the
shortest directed path from $u$ to $v$ has length $i$.  This path is
also a $\sinrule{}$-sequence. Further, the $\visrule$-boundary of any
$\sinrule{}$-sequence of $G$ included in the first $i$ levels is a
subset of the first $i{+}1$ levels.

We say that level $L(i)$ is a \emph{bottleneck} if $|L(i)| \leq \gamma
\sqrt{\deout{}n}$ and let $i_1< i_2< \ldots < i_k$ denote the indices
of the bottlenecks. Here, $\gamma>0$ is a constant, whose value will
be determined in the course of the proof. Conventionally, we also let
$i_{k+1}=\ell_s{+1}$ and $L(\ell_s+1)=\emptyset$. Since $L(1)=\{u\}$,
we have that $i_1=1$.  We group consecutive levels into \emph{blocks}
$V_j=\cup_{i_j\leq l < i_{j+1}}L\left(l\right)$, for $j=1,2, \ldots,
k$, so that the $j$-th block begins with the $j$-th bottleneck and
ends just before the $(j{+}1)$-st one, or with the last level,
$L(\ell_s)$, if $j=k$.  The number of levels of a block is upper
bounded as $i_{j+1}-i_{j} \leq \frac{\sqrt{n}}{\gamma
  \sqrt{\deout{}}}$.

We construct an $\sinrule$-visit of the form
$\psi=\psi_{1}\psi_{2}\ldots\psi_{k}$, where $\psi_{j}$ is a
$\sinrule{}$-visit of the sub-DAG $G_j$ induced by block $V_j$. Since
the $V_j$'s partition $V$, $\psi\in\visitset{\sinrule{}}{G}$.  By the
properties of the levels mentioned above,
$B_{\sinrule{}}(\psi_{1}\ldots\psi_{j-1})=L\left(i_{j}\right)$.
% By construction, the $i$-th block, for $i=1,2,\ldots$, begins with a
% $\beta$-bottleneck level $L_\visrule{}(start_i)$ whose vertices are
% all be enabled after $\visit{}_1\visit{}_2\ldots\visit{}_{i-1}$
% (hence,
% $B_\visrule{}[\visit{}_1\visit{}_2\ldots\visit{}_{i-1}]=L_\visrule{}(start_i)$) and which are all visited in $\visit{}_{i}$ (hence, will not not be in the boundary of any $\visrule{}$-sequence with prefix $\visit{}_1\visit{}_2\ldots\visit{}_{i}$). By the constructions of $\visit{}$ and of the $\visrule$-greedy schedule the vertices in the $(i+1)$-th and in the following blocks (if any) are not enabled (and, hence, appear in the boundary) of any prefix of $\visit{}$ that does not include at least at least one of the vertices of the $\beta$-bottleneck which is, by construction, the first level of the $(i+1)$-th block has been visited.
Furthermore, as $V_1, \ldots, V_{j-1}$ have already been visited by
$\psi_1 \ldots \psi_{i-1}$ and $\visit{}_j$ is a $\sinrule$-visit of
$G_j$, any of its prefixes $\psi_j'$ may only enable vertices in $V_j$
(\emph{i.e.}, the boundary of $\psi_j'$ in $G_j$) and vertices in
$L(i_{j+1})$ (\emph{i.e.}, the children of vertices in $\psi_j' \subseteq
V_j$, which are not in $V_1 \cup \ldots \cup V_{j}$). Therefore,
\begin{align*}
    b_\sinrule{}(\psi) &\leq
    \max_j\{|L\left(i_j\right)|+|L\left(i_{j+1}\right)|+b_{\sinrule{}}\left(\visit{}_j\right)\}\\ &\leq 2 \gamma
    \sqrt{\deout{}n}+\max_j\{b_{\sinrule{}}\left(\visit{}_j\right)\}.
\end{align*}
A case analysis shows how each $\psi_{j}$ can be chosen so
that the above term is at most $c\sqrt{\deout n}$.\\
\noindent\textbf{Case 2.1:} $|V_j|\leq n/2$. By the inductive
hypothesis, there exists $\visit{}_j\in \visitset{\sinrule{}}{G_j}$
such that $b_{\sinrule{}}(\psi_j) \leq c\sqrt{\deout{}n/2} =
\frac{c}{\sqrt{2}}\sqrt{\deout{}n}$. A sufficient condition for the
desired result is that $2\gamma+\frac{c}{\sqrt{2}} \leq c$, which we
will discuss below.\\
\noindent\textbf{Case 2.2:} $|V_j| > n/2$.  Let
$u_1,u_2,\ldots,u_d$ be the input vertices of $G_j$.\\
\noindent\textbf{Case 2.2.a:} No input of $G_j$ has more than $n/2$
descendants, in $G_j$. Then, we construct an $\sinrule{}$-visit of
$G_j$ as
$\visit{}_j=u_1\visit{}_{u_1}u_2\visit{}_{u_2}\ldots,u_d\visit{}_{u_d}$,
where $\visit{}_{u_l}$ is a $\sinrule{}$-visit, with minimum boundary
complexity, of the sub-DAG induced by the descendants of $u_l$ in
$G_j$ which have not been visited in
$u_1\visit{}_{u_1}u_2\visit{}_{u_2}\ldots,u_l$. This is the
$\sinrule{}$-enabled reach of $u_l$ in $G_j$, given
$u_1\visit{}_{u_1}u_2\visit{}_{u_2}\ldots,u_l$.  Since the input
vertices of $G_j$ are not in the boundary of any prefix of $\psi_j$,
by Lemma~\ref{lem:enabledreachproperty},
% $B_\sinrule{}(\visit{}_{u_1}\visit{}_{u_2}\ldots,\visit{}_{u_t})=\emptyset$ for $t=1,2,
\begin{align*}
    b_\sinrule{}(\visit{}_j)&\leq
    \max_{l=1,\ldots,d}b_{\sinrule{}}(\visit{}_{u_j}) \leq
    c\sqrt{\deout{}\frac{n}{2}},
\end{align*}
 where the last step follows by the inductive hypothesis, considering
 that $\visit{}_{u_j}\subseteq \des{u_j}$ and, by the assumption of
 this case, $|\des{u_j}| \leq n/2$ (here, and throughout the rest of
 this proof, $\des{v}$ refers to the descendants of $v$ in $G_j$). The
 sufficient condition for the desired result is the same as in Case 2.1\\
\noindent\textbf{Case 2.2.b:} There is an input of $G_j$, w.l.o.g, say
$u_1$, such that $|\des{u_1}|>n/2$.  In order to break down $V_j$ into
pieces of size smaller that $n/2$, to be visited one at the time, we
select a vertex $y \in \des{u_1}$ such that $|\des{y}|\geq n/2$ and
$\max_{z\in \mathrm{suc}(y)}|\des{z}|<n/2$.  Specifically, we can
choose $y$ as the last vertex, in a topological ordering of $G_j$,
with at least $n/2$ descendants. Let $y \in L(i)$, where clearly $i_j
\leq i < i_{j+1}$, and consider a shortest path $\pi y$ among those
from input vertices of $G_i$ to $y$. We have
$|\pi|=i-i_{j}<i_{j+1}-i_{j}\leq \frac{\sqrt{n}}{\gamma
  \sqrt{\deout{}}}$.

Consider now the visit $\visit{}_j=\pi y\visit{}_{y}\visit{}'$ such
that $\visit{}_{y}$ is a $\sinrule{}$-visit of the sub-DAG induced by
$V_j\cap \mathrm{reach}_{\sinrule{}}\left(y|\pi y\right)$ constructed
as discussed in Case a, and $\visit{}'$ is a $\sinrule{}$-visit, with
minimum boundary complexity, of the sub-DAG induced by $V_j\setminus
\pi y\visit{}_{y}$. The boundary associated with any prefix of $\pi$
includes at most $\deout{}$ successors for each of its vertices, which
are at most $|\pi| \leq \frac{\sqrt{n}}{\gamma
  \sqrt{\deout{}}}$. Thus, the total contribution of $\pi$ to the
boundary is at most $\frac{\sqrt{\deout{}n}}{\gamma}$.  Arguing along
  the lines of Case a, it can be shown that
  $b_{\sinrule{}}(\visit{}_{y})\leq \frac{c}{\sqrt{2}}\sqrt{\deout{}
    n}$. Finally, since $|\psi_{v^*}|\geq n/2$, then $|V_j\setminus
  \pi y\visit{}_{y}|<n/2$. Hence, by the inductive hypothesis,
  $b_{\sinrule{}}\left(\visit{}'\right)<\frac{c}{\sqrt{2}}\sqrt{\deout{}n}$.

By Lemma~\ref{lem:enabledreachproperty}, we can conclude that
  $\psi_j\in\visitset{\sinrule{}}{G_j}$ and 
\begin{align*}
  b_{\sinrule{}}(\visit{}_j)&\leq \frac{1}{\gamma}\sqrt{\deout{}n}+ \mymax{b_{\sinrule{}}\left(\visit{}_{y}\right),b_{\sinrule{}}\left(\visit{}'\right)}
  \leq  \left(\frac{1}{\gamma}+\frac{c}{\sqrt{2}}\right) \sqrt{\deout{}n}.
\end{align*}
To establish the stated result, we need to satisfy the bound
$2\gamma+\frac{1}{\gamma}+\frac{c}{\sqrt{2}} \leq c$. This requirement
is more stringent than the one for Case 2.1 and Case 2.2.a.  Solving for
$c$, the sufficient condition is $c \geq
\sqrt{2}\left(\sqrt{2}+1\right)\left(2\gamma+\frac{1}{\gamma}\right)$.
The r.h.s. is minimized when we let $\gamma=\frac{1}{\sqrt{2}}$, which
yields $c \geq 4\left(\sqrt{2}+1 \right)$.
\end{proof}
% Uncomment for CR
% The upper bound in Theorem~\ref{thm:gensin} is existentially tight as there exists DAGs, such as $q$-pyramids  discussed in Appendix~\ref{sec:rpyramid}, of matching $\sinrule{}$-boundary complexity. 
% Comment for CR
The upper bound in Theorem~\ref{thm:gensin} is existentially tight as there exists DAGs, such as $q$-pyramids  discussed in Section~\ref{sec:rpyramid}, of matching $\sinrule{}$-boundary complexity. 
\subsection{Topological rule $\toporule{}$}\label{sec:visitstop}
The following property is peculiar to the $\toporule{}$-enabled reach:
\begin{lemma}
%  Uncomment for CR
% [Proof in Appendix~\ref{app:topo}]
\label{lem:disttop}
Let $\psi$ be a $\toporule{}$-sequence of DAG $\gra{}$ and let
$u,v\in \bs{}_{\toporule{}}\left(\visit{}\right)$ be distinct
vertices. Then,
$\mathrm{reach}_{\toporule{}}\left(u|\psi
u\right)\cap \mathrm{reach}_{\toporule{}}\left(v|\psi v\right)=\emptyset$.
\end{lemma}
\begin{proof}Let $w \in \mathrm{reach}_{\toporule{}}\left(u|\psi u\right)$; then,
by Definition~\ref{defEnabledReach}, (i) $\anc{w} \subseteq Y
=\psi \cup u \cup \mathrm{reach}_{\toporule{}}\left(u|\psi u\right)$.
Next, we claim that (ii) $v \notin Y$. In fact, $v \notin \psi$
since $v$ belongs to the boundary of $\psi$; $v \not= u$, by
assumption; and $v \notin
\mathrm{reach}_{\toporule{}}\left(u|\psi u\right)$, since
$v\in \bs{}_{\toporule{}}\left(\visit{}\right)$ implies that
$\anc{v} \subseteq \psi$, whereas $u \notin \psi$, so that
$v \notin \des{u} \supseteq \mathrm{reach}_{\toporule{}}\left(u|\psi
u\right)$.  Combining (i) and (ii), we have that $v \notin \anc{w}$
or, equivalently, that $w \notin \des{v}$ whence, by
Definition~\ref{defEnabledReach},
$w \notin \mathrm{reach}_{\toporule{}}\left(v|\psi v\right)$.
\end{proof}
The disjointedness of the enabled-reach sets ensures that, if there are $k$ vertices in the boundary, at least one of them has an enabled reach with fewer than $n/k$ vertices. By leveraging this property, we obtain the following result.
\begin{theorem}
% Uncomment for CR
% [Proof in Appendix~\ref{app:topo}]
\label{thm:upptopo}
For any DAG $\gra{}$, there exists a visit
$\psi\in\visitset{\toporule{}}{G}$ such that (a) if $\deout{}=0$, then
$b_{\toporule{}}\left(\psi\right)=0$; (b) if $\deout =1$, then
$b_{\toporule{}}\left(\psi\right)=1$; and (c) if $\deout \leq D$, for
some $D \geq 2$, then
$b_{\toporule{}}\left(\psi\right)\leq \frac{D-1}{\log_2 D}\log_2 n
+1$.
\end{theorem}
\begin{proof}
(a) If $\deout{}=0$, then all vertices of $G$ are input vertices and any
ordering $\psi$ of $V$ is a legal topological ordering, with
$b_{\toporule{}}\left(\psi\right)=0$, since input vertices are never
part of the boundary.

(b) If $\deout{}=1$, then it is easy to see that $G$ consists of a
collection of chains (one for each input), at least one of which has
length greater than $1$. Clearly, visiting one chain at the time
yields a visit $\psi$ with $b_{\toporule{}}\left(\psi\right)=1$.

(c) If $\deout{} \leq D$, with $D \geq 2$, we proceed by induction on
$n$.  In the base case, $n=1$, clearly there is a unique schedule
$\psi$ with $b_{\toporule{}}\left(\psi\right)=0$.  For the inductive
step ($n>1$), we recursively construct a visit $\psi=topv(G)$ as
follows. Arbitrarily choose an input vertex $u$ of $G$; let
$B_{\toporule{}}\left(u\right)=\{v_1,v_2,\ldots,v_k\}\subseteq \sucx{u}$,
with $0 \leq k \leq \dout{u} \leq D$; and output $\psi = u
v_{j_1} \phi_1 \ldots v_{j_k} \phi_k \phi$ where, informally, $\phi_h$
is a visit of the enabled reach of $v_{j_h}$, given the prefix of
$\psi$ to the left of $v_{j_h}$.  Each $v_{j_h}$ is chosen among the
successors of $u$ (initially) enabled by $u$ and not yet visited so
as to minimize the size of the enabled reach.  Finally, $\phi$ is a
visit of the vertices not yet traversed by the end of $\phi_k$.  More
formally:\\[-5mm]
\begin{itemize}
\item $j_1,\ldots,j_k$ is a permutation of $1, \ldots, k$, specified
below.
\item For $h=1,\ldots,k$, $v_{j_h}\phi_h=topv(G_h)$, where $G_h$ is
the subgraph induced by
$\{v_{j_h}\} \cup \mathrm{reach}_{\toporule{}}\left(v_{j_h}| u
v_{j_1} \phi_1 \ldots v_{j_{h-1}} \phi_{h-1}v_{j_h} \right)$.  It is
easy to see that any visit of $G_h$ must begin with vertex $v_{j_h}$,
the only input vertex of $G_h$, every other vertex of $G_h$ being a
proper descendant of $v_{j_h}$.
\item $j_h
= \underset{j \in \{1,\ldots,k\} \setminus \{j_1, \ldots,j_{h-1}\}}
{\arg\min} \hspace*{1mm} |\mathrm{reach}_{\toporule{}}\left(v_j| u
v_{j_1} \phi_1 \ldots v_{j_{h-1}} \phi_{h-1}v_j \right)|$.
\item $\phi = topv(\bar{G})$, where $\bar{G}$ is the subgraph
induced by $V\setminus \{u v_{j_1} \phi_1 \ldots v_{j_k} \phi_k\}$.
\end{itemize}
To establish the claimed bound on $b_{\toporule{}}\left(\psi\right)$,
we make the following observations.
\begin{enumerate}
\item For $h=1,\ldots,k$, the boundary of the prefix of $\psi$ ending
just before $v_{j_h}$ equals $\{v_{j_h}, \ldots, v_{j_k}\}$, hence it
has size $k-h+1 \leq k$.
\item By Lemma~\ref{lem:enabledreachproperty}, while visiting $G_h$, 
the boundary is at most
$(k-h)+b_{\toporule{}}\left(v_{j_h}\phi_h\right)$.
\item Collectively, the $k-h+1$ enabled
reaches from which the next one to be visited is chosen contain at most
$n-1-k$ vertices (because $u$ and $v_1,\ldots,v_k$ are not contained
in any of them). Since, by Lemma~\ref{lem:disttop}, these reaches are
disjoint, the smallest ones contain at most
$\lfloor \frac{n-1-k}{k-h+1} \rfloor$ vertices. Additionally
accounting for $v_{j_h}$, $G_h$ contains at most
$1+\lfloor \frac{n-1-k}{k-h+1} \rfloor =
\lfloor \frac{n-h}{k-h+1} \rfloor$ vertices.
\item For convenience, let
$f(n,D)=\frac{D-1}{\log_2 D}\log_2 n +1$ and observe that $f$ is
increasing with both arguments. By the inductive hypothesis, we have:
\begin{equation*}\label{eqn:f1}
b_{\toporule{}}\left(\psi\right) \leq \max\{k,
\max_{h=1}^{k} \{(k-h)+f(\frac{n-h}{k-h+1},D)\}, f(n-1-k,D)\}.
\end{equation*}
Using $f(\frac{n-h}{k-h+1},D) \leq f(\frac{n}{k-h+1},D)$ and letting
$q=k-h+1$, the previous yields
\begin{equation*}\label{eqn:f2}
b_{\toporule{}}\left(\psi\right) \leq \max\{k,
\max_{q=1}^{k} \{q-1+f(\frac{n}{q},D)\}, f(n-1-k,D)\}.
\end{equation*}
\item To complete the inductive step, it remains to show that each of the
three terms in the outer $\max$ is no larger than $f(n,D)$.  This is
obvious for the third term $f(n-1-k,D)$, given the monotonicity of
$f$. For the second term, we observe that each argument in the inner
$\max$ satisfies
\begin{eqnarray*}
q-1+f(\frac{n}{q},D) &=& q-1+\frac{D-1}{\log_2 D}\log_2(\frac{n}{q})+1\\
&=& \left(\frac{D-1}{\log_2 D}\log_2 n+1\right)+
(q-1)\left(1-\frac{D-1}{\log_2 D}\frac{\log_2 q}{q-1}\right)\\
&\leq& \frac{D-1}{\log_2 D}\log_2 n+1,
\end{eqnarray*}
where the term dropped in the last step is negative or null, since
$\frac{q-1}{\log_2 q} \leq \frac{D-1}{\log_2 D}$, for $q=1,\ldots,D$,
as can be shown by straightforward calculus.  Finally, to bound the
first term, $k$, we observe that, if $k \leq 1$, then the target bound
trivially holds.  Otherwise, consider that $n \geq k+1$ ($G$ contains
at least the distinct vertices $u$ and $v_1,\ldots, v_k$) and that,
for $2 \leq k \leq D$, we have $\frac{k-1}{\log_2
k} \leq \frac{D-1}{\log_2 D}$ (as mentioned above). Then, the target
inequality follows:
\begin{eqnarray*}
k = (k-1)+1 \leq \frac{k-1}{\log_2 k}\log_2 k
+1 < \frac{D-1}{\log_2 D} \log_2 n+1.
\end{eqnarray*}
\end{enumerate}
\end{proof}
The upper bound in Theorem~\ref{thm:upptopo} is existentially tight as there exist DAGs such as inverted $q$-trees (discussed in 
% Uncomment for CR
% Appendix~\ref{sec:examplespace})
% Comment for CR
Section~\ref{sec:rtrees})
with matching $\toporule{}$-boundary complexity. 
% Uncomment for CR
% Interestingly, there exist DAGs for which  the $\toporule{}$-boundary complexity is asymptotically higher than the $\sinrule{}$-boundary complexity. One such DAG $G$ is shown in Figure~\ref{fig:treechain2}. It is a simple exercise to prove that $b_\sinrule{}\left(G\right)=1$, $b_\toporule{}\left(G\right)=\Theta\left(\log n\right)$, and  $p\left(G_R\right)=\Theta\left(\log n\right)$.
% Comment for CR
\subsection{Comparison of pebbling lower bounds achievable with $\sinrule{}$ vs $\toporule{}$}
Interestingly, there exist DAGs for which  the $\toporule{}$-boundary complexity is asymptotically higher than the $\sinrule{}$-boundary complexity. One such DAG $G$, shown in Figure~\ref{fig:treechain2}, is obtained by connecting the vertices of a directed binary arborescence with a directed chain. Its $\sinrule{}$-boundary complexity is $2$: consider a visit starting in the leftmost vertex of the chain and that visits vertices of the tree as soon as they are enabled. Instead, its $\toporule{}$-boundary complexity is $\log_2 \left(n+2\right)-1$, where $n$ denotes the number of its vertices. In every possible $\toporule{}$-visit, the directed chain must be visited first, and then the arborescence is left to be visited. The bound on the $\toporule{}$- boundary complexity follows an argument similar to that in the analysis of complete $q$-trees in Section~\ref{sec:rtrees}. 
    Consider the reverse DAG. Its pebbling number is $\left(n+2\right)-1$, where $n$ denotes the number of vertices. (The proof of this statement is a simple exercise.)  By the previous considerations, while using the $\sinrule{}$ yields in Theorem~\ref{thm:visitlwb} yields a trivial $\Omega(1)$ lower bound to the pebbling number. Instead, using $\toporule{}$ yields a tight lower bound to the pebbling number of the DAG.

\begin{figure}[]
\centering
    \resizebox{\textwidth}{!}{
    \input{fig3}}
    \caption{Example of a DAG which allows us to study the relationship between the $\sinrule{}$-boundary complexity and the $\toporule{}$-boundary complexity}
\label{fig:treechain2}
\end{figure}

\section{Visits and I/O complexity}\label{sec:iocomp}
In this section, we show how the visit framework is fruitful in the
investigation not just of space complexity but also of \io{}
complexity, by developing a new \io{} lower bound technique, named
``\emph{visit partition}''. This extends a result by
Hong and Kung~\cite{jia1981complexity}, which has provided the basis
for many \io{} lower bounds in the literature thus far.

The \io{} model of computation is based on a system with a memory
hierarchy of two levels: a fast memory or \emph{cache} of $M$ memory
words and a \emph{slow memory}, with an unlimited number of words.  We
assume that any value associated with a DAG vertex can be stored
within a single memory word. A \emph{computation} is a sequence of
steps of the following types: (i) \emph{operations}, with operands and
results in cache; (ii) \emph{reads}, that copy the content of a memory
location into a cache location; and (iii) \emph{writes}, that copy the
content of a cache location into a memory location. Reads and writes
are also called \io{} operations. Input values are assumed to be
available in the slow memory at the beginning of the computation.
Output values are required to be in the slow memory at the end of the
computation. The conditions under which a computation is a valid
execution of a given DAG are intuitively clear. They can be formalized
in terms of the ``\emph{red-blue pebble game}'', introduced by Hong
and Kung in ~\cite{jia1981complexity}.
% and reviewed in Appendix~\ref{app:rbpg}.  
The \emph{\io{} complexity}
$IO\left(G,M\right)$ of a DAG $G$ is defined as the minimum number
of \io{} operations over all possible computations of $G$. The
\emph{\io{} write complexity} $IO_{\mathcal W}\left(G,M\right)$ and
the \emph{\io{} read complexity} $IO_{\mathcal R}\left(G,M\right)$
are similarly defined.

The visit partition approach to \io{}  lower bounds develops
along the following lines:
\begin{itemize}
\item A visit rule $\visrule$ is chosen for the analysis.
\item A procedure is specified to map each computation $\phi$ of the given DAG
$G$ to an $\visrule$-visit $\psi$ of the reverse DAG $G^R$.
\item Given any partition of visit $\psi$ into consecutive segments,
a set of \io{} operations is identified for each segment, the sets of
different segments being disjoint, whence their contributions can be
added in the \io{} lower bound.
\item The number of read operations associated with a segment is
lower bounded in terms of the size of a minimum
post-dominator of the segment.
\item The number of write operations associated with a segment is
lower bounded in terms of the number of segment vertices that either
inputs of $G_R$ (hence, outputs of $G$), or belong to the boundary of
the visit at the beginning of the segment. The resulting global lower bound,
modified by the addition of the term $|I|-|O|$, also applies to
read operations.
\item A lower bound, $q$, to the \io{} of a given DAG, can be established
by showing that, \emph{for each} visit, \emph{there exists} a segment
partition that requires at least $q$ \io{} operations, between reads
and writes.
\end{itemize}
The technical details of this outline are presented in the next subsections.

\subsection{Segment partitions of a visit}
The concept of \emph{post-dominator set} mirrors that
of \emph{dominator set} used in~\cite{jia1981complexity}:

\begin{definition}[Post-dominator set]\label{def:postdominato}
Given $H=(W,F)$ and $X\subseteq W$, a set $P\subseteq W$ is a
\emph{post-dominator} set of $V'\subseteq W$ if every directed
path from a vertex in $V'$ to an output vertex intersects
$P$. 
We denote as $pd_{min}\left(X\right)$ the minimum size of any post-dominator set of $X$.
\end{definition}
It is simple to see that $P$ is a post-dominator set of $X$
in $H$ if and only if $P$ is a dominator set of $X$ in the reverse DAG $H_R$.

Let $\visit{}$ be an $\visrule$-visit of
$G_R=(V,E_R)$. A \emph{segment partition} of $\visit{}$ into
$k$ \emph{segments} is identified by a sequence of indices ${\bf
i}=(i_1,i_2, \ldots, i_k)$, with $1\leq i_1< i_2 < \ldots< i_k=n$. We
also let $i_{0}=0$, for convenience.  Since $\psi$ is a permutation of
the vertices in $V$, the segments partition
$V$.
For $1\leq j\leq k$, $\psi(i_{j-1}..i_{j}]$ is called the
$j$-th \emph{segment} of the partition.
Two measures play a role in our \io{} lower bound analysis of any segment:
\begin{itemize}
\item The size, $pd_{min}\left(\visit(i_{j-1}..i_j]\right)$, of
minimum post-dominator sets of $\psi(i_{j-1}..i_j]$ .
\item The $\visrule$-\emph{entering boundary size}
\begin{equation*}
    b^{(ent)}_\visrule\left(\visit(i_{j-1}..i_j]\right) = |B^{(ent)}_\visrule\left(\visit(i_{j-1}..i_j]\right)|
\end{equation*}
where, denoting as $I_R$  the set of input vertices of $G_R$, 
\begin{equation*}
    B^{(ent)}_\visrule\left(\visit(i_{j-1}..i_j]\right) = \left( I_R \cup B_{\visrule}\left( \visit[1..i_{j-1}]\right)\right)\cap \visit(i_{j-1}..i_j].
\end{equation*}
\end{itemize}

\subsection{Lower bound}\label{sec:iolwb}
Let $\phi=(\phi_1,\phi_2, \ldots, \phi_T)$ be a computation of a DAG
$\gra{}$ in the \io{} model, in $T$ steps, where $\phi_t$ is the
$t$-th step.  Given an $\visrule\in\visset{G_R}$, we construct an
$\visrule$-visit $\psi$ of $G_R$ corresponding to $\phi$. Our
\io{} lower bounds will be based solely on properties of
the visit. To construct $\psi$, the computation is examined backward,
one step at a time.  The visit is constructed incrementally, by
extending an initially empty prefix.  Let $\psi[1..i(t)]$ be the
prefix already constructed just before processing computation step
$\phi_t$ (initially, $i(T)=0$).  For $t=T,T{-}1, \ldots 1$, if
$\phi_t$ is either a functional operation evaluating vertex $v \in
V\setminus I$ or a read operation copying a vertex $v \in I$ (input of $G$)
into the cache, then, if $v$ has not already been visited
(\myie{} $v \notin \psi[1..i(t)]$) and at least one enabler set
$Q \in
\ruleof{\visrule}{v}$ has already been visited (\myie{}
$Q \subseteq \psi[1..i(t)]$), then $v$ is added to the visit
(\myie{} $i(t{-}1)=i(t){+}1$ and
$\psi[1..i(t{-}1)]=\psi[1..i(t)]v$). Otherwise, the visit constructed thus far remains unchanged (\myie{} $i(t-1)=i(t)$).

By construction, a vertex is included in $\psi$ at most once and only
after at least one of its enablers has been visited. To conclude that
$\psi$ is indeed an $\visrule$-visit of $G_R$ it remains to show that
it contains all the vertices. The vertices in $I_R$ (which are the
inputs of $G_R$, that is, the outputs of $G$) are added to the visit
when they are first encountered in the backward processing of $\phi$,
since they are enabled by the empty set.  Suppose now, by
contradiction, that there are vertices in $V\setminus I_R$ that are not
included in $\psi$. Then, let $t_v$ be the smallest $t$ such that $v$
is computed (if $v\in V\setminus I$) or read from slow memory (if $v \in I$) in
step $\phi_t$, and let $u$ be the vertex with the largest $t_u$ that
is not in $\psi$. It must be the case that if $\psi[1..i(t_u)]$ does not
include any enabler of $u$, which in turn implies that there exists a
predecessor $w$ of $u$ (in $G_R$) that does not belong to
$\psi[1..i(t_u)]$. Since $t_w >t_u$, this implies that
$\psi[1..i(t_w)]$ does not include any enabler of $w$, whence
$w \notin \psi$, which contradicts the definition of $u$.

The procedure just described to construct an $\visrule$-visit of $G_R$
from an \io{} computation $\phi$ of $G$ is quite similar to the one
in the proof of Theorem~\ref{thm:visitlwb} for the analysis of the
pebbling number. The differences are due to the circumstance that the
standard \io{} model assumes the inputs to be initially available in
slow memory, whereas the pebbling model assumes that the inputs can be
(repeatedly) loaded into the working space at any time.

\begin{lemma}[Visit partition]\label{lem:iostep}
Let $\phi$ be a computation of DAG $\gra{}$ on the \io{} model with a
cache of $M$ words. Let $\visrule\in\visset{G_R}$ and let $\psi$
be the $\visrule$-visit of $G_R$ constructed from $\phi$, as described above,
and ${\bf i}=\left(i_1,i_2,\ldots, i_{k}\right)$ any of its segment
partitions. Then, the number $IO_\mathcal{W}(\phi,M)$ of write \io{}
operations and the number $IO_\mathcal{R}(\phi,M)$ of read \io{}
operations executed by $\phi$ satisfy the bounds
\begin{eqnarray}\label{eq:write}
    IO_\mathcal{W}(\phi,M) &\geq& \mathcal{W}_{\visrule}({\bf i},\psi,M)
    := \sum_{j=1}^{k} \max\{0,b^{(ent)}_\visrule\left(\visit(i_{j-1}..i_j]\right)-M\},\\
\label{eq:read0}
    IO_\mathcal{R}(\phi,M) &\geq& \mathcal{W}_{\visrule}({\bf i},\psi,M)+|I|-|O|, \\
\label{eq:read}
    IO_\mathcal{R}(\phi,M) &\geq & \mathcal{R}_{\visrule}({\bf i},\psi,M)
    := \sum_{j=1}^{k}\max\{0,pd_{min}\left(\visit(i_{j-1}..i_j]\right)-M\}.
\end{eqnarray}
The total number $IO(\phi,M)=IO_\mathcal{W}(\phi,M)+IO_\mathcal{R}(\phi,M)$
of \io{} operations executed by $\phi$ satisfies the bound
\begin{equation}\label{eq:tot}
    IO(\phi,M)\geq \mathcal{W}_{\psi}({\bf i},M)
        + \max\{\mathcal{R}_{\psi}({\bf i},M), \mathcal{W}_{\psi}({\bf
        i},M)+|I|-|O|\}.
\end{equation}
\end{lemma}
\begin{proof}
We will analyze the entering boundary and the post-dominator
contributions to the lower bound on the number of, respectively, write
and read \io{} operations for a generic segment of the visit
$\psi(h..i]$, with $0 \leq h < i \leq n$, and then compose the
contributions of the segments in partition ${\bf i}$.  For
$l=1,\ldots,n$, we let $\tau_l$ be such that $\psi[l]$ has been added
to visit $\psi$ in correspondence of computation step
$\phi_{\tau_{l}}$. Observe that $\tau_1 > \tau_2 >\ldots >\tau_n$,
since the visit is constructed from the computation in reverse. When
speaking of cache or of the slow memory at time $t$, we refer to their
state just before the execution of computation step $\phi_t$.

\emph{Proof of~\eqref{eq:write} - Boundary bound}: We claim that, at time $\tau_h$,
the value of each vertex of set
$B^{(ent)}_\visrule\left(\visit(h..i]\right)$ is stored in cache or in
slow memory. Let
\[v \in B^{(ent)}_\visrule\left(\visit(h..i]\right) =
(B_{\visrule}\left(\visit[1..h]\right)\cup O)\cap \visit(h..i],\]
where $O=I_R$ is the set of output vertices of $G$ (also, input vertices
of $G_R$). Let $v = \psi[g]$, with $h<g\leq i$.

\underline{Case 1}. If $v \in O$,
then at the time $\tau_g \in [\tau_i,\tau_h)$ when it has been visited
$v$ has been computed for the last time. Thereafter, by the rules of
the \io{} model, the value of $v$ must be kept in memory till the end
of the computation. At time $\tau_h$, $v$ can be either in cache or
slow memory, but at some time $t>\tau_g$ it must be written in slow
memory.

\underline{Case 2}. If $v\in B_{\visrule}\left(\visit[1..h]\right)$,
then $\visit[1..h]$ includes an $\visrule$-enabler of $v$ and we
consider the smallest index $f$ such that $\visit[1..f])$ includes an
$\visrule$-enabler of $v$. Clearly $f \leq h$ and $u=\psi[f]$ is a
successor of $v$ in $G$.  We argue that the value of $v$ must be in
memory during the interval $(\tau_g,\tau_f]$. In fact, when $u$ is
computed, $v$ must be in cache. We separately analyze two subcases.

\underline{Case 2.1}. If $v \in V\setminus I$, then it is not computed at any time
$t \in (\tau_g,\tau_f]$, otherwise $v$ would be visited at $t$. Since
$\tau_g < \tau_h \leq \tau_f$, we have that (a copy of the value of)
$v$ is in memory at time $\tau_h$.

\underline{Case 2.2}. If $v \in I$, then it is not read from slow memory
at any time $t \in (\tau_g,\tau_f]$, otherwise $v$ would be visited at
$t$. Then, throughout this interval, $v$ must be kept in cache. Since
$\tau_g < \tau_h \leq \tau_f$, we have that (a copy of the value of)
$v$ is in cache at time $\tau_h$.

Given that at most $M$ vertices of
$B^{(ent)}_\visrule\left(\visit(h..i]\right)$ can be in cache at
$\tau_h$, we conclude that at least
$\max\{0,|B^{(ent)}_\visrule\left(\visit(h..i]\right)|-M\} =
\max\{0, b_r^{ent}(\psi(h..i])-M\}$ vertices must be in slow memory.
Such vertices must all fall under Case 2.1 since the vertices in Case
2.2 must be in the cache. Then, they must have been written into slow
memory, thus contributing to the number of write \io{}, like the
vertices in $O$ (Case 1). Moreover, the contributions of the
(disjoint) segments of partition ${\bf i}$ can be added, since they
count write operations involving vertices that belong to disjoint
sets. This concludes the proof for~\eqref{eq:write}.

\emph{Proof of \eqref{eq:read0} - Modified Boundary bound:} By examining the argument for Case 2.1, we see that the vertices
involved must also be read, at some time $t \geq \tau_h$, so that,
those that are not in cache at time $\tau_h$ contribute to the number
of read \io{}. Each vertex in $I$ is read at least once from
slow memory. Finally, considering that no read \io{} has been argued
when $v \in O$, we reach~\eqref{eq:read0}.

\emph{Proof of~\eqref{eq:read} - Post-dominator bound}: We claim that
the set $Y$ of vertices that are in cache at time $\tau_i$ or are read
into the cache during the interval $[\tau_i,\tau_h)$ is a dominator
set of $\psi(h..i]$ in $G$ or, equivalently, a post-dominator set of
$\psi(h..i]$ in $G_R$. Let $v = \psi[g]$, with $h<g\leq i$.

If $v \in I$, then at the time $\tau_g \in [\tau_i,\tau_h)$, when $v$
has been visited, it has been read into the cache, so that $v \in Y$,
whence $v$ is dominated by $Y$.

If $v \in V\setminus I$, then at the time $\tau_g \in [\tau_i,\tau_h)$ when $v$
has been visited it has been computed. If, by way of contradiction,
$v$ is not dominated by $Y$, then there is a directed path in $G$, say
$(v_1,v_2, \ldots v_q=v)$, with no vertex in $Y$. Let $v_s$ be the 
first vertex on this path computed during interval $[\tau_i,\tau_h)$,
which must exist since $v$ is computed during such interval. When
$v_s$ is computed, $v_{s-1}$ must be in cache, since it is one of its
operands. However, during $[\tau_i,\tau_h)$, $v_{s-1}$ cannot be in
cache since is not computed (by the definition of $v_s$), nor is it
available  at the initial time $\tau_i$ or read from slow memory
(since $v_s \notin Y$). Thus, we have reached a contradiction, which
shows that $v$ is actually dominated by $Y$.

Given that at most $M$ vertices of $Y$ can be initially in the cache, we
conclude that at least
$\max\{0,|Y|-M\} \geq \max\{0,pd_{min}(\psi(h..i])-M\}$ must be
brought into cache by read operations that occur during the interval
$[\tau_i,\tau_h)$. Moreover, the contributions of the (disjoint)
segments of partition ${\bf i}$ can be added since they count read
operations occurring in different time intervals. This concludes the
proof for~\eqref{eq:read}.

Finally, a straightforward combination of
bounds~\eqref{eq:write},~\eqref{eq:read0}, and~\eqref{eq:read}
yields bound~\eqref{eq:tot}.
\end{proof}
We observe that, in Lemma~\ref{lem:iostep}, we can choose the visit
rule and then, for the visit $\psi$ corresponding to a given
computation $\phi$, we can choose the segment partition ${\bf i}$ with
the goal of maximizing the resulting lower bound for the cost metric
of interest. However, for the lower bound to apply to (all
computations of) DAG $G$, we have to consider the minimum lower bound
over all visits. The preceding observations are made more formal in
the next theorem. It is generally possible that none of the
computations that minimize the number of read operations also
minimizes the number of write operations, so that the \io{} complexity
may be larger than the sum of the read and of the write complexity.

\begin{theorem}[\io{} lower bound]\label{thm:finiolwb}
Given a DAG $\gra$, with input set $I$ and output set $O$, a visit
rule $\visrule\in\visset{G_R}$, a visit of $G_R$ according to this
rule $\visit{}\in\visitset{\visrule}{G_R}$, and a cache size $M$, we
define the quantities
\begin{eqnarray}\label{eq:visitWriteLB}
{\mathcal{W}}_\visrule(\psi,M) &=&
\max_{{\bf i} \in \mathcal{I}(\visit{})} \mathcal{W}_{\visrule}({\bf i},\psi,M),\\
\label{eq:visitReadLB}
{\mathcal{R}}_{\visrule}(\psi,M) &=&
\max_{{\bf i} \in \mathcal{I}(\visit{})} \mathcal{R}_{\visrule}({\bf i},\visit,M).
\end{eqnarray}
where $\mathcal{I}(\visit{})$ denotes the set of all segment
partitions of $\psi$ and the quantities $\mathcal{W}_{\psi}({\bf
i},M)$ and $\mathcal{R}_{\psi}({\bf i},M)$ are those introduced in
Equations~\eqref{eq:write} and~\eqref{eq:read}, respectively. Then,
the write \io{} complexity $IO_{\mathcal W}\left(G,M\right)$, the
read \io{} complexity $IO_{\mathcal W}\left(G,M\right)$, and the
total \io{} complexity $IO\left(G,M\right)$ satisfy the following
bounds:
\begin{eqnarray}\label{eq:totwrite}
    IO_\mathcal{W}\left(G,M\right) & \geq &
    \min_{\visit{}\in\visitset{\visrule}{G_R}}
    {\mathcal{W}}_{\visrule}(\psi,M),\\
    \label{eq:totread}
    IO_\mathcal{R}\left(G,M\right) & \geq &
    \min_{\visit{}\in\visitset{\visrule}{G_R}}
    \max\{{\mathcal{R}}_{\visrule}(\psi,M),~{\mathcal{W}}_{\visrule}(\psi,M)+|I|-|O|\},\\
    \label{eq:totot}
    IO\left(G,M\right)& \geq &\min_{\visit{}\in\visitset{\visrule}{G_R}}
    \{{\mathcal{W}}_{\visrule}(\psi,M) + 
    \max\{{\mathcal{R}}_{\visrule}(\psi,M),{\mathcal{W}}_{\visrule}(\psi,M)+|I|-|O|\} \}.
\end{eqnarray}
\end{theorem}
\begin{proof}
Bound~\eqref{eq:write} of Lemma~\ref{lem:iostep} applies to a
computation $\phi$ of $G$ for which the procedure constructing the
visit outputs $\psi$.  The bound holds for any segment partition ${\bf
i}$ and, in particular, for the partition that maximizes
$\mathcal{W}_{\visrule}({\bf i},\psi,M)$.  Using
Definition~\eqref{eq:visitWriteLB}, we obtain
$IO_{W}(\phi,M) \geq {\mathcal{W}}_\visrule(\psi,M)$.  To formulate a
lower bound that holds for any computation $\phi$, hence for any visit
$\psi$, we need to minimize with respect to
$\visit{}\in\visitset{\visrule}{G_R}$, arriving at
Bound~\eqref{eq:totwrite}.

Bounds~\eqref{eq:totread} and~\eqref{eq:totot} are established by
analogous arguments.
\end{proof}
\subsection{Comparison with Hong and Kung's S-partition technique}
Hong and Kung~\cite{jia1981complexity}
introduced the ``\emph{S-partition technique}'', for \io{} lower
bounds. In this section, we show that the central result of their
approach can be derived as a corollary of the visit partition
approach, when the latter is specialized to the topological visit rule
$\toporule$.

The $S$-partitions of a DAG are defined in terms of dominator and
minimum sets. Given a DAG $G=(V,E)$ and a set $V'\subseteq V$, we say
that $D\subseteq V$ is a \emph{dominator set} of $V'$ if every
directed path from an input vertex of $G$ to a vertex in $V'$
intersects $D$. The \emph{minimum set} of $V'$ is the set of all
vertices of $V'$ that have no successors in $V'$. An $S$-partition is
a sequence $(V_1,V_2,\ldots,V_k)$ of sets such that (a) they are
disjoint and their union equals $V$; (b) each $V_j$ has a dominator
set of size at most $S$; (c) the minimum set of each $V_j$ has size at
most $S$; (d) there is no edge from a vertex in $V_j$ to a vertex in
$\cup_{i=1}^{j-1} V_i$.
\begin{theorem}[Adapted from {\cite[Theorem
        3.1]{jia1981complexity}}]\label{thm:hongkung} Any computation
  of a DAG $\gra$ on the \io{} model with a cache of $M$ words, executing
  $q$ \io{} operations, is associated with a $2M$-partition of $G$ with
  $k$ sets, such that $M k > q > M(k-1)$. Therefore, if $k(G,2M)$ is the minimum size of a $2M$-partition of
$G$, the \io{} complexity of $G$ satisfies:
\begin{equation}\label{eq:hongkungbound}
    IO\left(G,M\right)\geq M\left(k(G,2M)-1\right).
\end{equation}
\end{theorem}
Next, we show that when $\visrule$ is the topological rule,
Bound~\eqref{eq:totot} of Theorem~\ref{thm:finiolwb} implies
Bound~\eqref{eq:hongkungbound} of Theorem~\ref{thm:hongkung}. In the visit framework, minimum sets arise as the boundary
of topological visits.
\begin{theorem}\label{lem:comparison1}
Given a DAG $\gra{}$, let $\psi$ be any $\toporule$-visit of
$G_R$. There exists at least one segment partition
${\bf i}=\left(i_1,i_2,\ldots,i_k\right)$ of $\psi$ such that
\begin{equation*}
\mathcal{W}_{\visrule}({\bf i},\psi,M)+\mathcal{R}_{\visrule}({\bf
  i},\psi,M) \geq M\left(k(G,2M)-1\right),
\end{equation*}
whence $IO\left(G,M\right) \geq M\left(k(G,2M)-1\right)$.
\end{theorem}% 
\begin{proof}
For notational simplicity, throughout this proof, $\visrule$ stands
for $\toporule$. We preliminarily observe that if
${\visit{}\in\visitset{\toporule}{G_R}}$, then, for any segment
$\psi(h,i]$, $B^{(ent)}_{\visrule}\left(\psi(h,i]\right)$ is the
    minimum set of $\psi(h,i]$ in $G$.  In fact, by the definition of
      $\toporule$, $v \in B_{\visrule}\left(\psi[1,h]\right)$ if and
      only if \emph{all} of its predecessors in $G_R$ (\emph{i.e.},
      its successors in $G$) are in $\visit[1..h]$.  Hence,
      $B^{(ent)}_\visrule\left(\visit(h..i]\right)$ contains exactly
        those vertices in $\visit(h,i]$ which are either inputs of
          $G_R$ (outputs of $G$) or for which $\visit(h,i]$ contains
            no predecessor in $G_R$ (thus, no successor in $G$).

Below we will make use of the following two properties, whose simple
proof is omitted here. Let $U \subseteq V$ and $x \in V$. Then
$min_{pd}(U \cup \{x\}) \leq min_{pd}(U)+1$ and $sms(U \cup \{x\}) \leq
sms(U)+1$, where $sms(X)$ denotes the size of the minimum set of $X
\subseteq V$.

Given ${\visit{}\in\visitset{\toporule}{G_R}}$, we consider a segment
partition ${\bf i}=\left(i_1,i_2,\ldots,i_k\right)$ where each
segment, with the possible exception of the last one, has a minimum
postdominator or the minimum set of size $2M$. Based on the properties
stated in the preceding paragraph, such a partition can be easily
constructed by scanning the visit one vertex at a time and closing a
segment as soon as the desired condition is met, or all vertices have
been scanned.

For $j=1,\ldots,k$, let $V_{j}=\{\psi[i_{k-j}{+}1], \ldots,
\psi[i_{k+1-j}]\}$. We show that sequence $(V_1,\ldots,V_k)$ is a
$2M$-partition of $G$, by proving the defining properties: (a) As they
correspond to the segments of a visit, the $V_j$'s are disjoint and
their union equals $V$. (b) By construction, the dominator of each
$V_j$ in $G$ has size at most $2M$.  (c) By construction, the minimum
set of each $V_j$ in $G$ has size at most $2M$. (d) Finally, since
$\psi$ is an $\toporule$-visit of $G_R$, it is a reverse topological
ordering of the vertices in $G$. Therefore, there are no edges of $G$
from $V_j$ to $\cup_{i=1}^{j-1} V_i$. Clearly, $k \geq k(G,2M)$, by
the definition of the latter quantity.

To analyze the \io{} requirements of $\psi$, consider the
following sequence of inequalities:
\begin{eqnarray*}
  \mathcal{W}_{\visrule}(\psi,M)
  &+&\mathcal{R}_{\visrule}(\psi,M) \geq
    \mathcal{W}_{\visrule}({\bf i},\psi,M)
   + \mathcal{R}_{\visrule}({\bf i},\psi,M)\\
  & \geq & \sum_{j=1}^{k}\max
  \{0, b^{(ent)}_\visrule\left(\visit(i_{j-1}..i_{j+1}]\right)-M,
    min_{pd}\left(\visit(i_{j-1}..i_{j}]\right) -M\}  \\
  & \geq & \sum_{j=1}^{k-1} (2M-M) = (k-1)M\\
  & \geq & M\left(k(G,2M)-1\right).
\end{eqnarray*}  
These four inequalities respectively take into account (i) the
definitions in~\eqref{eq:visitWriteLB} and~\eqref{eq:visitReadLB};
(ii) the definitions in~\eqref{eq:write} and~\eqref{eq:read}; (iii)
the fact that, for $j=1,\ldots,k-1$, at least one of the arguments of
the $\max$ operator equals $M$; and (iv) the relationship $k \geq
k(G,2M)$, seen above. Finally, since the chain of inequalities
applies to any visit ${\visit{}\in\visitset{\toporule}{G_R}}$, we can
invoke inequality~\eqref{eq:totot} to conclude that
$IO\left(G,M\right) \geq M\left(k(G,2M)-1\right)$.
\end{proof}
\input{diamond.tex}
\subsection{Extensions to related I/O models}\label{sec:otherModels}
\subparagraph{External
Memory Model:}The result in Theorem~\ref{thm:finiolwb} can be straightforwardly extended to the External
Memory Model of  Aggarwal and Vitter~\cite{Aggarwal:1988:ICS:48529.48535},
where a single \io{} operation can move $L\geq 1$ memory words between
cache and consecutive slow-memory locations. 

\subparagraph{Models with asymmetric cost of read and write I/O operations:} Since our method distinguishes the contribution of write and reads \io{} operations, it would be interesting to use it to investigate \io{} lower bounds where reads and writes have different cost~\cite{Blelloch1,Blelloch2}, including the case in which only the cost of write \io{} operations is considered~\cite{BallardBDDDPSTY14,CarsonDGKKSS16}. To this end, the lower bound in~\eqref{eq:totot} can be modified to include multiplicative scaling for each component.

\subparagraph{Dropping the slow memory requirement for output values:}By using a modified concept of  $\visrule$-entering boundary of a segment, defined as $\hat{B}^{(ent)}_\visrule\left(\visit(i_{j-1}..i_j]\right)=B_{\visrule}\left(\visit[1..i_{j-1}]\right)\cap \visit(i_{j-1}..i_j]$, our method yields I/O lower bounds in a modified version of the \io{} models where output values are not required to be written into the slow memory.

\subparagraph{Free-input model:} The visit partition technique can also be adapted to the ``\emph{free input}'' model, more akin to the pebbling model, in which input values can be generated into cache at any time (e.g., they are read from a dedicated ROM memory), rather than being initially stored in the slow memory. While our lower bound to the number of write \io{} operations~\eqref{eq:totwrite} remains unchanged,  the lower bound to the number of read \io{} operations~\eqref{eq:totread} must be revised removing the contribution of the post-dominator bound and the read \io{} term of the boundary-bound. For any $\visrule\in\visset{G_R}$ we have:
\begin{equation*}
IO^{fi}_\mathcal{R}\left(G,M\right)  \geq \min_{\psi\in\visitset{\visrule}{G_R}}
    {\mathcal{W}}_\visrule(\psi,M)-|O|,
\end{equation*}
and, thus, 
% $IO^{fi}\left(G,M\right) \geq \min_{\visit{}\in\visitset{\visrule}{G_R}}
    % 2\mathcal{W}_\visrule(\psi,M) -|O|.$
\begin{equation*}
IO^{fi}\left(G,M\right) \geq \min_{\visit{}\in\visitset{\visrule}{G_R}}
    2\mathcal{W}_\visrule(\psi,M) -|O|.
\end{equation*}
\subparagraph{Execution with no recomputation:} Finally, our result in Theorem~\ref{thm:finiolwb} can be adapted and simplified to yield \io{} lower bounds assuming that the value associated to any vertex $V\setminus I$ is computed exactly once (the no-recomputation assumption). This simplifying assumption is often of interest as it focuses the analysis on schedules with a minimum number of computational steps. Moreover, it may provide a stepping stone towards the analysis of the more general and challenging case where recomputation is allowed. Without recomputation, computational schedules correspond to the topological orderings of $G$. Thus, for any $\visrule\in\visset{G_R}$, the corresponding $\visrule$-visits of $G_R$ constructed according to the procedure discussed in Section~\ref{sec:iolwb}, are the  topological orderings of $G_R$. Therefore, we can restrict our attention  to such orderings obtaining the following corollary:
\begin{corollary}
Given a DAG $\gra$, with input set $I$ and output set $O$, consider its computations in the I/O model using a cache of size $M$ such that no value is ever computed more than once. For any $\visrule\in\visset{G_R}$ we have:
\begin{eqnarray*}
    IO^{nr}_\mathcal{W}\left(G,M\right) & \geq &
    \min_{\visit{}\in\visitset{\toporule}{G_R}}
    {\mathcal{W}}_\visrule(\psi,M),\\
    IO^{nr}_\mathcal{R}\left(G,M\right) & \geq &
    \min_{\visit{}\in\visitset{\toporule}{G_R}}
    \max\{{\mathcal{R}}_\visrule(\psi,M),~{\mathcal{W}}_\visrule(\psi,M)+|I|-|O|\},\\
    IO^{nr}\left(G,M\right)& \geq &\min_{\visit{}\in\visitset{\toporule}{G_R}}
    \{{\mathcal{W}}_\visrule(\psi,M) + 
    \max\{{\mathcal{R}}_\visrule(\psi,M),{\mathcal{W}}_\visrule(\psi,M)+|I|-|O|\} \}.
\end{eqnarray*}
\end{corollary}
The bounds obtained for computations without recomputation are generally higher than the general ones in Theorem~\ref{thm:finiolwb}, as while for each rule $\visrule$ we still analyze the entering boundary and the minimum post-dominator size of visit partitions, the set of visits to be considered is restricted to the subset $\visitset{\toporule}{G_R}\subseteq\visitset{\visrule}{G_R}$, thus possibly eliminating some visit $\psi$ with low $\mathcal{W}_\visrule(\psi',M)$ and/or $\mathcal{R}_\visrule(\psi',M)$. It is easy to see that the best lower bounds are obtained for the choice $\visrule=\sinrule{}$.

\section{Conclusions}\label{sec:conclusion}
We have proposed the visit framework to investigate both space and I/O
complexity lower bounds. The universal upper bounds we have obtained
for both the singleton and the topological types of visits show that
these types cannot yield tight pebbling lower bounds for all DAGs,
although they do for some DAGs.  The framework gives ample flexibility
to tailor the type of visit to the given DAG, but we do not yet have
good insights either on how to exploit this flexibility or on how to
show that this flexibility is not ultimately helpful.

The spectrum of visit types exhibits the following tradeoff.  As we
go from the topological rule to, say, the singleton rule, by relaxing
the enablement constraints, the set of vertex sequences that qualifies
as a visit increases (which goes in the direction of reducing the
boundary complexity of the DAG, since the minimization takes place
over a larger domain), but the boundary complexity of a specific visit
also increases (which goes in the direction of increasing the boundary
complexity of the DAG, since the function to be minimized increases).
The tension between these opposite forces has proven difficult to
analyze quantitatively. The arguments used to establish universal
upper bounds in the singleton and in the topological cases are
significantly different, and it is not clear how to interpolate them
for an intermediate visit type.  Further research is clearly needed
to make progress on what appears to be a rich combinatorial problem.

Another contribution of the visit framework is a step toward a unified
treatment of pebbling and I/O complexity.  Within the framework, we
have already seen how to generalize the by now classical Hong-Kung
partition technique, based on dominator and minimum sets, thus
achieving much better lower bounds for some DAGs.

We conjecture that, for other significant DAGs whose \io{} complexity
cannot be well captured by the $S$-partition technique, visit
partitions will help obtain good lower bounds. Good candidates are
DAGs with a constant degree and a space complexity $S(N)$ superlinear,
say polynomial, in the number of inputs $N=|I|$. For such DAGs, the
size of the minimum dominator set cannot exceed $N$ and, as shown in
Theorem~\ref{lem:comparison1}, the size of the topological boundary,
\myie{} of the minimum set, is $O(\log N)$ at any point in the
computation.  On the other hand, the singleton boundary could be
significantly higher.

Although we have not explicitly discussed the issue in this work, we
do not have general \io{} upper bounds matching our visit partition lower
bounds. Therefore, further work is needed to achieve a full
characterization of the I/O complexity of a DAG.

\bibliography{bibliography.bib}
\end{document}

%% file: examplespace.tex
A reformulation of these arguments within the visit framework can be found in~\cite{thesis}  for stacks of superconcentrators and in the following section for $q$-pyramid and $q$-tree DAGs. 

\subsection{Examples of application of the visit method for bounding the pebbling number of DAGs}\label{sec:examplespace}

In this section, we show applications of the lower bound in Theorem~\ref{thm:visitlwb}. While these results are not novel, they are meant to showcase the potential benefit of the visits as a unifying method to analyze the pebbling number of DAGs

\subsubsection{$q$-pyramid DAGs}\label{sec:rpyramid}
\begin{figure}
\centering
\begin{minipage}{.48\textwidth}
  \centering
  \input{pyramid.tex}
  \captionof{figure}{A $2$-pyramid  with $b=8$ DAG $G\left(2,8\right)$.}
  \label{fig:2pyra}
\end{minipage}%
\hfill
\begin{minipage}{.48\textwidth}
  \centering
  \input{3pyramid.tex}
  \captionof{figure}{A $3$-pyramid with  $b=6$ DAG $G\left(3,6\right)$.}
  \label{fig:3pyra}
\end{minipage}
\end{figure}
We use the lower bound technique in Theorem~\ref{thm:visitlwb} to obtain an asymptotically tight lower bound to the pebbling number of a family of $q$-pyramids DAGs defined below. While this result has already been presented in the literature~\cite{bilardi2000space,ranjan2012upper,savage97models}, we present here an alternative derivation based on the visit method outlined in Theorem~\ref{thm:visitlwb}. Furthermore, the result on $q$-pyramids yields a match for the upper bounds for general visits based on topological depth (Theorem~\ref{thm:depthbound}) and for the upper bound on the $\sinrule{}$-boundary complexity (Theorem~\ref{thm:gensin}).

\begin{definition}[$q$-pyramid DAGs]\label{def:pyramid}
An $q$-pyramid of height $b$ is a layered DAG $\gra{}$ such that:
\begin{itemize}
    \item Let 
$V_i=\{v_{i,1},v_{i,2},\ldots,v_{i,\left(r-1\right)\left(h-i\right)+1}\}$ denote the set of vertices at the $i$-th layer, for $i=1,2,\ldots,b$. $\{V_1,V_2,\ldots,V_b\}$ partition $V$.
\item For all $i=2,\ldots,b$ each vertex $v_{i,j}\in V_i$, for $j=1,2,\ldots,\left(r-1\right)\left(b-i\right)+1$, has as immediate predecessors the vertices $v_{i-1,j}, v_{i-1,j+1},\ldots,v_{i-1,j+r-1}\in V_{i-1}$.
\end{itemize}
\end{definition}

By Definition~\ref{def:pyramid}, the vertices (resp., the vertex) on the first (resp., last) layer have no direct predecessors (resp., has no successors), and are henceforth referred to as the input vertices (resp., output vertex) of the pyramid DAG. An example of a $2$-pyramid (resp., $3$-pyramid) is presented in Figure~\ref{fig:2pyra} (resp., Figure~\ref{fig:3pyra}). Further, $|V|=1+\left(b-1\right)\left(p+b/2\right)$ and $\dein=\deout=p$ if $b\geq 2$.

The following lemma captures an important structural property of $q$-pyramids:

\begin{lemma}\label{lem:piramid}
\sloppy Given an $q$-pyramid DAG with $b$ levels $\gra{}$. Let $\pi = \{v_2,v_3,\ldots,v_b\}$ denote a set of vertices which compose a path directed from one of the input vertices of $G$ (excluded) to the output vertex $v_b$ of $G$. There exist $\left(b-1\right)\left(r-1\right)+1$ vertex disjoint paths from input vertices of $G$ to vertices in $\pi$ which share vertices only in $\pi$ such that for each vertex $v\in\pi$, at least $\left(r-1\right)$ such paths include $v$.
\end{lemma}
\begin{proof}
The proof is by induction on the number of levels of $G$. In the base case $b=1$. In this case, $\pi=\emptyset$ and the statement trivially holds.

We assume inductively that the statement holds for $b\geq 1$, and we proceed to show that it holds for a $q$-pyramid with $b+1$ levels $G$. Let $\pi=\{v_2,v_2,\ldots, v_{b},v_{b+1}\}$. We define as $G'=\left(V',E'\right)$ the sub-DAG of $G$ which corresponds to the $q$-pyramid with $b$ levels whose set of vertices corresponds to $v_b$ and all its ancestors (\myie{} $V'=\{v_b\}\cup \anc{v_b}$), and whose set of edges includes all edges of $G'=\left(V',E'\right)$ connecting vertices in $V'$. By construction $G'$ is an $q$-pyramid with $b$ levels. By inductive hypothesis the statement of the lemma therefore holds for $\{v_2,v_3,\ldots, v_{b-1},v_b\}\subseteq \Pi$.

To complete the proof, we shall now show that there exist $q-1$ paths connecting input vertices of $G$ to $v_{b+1}$ which only share $v_{b+1}$ and do not share any vertex with the paths obtained using the inductive hypothesis.
By definition of $q$-pyramid $v_{b+1}$ has $r$ predecessors $\{u_1,u_2,\ldots,u_r\}$ among whom $v_b$. Without loss of generality, let us assume $v_b=u_i$ for $i\in\{1,2,\ldots,q\}$. For each $u_j$ with $i\neq j$ we construct a path $\pi_j$ from an input vertex of $G$ to $u_j$ (and, hence, $v_{b+1}$) as follows: if $j<i$ (resp., $j>1$) start by adding $v_{b+1}$, $u_j$ and the leftmost (resp., rightmost) predecessor of $u_j$ to $\pi_j$, we then proceed ``\emph{descending}'' the pyramid by adding to $\pi_j$ the leftmost (resp., rightmost) predecessor of the last vertex added to $\pi_j$ until we reach an input vertex of $G$. By Definition~\ref{def:pyramid}, any pair of vertices on the same level have different leftmost and rightmost predecessors. Hence all the $q-1$ paths previously described do not share any vertex but $v_{b+1}$. Further, by construction, none of these paths include vertices in $V'$. Hence, they are vertex disjoint with respect to the paths obtained using the inductive hypothesis on the sub-DAG $G'$.
\end{proof}

Lemma~\ref{lem:piramid} is a modified version of results previously presented in the literature~\cite{cook1974storage,ranjan2012upper}. The property synthesized in it allows to obtain the following result:

\begin{theorem}\label{thm:spacepyra}
Let $\gra{}$ be a $q$-pyramid with $b$ levels:
\begin{equation*}
    p\left(G\right) \geq \left(q-1\right)\left(b-1\right)
\end{equation*}
\end{theorem}
\begin{proof}
\begin{figure}
    \centering
    \input{proofpyramid.tex}
    \caption{Graphic representation for the proof of Theorem~\ref{thm:spacepyra}. $G_R$ is a reverse $2$-pyramid of height $8$. The path $\Pi$ is highlighted in red. The vertex $\psi[i^*]$ is highlighted in green. The $8$ vertex disjoint paths from vertices in $\Pi$ to the output vertices of $G_R$ whose existence is formalized by Lemma~\ref{lem:piramid} are highlighted in blue dashed rectangles.}
    \label{fig:proofpyra}
\end{figure}
For $b=1$ the entire DAG corresponds to a single vertex. By the rules of the pebble game, a single pebble is necessary and sufficient. In the following, we assume $b>1$.
Consider the reverse DAG of $G$, henceforth referred as $G_R$. $G_R$ is a \emph{reverse} $q$-pyramid of height $b$ which has the same set of vertices as $G$ and whose edges correspond to those of $G$ but with the orientation of the edges being reversed. $G_R$  input vertex (resp., output vertices) correspond to the output vertex (resp., input vertices) of $G$ (use Figure~\ref{fig:proofpyra} as a reference). In the following, we prove that 
$$
min_{\visit\in\visitset{\sinrule}{G_R}}b_{\sinrule}(\visit{})\geq \left(q-1\right)\left(b-1\right),
$$
from whence, by Theorem~\ref{thm:visitlwb}, the statement follows.

Consider the singleton visit rule for $G_R$ and let $\psi$ be a $\sinrule$-visit of $G_R$. Let $i^*$ denote the step of the visit $\psi$ during which the first output vertex of $G_R$ is visited. By definition of $\sinrule$, there exists a set $\Pi\in\infix{\psi}{1}{i^*-1}$ such that the vertices in $\Pi$ form a path directed from the input vertex of $G_R$ to $\psi[i^*]$ (\myie{} the output vertex of $G_R$ visited at the $i^*$-th step of $\psi$), with $\psi[i^*]$ excluded.  This follows from the properties of $\sinrule$: In order for vertex $\psi[i^*]$ to be visited at step $i^*$, at least one of its predecessors must have been previously visited during $\infix{\psi}{1}{i^*-1}$. In order for such a vertex to have been visited, one of its predecessors must have been visited previously. The same reasoning can be iteratively repeated until the input vertex of $G_j$, which, by the construction of $\visrule^*$, is enabled by the empty set. From this consideration, it follows that $i^*>1$.

Consider now the $\sinrule$-boundary of the visit at the $i^*-1$ step $B_{\sinrule}\left(\psi[1..i^*-1]\right)$. By Lemma~\ref{lem:piramid}, there are $\left(b-1\right)\left(r-1\right)$ vertex-disjoint paths connecting vertices of $\Pi\subseteq \infix{\psi}{1}{i^*-1}$ to the output vertices of $G_R$. This holds due to the definition of reverse DAG. By the definition of $\sinrule$, for each of these paths there must be at least one distinct vertex in $B_{\sinrule}\left(\psi[1..i^*-1]\right)$. The statement follows. 
\end{proof}

The proof technique used in Theorem~\ref{thm:spacepyra} is similar to the one presented for $2$-pyramids in~\cite{bilardi2000space}, with opportune modifications due to the differences of the technique based on visits. As shown in~\cite{ranjan2012upper}, the lower bound in Theorem~\ref{thm:spacepyra} is asymptotically tight.

Recall that a reverse $q$-pyramid with $b$ levels has $1 + (b-1)(q + b/2)$ vertices and $\deout{}=q$. By Theorem~\ref{thm:spacepyra} we have 
\begin{align*}
\min_{\psi\in\visitset{\sinrule{}}{G_R}} &\geq \left(q-1\right)\left(b-1\right)\\
&\geq \BOme{\sqrt{(q-1)n}}\\
&\geq \BOme{\sqrt{\deout{}n}}.
\end{align*}
Hence, we can conclude that the upper bound on the $b_\sinrule{}$ complexity of DAGs in Theorem~\ref{thm:gensin} is existentially tight.

\subsubsection{Complete $q$-trees}\label{sec:rtrees}
 Here we present  an  alternative  derivation of the known result on  the lower bound of the pebbling number for complete $q$-trees~\cite{paterson1970comparative}  based  on  the  visit  method  outlined  in  Theorem~\ref{thm:visitlwb} in order to both provide more intuition for the reader, and to provide further evidence of the generality of the method.

A \emph{Complete $q$-tree with $q^i$ leaves} DAG is a rooted \emph{in-tree} (or \emph{anti-arborescence}) $q$-ary tree with $q^i$ leaves for $i\geq 0$, $n=\sum_{j=0}^{i}q^j=\frac{q^{i+1}-1}{q-1}$ nodes. Thus, there are $p$ inputs (the leaves), and one single output (the root).

\begin{theorem}
\label{thm:spacetree}
Let $G$ be an $r$-tree DAG with $q^i$ leaves. We have:
\begin{equation*}
    p\left(G\right) \geq \left(q-1\right)\log_q q^i= \frac{q-1}{\log_2 q}\log_2 \frac{n(q-1)+1}{q}.
\end{equation*}
\end{theorem}

The proof follows steps analogous steps to that of the proof of Theorem~\ref{thm:spacepyra} analyzing the $\toporule{}$-boundary complexity of the reverse DAG which is an inverse complete $q$-tree or a complete $q$-arborescence. The reasoning used in the proof is  based on an observation analogous to the one in Lemma~\ref{lem:piramid}: for $q$-trees (and, hence, $q$-arborescence) it is easy to  show that for any path $\pi$ connecting an input leave vertex (excluded) to the root output vertex, there exist $\left(q-1\right)\log_q p+1$ paths connecting the $p$ leaves to vertices in $\pi$ which do not share any vertex not in $\Pi$. By extending the argument originally presented for binary trees in~\cite{paterson1970comparative}, we can conclude that the bound in Theorem~\ref{thm:spacetree} is tight.
As $\frac{q-1}{\log_2 q}\log_2 \frac{n(q-1)+1}{q}=O(\frac{q-1}{\log q}\log_2 n)$, we have that the $\toporule{}$-boundary complexity of the for complete $q$-arborescences matches the upper bound given in Theorem~\ref{thm:upptopo}, which is, thus, existentially tight.

%% file: pyramid.tex
\resizebox{\textwidth}{!}{

\tikzset{every picture/.style={line width=0.75pt}} %set default line width to 0.75pt        

\begin{tikzpicture}[x=0.75pt,y=0.75pt,yscale=-1,xscale=1]

%Shape: Circle [id:dp3238827616392994] 
\draw   (236.45,204.47) .. controls (234.22,202.21) and (234.24,198.57) .. (236.5,196.33) .. controls (238.76,194.1) and (242.4,194.13) .. (244.63,196.39) .. controls (246.86,198.65) and (246.84,202.29) .. (244.58,204.52) .. controls (242.32,206.75) and (238.68,206.73) .. (236.45,204.47) -- cycle ;
%Shape: Circle [id:dp3691777982612152] 
\draw   (279.59,204.03) .. controls (277.36,201.77) and (277.38,198.13) .. (279.64,195.9) .. controls (281.9,193.67) and (285.54,193.69) .. (287.77,195.95) .. controls (290,198.21) and (289.98,201.85) .. (287.72,204.08) .. controls (285.46,206.31) and (281.82,206.29) .. (279.59,204.03) -- cycle ;
%Shape: Circle [id:dp8412747088979771] 
\draw   (258.51,182.68) .. controls (256.28,180.42) and (256.3,176.78) .. (258.56,174.55) .. controls (260.82,172.32) and (264.46,172.34) .. (266.69,174.6) .. controls (268.92,176.86) and (268.9,180.5) .. (266.64,182.74) .. controls (264.38,184.97) and (260.74,184.94) .. (258.51,182.68) -- cycle ;
%Straight Lines [id:da6330564641332417] 
\draw    (278.94,195.19) -- (268.05,184.16) ;
\draw [shift={(266.64,182.74)}, rotate = 405.36] [fill={rgb, 255:red, 0; green, 0; blue, 0 }  ][line width=0.75]  [draw opacity=0] (8.93,-4.29) -- (0,0) -- (8.93,4.29) -- cycle    ;

%Shape: Circle [id:dp7575659591016339] 
\draw   (320.61,202.87) .. controls (318.38,200.61) and (318.4,196.97) .. (320.66,194.74) .. controls (322.92,192.51) and (326.56,192.53) .. (328.79,194.79) .. controls (331.02,197.05) and (331,200.69) .. (328.74,202.92) .. controls (326.48,205.16) and (322.84,205.13) .. (320.61,202.87) -- cycle ;
%Shape: Circle [id:dp06953946816545997] 
\draw   (300.93,182.95) .. controls (298.7,180.69) and (298.73,177.05) .. (300.99,174.82) .. controls (303.24,172.59) and (306.89,172.61) .. (309.12,174.87) .. controls (311.35,177.13) and (311.33,180.77) .. (309.07,183) .. controls (306.81,185.23) and (303.17,185.21) .. (300.93,182.95) -- cycle ;
%Shape: Circle [id:dp8758235950452324] 
\draw   (279.85,161.6) .. controls (277.62,159.34) and (277.65,155.7) .. (279.91,153.47) .. controls (282.17,151.24) and (285.81,151.26) .. (288.04,153.52) .. controls (290.27,155.78) and (290.25,159.42) .. (287.99,161.66) .. controls (285.73,163.89) and (282.09,163.86) .. (279.85,161.6) -- cycle ;
%Straight Lines [id:da4285865557886761] 
\draw    (320.66,194.74) -- (310.47,184.42) ;
\draw [shift={(309.07,183)}, rotate = 405.36] [fill={rgb, 255:red, 0; green, 0; blue, 0 }  ][line width=0.75]  [draw opacity=0] (8.93,-4.29) -- (0,0) -- (8.93,4.29) -- cycle    ;

%Straight Lines [id:da8144095381708423] 
\draw    (300.28,174.11) -- (289.39,163.08) ;
\draw [shift={(287.99,161.66)}, rotate = 405.36] [fill={rgb, 255:red, 0; green, 0; blue, 0 }  ][line width=0.75]  [draw opacity=0] (8.93,-4.29) -- (0,0) -- (8.93,4.29) -- cycle    ;

%Shape: Circle [id:dp8374702509352714] 
\draw   (364.44,204.56) .. controls (362.21,202.3) and (362.23,198.66) .. (364.49,196.43) .. controls (366.75,194.2) and (370.39,194.22) .. (372.62,196.48) .. controls (374.85,198.74) and (374.83,202.38) .. (372.57,204.61) .. controls (370.31,206.85) and (366.67,206.82) .. (364.44,204.56) -- cycle ;
%Shape: Circle [id:dp8352787017690564] 
\draw   (341.95,181.79) .. controls (339.72,179.53) and (339.75,175.89) .. (342.01,173.66) .. controls (344.27,171.43) and (347.91,171.45) .. (350.14,173.71) .. controls (352.37,175.97) and (352.35,179.61) .. (350.09,181.85) .. controls (347.83,184.08) and (344.19,184.05) .. (341.95,181.79) -- cycle ;
%Straight Lines [id:da35934407908701615] 
\draw    (364.49,196.43) -- (351.49,183.27) ;
\draw [shift={(350.09,181.85)}, rotate = 405.36] [fill={rgb, 255:red, 0; green, 0; blue, 0 }  ][line width=0.75]  [draw opacity=0] (8.93,-4.29) -- (0,0) -- (8.93,4.29) -- cycle    ;

%Shape: Circle [id:dp8108669352961486] 
\draw   (322.28,161.87) .. controls (320.05,159.61) and (320.07,155.97) .. (322.33,153.74) .. controls (324.59,151.51) and (328.23,151.53) .. (330.46,153.79) .. controls (332.69,156.05) and (332.67,159.69) .. (330.41,161.92) .. controls (328.15,164.15) and (324.51,164.13) .. (322.28,161.87) -- cycle ;
%Shape: Circle [id:dp3679504144687731] 
\draw   (301.2,140.52) .. controls (298.97,138.27) and (298.99,134.62) .. (301.25,132.39) .. controls (303.51,130.16) and (307.15,130.18) .. (309.38,132.44) .. controls (311.61,134.7) and (311.59,138.34) .. (309.33,140.58) .. controls (307.07,142.81) and (303.43,142.78) .. (301.2,140.52) -- cycle ;
%Straight Lines [id:da6678183355023364] 
\draw    (342.01,173.66) -- (331.82,163.35) ;
\draw [shift={(330.41,161.92)}, rotate = 405.36] [fill={rgb, 255:red, 0; green, 0; blue, 0 }  ][line width=0.75]  [draw opacity=0] (8.93,-4.29) -- (0,0) -- (8.93,4.29) -- cycle    ;

%Straight Lines [id:da8083849224311865] 
\draw    (321.63,153.03) -- (310.74,142) ;
\draw [shift={(309.33,140.58)}, rotate = 405.36] [fill={rgb, 255:red, 0; green, 0; blue, 0 }  ][line width=0.75]  [draw opacity=0] (8.93,-4.29) -- (0,0) -- (8.93,4.29) -- cycle    ;

%Straight Lines [id:da469520169358709] 
\draw    (244.63,196.39) -- (257.09,184.09) ;
\draw [shift={(258.51,182.68)}, rotate = 495.36] [fill={rgb, 255:red, 0; green, 0; blue, 0 }  ][line width=0.75]  [draw opacity=0] (8.93,-4.29) -- (0,0) -- (8.93,4.29) -- cycle    ;

%Straight Lines [id:da1612087273024969] 
\draw    (266.69,174.6) -- (278.43,163.01) ;
\draw [shift={(279.85,161.6)}, rotate = 495.36] [fill={rgb, 255:red, 0; green, 0; blue, 0 }  ][line width=0.75]  [draw opacity=0] (8.93,-4.29) -- (0,0) -- (8.93,4.29) -- cycle    ;

%Straight Lines [id:da5297038928365991] 
\draw    (288.04,153.52) -- (299.78,141.93) ;
\draw [shift={(301.2,140.52)}, rotate = 495.36] [fill={rgb, 255:red, 0; green, 0; blue, 0 }  ][line width=0.75]  [draw opacity=0] (8.93,-4.29) -- (0,0) -- (8.93,4.29) -- cycle    ;

%Straight Lines [id:da5914424292367955] 
\draw    (288.48,195.25) -- (300.22,183.65) ;
\draw [shift={(301.65,182.25)}, rotate = 495.36] [fill={rgb, 255:red, 0; green, 0; blue, 0 }  ][line width=0.75]  [draw opacity=0] (8.93,-4.29) -- (0,0) -- (8.93,4.29) -- cycle    ;

%Straight Lines [id:da5692285557815178] 
\draw    (309.83,174.17) -- (321.57,162.57) ;
\draw [shift={(322.99,161.17)}, rotate = 495.36] [fill={rgb, 255:red, 0; green, 0; blue, 0 }  ][line width=0.75]  [draw opacity=0] (8.93,-4.29) -- (0,0) -- (8.93,4.29) -- cycle    ;

%Straight Lines [id:da2765970977689345] 
\draw    (329.5,194.09) -- (341.24,182.5) ;
\draw [shift={(342.67,181.09)}, rotate = 495.36] [fill={rgb, 255:red, 0; green, 0; blue, 0 }  ][line width=0.75]  [draw opacity=0] (8.93,-4.29) -- (0,0) -- (8.93,4.29) -- cycle    ;

%Shape: Circle [id:dp3303952807828787] 
\draw   (153.01,205.36) .. controls (150.77,203.1) and (150.8,199.46) .. (153.06,197.22) .. controls (155.32,194.99) and (158.96,195.02) .. (161.19,197.28) .. controls (163.42,199.54) and (163.4,203.18) .. (161.14,205.41) .. controls (158.88,207.64) and (155.24,207.62) .. (153.01,205.36) -- cycle ;
%Shape: Circle [id:dp3723518649035322] 
\draw   (195.43,205.62) .. controls (193.2,203.36) and (193.22,199.72) .. (195.48,197.49) .. controls (197.74,195.26) and (201.38,195.28) .. (203.61,197.54) .. controls (205.84,199.8) and (205.82,203.44) .. (203.56,205.67) .. controls (201.3,207.91) and (197.66,207.88) .. (195.43,205.62) -- cycle ;
%Shape: Circle [id:dp45878471752653804] 
\draw   (174.35,184.28) .. controls (172.12,182.02) and (172.14,178.38) .. (174.4,176.14) .. controls (176.66,173.91) and (180.3,173.94) .. (182.53,176.2) .. controls (184.77,178.46) and (184.74,182.1) .. (182.48,184.33) .. controls (180.22,186.56) and (176.58,186.54) .. (174.35,184.28) -- cycle ;
%Straight Lines [id:da5719784849610465] 
\draw    (194.78,196.78) -- (183.89,185.75) ;
\draw [shift={(182.48,184.33)}, rotate = 405.36] [fill={rgb, 255:red, 0; green, 0; blue, 0 }  ][line width=0.75]  [draw opacity=0] (8.93,-4.29) -- (0,0) -- (8.93,4.29) -- cycle    ;

%Shape: Circle [id:dp33358912700110754] 
\draw   (216.78,184.54) .. controls (214.55,182.28) and (214.57,178.64) .. (216.83,176.41) .. controls (219.09,174.18) and (222.73,174.2) .. (224.96,176.46) .. controls (227.19,178.72) and (227.17,182.36) .. (224.91,184.59) .. controls (222.65,186.83) and (219.01,186.8) .. (216.78,184.54) -- cycle ;
%Shape: Circle [id:dp20653366475782908] 
\draw   (195.7,163.2) .. controls (193.47,160.94) and (193.49,157.3) .. (195.75,155.07) .. controls (198.01,152.83) and (201.65,152.86) .. (203.88,155.12) .. controls (206.11,157.38) and (206.09,161.02) .. (203.83,163.25) .. controls (201.57,165.48) and (197.93,165.46) .. (195.7,163.2) -- cycle ;
%Straight Lines [id:da6205410954472377] 
\draw    (236.5,196.33) -- (226.31,186.02) ;
\draw [shift={(224.91,184.59)}, rotate = 405.36] [fill={rgb, 255:red, 0; green, 0; blue, 0 }  ][line width=0.75]  [draw opacity=0] (8.93,-4.29) -- (0,0) -- (8.93,4.29) -- cycle    ;

%Straight Lines [id:da7948987937898975] 
\draw    (216.13,175.7) -- (205.23,164.67) ;
\draw [shift={(203.83,163.25)}, rotate = 405.36] [fill={rgb, 255:red, 0; green, 0; blue, 0 }  ][line width=0.75]  [draw opacity=0] (8.93,-4.29) -- (0,0) -- (8.93,4.29) -- cycle    ;

%Shape: Circle [id:dp01271299593777142] 
\draw   (238.83,162.76) .. controls (236.6,160.5) and (236.63,156.86) .. (238.89,154.63) .. controls (241.15,152.4) and (244.79,152.42) .. (247.02,154.68) .. controls (249.25,156.94) and (249.23,160.58) .. (246.97,162.81) .. controls (244.71,165.04) and (241.07,165.02) .. (238.83,162.76) -- cycle ;
%Shape: Circle [id:dp20127856532381472] 
\draw   (217.75,141.41) .. controls (215.52,139.16) and (215.55,135.51) .. (217.81,133.28) .. controls (220.07,131.05) and (223.71,131.07) .. (225.94,133.33) .. controls (228.17,135.59) and (228.15,139.23) .. (225.89,141.47) .. controls (223.63,143.7) and (219.99,143.67) .. (217.75,141.41) -- cycle ;
%Straight Lines [id:da7251498146882094] 
\draw    (258.56,174.55) -- (248.37,164.24) ;
\draw [shift={(246.97,162.81)}, rotate = 405.36] [fill={rgb, 255:red, 0; green, 0; blue, 0 }  ][line width=0.75]  [draw opacity=0] (8.93,-4.29) -- (0,0) -- (8.93,4.29) -- cycle    ;

%Straight Lines [id:da1281911930799462] 
\draw    (238.18,153.92) -- (227.29,142.89) ;
\draw [shift={(225.89,141.47)}, rotate = 405.36] [fill={rgb, 255:red, 0; green, 0; blue, 0 }  ][line width=0.75]  [draw opacity=0] (8.93,-4.29) -- (0,0) -- (8.93,4.29) -- cycle    ;

%Shape: Circle [id:dp25059960472082743] 
\draw   (260.18,141.68) .. controls (257.95,139.42) and (257.97,135.78) .. (260.23,133.55) .. controls (262.49,131.32) and (266.13,131.34) .. (268.36,133.6) .. controls (270.59,135.86) and (270.57,139.5) .. (268.31,141.73) .. controls (266.05,143.96) and (262.41,143.94) .. (260.18,141.68) -- cycle ;
%Shape: Circle [id:dp9571216804683162] 
\draw   (239.1,120.34) .. controls (236.87,118.08) and (236.89,114.44) .. (239.15,112.2) .. controls (241.41,109.97) and (245.05,110) .. (247.28,112.25) .. controls (249.51,114.51) and (249.49,118.16) .. (247.23,120.39) .. controls (244.97,122.62) and (241.33,122.59) .. (239.1,120.34) -- cycle ;
%Straight Lines [id:da6363786569769025] 
\draw    (279.91,153.47) -- (269.72,143.16) ;
\draw [shift={(268.31,141.73)}, rotate = 405.36] [fill={rgb, 255:red, 0; green, 0; blue, 0 }  ][line width=0.75]  [draw opacity=0] (8.93,-4.29) -- (0,0) -- (8.93,4.29) -- cycle    ;

%Straight Lines [id:da846219010326875] 
\draw    (259.53,132.84) -- (248.64,121.81) ;
\draw [shift={(247.23,120.39)}, rotate = 405.36] [fill={rgb, 255:red, 0; green, 0; blue, 0 }  ][line width=0.75]  [draw opacity=0] (8.93,-4.29) -- (0,0) -- (8.93,4.29) -- cycle    ;

%Shape: Circle [id:dp08430867635801298] 
\draw   (281.53,120.6) .. controls (279.3,118.34) and (279.32,114.7) .. (281.58,112.47) .. controls (283.84,110.24) and (287.48,110.26) .. (289.71,112.52) .. controls (291.94,114.78) and (291.92,118.42) .. (289.66,120.65) .. controls (287.4,122.88) and (283.76,122.86) .. (281.53,120.6) -- cycle ;
%Shape: Circle [id:dp20966612947862484] 
\draw   (260.45,99.26) .. controls (258.22,97) and (258.24,93.36) .. (260.5,91.12) .. controls (262.76,88.89) and (266.4,88.92) .. (268.63,91.18) .. controls (270.86,93.43) and (270.84,97.08) .. (268.58,99.31) .. controls (266.32,101.54) and (262.68,101.52) .. (260.45,99.26) -- cycle ;
%Straight Lines [id:da666460459420436] 
\draw    (301.25,132.39) -- (291.06,122.08) ;
\draw [shift={(289.66,120.65)}, rotate = 405.36] [fill={rgb, 255:red, 0; green, 0; blue, 0 }  ][line width=0.75]  [draw opacity=0] (8.93,-4.29) -- (0,0) -- (8.93,4.29) -- cycle    ;

%Straight Lines [id:da7511545864298335] 
\draw    (280.87,111.76) -- (269.98,100.73) ;
\draw [shift={(268.58,99.31)}, rotate = 405.36] [fill={rgb, 255:red, 0; green, 0; blue, 0 }  ][line width=0.75]  [draw opacity=0] (8.93,-4.29) -- (0,0) -- (8.93,4.29) -- cycle    ;

%Straight Lines [id:da908535770457783] 
\draw    (161.19,197.28) -- (172.93,185.68) ;
\draw [shift={(174.35,184.28)}, rotate = 495.36] [fill={rgb, 255:red, 0; green, 0; blue, 0 }  ][line width=0.75]  [draw opacity=0] (8.93,-4.29) -- (0,0) -- (8.93,4.29) -- cycle    ;

%Straight Lines [id:da8188558217468027] 
\draw    (182.53,176.2) -- (194.27,164.6) ;
\draw [shift={(195.7,163.2)}, rotate = 495.36] [fill={rgb, 255:red, 0; green, 0; blue, 0 }  ][line width=0.75]  [draw opacity=0] (8.93,-4.29) -- (0,0) -- (8.93,4.29) -- cycle    ;

%Straight Lines [id:da3141305937306551] 
\draw    (203.88,155.12) -- (216.33,142.82) ;
\draw [shift={(217.75,141.41)}, rotate = 495.36] [fill={rgb, 255:red, 0; green, 0; blue, 0 }  ][line width=0.75]  [draw opacity=0] (8.93,-4.29) -- (0,0) -- (8.93,4.29) -- cycle    ;

%Straight Lines [id:da18802901446315246] 
\draw    (225.94,133.33) -- (237.68,121.74) ;
\draw [shift={(239.1,120.34)}, rotate = 495.36] [fill={rgb, 255:red, 0; green, 0; blue, 0 }  ][line width=0.75]  [draw opacity=0] (8.93,-4.29) -- (0,0) -- (8.93,4.29) -- cycle    ;

%Straight Lines [id:da7611283820897998] 
\draw    (247.28,112.25) -- (259.02,100.66) ;
\draw [shift={(260.45,99.26)}, rotate = 495.36] [fill={rgb, 255:red, 0; green, 0; blue, 0 }  ][line width=0.75]  [draw opacity=0] (8.93,-4.29) -- (0,0) -- (8.93,4.29) -- cycle    ;

%Straight Lines [id:da07730291621714369] 
\draw    (204.32,196.84) -- (216.07,185.25) ;
\draw [shift={(217.49,183.84)}, rotate = 495.36] [fill={rgb, 255:red, 0; green, 0; blue, 0 }  ][line width=0.75]  [draw opacity=0] (8.93,-4.29) -- (0,0) -- (8.93,4.29) -- cycle    ;

%Straight Lines [id:da28443779697727933] 
\draw    (225.67,175.76) -- (238.12,163.46) ;
\draw [shift={(239.55,162.06)}, rotate = 495.36] [fill={rgb, 255:red, 0; green, 0; blue, 0 }  ][line width=0.75]  [draw opacity=0] (8.93,-4.29) -- (0,0) -- (8.93,4.29) -- cycle    ;

%Straight Lines [id:da3682391294152598] 
\draw    (247.73,153.98) -- (259.47,142.38) ;
\draw [shift={(260.89,140.98)}, rotate = 495.36] [fill={rgb, 255:red, 0; green, 0; blue, 0 }  ][line width=0.75]  [draw opacity=0] (8.93,-4.29) -- (0,0) -- (8.93,4.29) -- cycle    ;

%Straight Lines [id:da7589849032681268] 
\draw    (269.07,132.9) -- (280.81,121.3) ;
\draw [shift={(282.24,119.9)}, rotate = 495.36] [fill={rgb, 255:red, 0; green, 0; blue, 0 }  ][line width=0.75]  [draw opacity=0] (8.93,-4.29) -- (0,0) -- (8.93,4.29) -- cycle    ;

%Shape: Circle [id:dp9416980390202705] 
\draw   (406.86,204.83) .. controls (404.63,202.57) and (404.66,198.93) .. (406.92,196.7) .. controls (409.18,194.47) and (412.82,194.49) .. (415.05,196.75) .. controls (417.28,199.01) and (417.26,202.65) .. (415,204.88) .. controls (412.74,207.11) and (409.1,207.09) .. (406.86,204.83) -- cycle ;
%Shape: Circle [id:dp3709634529660728] 
\draw   (385.79,183.48) .. controls (383.55,181.22) and (383.58,177.58) .. (385.84,175.35) .. controls (388.1,173.12) and (391.74,173.14) .. (393.97,175.4) .. controls (396.2,177.66) and (396.18,181.3) .. (393.92,183.53) .. controls (391.66,185.77) and (388.02,185.74) .. (385.79,183.48) -- cycle ;
%Shape: Circle [id:dp6033752214042429] 
\draw   (363.3,160.71) .. controls (361.07,158.45) and (361.09,154.81) .. (363.35,152.58) .. controls (365.61,150.35) and (369.25,150.37) .. (371.48,152.63) .. controls (373.71,154.89) and (373.69,158.53) .. (371.43,160.77) .. controls (369.17,163) and (365.53,162.97) .. (363.3,160.71) -- cycle ;
%Straight Lines [id:da08301315303807777] 
\draw    (406.21,195.99) -- (395.32,184.96) ;
\draw [shift={(393.92,183.53)}, rotate = 405.36] [fill={rgb, 255:red, 0; green, 0; blue, 0 }  ][line width=0.75]  [draw opacity=0] (8.93,-4.29) -- (0,0) -- (8.93,4.29) -- cycle    ;

%Straight Lines [id:da896212681818765] 
\draw    (385.84,175.35) -- (372.84,162.19) ;
\draw [shift={(371.43,160.77)}, rotate = 405.36] [fill={rgb, 255:red, 0; green, 0; blue, 0 }  ][line width=0.75]  [draw opacity=0] (8.93,-4.29) -- (0,0) -- (8.93,4.29) -- cycle    ;

%Shape: Circle [id:dp7694174297332266] 
\draw   (343.63,140.79) .. controls (341.39,138.53) and (341.42,134.89) .. (343.68,132.66) .. controls (345.94,130.43) and (349.58,130.45) .. (351.81,132.71) .. controls (354.04,134.97) and (354.02,138.61) .. (351.76,140.84) .. controls (349.5,143.07) and (345.86,143.05) .. (343.63,140.79) -- cycle ;
%Shape: Circle [id:dp31328099297665646] 
\draw   (322.55,119.45) .. controls (320.32,117.19) and (320.34,113.55) .. (322.6,111.31) .. controls (324.86,109.08) and (328.5,109.11) .. (330.73,111.36) .. controls (332.96,113.62) and (332.94,117.27) .. (330.68,119.5) .. controls (328.42,121.73) and (324.78,121.7) .. (322.55,119.45) -- cycle ;
%Straight Lines [id:da7023760262036709] 
\draw    (363.35,152.58) -- (353.16,142.27) ;
\draw [shift={(351.76,140.84)}, rotate = 405.36] [fill={rgb, 255:red, 0; green, 0; blue, 0 }  ][line width=0.75]  [draw opacity=0] (8.93,-4.29) -- (0,0) -- (8.93,4.29) -- cycle    ;

%Straight Lines [id:da6570276374973496] 
\draw    (342.97,131.95) -- (332.08,120.92) ;
\draw [shift={(330.68,119.5)}, rotate = 405.36] [fill={rgb, 255:red, 0; green, 0; blue, 0 }  ][line width=0.75]  [draw opacity=0] (8.93,-4.29) -- (0,0) -- (8.93,4.29) -- cycle    ;

%Shape: Circle [id:dp04088905962801137] 
\draw   (449.99,205.81) .. controls (447.76,203.55) and (447.78,199.91) .. (450.04,197.68) .. controls (452.3,195.44) and (455.94,195.47) .. (458.18,197.73) .. controls (460.41,199.99) and (460.38,203.63) .. (458.12,205.86) .. controls (455.86,208.09) and (452.22,208.07) .. (449.99,205.81) -- cycle ;
%Shape: Circle [id:dp8820617213250592] 
\draw   (428.21,183.75) .. controls (425.98,181.49) and (426,177.85) .. (428.26,175.62) .. controls (430.52,173.39) and (434.16,173.41) .. (436.39,175.67) .. controls (438.62,177.93) and (438.6,181.57) .. (436.34,183.8) .. controls (434.08,186.03) and (430.44,186.01) .. (428.21,183.75) -- cycle ;
%Shape: Circle [id:dp7036878315810766] 
\draw   (407.13,162.4) .. controls (404.9,160.14) and (404.92,156.5) .. (407.18,154.27) .. controls (409.44,152.04) and (413.08,152.06) .. (415.31,154.32) .. controls (417.55,156.58) and (417.52,160.22) .. (415.26,162.46) .. controls (413,164.69) and (409.36,164.66) .. (407.13,162.4) -- cycle ;
%Shape: Circle [id:dp5186260720632401] 
\draw   (384.65,139.63) .. controls (382.42,137.38) and (382.44,133.73) .. (384.7,131.5) .. controls (386.96,129.27) and (390.6,129.29) .. (392.83,131.55) .. controls (395.06,133.81) and (395.04,137.45) .. (392.78,139.69) .. controls (390.52,141.92) and (386.88,141.89) .. (384.65,139.63) -- cycle ;
%Straight Lines [id:da7642435439141513] 
\draw    (450.04,197.68) -- (437.75,185.22) ;
\draw [shift={(436.34,183.8)}, rotate = 405.36] [fill={rgb, 255:red, 0; green, 0; blue, 0 }  ][line width=0.75]  [draw opacity=0] (8.93,-4.29) -- (0,0) -- (8.93,4.29) -- cycle    ;

%Straight Lines [id:da7962224253944465] 
\draw    (427.56,174.91) -- (416.67,163.88) ;
\draw [shift={(415.26,162.46)}, rotate = 405.36] [fill={rgb, 255:red, 0; green, 0; blue, 0 }  ][line width=0.75]  [draw opacity=0] (8.93,-4.29) -- (0,0) -- (8.93,4.29) -- cycle    ;

%Straight Lines [id:da6361465056729099] 
\draw    (407.18,154.27) -- (394.18,141.11) ;
\draw [shift={(392.78,139.69)}, rotate = 405.36] [fill={rgb, 255:red, 0; green, 0; blue, 0 }  ][line width=0.75]  [draw opacity=0] (8.93,-4.29) -- (0,0) -- (8.93,4.29) -- cycle    ;

%Shape: Circle [id:dp31359532162736103] 
\draw   (364.97,119.71) .. controls (362.74,117.45) and (362.76,113.81) .. (365.02,111.58) .. controls (367.28,109.35) and (370.92,109.37) .. (373.15,111.63) .. controls (375.39,113.89) and (375.36,117.53) .. (373.1,119.76) .. controls (370.84,121.99) and (367.2,121.97) .. (364.97,119.71) -- cycle ;
%Shape: Circle [id:dp028643126499286797] 
\draw   (343.89,98.37) .. controls (341.66,96.11) and (341.68,92.47) .. (343.94,90.23) .. controls (346.2,88) and (349.84,88.03) .. (352.08,90.29) .. controls (354.31,92.55) and (354.28,96.19) .. (352.02,98.42) .. controls (349.76,100.65) and (346.12,100.63) .. (343.89,98.37) -- cycle ;
%Straight Lines [id:da8361237650368669] 
\draw    (384.7,131.5) -- (374.51,121.19) ;
\draw [shift={(373.1,119.76)}, rotate = 405.36] [fill={rgb, 255:red, 0; green, 0; blue, 0 }  ][line width=0.75]  [draw opacity=0] (8.93,-4.29) -- (0,0) -- (8.93,4.29) -- cycle    ;

%Straight Lines [id:da913684394556654] 
\draw    (364.32,110.87) -- (353.43,99.84) ;
\draw [shift={(352.02,98.42)}, rotate = 405.36] [fill={rgb, 255:red, 0; green, 0; blue, 0 }  ][line width=0.75]  [draw opacity=0] (8.93,-4.29) -- (0,0) -- (8.93,4.29) -- cycle    ;

%Straight Lines [id:da9752298324669595] 
\draw    (309.38,132.44) -- (321.12,120.85) ;
\draw [shift={(322.55,119.45)}, rotate = 495.36] [fill={rgb, 255:red, 0; green, 0; blue, 0 }  ][line width=0.75]  [draw opacity=0] (8.93,-4.29) -- (0,0) -- (8.93,4.29) -- cycle    ;

%Straight Lines [id:da41357499067013537] 
\draw    (330.73,111.36) -- (342.47,99.77) ;
\draw [shift={(343.89,98.37)}, rotate = 495.36] [fill={rgb, 255:red, 0; green, 0; blue, 0 }  ][line width=0.75]  [draw opacity=0] (8.93,-4.29) -- (0,0) -- (8.93,4.29) -- cycle    ;

%Straight Lines [id:da839169654346277] 
\draw    (331.17,153.09) -- (342.91,141.49) ;
\draw [shift={(344.34,140.09)}, rotate = 495.36] [fill={rgb, 255:red, 0; green, 0; blue, 0 }  ][line width=0.75]  [draw opacity=0] (8.93,-4.29) -- (0,0) -- (8.93,4.29) -- cycle    ;

%Straight Lines [id:da3152793569127257] 
\draw    (352.52,132.01) -- (364.26,120.41) ;
\draw [shift={(365.68,119.01)}, rotate = 495.36] [fill={rgb, 255:red, 0; green, 0; blue, 0 }  ][line width=0.75]  [draw opacity=0] (8.93,-4.29) -- (0,0) -- (8.93,4.29) -- cycle    ;

%Straight Lines [id:da5951494119326688] 
\draw    (350.85,173.01) -- (362.59,161.42) ;
\draw [shift={(364.01,160.01)}, rotate = 495.36] [fill={rgb, 255:red, 0; green, 0; blue, 0 }  ][line width=0.75]  [draw opacity=0] (8.93,-4.29) -- (0,0) -- (8.93,4.29) -- cycle    ;

%Straight Lines [id:da1385543251339838] 
\draw    (372.19,151.93) -- (383.94,140.34) ;
\draw [shift={(385.36,138.93)}, rotate = 495.36] [fill={rgb, 255:red, 0; green, 0; blue, 0 }  ][line width=0.75]  [draw opacity=0] (8.93,-4.29) -- (0,0) -- (8.93,4.29) -- cycle    ;

%Straight Lines [id:da6798157754935625] 
\draw    (393.99,172.57) -- (405.73,160.98) ;
\draw [shift={(407.15,159.58)}, rotate = 495.36] [fill={rgb, 255:red, 0; green, 0; blue, 0 }  ][line width=0.75]  [draw opacity=0] (8.93,-4.29) -- (0,0) -- (8.93,4.29) -- cycle    ;

%Straight Lines [id:da8555006386103732] 
\draw    (415.05,196.75) -- (426.79,185.16) ;
\draw [shift={(428.21,183.75)}, rotate = 495.36] [fill={rgb, 255:red, 0; green, 0; blue, 0 }  ][line width=0.75]  [draw opacity=0] (8.93,-4.29) -- (0,0) -- (8.93,4.29) -- cycle    ;

%Straight Lines [id:da7501764290209991] 
\draw    (372.62,196.48) -- (384.36,184.89) ;
\draw [shift={(385.79,183.48)}, rotate = 495.36] [fill={rgb, 255:red, 0; green, 0; blue, 0 }  ][line width=0.75]  [draw opacity=0] (8.93,-4.29) -- (0,0) -- (8.93,4.29) -- cycle    ;

%Shape: Circle [id:dp5644577243042967] 
\draw   (302.87,99.52) .. controls (300.64,97.26) and (300.66,93.62) .. (302.92,91.39) .. controls (305.18,89.16) and (308.82,89.18) .. (311.06,91.44) .. controls (313.29,93.7) and (313.26,97.34) .. (311,99.57) .. controls (308.74,101.8) and (305.1,101.78) .. (302.87,99.52) -- cycle ;
%Shape: Circle [id:dp7652309452362334] 
\draw   (281.79,78.18) .. controls (279.56,75.92) and (279.58,72.28) .. (281.84,70.04) .. controls (284.1,67.81) and (287.74,67.84) .. (289.98,70.1) .. controls (292.21,72.36) and (292.18,76) .. (289.92,78.23) .. controls (287.67,80.46) and (284.02,80.44) .. (281.79,78.18) -- cycle ;
%Straight Lines [id:da3889497746504138] 
\draw    (322.6,111.31) -- (312.41,101) ;
\draw [shift={(311,99.57)}, rotate = 405.36] [fill={rgb, 255:red, 0; green, 0; blue, 0 }  ][line width=0.75]  [draw opacity=0] (8.93,-4.29) -- (0,0) -- (8.93,4.29) -- cycle    ;

%Straight Lines [id:da08037885506928388] 
\draw    (302.22,90.68) -- (291.33,79.65) ;
\draw [shift={(289.92,78.23)}, rotate = 405.36] [fill={rgb, 255:red, 0; green, 0; blue, 0 }  ][line width=0.75]  [draw opacity=0] (8.93,-4.29) -- (0,0) -- (8.93,4.29) -- cycle    ;

%Shape: Circle [id:dp1718776649824194] 
\draw   (324.22,78.44) .. controls (321.99,76.18) and (322.01,72.54) .. (324.27,70.31) .. controls (326.53,68.08) and (330.17,68.1) .. (332.4,70.36) .. controls (334.63,72.62) and (334.61,76.26) .. (332.35,78.49) .. controls (330.09,80.73) and (326.45,80.7) .. (324.22,78.44) -- cycle ;
%Shape: Circle [id:dp4765212840553168] 
\draw   (303.14,57.1) .. controls (300.91,54.84) and (300.93,51.2) .. (303.19,48.97) .. controls (305.45,46.73) and (309.09,46.76) .. (311.32,49.02) .. controls (313.55,51.28) and (313.53,54.92) .. (311.27,57.15) .. controls (309.01,59.38) and (305.37,59.36) .. (303.14,57.1) -- cycle ;
%Straight Lines [id:da8955847847833076] 
\draw    (343.94,90.23) -- (333.76,79.92) ;
\draw [shift={(332.35,78.49)}, rotate = 405.36] [fill={rgb, 255:red, 0; green, 0; blue, 0 }  ][line width=0.75]  [draw opacity=0] (8.93,-4.29) -- (0,0) -- (8.93,4.29) -- cycle    ;

%Straight Lines [id:da38895428141865174] 
\draw    (323.57,69.6) -- (312.68,58.57) ;
\draw [shift={(311.27,57.15)}, rotate = 405.36] [fill={rgb, 255:red, 0; green, 0; blue, 0 }  ][line width=0.75]  [draw opacity=0] (8.93,-4.29) -- (0,0) -- (8.93,4.29) -- cycle    ;

%Straight Lines [id:da052553837951357485] 
\draw    (268.63,91.18) -- (280.37,79.58) ;
\draw [shift={(281.79,78.18)}, rotate = 495.36] [fill={rgb, 255:red, 0; green, 0; blue, 0 }  ][line width=0.75]  [draw opacity=0] (8.93,-4.29) -- (0,0) -- (8.93,4.29) -- cycle    ;

%Straight Lines [id:da44860578896129955] 
\draw    (289.98,70.1) -- (301.72,58.5) ;
\draw [shift={(303.14,57.1)}, rotate = 495.36] [fill={rgb, 255:red, 0; green, 0; blue, 0 }  ][line width=0.75]  [draw opacity=0] (8.93,-4.29) -- (0,0) -- (8.93,4.29) -- cycle    ;

%Straight Lines [id:da6797967933324092] 
\draw    (290.42,111.82) -- (302.16,100.23) ;
\draw [shift={(303.58,98.82)}, rotate = 495.36] [fill={rgb, 255:red, 0; green, 0; blue, 0 }  ][line width=0.75]  [draw opacity=0] (8.93,-4.29) -- (0,0) -- (8.93,4.29) -- cycle    ;

%Straight Lines [id:da7766742827456306] 
\draw    (311.77,90.74) -- (323.51,79.15) ;
\draw [shift={(324.93,77.74)}, rotate = 495.36] [fill={rgb, 255:red, 0; green, 0; blue, 0 }  ][line width=0.75]  [draw opacity=0] (8.93,-4.29) -- (0,0) -- (8.93,4.29) -- cycle    ;

\end{tikzpicture}
}

%% file: 3pyramid.tex
\resizebox{\textwidth}{!}{

\tikzset{every picture/.style={line width=0.75pt}} %set default line width to 0.75pt        

\begin{tikzpicture}[x=0.75pt,y=0.75pt,yscale=-1,xscale=1]

%Straight Lines [id:da34415704321784535] 
\draw    (305.88,184.14) -- (305.85,163.04) ;
\draw [shift={(305.84,161.04)}, rotate = 449.9] [fill={rgb, 255:red, 0; green, 0; blue, 0 }  ][line width=0.75]  [draw opacity=0] (8.93,-4.29) -- (0,0) -- (8.93,4.29) -- cycle    ;

%Straight Lines [id:da7970265924003781] 
\draw [fill={rgb, 255:red, 255; green, 255; blue, 255 }  ,fill opacity=1 ]   (158.98,182.1) -- (178.85,160.9) ;
\draw [shift={(180.22,159.44)}, rotate = 493.16] [fill={rgb, 255:red, 0; green, 0; blue, 0 }  ][line width=0.75]  [draw opacity=0] (8.93,-4.29) -- (0,0) -- (8.93,4.29) -- cycle    ;

%Straight Lines [id:da5795795565839015] 
\draw    (185.98,181.1) -- (209.2,161.13) ;
\draw [shift={(210.71,159.83)}, rotate = 499.32] [fill={rgb, 255:red, 0; green, 0; blue, 0 }  ][line width=0.75]  [draw opacity=0] (8.93,-4.29) -- (0,0) -- (8.93,4.29) -- cycle    ;

%Straight Lines [id:da8359459464227736] 
\draw    (215.98,181.1) -- (239.2,161.13) ;
\draw [shift={(240.71,159.83)}, rotate = 499.32] [fill={rgb, 255:red, 0; green, 0; blue, 0 }  ][line width=0.75]  [draw opacity=0] (8.93,-4.29) -- (0,0) -- (8.93,4.29) -- cycle    ;

%Straight Lines [id:da40040091207978756] 
\draw    (247.98,181.1) -- (268.79,160.84) ;
\draw [shift={(270.22,159.44)}, rotate = 495.77] [fill={rgb, 255:red, 0; green, 0; blue, 0 }  ][line width=0.75]  [draw opacity=0] (8.93,-4.29) -- (0,0) -- (8.93,4.29) -- cycle    ;

%Straight Lines [id:da8861162994552958] 
\draw    (277.98,181.1) -- (300.76,159.81) ;
\draw [shift={(302.22,158.44)}, rotate = 496.94] [fill={rgb, 255:red, 0; green, 0; blue, 0 }  ][line width=0.75]  [draw opacity=0] (8.93,-4.29) -- (0,0) -- (8.93,4.29) -- cycle    ;

%Straight Lines [id:da04517691068355534] 
\draw    (309.98,180.1) -- (330.79,159.84) ;
\draw [shift={(332.22,158.44)}, rotate = 495.77] [fill={rgb, 255:red, 0; green, 0; blue, 0 }  ][line width=0.75]  [draw opacity=0] (8.93,-4.29) -- (0,0) -- (8.93,4.29) -- cycle    ;

%Straight Lines [id:da33262671456736714] 
\draw    (340.98,181.1) -- (361.79,160.84) ;
\draw [shift={(363.22,159.44)}, rotate = 495.77] [fill={rgb, 255:red, 0; green, 0; blue, 0 }  ][line width=0.75]  [draw opacity=0] (8.93,-4.29) -- (0,0) -- (8.93,4.29) -- cycle    ;

%Straight Lines [id:da2301093736231825] 
\draw    (369.98,180.1) -- (390.28,160.23) ;
\draw [shift={(391.71,158.83)}, rotate = 495.63] [fill={rgb, 255:red, 0; green, 0; blue, 0 }  ][line width=0.75]  [draw opacity=0] (8.93,-4.29) -- (0,0) -- (8.93,4.29) -- cycle    ;

%Straight Lines [id:da67797382743053] 
\draw    (401.48,180.71) -- (421.79,160.84) ;
\draw [shift={(423.22,159.44)}, rotate = 495.63] [fill={rgb, 255:red, 0; green, 0; blue, 0 }  ][line width=0.75]  [draw opacity=0] (8.93,-4.29) -- (0,0) -- (8.93,4.29) -- cycle    ;

%Straight Lines [id:da6077594569910298] 
\draw    (335.88,184.14) -- (335.85,162.04) ;
\draw [shift={(335.84,160.04)}, rotate = 449.9] [fill={rgb, 255:red, 0; green, 0; blue, 0 }  ][line width=0.75]  [draw opacity=0] (8.93,-4.29) -- (0,0) -- (8.93,4.29) -- cycle    ;

%Straight Lines [id:da10304657875287093] 
\draw    (273.88,185.14) -- (275.24,161.49) ;
\draw [shift={(275.35,159.49)}, rotate = 453.27] [fill={rgb, 255:red, 0; green, 0; blue, 0 }  ][line width=0.75]  [draw opacity=0] (8.93,-4.29) -- (0,0) -- (8.93,4.29) -- cycle    ;

%Straight Lines [id:da5205637660331306] 
\draw    (243.88,185.14) -- (244.77,162.04) ;
\draw [shift={(244.84,160.04)}, rotate = 452.19] [fill={rgb, 255:red, 0; green, 0; blue, 0 }  ][line width=0.75]  [draw opacity=0] (8.93,-4.29) -- (0,0) -- (8.93,4.29) -- cycle    ;

%Straight Lines [id:da031422290220262106] 
\draw    (365.88,184.14) -- (366.73,163.04) ;
\draw [shift={(366.81,161.04)}, rotate = 452.3] [fill={rgb, 255:red, 0; green, 0; blue, 0 }  ][line width=0.75]  [draw opacity=0] (8.93,-4.29) -- (0,0) -- (8.93,4.29) -- cycle    ;

%Straight Lines [id:da9526248779007993] 
\draw    (214.98,180.1) -- (214.77,164.24) ;
\draw [shift={(214.75,162.24)}, rotate = 449.26] [fill={rgb, 255:red, 0; green, 0; blue, 0 }  ][line width=0.75]  [draw opacity=0] (8.93,-4.29) -- (0,0) -- (8.93,4.29) -- cycle    ;

%Straight Lines [id:da9790044023990057] 
\draw    (183.98,180.1) -- (183.77,164.24) ;
\draw [shift={(183.75,162.24)}, rotate = 449.26] [fill={rgb, 255:red, 0; green, 0; blue, 0 }  ][line width=0.75]  [draw opacity=0] (8.93,-4.29) -- (0,0) -- (8.93,4.29) -- cycle    ;

%Straight Lines [id:da6435884736355861] 
\draw    (396.88,185.14) -- (396.82,164.04) ;
\draw [shift={(396.81,162.04)}, rotate = 449.82] [fill={rgb, 255:red, 0; green, 0; blue, 0 }  ][line width=0.75]  [draw opacity=0] (8.93,-4.29) -- (0,0) -- (8.93,4.29) -- cycle    ;

%Straight Lines [id:da2523045829916386] 
\draw    (426.58,178.35) -- (426.38,162.49) ;
\draw [shift={(426.35,160.49)}, rotate = 449.26] [fill={rgb, 255:red, 0; green, 0; blue, 0 }  ][line width=0.75]  [draw opacity=0] (8.93,-4.29) -- (0,0) -- (8.93,4.29) -- cycle    ;

%Straight Lines [id:da6934282479945018] 
\draw    (450.84,181.04) -- (432.69,160.98) ;
\draw [shift={(431.35,159.49)}, rotate = 407.87] [fill={rgb, 255:red, 0; green, 0; blue, 0 }  ][line width=0.75]  [draw opacity=0] (8.93,-4.29) -- (0,0) -- (8.93,4.29) -- cycle    ;

%Straight Lines [id:da7836727311766261] 
\draw    (422.84,181.04) -- (402.76,160.91) ;
\draw [shift={(401.35,159.49)}, rotate = 405.08000000000004] [fill={rgb, 255:red, 0; green, 0; blue, 0 }  ][line width=0.75]  [draw opacity=0] (8.93,-4.29) -- (0,0) -- (8.93,4.29) -- cycle    ;

%Straight Lines [id:da40547829490122145] 
\draw    (392.84,181.04) -- (371.76,159.91) ;
\draw [shift={(370.35,158.49)}, rotate = 405.07] [fill={rgb, 255:red, 0; green, 0; blue, 0 }  ][line width=0.75]  [draw opacity=0] (8.93,-4.29) -- (0,0) -- (8.93,4.29) -- cycle    ;

%Straight Lines [id:da018187596152810848] 
\draw    (361.84,180.04) -- (341.76,159.91) ;
\draw [shift={(340.35,158.49)}, rotate = 405.08000000000004] [fill={rgb, 255:red, 0; green, 0; blue, 0 }  ][line width=0.75]  [draw opacity=0] (8.93,-4.29) -- (0,0) -- (8.93,4.29) -- cycle    ;

%Straight Lines [id:da026962255306380634] 
\draw    (331.84,180.04) -- (311.76,159.91) ;
\draw [shift={(310.35,158.49)}, rotate = 405.08000000000004] [fill={rgb, 255:red, 0; green, 0; blue, 0 }  ][line width=0.75]  [draw opacity=0] (8.93,-4.29) -- (0,0) -- (8.93,4.29) -- cycle    ;

%Straight Lines [id:da18467421198772405] 
\draw    (301.84,180.04) -- (279.86,160.81) ;
\draw [shift={(278.35,159.49)}, rotate = 401.18] [fill={rgb, 255:red, 0; green, 0; blue, 0 }  ][line width=0.75]  [draw opacity=0] (8.93,-4.29) -- (0,0) -- (8.93,4.29) -- cycle    ;

%Straight Lines [id:da11293280584115539] 
\draw    (239.84,181.04) -- (219.76,160.91) ;
\draw [shift={(218.35,159.49)}, rotate = 405.08000000000004] [fill={rgb, 255:red, 0; green, 0; blue, 0 }  ][line width=0.75]  [draw opacity=0] (8.93,-4.29) -- (0,0) -- (8.93,4.29) -- cycle    ;

%Straight Lines [id:da9864959093534982] 
\draw    (269.84,181.04) -- (249.76,160.91) ;
\draw [shift={(248.35,159.49)}, rotate = 405.08000000000004] [fill={rgb, 255:red, 0; green, 0; blue, 0 }  ][line width=0.75]  [draw opacity=0] (8.93,-4.29) -- (0,0) -- (8.93,4.29) -- cycle    ;

%Straight Lines [id:da40520823909394377] 
\draw    (209.84,181.04) -- (189.76,160.91) ;
\draw [shift={(188.35,159.49)}, rotate = 405.08000000000004] [fill={rgb, 255:red, 0; green, 0; blue, 0 }  ][line width=0.75]  [draw opacity=0] (8.93,-4.29) -- (0,0) -- (8.93,4.29) -- cycle    ;

%Straight Lines [id:da9339041531024501] 
\draw    (306.31,154.4) -- (305.36,133.15) ;
\draw [shift={(305.27,131.15)}, rotate = 447.44] [fill={rgb, 255:red, 0; green, 0; blue, 0 }  ][line width=0.75]  [draw opacity=0] (8.93,-4.29) -- (0,0) -- (8.93,4.29) -- cycle    ;

%Straight Lines [id:da16368604455467728] 
\draw [fill={rgb, 255:red, 255; green, 255; blue, 255 }  ,fill opacity=1 ]   (188.4,151.36) -- (208.71,131.5) ;
\draw [shift={(210.14,130.1)}, rotate = 495.63] [fill={rgb, 255:red, 0; green, 0; blue, 0 }  ][line width=0.75]  [draw opacity=0] (8.93,-4.29) -- (0,0) -- (8.93,4.29) -- cycle    ;

%Straight Lines [id:da3057310841762293] 
\draw    (218.4,151.36) -- (238.71,131.5) ;
\draw [shift={(240.14,130.1)}, rotate = 495.63] [fill={rgb, 255:red, 0; green, 0; blue, 0 }  ][line width=0.75]  [draw opacity=0] (8.93,-4.29) -- (0,0) -- (8.93,4.29) -- cycle    ;

%Straight Lines [id:da6525821341661442] 
\draw    (248.4,151.36) -- (268.71,131.5) ;
\draw [shift={(270.14,130.1)}, rotate = 495.63] [fill={rgb, 255:red, 0; green, 0; blue, 0 }  ][line width=0.75]  [draw opacity=0] (8.93,-4.29) -- (0,0) -- (8.93,4.29) -- cycle    ;

%Straight Lines [id:da3144621910008789] 
\draw    (278.4,151.36) -- (300.68,130.47) ;
\draw [shift={(302.14,129.1)}, rotate = 496.83] [fill={rgb, 255:red, 0; green, 0; blue, 0 }  ][line width=0.75]  [draw opacity=0] (8.93,-4.29) -- (0,0) -- (8.93,4.29) -- cycle    ;

%Straight Lines [id:da9342827960910687] 
\draw    (310.4,150.36) -- (330.71,130.5) ;
\draw [shift={(332.14,129.1)}, rotate = 495.63] [fill={rgb, 255:red, 0; green, 0; blue, 0 }  ][line width=0.75]  [draw opacity=0] (8.93,-4.29) -- (0,0) -- (8.93,4.29) -- cycle    ;

%Straight Lines [id:da7082496382440038] 
\draw    (340.4,150.36) -- (360.71,130.5) ;
\draw [shift={(362.14,129.1)}, rotate = 495.63] [fill={rgb, 255:red, 0; green, 0; blue, 0 }  ][line width=0.75]  [draw opacity=0] (8.93,-4.29) -- (0,0) -- (8.93,4.29) -- cycle    ;

%Straight Lines [id:da7630794946701083] 
\draw    (371.4,151.36) -- (391.71,131.5) ;
\draw [shift={(393.14,130.1)}, rotate = 495.63] [fill={rgb, 255:red, 0; green, 0; blue, 0 }  ][line width=0.75]  [draw opacity=0] (8.93,-4.29) -- (0,0) -- (8.93,4.29) -- cycle    ;

%Straight Lines [id:da158650073064176] 
\draw    (336.31,154.4) -- (336.27,132.15) ;
\draw [shift={(336.27,130.15)}, rotate = 449.91] [fill={rgb, 255:red, 0; green, 0; blue, 0 }  ][line width=0.75]  [draw opacity=0] (8.93,-4.29) -- (0,0) -- (8.93,4.29) -- cycle    ;

%Straight Lines [id:da28760565600846544] 
\draw    (274.27,145.15) -- (274.27,133.15) ;
\draw [shift={(274.27,131.15)}, rotate = 450] [fill={rgb, 255:red, 0; green, 0; blue, 0 }  ][line width=0.75]  [draw opacity=0] (8.93,-4.29) -- (0,0) -- (8.93,4.29) -- cycle    ;

%Straight Lines [id:da37200782111022934] 
\draw    (244.31,155.4) -- (245.19,132.15) ;
\draw [shift={(245.27,130.15)}, rotate = 452.18] [fill={rgb, 255:red, 0; green, 0; blue, 0 }  ][line width=0.75]  [draw opacity=0] (8.93,-4.29) -- (0,0) -- (8.93,4.29) -- cycle    ;

%Straight Lines [id:da3994928585998716] 
\draw    (366.31,154.4) -- (365.99,133.15) ;
\draw [shift={(365.96,131.15)}, rotate = 449.13] [fill={rgb, 255:red, 0; green, 0; blue, 0 }  ][line width=0.75]  [draw opacity=0] (8.93,-4.29) -- (0,0) -- (8.93,4.29) -- cycle    ;

%Straight Lines [id:da31821198383811233] 
\draw    (214.5,149) -- (214.3,133.15) ;
\draw [shift={(214.27,131.15)}, rotate = 449.26] [fill={rgb, 255:red, 0; green, 0; blue, 0 }  ][line width=0.75]  [draw opacity=0] (8.93,-4.29) -- (0,0) -- (8.93,4.29) -- cycle    ;

%Straight Lines [id:da4227711677654631] 
\draw    (397.31,155.4) -- (397.9,134.15) ;
\draw [shift={(397.96,132.15)}, rotate = 451.6] [fill={rgb, 255:red, 0; green, 0; blue, 0 }  ][line width=0.75]  [draw opacity=0] (8.93,-4.29) -- (0,0) -- (8.93,4.29) -- cycle    ;

%Straight Lines [id:da31482985109682327] 
\draw    (423.27,151.31) -- (402.71,131.53) ;
\draw [shift={(401.27,130.15)}, rotate = 403.89] [fill={rgb, 255:red, 0; green, 0; blue, 0 }  ][line width=0.75]  [draw opacity=0] (8.93,-4.29) -- (0,0) -- (8.93,4.29) -- cycle    ;

%Straight Lines [id:da7702210470997595] 
\draw    (393.27,151.31) -- (371.71,130.54) ;
\draw [shift={(370.27,129.15)}, rotate = 403.94] [fill={rgb, 255:red, 0; green, 0; blue, 0 }  ][line width=0.75]  [draw opacity=0] (8.93,-4.29) -- (0,0) -- (8.93,4.29) -- cycle    ;

%Straight Lines [id:da03309405942824695] 
\draw    (362.27,150.31) -- (341.71,130.53) ;
\draw [shift={(340.27,129.15)}, rotate = 403.89] [fill={rgb, 255:red, 0; green, 0; blue, 0 }  ][line width=0.75]  [draw opacity=0] (8.93,-4.29) -- (0,0) -- (8.93,4.29) -- cycle    ;

%Straight Lines [id:da6817839869866729] 
\draw    (332.27,150.31) -- (311.71,130.53) ;
\draw [shift={(310.27,129.15)}, rotate = 403.89] [fill={rgb, 255:red, 0; green, 0; blue, 0 }  ][line width=0.75]  [draw opacity=0] (8.93,-4.29) -- (0,0) -- (8.93,4.29) -- cycle    ;

%Straight Lines [id:da8451820955654092] 
\draw    (302.27,150.31) -- (279.8,131.43) ;
\draw [shift={(278.27,130.15)}, rotate = 400.03999999999996] [fill={rgb, 255:red, 0; green, 0; blue, 0 }  ][line width=0.75]  [draw opacity=0] (8.93,-4.29) -- (0,0) -- (8.93,4.29) -- cycle    ;

%Straight Lines [id:da1177967882488753] 
\draw    (270.27,151.31) -- (249.71,131.53) ;
\draw [shift={(248.27,130.15)}, rotate = 403.89] [fill={rgb, 255:red, 0; green, 0; blue, 0 }  ][line width=0.75]  [draw opacity=0] (8.93,-4.29) -- (0,0) -- (8.93,4.29) -- cycle    ;

%Straight Lines [id:da1627014289367037] 
\draw    (240.27,151.31) -- (219.71,131.53) ;
\draw [shift={(218.27,130.15)}, rotate = 403.89] [fill={rgb, 255:red, 0; green, 0; blue, 0 }  ][line width=0.75]  [draw opacity=0] (8.93,-4.29) -- (0,0) -- (8.93,4.29) -- cycle    ;

%Straight Lines [id:da28404361022703495] 
\draw    (306.23,125.06) -- (306.19,101.97) ;
\draw [shift={(306.19,99.97)}, rotate = 449.91] [fill={rgb, 255:red, 0; green, 0; blue, 0 }  ][line width=0.75]  [draw opacity=0] (8.93,-4.29) -- (0,0) -- (8.93,4.29) -- cycle    ;

%Straight Lines [id:da6384197771189963] 
\draw [fill={rgb, 255:red, 255; green, 255; blue, 255 }  ,fill opacity=1 ]   (218.32,122.02) -- (238.39,101.6) ;
\draw [shift={(239.79,100.18)}, rotate = 494.51] [fill={rgb, 255:red, 0; green, 0; blue, 0 }  ][line width=0.75]  [draw opacity=0] (8.93,-4.29) -- (0,0) -- (8.93,4.29) -- cycle    ;

%Straight Lines [id:da4311755896667351] 
\draw [fill={rgb, 255:red, 255; green, 255; blue, 255 }  ,fill opacity=1 ]   (248.32,122.02) -- (268.39,101.6) ;
\draw [shift={(269.79,100.18)}, rotate = 494.51] [fill={rgb, 255:red, 0; green, 0; blue, 0 }  ][line width=0.75]  [draw opacity=0] (8.93,-4.29) -- (0,0) -- (8.93,4.29) -- cycle    ;

%Straight Lines [id:da39363530649020406] 
\draw    (336.23,125.06) -- (335.56,103) ;
\draw [shift={(335.5,101)}, rotate = 448.26] [fill={rgb, 255:red, 0; green, 0; blue, 0 }  ][line width=0.75]  [draw opacity=0] (8.93,-4.29) -- (0,0) -- (8.93,4.29) -- cycle    ;

%Straight Lines [id:da035572789190215115] 
\draw    (274.23,126.06) -- (274.48,104) ;
\draw [shift={(274.5,102)}, rotate = 450.64] [fill={rgb, 255:red, 0; green, 0; blue, 0 }  ][line width=0.75]  [draw opacity=0] (8.93,-4.29) -- (0,0) -- (8.93,4.29) -- cycle    ;

%Straight Lines [id:da5367835029099606] 
\draw    (244.23,124.06) -- (244.48,104) ;
\draw [shift={(244.5,102)}, rotate = 450.7] [fill={rgb, 255:red, 0; green, 0; blue, 0 }  ][line width=0.75]  [draw opacity=0] (8.93,-4.29) -- (0,0) -- (8.93,4.29) -- cycle    ;

%Straight Lines [id:da15673874155158285] 
\draw    (366.23,125.06) -- (365.56,104) ;
\draw [shift={(365.5,102)}, rotate = 448.19] [fill={rgb, 255:red, 0; green, 0; blue, 0 }  ][line width=0.75]  [draw opacity=0] (8.93,-4.29) -- (0,0) -- (8.93,4.29) -- cycle    ;

%Straight Lines [id:da8964469886934598] 
\draw    (302.19,120.97) -- (279.45,101.53) ;
\draw [shift={(277.92,100.23)}, rotate = 400.52] [fill={rgb, 255:red, 0; green, 0; blue, 0 }  ][line width=0.75]  [draw opacity=0] (8.93,-4.29) -- (0,0) -- (8.93,4.29) -- cycle    ;

%Straight Lines [id:da7266712694618267] 
\draw    (270.19,121.97) -- (249.36,101.62) ;
\draw [shift={(247.92,100.23)}, rotate = 404.31] [fill={rgb, 255:red, 0; green, 0; blue, 0 }  ][line width=0.75]  [draw opacity=0] (8.93,-4.29) -- (0,0) -- (8.93,4.29) -- cycle    ;

%Straight Lines [id:da6699382335171968] 
\draw [fill={rgb, 255:red, 255; green, 255; blue, 255 }  ,fill opacity=1 ]   (278.32,122.02) -- (300.36,100.57) ;
\draw [shift={(301.79,99.18)}, rotate = 495.78] [fill={rgb, 255:red, 0; green, 0; blue, 0 }  ][line width=0.75]  [draw opacity=0] (8.93,-4.29) -- (0,0) -- (8.93,4.29) -- cycle    ;

%Straight Lines [id:da4895902666846237] 
\draw [fill={rgb, 255:red, 255; green, 255; blue, 255 }  ,fill opacity=1 ]   (310.32,121.02) -- (330.39,100.6) ;
\draw [shift={(331.79,99.18)}, rotate = 494.51] [fill={rgb, 255:red, 0; green, 0; blue, 0 }  ][line width=0.75]  [draw opacity=0] (8.93,-4.29) -- (0,0) -- (8.93,4.29) -- cycle    ;

%Straight Lines [id:da29192389325117474] 
\draw [fill={rgb, 255:red, 255; green, 255; blue, 255 }  ,fill opacity=1 ]   (340.32,121.02) -- (360.39,100.6) ;
\draw [shift={(361.79,99.18)}, rotate = 494.51] [fill={rgb, 255:red, 0; green, 0; blue, 0 }  ][line width=0.75]  [draw opacity=0] (8.93,-4.29) -- (0,0) -- (8.93,4.29) -- cycle    ;

%Straight Lines [id:da14163839253850852] 
\draw [fill={rgb, 255:red, 255; green, 255; blue, 255 }  ,fill opacity=1 ]   (393.19,121.97) -- (371.35,100.63) ;
\draw [shift={(369.92,99.23)}, rotate = 404.34000000000003] [fill={rgb, 255:red, 0; green, 0; blue, 0 }  ][line width=0.75]  [draw opacity=0] (8.93,-4.29) -- (0,0) -- (8.93,4.29) -- cycle    ;

%Straight Lines [id:da7343506870477885] 
\draw [fill={rgb, 255:red, 255; green, 255; blue, 255 }  ,fill opacity=1 ]   (362.19,120.97) -- (341.36,100.62) ;
\draw [shift={(339.92,99.23)}, rotate = 404.31] [fill={rgb, 255:red, 0; green, 0; blue, 0 }  ][line width=0.75]  [draw opacity=0] (8.93,-4.29) -- (0,0) -- (8.93,4.29) -- cycle    ;

%Straight Lines [id:da1027859445338446] 
\draw [fill={rgb, 255:red, 255; green, 255; blue, 255 }  ,fill opacity=1 ]   (332.19,120.97) -- (311.36,100.62) ;
\draw [shift={(309.92,99.23)}, rotate = 404.31] [fill={rgb, 255:red, 0; green, 0; blue, 0 }  ][line width=0.75]  [draw opacity=0] (8.93,-4.29) -- (0,0) -- (8.93,4.29) -- cycle    ;

%Straight Lines [id:da020290930951866626] 
\draw [fill={rgb, 255:red, 255; green, 255; blue, 255 }  ,fill opacity=1 ]   (247.98,92.1) -- (267.82,71.87) ;
\draw [shift={(269.22,70.44)}, rotate = 494.45] [fill={rgb, 255:red, 0; green, 0; blue, 0 }  ][line width=0.75]  [draw opacity=0] (8.93,-4.29) -- (0,0) -- (8.93,4.29) -- cycle    ;

%Straight Lines [id:da0930475777907267] 
\draw [fill={rgb, 255:red, 255; green, 255; blue, 255 }  ,fill opacity=1 ]   (277.98,92.1) -- (299.79,70.84) ;
\draw [shift={(301.22,69.44)}, rotate = 495.74] [fill={rgb, 255:red, 0; green, 0; blue, 0 }  ][line width=0.75]  [draw opacity=0] (8.93,-4.29) -- (0,0) -- (8.93,4.29) -- cycle    ;

%Straight Lines [id:da7288782256871695] 
\draw [fill={rgb, 255:red, 255; green, 255; blue, 255 }  ,fill opacity=1 ]   (309.98,91.1) -- (329.82,70.87) ;
\draw [shift={(331.22,69.44)}, rotate = 494.45] [fill={rgb, 255:red, 0; green, 0; blue, 0 }  ][line width=0.75]  [draw opacity=0] (8.93,-4.29) -- (0,0) -- (8.93,4.29) -- cycle    ;

%Straight Lines [id:da038833494021750736] 
\draw [fill={rgb, 255:red, 255; green, 255; blue, 255 }  ,fill opacity=1 ]   (335.88,95.14) -- (335.53,73) ;
\draw [shift={(335.5,71)}, rotate = 449.09] [fill={rgb, 255:red, 0; green, 0; blue, 0 }  ][line width=0.75]  [draw opacity=0] (8.93,-4.29) -- (0,0) -- (8.93,4.29) -- cycle    ;

%Straight Lines [id:da5550518281461845] 
\draw [fill={rgb, 255:red, 255; green, 255; blue, 255 }  ,fill opacity=1 ]   (361.84,91.04) -- (340.79,70.88) ;
\draw [shift={(339.35,69.49)}, rotate = 403.77] [fill={rgb, 255:red, 0; green, 0; blue, 0 }  ][line width=0.75]  [draw opacity=0] (8.93,-4.29) -- (0,0) -- (8.93,4.29) -- cycle    ;

%Straight Lines [id:da6081026345866667] 
\draw [fill={rgb, 255:red, 255; green, 255; blue, 255 }  ,fill opacity=1 ]   (331.84,91.04) -- (310.79,70.88) ;
\draw [shift={(309.35,69.49)}, rotate = 403.77] [fill={rgb, 255:red, 0; green, 0; blue, 0 }  ][line width=0.75]  [draw opacity=0] (8.93,-4.29) -- (0,0) -- (8.93,4.29) -- cycle    ;

%Straight Lines [id:da5175384753522947] 
\draw [fill={rgb, 255:red, 255; green, 255; blue, 255 }  ,fill opacity=1 ]   (301.84,91.04) -- (278.88,71.78) ;
\draw [shift={(277.35,70.49)}, rotate = 400] [fill={rgb, 255:red, 0; green, 0; blue, 0 }  ][line width=0.75]  [draw opacity=0] (8.93,-4.29) -- (0,0) -- (8.93,4.29) -- cycle    ;

%Straight Lines [id:da35892265702281834] 
\draw [fill={rgb, 255:red, 255; green, 255; blue, 255 }  ,fill opacity=1 ]   (273.88,96.14) -- (273.32,73.15) ;
\draw [shift={(273.27,71.15)}, rotate = 448.59] [fill={rgb, 255:red, 0; green, 0; blue, 0 }  ][line width=0.75]  [draw opacity=0] (8.93,-4.29) -- (0,0) -- (8.93,4.29) -- cycle    ;

%Straight Lines [id:da8680999986215057] 
\draw [fill={rgb, 255:red, 255; green, 255; blue, 255 }  ,fill opacity=1 ]   (305.88,95.14) -- (304.62,74) ;
\draw [shift={(304.5,72)}, rotate = 446.58] [fill={rgb, 255:red, 0; green, 0; blue, 0 }  ][line width=0.75]  [draw opacity=0] (8.93,-4.29) -- (0,0) -- (8.93,4.29) -- cycle    ;

%Straight Lines [id:da44860578896129955] 
\draw [fill={rgb, 255:red, 255; green, 255; blue, 255 }  ,fill opacity=1 ]   (305.31,65.4) -- (305.48,45) ;
\draw [shift={(305.5,43)}, rotate = 450.49] [fill={rgb, 255:red, 0; green, 0; blue, 0 }  ][line width=0.75]  [draw opacity=0] (8.93,-4.29) -- (0,0) -- (8.93,4.29) -- cycle    ;

%Shape: Circle [id:dp1718776649824194] 
\draw  [fill={rgb, 255:red, 255; green, 255; blue, 255 }  ,fill opacity=1 ] (301.22,69.44) .. controls (298.99,67.18) and (299.01,63.54) .. (301.27,61.31) .. controls (303.53,59.08) and (307.17,59.1) .. (309.4,61.36) .. controls (311.63,63.62) and (311.61,67.26) .. (309.35,69.49) .. controls (307.09,71.73) and (303.45,71.7) .. (301.22,69.44) -- cycle ;
%Shape: Circle [id:dp4765212840553168] 
\draw   (302.14,41.1) .. controls (299.91,38.84) and (299.93,35.2) .. (302.19,32.97) .. controls (304.45,30.73) and (308.09,30.76) .. (310.32,33.02) .. controls (312.55,35.28) and (312.53,38.92) .. (310.27,41.15) .. controls (308.01,43.38) and (304.37,43.36) .. (302.14,41.1) -- cycle ;
%Straight Lines [id:da052553837951357485] 
\draw [fill={rgb, 255:red, 255; green, 255; blue, 255 }  ,fill opacity=1 ]   (277.4,62.36) -- (300.62,42.4) ;
\draw [shift={(302.14,41.1)}, rotate = 499.32] [fill={rgb, 255:red, 0; green, 0; blue, 0 }  ][line width=0.75]  [draw opacity=0] (8.93,-4.29) -- (0,0) -- (8.93,4.29) -- cycle    ;

%Shape: Circle [id:dp03303118639570024] 
\draw  [fill={rgb, 255:red, 255; green, 255; blue, 255 }  ,fill opacity=1 ] (301.79,188.18) .. controls (299.56,185.92) and (299.58,182.28) .. (301.84,180.04) .. controls (304.1,177.81) and (307.74,177.84) .. (309.98,180.1) .. controls (312.21,182.36) and (312.18,186) .. (309.92,188.23) .. controls (307.67,190.46) and (304.02,190.44) .. (301.79,188.18) -- cycle ;
%Shape: Circle [id:dp02519422967128726] 
\draw  [fill={rgb, 255:red, 255; green, 255; blue, 255 }  ,fill opacity=1 ] (302.22,158.44) .. controls (299.99,156.18) and (300.01,152.54) .. (302.27,150.31) .. controls (304.53,148.08) and (308.17,148.1) .. (310.4,150.36) .. controls (312.63,152.62) and (312.61,156.26) .. (310.35,158.49) .. controls (308.09,160.73) and (304.45,160.7) .. (302.22,158.44) -- cycle ;
%Shape: Circle [id:dp24523844604925094] 
\draw  [fill={rgb, 255:red, 255; green, 255; blue, 255 }  ,fill opacity=1 ] (302.14,129.1) .. controls (299.91,126.84) and (299.93,123.2) .. (302.19,120.97) .. controls (304.45,118.73) and (308.09,118.76) .. (310.32,121.02) .. controls (312.55,123.28) and (312.53,126.92) .. (310.27,129.15) .. controls (308.01,131.38) and (304.37,131.36) .. (302.14,129.1) -- cycle ;
%Shape: Circle [id:dp27902184209201586] 
\draw  [fill={rgb, 255:red, 255; green, 255; blue, 255 }  ,fill opacity=1 ] (331.22,69.44) .. controls (328.99,67.18) and (329.01,63.54) .. (331.27,61.31) .. controls (333.53,59.08) and (337.17,59.1) .. (339.4,61.36) .. controls (341.63,63.62) and (341.61,67.26) .. (339.35,69.49) .. controls (337.09,71.73) and (333.45,71.7) .. (331.22,69.44) -- cycle ;
%Shape: Circle [id:dp171692383948848] 
\draw  [fill={rgb, 255:red, 255; green, 255; blue, 255 }  ,fill opacity=1 ] (331.79,188.18) .. controls (329.56,185.92) and (329.58,182.28) .. (331.84,180.04) .. controls (334.1,177.81) and (337.74,177.84) .. (339.98,180.1) .. controls (342.21,182.36) and (342.18,186) .. (339.92,188.23) .. controls (337.67,190.46) and (334.02,190.44) .. (331.79,188.18) -- cycle ;
%Shape: Circle [id:dp7282533454545643] 
\draw  [fill={rgb, 255:red, 255; green, 255; blue, 255 }  ,fill opacity=1 ] (332.22,158.44) .. controls (329.99,156.18) and (330.01,152.54) .. (332.27,150.31) .. controls (334.53,148.08) and (338.17,148.1) .. (340.4,150.36) .. controls (342.63,152.62) and (342.61,156.26) .. (340.35,158.49) .. controls (338.09,160.73) and (334.45,160.7) .. (332.22,158.44) -- cycle ;
%Shape: Circle [id:dp1550509585866] 
\draw  [fill={rgb, 255:red, 255; green, 255; blue, 255 }  ,fill opacity=1 ] (332.14,129.1) .. controls (329.91,126.84) and (329.93,123.2) .. (332.19,120.97) .. controls (334.45,118.73) and (338.09,118.76) .. (340.32,121.02) .. controls (342.55,123.28) and (342.53,126.92) .. (340.27,129.15) .. controls (338.01,131.38) and (334.37,131.36) .. (332.14,129.1) -- cycle ;
%Shape: Circle [id:dp34719276181715064] 
\draw  [fill={rgb, 255:red, 255; green, 255; blue, 255 }  ,fill opacity=1 ] (269.22,70.44) .. controls (266.99,68.18) and (267.01,64.54) .. (269.27,62.31) .. controls (271.53,60.08) and (275.17,60.1) .. (277.4,62.36) .. controls (279.63,64.62) and (279.61,68.26) .. (277.35,70.49) .. controls (275.09,72.73) and (271.45,72.7) .. (269.22,70.44) -- cycle ;
%Shape: Circle [id:dp5407200592901367] 
\draw  [fill={rgb, 255:red, 255; green, 255; blue, 255 }  ,fill opacity=1 ] (269.79,189.18) .. controls (267.56,186.92) and (267.58,183.28) .. (269.84,181.04) .. controls (272.1,178.81) and (275.74,178.84) .. (277.98,181.1) .. controls (280.21,183.36) and (280.18,187) .. (277.92,189.23) .. controls (275.67,191.46) and (272.02,191.44) .. (269.79,189.18) -- cycle ;
%Shape: Circle [id:dp14794516159403948] 
\draw  [fill={rgb, 255:red, 255; green, 255; blue, 255 }  ,fill opacity=1 ] (270.22,159.44) .. controls (267.99,157.18) and (268.01,153.54) .. (270.27,151.31) .. controls (272.53,149.08) and (276.17,149.1) .. (278.4,151.36) .. controls (280.63,153.62) and (280.61,157.26) .. (278.35,159.49) .. controls (276.09,161.73) and (272.45,161.7) .. (270.22,159.44) -- cycle ;
%Shape: Circle [id:dp8145300742294805] 
\draw  [fill={rgb, 255:red, 255; green, 255; blue, 255 }  ,fill opacity=1 ] (270.14,130.1) .. controls (267.91,127.84) and (267.93,124.2) .. (270.19,121.97) .. controls (272.45,119.73) and (276.09,119.76) .. (278.32,122.02) .. controls (280.55,124.28) and (280.53,127.92) .. (278.27,130.15) .. controls (276.01,132.38) and (272.37,132.36) .. (270.14,130.1) -- cycle ;
%Shape: Circle [id:dp29066125181630076] 
\draw  [fill={rgb, 255:red, 255; green, 255; blue, 255 }  ,fill opacity=1 ] (361.79,188.18) .. controls (359.56,185.92) and (359.58,182.28) .. (361.84,180.04) .. controls (364.1,177.81) and (367.74,177.84) .. (369.98,180.1) .. controls (372.21,182.36) and (372.18,186) .. (369.92,188.23) .. controls (367.67,190.46) and (364.02,190.44) .. (361.79,188.18) -- cycle ;
%Shape: Circle [id:dp5494883445846179] 
\draw  [fill={rgb, 255:red, 255; green, 255; blue, 255 }  ,fill opacity=1 ] (362.22,158.44) .. controls (359.99,156.18) and (360.01,152.54) .. (362.27,150.31) .. controls (364.53,148.08) and (368.17,148.1) .. (370.4,150.36) .. controls (372.63,152.62) and (372.61,156.26) .. (370.35,158.49) .. controls (368.09,160.73) and (364.45,160.7) .. (362.22,158.44) -- cycle ;
%Shape: Circle [id:dp2914348302157945] 
\draw  [fill={rgb, 255:red, 255; green, 255; blue, 255 }  ,fill opacity=1 ] (362.14,129.1) .. controls (359.91,126.84) and (359.93,123.2) .. (362.19,120.97) .. controls (364.45,118.73) and (368.09,118.76) .. (370.32,121.02) .. controls (372.55,123.28) and (372.53,126.92) .. (370.27,129.15) .. controls (368.01,131.38) and (364.37,131.36) .. (362.14,129.1) -- cycle ;
%Shape: Circle [id:dp6905727669500488] 
\draw  [fill={rgb, 255:red, 255; green, 255; blue, 255 }  ,fill opacity=1 ] (239.79,189.18) .. controls (237.56,186.92) and (237.58,183.28) .. (239.84,181.04) .. controls (242.1,178.81) and (245.74,178.84) .. (247.98,181.1) .. controls (250.21,183.36) and (250.18,187) .. (247.92,189.23) .. controls (245.67,191.46) and (242.02,191.44) .. (239.79,189.18) -- cycle ;
%Shape: Circle [id:dp8002584252784148] 
\draw  [fill={rgb, 255:red, 255; green, 255; blue, 255 }  ,fill opacity=1 ] (240.22,159.44) .. controls (237.99,157.18) and (238.01,153.54) .. (240.27,151.31) .. controls (242.53,149.08) and (246.17,149.1) .. (248.4,151.36) .. controls (250.63,153.62) and (250.61,157.26) .. (248.35,159.49) .. controls (246.09,161.73) and (242.45,161.7) .. (240.22,159.44) -- cycle ;
%Shape: Circle [id:dp700282193402612] 
\draw  [fill={rgb, 255:red, 255; green, 255; blue, 255 }  ,fill opacity=1 ] (240.14,130.1) .. controls (237.91,127.84) and (237.93,124.2) .. (240.19,121.97) .. controls (242.45,119.73) and (246.09,119.76) .. (248.32,122.02) .. controls (250.55,124.28) and (250.53,127.92) .. (248.27,130.15) .. controls (246.01,132.38) and (242.37,132.36) .. (240.14,130.1) -- cycle ;
%Shape: Circle [id:dp1974392688470603] 
\draw  [fill={rgb, 255:red, 255; green, 255; blue, 255 }  ,fill opacity=1 ] (392.79,189.18) .. controls (390.56,186.92) and (390.58,183.28) .. (392.84,181.04) .. controls (395.1,178.81) and (398.74,178.84) .. (400.98,181.1) .. controls (403.21,183.36) and (403.18,187) .. (400.92,189.23) .. controls (398.67,191.46) and (395.02,191.44) .. (392.79,189.18) -- cycle ;
%Shape: Circle [id:dp22189309361353438] 
\draw  [fill={rgb, 255:red, 255; green, 255; blue, 255 }  ,fill opacity=1 ] (393.22,159.44) .. controls (390.99,157.18) and (391.01,153.54) .. (393.27,151.31) .. controls (395.53,149.08) and (399.17,149.1) .. (401.4,151.36) .. controls (403.63,153.62) and (403.61,157.26) .. (401.35,159.49) .. controls (399.09,161.73) and (395.45,161.7) .. (393.22,159.44) -- cycle ;
%Shape: Circle [id:dp21427526912184391] 
\draw  [fill={rgb, 255:red, 255; green, 255; blue, 255 }  ,fill opacity=1 ] (393.14,130.1) .. controls (390.91,127.84) and (390.93,124.2) .. (393.19,121.97) .. controls (395.45,119.73) and (399.09,119.76) .. (401.32,122.02) .. controls (403.55,124.28) and (403.53,127.92) .. (401.27,130.15) .. controls (399.01,132.38) and (395.37,132.36) .. (393.14,130.1) -- cycle ;
%Shape: Circle [id:dp4616878553245696] 
\draw  [fill={rgb, 255:red, 255; green, 255; blue, 255 }  ,fill opacity=1 ] (209.79,189.18) .. controls (207.56,186.92) and (207.58,183.28) .. (209.84,181.04) .. controls (212.1,178.81) and (215.74,178.84) .. (217.98,181.1) .. controls (220.21,183.36) and (220.18,187) .. (217.92,189.23) .. controls (215.67,191.46) and (212.02,191.44) .. (209.79,189.18) -- cycle ;
%Shape: Circle [id:dp18626616316323408] 
\draw  [fill={rgb, 255:red, 255; green, 255; blue, 255 }  ,fill opacity=1 ] (210.22,159.44) .. controls (207.99,157.18) and (208.01,153.54) .. (210.27,151.31) .. controls (212.53,149.08) and (216.17,149.1) .. (218.4,151.36) .. controls (220.63,153.62) and (220.61,157.26) .. (218.35,159.49) .. controls (216.09,161.73) and (212.45,161.7) .. (210.22,159.44) -- cycle ;
%Shape: Circle [id:dp42529060936428076] 
\draw  [fill={rgb, 255:red, 255; green, 255; blue, 255 }  ,fill opacity=1 ] (210.14,130.1) .. controls (207.91,127.84) and (207.93,124.2) .. (210.19,121.97) .. controls (212.45,119.73) and (216.09,119.76) .. (218.32,122.02) .. controls (220.55,124.28) and (220.53,127.92) .. (218.27,130.15) .. controls (216.01,132.38) and (212.37,132.36) .. (210.14,130.1) -- cycle ;
%Shape: Circle [id:dp5247371971842325] 
\draw  [fill={rgb, 255:red, 255; green, 255; blue, 255 }  ,fill opacity=1 ] (422.79,189.18) .. controls (420.56,186.92) and (420.58,183.28) .. (422.84,181.04) .. controls (425.1,178.81) and (428.74,178.84) .. (430.98,181.1) .. controls (433.21,183.36) and (433.18,187) .. (430.92,189.23) .. controls (428.67,191.46) and (425.02,191.44) .. (422.79,189.18) -- cycle ;
%Shape: Circle [id:dp7021240139579497] 
\draw  [fill={rgb, 255:red, 255; green, 255; blue, 255 }  ,fill opacity=1 ] (423.22,159.44) .. controls (420.99,157.18) and (421.01,153.54) .. (423.27,151.31) .. controls (425.53,149.08) and (429.17,149.1) .. (431.4,151.36) .. controls (433.63,153.62) and (433.61,157.26) .. (431.35,159.49) .. controls (429.09,161.73) and (425.45,161.7) .. (423.22,159.44) -- cycle ;
%Shape: Circle [id:dp4768836121111828] 
\draw  [fill={rgb, 255:red, 255; green, 255; blue, 255 }  ,fill opacity=1 ] (179.79,189.18) .. controls (177.56,186.92) and (177.58,183.28) .. (179.84,181.04) .. controls (182.1,178.81) and (185.74,178.84) .. (187.98,181.1) .. controls (190.21,183.36) and (190.18,187) .. (187.92,189.23) .. controls (185.67,191.46) and (182.02,191.44) .. (179.79,189.18) -- cycle ;
%Shape: Circle [id:dp5314669295070624] 
\draw  [fill={rgb, 255:red, 255; green, 255; blue, 255 }  ,fill opacity=1 ] (180.22,159.44) .. controls (177.99,157.18) and (178.01,153.54) .. (180.27,151.31) .. controls (182.53,149.08) and (186.17,149.1) .. (188.4,151.36) .. controls (190.63,153.62) and (190.61,157.26) .. (188.35,159.49) .. controls (186.09,161.73) and (182.45,161.7) .. (180.22,159.44) -- cycle ;
%Shape: Circle [id:dp8736331813031579] 
\draw  [fill={rgb, 255:red, 255; green, 255; blue, 255 }  ,fill opacity=1 ] (450.79,189.18) .. controls (448.56,186.92) and (448.58,183.28) .. (450.84,181.04) .. controls (453.1,178.81) and (456.74,178.84) .. (458.98,181.1) .. controls (461.21,183.36) and (461.18,187) .. (458.92,189.23) .. controls (456.67,191.46) and (453.02,191.44) .. (450.79,189.18) -- cycle ;
%Shape: Circle [id:dp5772809482966939] 
\draw  [fill={rgb, 255:red, 255; green, 255; blue, 255 }  ,fill opacity=1 ] (150.79,190.18) .. controls (148.56,187.92) and (148.58,184.28) .. (150.84,182.04) .. controls (153.1,179.81) and (156.74,179.84) .. (158.98,182.1) .. controls (161.21,184.36) and (161.18,188) .. (158.92,190.23) .. controls (156.67,192.46) and (153.02,192.44) .. (150.79,190.18) -- cycle ;
%Straight Lines [id:da348819440860183] 
\draw [fill={rgb, 255:red, 255; green, 255; blue, 255 }  ,fill opacity=1 ]   (331.27,61.31) -- (311.71,42.53) ;
\draw [shift={(310.27,41.15)}, rotate = 403.84000000000003] [fill={rgb, 255:red, 0; green, 0; blue, 0 }  ][line width=0.75]  [draw opacity=0] (8.93,-4.29) -- (0,0) -- (8.93,4.29) -- cycle    ;

%Shape: Circle [id:dp7652309452362334] 
\draw  [fill={rgb, 255:red, 255; green, 255; blue, 255 }  ,fill opacity=1 ] (301.79,99.18) .. controls (299.56,96.92) and (299.58,93.28) .. (301.84,91.04) .. controls (304.1,88.81) and (307.74,88.84) .. (309.98,91.1) .. controls (312.21,93.36) and (312.18,97) .. (309.92,99.23) .. controls (307.67,101.46) and (304.02,101.44) .. (301.79,99.18) -- cycle ;
%Shape: Circle [id:dp4672509686718871] 
\draw  [fill={rgb, 255:red, 255; green, 255; blue, 255 }  ,fill opacity=1 ] (331.79,99.18) .. controls (329.56,96.92) and (329.58,93.28) .. (331.84,91.04) .. controls (334.1,88.81) and (337.74,88.84) .. (339.98,91.1) .. controls (342.21,93.36) and (342.18,97) .. (339.92,99.23) .. controls (337.67,101.46) and (334.02,101.44) .. (331.79,99.18) -- cycle ;
%Shape: Circle [id:dp5781399751737633] 
\draw  [fill={rgb, 255:red, 255; green, 255; blue, 255 }  ,fill opacity=1 ] (361.79,99.18) .. controls (359.56,96.92) and (359.58,93.28) .. (361.84,91.04) .. controls (364.1,88.81) and (367.74,88.84) .. (369.98,91.1) .. controls (372.21,93.36) and (372.18,97) .. (369.92,99.23) .. controls (367.67,101.46) and (364.02,101.44) .. (361.79,99.18) -- cycle ;
%Shape: Circle [id:dp27491840594868155] 
\draw  [fill={rgb, 255:red, 255; green, 255; blue, 255 }  ,fill opacity=1 ] (239.79,100.18) .. controls (237.56,97.92) and (237.58,94.28) .. (239.84,92.04) .. controls (242.1,89.81) and (245.74,89.84) .. (247.98,92.1) .. controls (250.21,94.36) and (250.18,98) .. (247.92,100.23) .. controls (245.67,102.46) and (242.02,102.44) .. (239.79,100.18) -- cycle ;
%Shape: Circle [id:dp3123422080945424] 
\draw  [fill={rgb, 255:red, 255; green, 255; blue, 255 }  ,fill opacity=1 ] (269.79,100.18) .. controls (267.56,97.92) and (267.58,94.28) .. (269.84,92.04) .. controls (272.1,89.81) and (275.74,89.84) .. (277.98,92.1) .. controls (280.21,94.36) and (280.18,98) .. (277.92,100.23) .. controls (275.67,102.46) and (272.02,102.44) .. (269.79,100.18) -- cycle ;

\end{tikzpicture}
}

%% file: proofpyramid.tex
\tikzset{every picture/.style={line width=0.75pt}} %set default line width to 0.75pt        

\begin{tikzpicture}[x=0.75pt,y=0.75pt,yscale=-1,xscale=1]

%Shape: Circle [id:dp28860743959038926] 
\draw  [fill={rgb, 255:red, 208; green, 2; blue, 27 }  ,fill opacity=1 ] (282.97,238.48) .. controls (280.74,236.21) and (280.76,232.57) .. (283.03,230.34) .. controls (285.29,228.12) and (288.93,228.14) .. (291.16,230.4) .. controls (293.39,232.67) and (293.36,236.31) .. (291.1,238.54) .. controls (288.84,240.77) and (285.19,240.74) .. (282.97,238.48) -- cycle ;
%Shape: Circle [id:dp20210006473129538] 
\draw  [fill={rgb, 255:red, 208; green, 2; blue, 27 }  ,fill opacity=1 ] (262.61,217.82) .. controls (260.39,215.55) and (260.41,211.91) .. (262.68,209.68) .. controls (264.94,207.46) and (268.58,207.48) .. (270.81,209.75) .. controls (273.04,212.01) and (273.01,215.65) .. (270.75,217.88) .. controls (268.48,220.11) and (264.84,220.08) .. (262.61,217.82) -- cycle ;
%Shape: Circle [id:dp3230818484130964] 
\draw  [color={rgb, 255:red, 0; green, 0; blue, 0 }  ,draw opacity=1 ] (240.86,195.73) .. controls (238.63,193.47) and (238.66,189.83) .. (240.92,187.6) .. controls (243.18,185.37) and (246.82,185.4) .. (249.05,187.66) .. controls (251.28,189.92) and (251.25,193.56) .. (248.99,195.79) .. controls (246.73,198.02) and (243.09,197.99) .. (240.86,195.73) -- cycle ;
%Shape: Circle [id:dp813623866828254] 
\draw  [color={rgb, 255:red, 0; green, 0; blue, 0 }  ,draw opacity=1 ] (219.81,174.36) .. controls (217.58,172.1) and (217.6,168.46) .. (219.87,166.23) .. controls (222.13,164) and (225.77,164.03) .. (228,166.29) .. controls (230.23,168.55) and (230.2,172.19) .. (227.94,174.42) .. controls (225.67,176.65) and (222.03,176.62) .. (219.81,174.36) -- cycle ;
%Shape: Circle [id:dp5814591326927978] 
\draw  [color={rgb, 255:red, 0; green, 0; blue, 0 }  ,draw opacity=1 ] (197.35,151.56) .. controls (195.12,149.3) and (195.15,145.66) .. (197.41,143.43) .. controls (199.67,141.2) and (203.31,141.23) .. (205.54,143.49) .. controls (207.77,145.76) and (207.74,149.4) .. (205.48,151.62) .. controls (203.22,153.85) and (199.58,153.83) .. (197.35,151.56) -- cycle ;
%Straight Lines [id:da7882484334061957] 
\draw [color={rgb, 255:red, 208; green, 2; blue, 27 }  ,draw opacity=1 ][fill={rgb, 255:red, 208; green, 2; blue, 27 }  ,fill opacity=1 ]   (283.03,230.34) -- (272.15,219.3) ;
\draw [shift={(270.75,217.88)}, rotate = 405.43] [fill={rgb, 255:red, 208; green, 2; blue, 27 }  ,fill opacity=1 ][line width=0.75]  [draw opacity=0] (8.93,-4.29) -- (0,0) -- (8.93,4.29) -- cycle    ;

%Straight Lines [id:da6905513194917805] 
\draw [color={rgb, 255:red, 0; green, 0; blue, 0 }  ,draw opacity=1 ]   (262.68,209.68) -- (250.39,197.22) ;
\draw [shift={(248.99,195.79)}, rotate = 405.43] [fill={rgb, 255:red, 0; green, 0; blue, 0 }  ,fill opacity=1 ][line width=0.75]  [draw opacity=0] (8.93,-4.29) -- (0,0) -- (8.93,4.29) -- cycle    ;

%Straight Lines [id:da7647173172293733] 
\draw [color={rgb, 255:red, 0; green, 0; blue, 0 }  ,draw opacity=1 ]   (240.22,186.89) -- (229.34,175.85) ;
\draw [shift={(227.94,174.42)}, rotate = 405.43] [fill={rgb, 255:red, 0; green, 0; blue, 0 }  ,fill opacity=1 ][line width=0.75]  [draw opacity=0] (8.93,-4.29) -- (0,0) -- (8.93,4.29) -- cycle    ;

%Straight Lines [id:da7380110843277776] 
\draw [color={rgb, 255:red, 0; green, 0; blue, 0 }  ,draw opacity=1 ]   (219.87,166.23) -- (206.88,153.05) ;
\draw [shift={(205.48,151.62)}, rotate = 405.43] [fill={rgb, 255:red, 0; green, 0; blue, 0 }  ,fill opacity=1 ][line width=0.75]  [draw opacity=0] (8.93,-4.29) -- (0,0) -- (8.93,4.29) -- cycle    ;

%Shape: Circle [id:dp8504133946804597] 
\draw  [color={rgb, 255:red, 0; green, 0; blue, 0 }  ,draw opacity=1 ] (177.7,131.62) .. controls (175.47,129.35) and (175.5,125.71) .. (177.76,123.48) .. controls (180.02,121.26) and (183.66,121.28) .. (185.89,123.55) .. controls (188.12,125.81) and (188.09,129.45) .. (185.83,131.68) .. controls (183.57,133.91) and (179.93,133.88) .. (177.7,131.62) -- cycle ;
%Shape: Circle [id:dp340141764027726] 
\draw  [color={rgb, 255:red, 0; green, 0; blue, 0 }  ,draw opacity=1 ] (156.65,110.24) .. controls (154.42,107.98) and (154.44,104.34) .. (156.71,102.11) .. controls (158.97,99.88) and (162.61,99.91) .. (164.84,102.17) .. controls (167.07,104.44) and (167.04,108.08) .. (164.78,110.31) .. controls (162.51,112.53) and (158.87,112.51) .. (156.65,110.24) -- cycle ;
%Straight Lines [id:da25852576893172774] 
\draw [color={rgb, 255:red, 0; green, 0; blue, 0 }  ,draw opacity=1 ]   (197.41,143.43) -- (187.23,133.1) ;
\draw [shift={(185.83,131.68)}, rotate = 405.43] [fill={rgb, 255:red, 0; green, 0; blue, 0 }  ,fill opacity=1 ][line width=0.75]  [draw opacity=0] (8.93,-4.29) -- (0,0) -- (8.93,4.29) -- cycle    ;

%Straight Lines [id:da4975822554771765] 
\draw [color={rgb, 255:red, 0; green, 0; blue, 0 }  ,draw opacity=1 ]   (177.06,122.77) -- (166.18,111.73) ;
\draw [shift={(164.78,110.31)}, rotate = 405.43] [fill={rgb, 255:red, 0; green, 0; blue, 0 }  ,fill opacity=1 ][line width=0.75]  [draw opacity=0] (8.93,-4.29) -- (0,0) -- (8.93,4.29) -- cycle    ;

%Shape: Circle [id:dp4521969597706659] 
\draw  [color={rgb, 255:red, 0; green, 0; blue, 0 }  ,draw opacity=1 ] (305.05,216.72) .. controls (302.82,214.46) and (302.85,210.82) .. (305.11,208.59) .. controls (307.37,206.36) and (311.01,206.39) .. (313.24,208.65) .. controls (315.47,210.91) and (315.44,214.55) .. (313.18,216.78) .. controls (310.92,219.01) and (307.28,218.98) .. (305.05,216.72) -- cycle ;
%Shape: Circle [id:dp997574545597677] 
\draw  [fill={rgb, 255:red, 208; green, 2; blue, 27 }  ,fill opacity=1 ] (284.7,196.06) .. controls (282.47,193.8) and (282.5,190.16) .. (284.76,187.93) .. controls (287.02,185.7) and (290.66,185.73) .. (292.89,187.99) .. controls (295.12,190.25) and (295.09,193.89) .. (292.83,196.12) .. controls (290.57,198.35) and (286.93,198.32) .. (284.7,196.06) -- cycle ;
%Shape: Circle [id:dp07762260711286384] 
\draw  [color={rgb, 255:red, 0; green, 0; blue, 0 }  ,draw opacity=1 ] (262.94,173.98) .. controls (260.71,171.71) and (260.74,168.07) .. (263,165.85) .. controls (265.27,163.62) and (268.91,163.64) .. (271.14,165.91) .. controls (273.36,168.17) and (273.34,171.81) .. (271.07,174.04) .. controls (268.81,176.27) and (265.17,176.24) .. (262.94,173.98) -- cycle ;
%Shape: Circle [id:dp15380558924814602] 
\draw  [color={rgb, 255:red, 0; green, 0; blue, 0 }  ,draw opacity=1 ] (241.89,152.6) .. controls (239.66,150.34) and (239.69,146.7) .. (241.95,144.47) .. controls (244.21,142.24) and (247.85,142.27) .. (250.08,144.53) .. controls (252.31,146.8) and (252.28,150.44) .. (250.02,152.67) .. controls (247.76,154.89) and (244.12,154.87) .. (241.89,152.6) -- cycle ;
%Shape: Circle [id:dp3340315078454674] 
\draw  [color={rgb, 255:red, 0; green, 0; blue, 0 }  ,draw opacity=1 ] (219.43,129.81) .. controls (217.2,127.55) and (217.23,123.91) .. (219.49,121.68) .. controls (221.76,119.45) and (225.4,119.48) .. (227.63,121.74) .. controls (229.85,124) and (229.83,127.64) .. (227.56,129.87) .. controls (225.3,132.1) and (221.66,132.07) .. (219.43,129.81) -- cycle ;
%Straight Lines [id:da4193798422358561] 
\draw [color={rgb, 255:red, 0; green, 0; blue, 0 }  ,draw opacity=1 ]   (305.11,208.59) -- (294.23,197.55) ;
\draw [shift={(292.83,196.12)}, rotate = 405.43] [fill={rgb, 255:red, 0; green, 0; blue, 0 }  ,fill opacity=1 ][line width=0.75]  [draw opacity=0] (8.93,-4.29) -- (0,0) -- (8.93,4.29) -- cycle    ;

%Straight Lines [id:da8026976527815206] 
\draw [color={rgb, 255:red, 0; green, 0; blue, 0 }  ,draw opacity=1 ][fill={rgb, 255:red, 74; green, 144; blue, 226 }  ,fill opacity=1 ]   (284.76,187.93) -- (272.48,175.46) ;
\draw [shift={(271.07,174.04)}, rotate = 405.43] [fill={rgb, 255:red, 0; green, 0; blue, 0 }  ,fill opacity=1 ][line width=0.75]  [draw opacity=0] (8.93,-4.29) -- (0,0) -- (8.93,4.29) -- cycle    ;

%Straight Lines [id:da391960949852002] 
\draw [color={rgb, 255:red, 0; green, 0; blue, 0 }  ,draw opacity=1 ][fill={rgb, 255:red, 74; green, 144; blue, 226 }  ,fill opacity=1 ]   (262.3,165.13) -- (251.42,154.09) ;
\draw [shift={(250.02,152.67)}, rotate = 405.43] [fill={rgb, 255:red, 0; green, 0; blue, 0 }  ,fill opacity=1 ][line width=0.75]  [draw opacity=0] (8.93,-4.29) -- (0,0) -- (8.93,4.29) -- cycle    ;

%Straight Lines [id:da49140665908383063] 
\draw [color={rgb, 255:red, 0; green, 0; blue, 0 }  ,draw opacity=1 ][fill={rgb, 255:red, 74; green, 144; blue, 226 }  ,fill opacity=1 ]   (241.95,144.47) -- (228.97,131.29) ;
\draw [shift={(227.56,129.87)}, rotate = 405.43] [fill={rgb, 255:red, 0; green, 0; blue, 0 }  ,fill opacity=1 ][line width=0.75]  [draw opacity=0] (8.93,-4.29) -- (0,0) -- (8.93,4.29) -- cycle    ;

%Shape: Circle [id:dp7232604775286668] 
\draw  [color={rgb, 255:red, 0; green, 0; blue, 0 }  ,draw opacity=1 ] (199.78,109.86) .. controls (197.55,107.6) and (197.58,103.96) .. (199.84,101.73) .. controls (202.11,99.5) and (205.75,99.53) .. (207.98,101.79) .. controls (210.2,104.05) and (210.18,107.69) .. (207.91,109.92) .. controls (205.65,112.15) and (202.01,112.12) .. (199.78,109.86) -- cycle ;
%Shape: Circle [id:dp0937520102216427] 
\draw   (178.73,88.49) .. controls (176.5,86.23) and (176.53,82.59) .. (178.79,80.36) .. controls (181.05,78.13) and (184.69,78.16) .. (186.92,80.42) .. controls (189.15,82.68) and (189.12,86.32) .. (186.86,88.55) .. controls (184.6,90.78) and (180.96,90.75) .. (178.73,88.49) -- cycle ;
%Straight Lines [id:da4902185035627853] 
\draw [color={rgb, 255:red, 0; green, 0; blue, 0 }  ,draw opacity=1 ][fill={rgb, 255:red, 74; green, 144; blue, 226 }  ,fill opacity=1 ]   (219.49,121.68) -- (209.32,111.35) ;
\draw [shift={(207.91,109.92)}, rotate = 405.43] [fill={rgb, 255:red, 0; green, 0; blue, 0 }  ,fill opacity=1 ][line width=0.75]  [draw opacity=0] (8.93,-4.29) -- (0,0) -- (8.93,4.29) -- cycle    ;

%Straight Lines [id:da15970257558480383] 
\draw [color={rgb, 255:red, 0; green, 0; blue, 0 }  ,draw opacity=1 ][fill={rgb, 255:red, 74; green, 144; blue, 226 }  ,fill opacity=1 ]   (199.14,101.02) -- (188.26,89.98) ;
\draw [shift={(186.86,88.55)}, rotate = 405.43] [fill={rgb, 255:red, 0; green, 0; blue, 0 }  ,fill opacity=1 ][line width=0.75]  [draw opacity=0] (8.93,-4.29) -- (0,0) -- (8.93,4.29) -- cycle    ;

%Shape: Circle [id:dp752373829089108] 
\draw  [color={rgb, 255:red, 0; green, 0; blue, 0 }  ,draw opacity=1 ] (326.42,195.67) .. controls (324.19,193.4) and (324.22,189.76) .. (326.48,187.54) .. controls (328.75,185.31) and (332.39,185.33) .. (334.61,187.6) .. controls (336.84,189.86) and (336.82,193.5) .. (334.55,195.73) .. controls (332.29,197.96) and (328.65,197.93) .. (326.42,195.67) -- cycle ;
%Shape: Circle [id:dp8903709069996251] 
\draw  [fill={rgb, 255:red, 208; green, 2; blue, 27 }  ,fill opacity=1 ] (306.07,175.01) .. controls (303.84,172.75) and (303.87,169.1) .. (306.13,166.88) .. controls (308.39,164.65) and (312.03,164.67) .. (314.26,166.94) .. controls (316.49,169.2) and (316.46,172.84) .. (314.2,175.07) .. controls (311.94,177.3) and (308.3,177.27) .. (306.07,175.01) -- cycle ;
%Shape: Circle [id:dp08276655816901846] 
\draw  [color={rgb, 255:red, 0; green, 0; blue, 0 }  ,draw opacity=1 ] (284.32,152.92) .. controls (282.09,150.66) and (282.11,147.02) .. (284.38,144.79) .. controls (286.64,142.56) and (290.28,142.59) .. (292.51,144.85) .. controls (294.74,147.12) and (294.71,150.76) .. (292.45,152.98) .. controls (290.18,155.21) and (286.54,155.19) .. (284.32,152.92) -- cycle ;
%Shape: Circle [id:dp7602725273430508] 
\draw  [color={rgb, 255:red, 0; green, 0; blue, 0 }  ,draw opacity=1 ] (263.26,131.55) .. controls (261.03,129.29) and (261.06,125.65) .. (263.32,123.42) .. controls (265.59,121.19) and (269.23,121.22) .. (271.45,123.48) .. controls (273.68,125.74) and (273.66,129.38) .. (271.39,131.61) .. controls (269.13,133.84) and (265.49,133.81) .. (263.26,131.55) -- cycle ;
%Shape: Circle [id:dp4916951257307052] 
\draw  [color={rgb, 255:red, 0; green, 0; blue, 0 }  ,draw opacity=1 ] (240.8,108.75) .. controls (238.58,106.49) and (238.6,102.85) .. (240.87,100.62) .. controls (243.13,98.39) and (246.77,98.42) .. (249,100.68) .. controls (251.23,102.95) and (251.2,106.59) .. (248.94,108.82) .. controls (246.67,111.04) and (243.03,111.02) .. (240.8,108.75) -- cycle ;
%Straight Lines [id:da6691435669529446] 
\draw    (326.48,187.54) -- (315.61,176.49) ;
\draw [shift={(314.2,175.07)}, rotate = 405.43] [fill={rgb, 255:red, 0; green, 0; blue, 0 }  ][line width=0.75]  [draw opacity=0] (8.93,-4.29) -- (0,0) -- (8.93,4.29) -- cycle    ;

%Straight Lines [id:da757900962880528] 
\draw [color={rgb, 255:red, 0; green, 0; blue, 0 }  ,draw opacity=1 ]   (306.13,166.88) -- (293.85,154.41) ;
\draw [shift={(292.45,152.98)}, rotate = 405.43] [fill={rgb, 255:red, 0; green, 0; blue, 0 }  ,fill opacity=1 ][line width=0.75]  [draw opacity=0] (8.93,-4.29) -- (0,0) -- (8.93,4.29) -- cycle    ;

%Straight Lines [id:da03636328877130035] 
\draw [color={rgb, 255:red, 0; green, 0; blue, 0 }  ,draw opacity=1 ]   (283.67,144.08) -- (272.8,133.04) ;
\draw [shift={(271.39,131.61)}, rotate = 405.43] [fill={rgb, 255:red, 0; green, 0; blue, 0 }  ,fill opacity=1 ][line width=0.75]  [draw opacity=0] (8.93,-4.29) -- (0,0) -- (8.93,4.29) -- cycle    ;

%Straight Lines [id:da20617663839458578] 
\draw [color={rgb, 255:red, 0; green, 0; blue, 0 }  ,draw opacity=1 ]   (263.32,123.42) -- (250.34,110.24) ;
\draw [shift={(248.94,108.82)}, rotate = 405.43] [fill={rgb, 255:red, 0; green, 0; blue, 0 }  ,fill opacity=1 ][line width=0.75]  [draw opacity=0] (8.93,-4.29) -- (0,0) -- (8.93,4.29) -- cycle    ;

%Shape: Circle [id:dp6254258889586082] 
\draw   (221.15,88.81) .. controls (218.93,86.55) and (218.95,82.9) .. (221.22,80.68) .. controls (223.48,78.45) and (227.12,78.48) .. (229.35,80.74) .. controls (231.58,83) and (231.55,86.64) .. (229.29,88.87) .. controls (227.02,91.1) and (223.38,91.07) .. (221.15,88.81) -- cycle ;
%Straight Lines [id:da9218509723574893] 
\draw [color={rgb, 255:red, 0; green, 0; blue, 0 }  ,draw opacity=1 ]   (240.87,100.62) -- (230.69,90.29) ;
\draw [shift={(229.29,88.87)}, rotate = 405.43] [fill={rgb, 255:red, 0; green, 0; blue, 0 }  ,fill opacity=1 ][line width=0.75]  [draw opacity=0] (8.93,-4.29) -- (0,0) -- (8.93,4.29) -- cycle    ;

%Shape: Circle [id:dp2707132693449785] 
\draw  [color={rgb, 255:red, 0; green, 0; blue, 0 }  ,draw opacity=1 ] (347.79,174.61) .. controls (345.57,172.35) and (345.59,168.71) .. (347.85,166.48) .. controls (350.12,164.25) and (353.76,164.28) .. (355.99,166.54) .. controls (358.21,168.81) and (358.19,172.45) .. (355.93,174.67) .. controls (353.66,176.9) and (350.02,176.88) .. (347.79,174.61) -- cycle ;
%Shape: Circle [id:dp2008397937239026] 
\draw  [fill={rgb, 255:red, 208; green, 2; blue, 27 }  ,fill opacity=1 ] (327.44,153.95) .. controls (325.21,151.69) and (325.24,148.05) .. (327.5,145.82) .. controls (329.77,143.59) and (333.41,143.62) .. (335.63,145.88) .. controls (337.86,148.15) and (337.84,151.79) .. (335.57,154.02) .. controls (333.31,156.24) and (329.67,156.22) .. (327.44,153.95) -- cycle ;
%Shape: Circle [id:dp873274286120407] 
\draw  [fill={rgb, 255:red, 208; green, 2; blue, 27 }  ,fill opacity=1 ] (305.69,131.87) .. controls (303.46,129.61) and (303.49,125.97) .. (305.75,123.74) .. controls (308.01,121.51) and (311.65,121.54) .. (313.88,123.8) .. controls (316.11,126.06) and (316.08,129.7) .. (313.82,131.93) .. controls (311.56,134.16) and (307.92,134.13) .. (305.69,131.87) -- cycle ;
%Shape: Circle [id:dp696949764416567] 
\draw  [color={rgb, 255:red, 0; green, 0; blue, 0 }  ,draw opacity=1 ][fill={rgb, 255:red, 208; green, 2; blue, 27 }  ,fill opacity=1 ] (284.63,110.5) .. controls (282.4,108.24) and (282.43,104.6) .. (284.69,102.37) .. controls (286.96,100.14) and (290.6,100.17) .. (292.83,102.43) .. controls (295.05,104.69) and (295.03,108.33) .. (292.77,110.56) .. controls (290.5,112.79) and (286.86,112.76) .. (284.63,110.5) -- cycle ;
%Shape: Circle [id:dp6953096222611246] 
\draw  [fill={rgb, 255:red, 65; green, 117; blue, 5 }  ,fill opacity=1 ] (262.18,87.7) .. controls (259.95,85.44) and (259.98,81.8) .. (262.24,79.57) .. controls (264.5,77.34) and (268.14,77.37) .. (270.37,79.63) .. controls (272.6,81.89) and (272.57,85.53) .. (270.31,87.76) .. controls (268.05,89.99) and (264.41,89.96) .. (262.18,87.7) -- cycle ;
%Straight Lines [id:da21155836700235775] 
\draw [color={rgb, 255:red, 0; green, 0; blue, 0 }  ,draw opacity=1 ]   (347.85,166.48) -- (336.98,155.44) ;
\draw [shift={(335.57,154.02)}, rotate = 405.43] [fill={rgb, 255:red, 0; green, 0; blue, 0 }  ,fill opacity=1 ][line width=0.75]  [draw opacity=0] (8.93,-4.29) -- (0,0) -- (8.93,4.29) -- cycle    ;

%Straight Lines [id:da19120451395789462] 
\draw [color={rgb, 255:red, 208; green, 2; blue, 27 }  ,draw opacity=1 ][fill={rgb, 255:red, 208; green, 2; blue, 27 }  ,fill opacity=1 ]   (327.5,144.82) -- (315.22,132.36) ;
\draw [shift={(313.82,130.93)}, rotate = 405.43] [fill={rgb, 255:red, 208; green, 2; blue, 27 }  ,fill opacity=1 ][line width=0.75]  [draw opacity=0] (8.93,-4.29) -- (0,0) -- (8.93,4.29) -- cycle    ;

%Straight Lines [id:da1766545976921079] 
\draw [color={rgb, 255:red, 208; green, 2; blue, 27 }  ,draw opacity=1 ][fill={rgb, 255:red, 208; green, 2; blue, 27 }  ,fill opacity=1 ]   (305.05,122.03) -- (294.17,110.98) ;
\draw [shift={(292.77,109.56)}, rotate = 405.43] [fill={rgb, 255:red, 208; green, 2; blue, 27 }  ,fill opacity=1 ][line width=0.75]  [draw opacity=0] (8.93,-4.29) -- (0,0) -- (8.93,4.29) -- cycle    ;

%Straight Lines [id:da37651568542174485] 
\draw [color={rgb, 255:red, 0; green, 0; blue, 0 }  ,draw opacity=1 ][fill={rgb, 255:red, 208; green, 2; blue, 27 }  ,fill opacity=1 ]   (284.69,102.37) -- (271.71,89.19) ;
\draw [shift={(270.31,87.76)}, rotate = 405.43] [fill={rgb, 255:red, 0; green, 0; blue, 0 }  ,fill opacity=1 ][line width=0.75]  [draw opacity=0] (8.93,-4.29) -- (0,0) -- (8.93,4.29) -- cycle    ;

%Shape: Circle [id:dp9601381885931519] 
\draw   (369.88,152.86) .. controls (367.65,150.6) and (367.68,146.96) .. (369.94,144.73) .. controls (372.2,142.5) and (375.84,142.53) .. (378.07,144.79) .. controls (380.3,147.05) and (380.27,150.69) .. (378.01,152.92) .. controls (375.75,155.15) and (372.11,155.12) .. (369.88,152.86) -- cycle ;
%Shape: Circle [id:dp3742404427048862] 
\draw  [color={rgb, 255:red, 0; green, 0; blue, 0 }  ,draw opacity=1 ] (349.53,132.2) .. controls (347.3,129.94) and (347.33,126.3) .. (349.59,124.07) .. controls (351.85,121.84) and (355.49,121.87) .. (357.72,124.13) .. controls (359.95,126.39) and (359.92,130.03) .. (357.66,132.26) .. controls (355.4,134.49) and (351.75,134.46) .. (349.53,132.2) -- cycle ;
%Shape: Circle [id:dp4280787443629057] 
\draw  [color={rgb, 255:red, 0; green, 0; blue, 0 }  ,draw opacity=1 ] (327.77,110.11) .. controls (325.54,107.85) and (325.57,104.21) .. (327.83,101.98) .. controls (330.09,99.75) and (333.74,99.78) .. (335.96,102.04) .. controls (338.19,104.31) and (338.16,107.95) .. (335.9,110.18) .. controls (333.64,112.4) and (330,112.38) .. (327.77,110.11) -- cycle ;
%Shape: Circle [id:dp805599568592626] 
\draw   (306.72,88.74) .. controls (304.49,86.48) and (304.52,82.84) .. (306.78,80.61) .. controls (309.04,78.38) and (312.68,78.41) .. (314.91,80.67) .. controls (317.14,82.93) and (317.11,86.58) .. (314.85,88.8) .. controls (312.59,91.03) and (308.95,91.01) .. (306.72,88.74) -- cycle ;
%Straight Lines [id:da7389189104680325] 
\draw [color={rgb, 255:red, 0; green, 0; blue, 0 }  ,draw opacity=1 ]   (369.94,144.73) -- (359.06,133.68) ;
\draw [shift={(357.66,132.26)}, rotate = 405.43] [fill={rgb, 255:red, 0; green, 0; blue, 0 }  ,fill opacity=1 ][line width=0.75]  [draw opacity=0] (8.93,-4.29) -- (0,0) -- (8.93,4.29) -- cycle    ;

%Straight Lines [id:da9643633578506188] 
\draw [color={rgb, 255:red, 0; green, 0; blue, 0 }  ,draw opacity=1 ]   (349.59,124.07) -- (337.31,111.6) ;
\draw [shift={(335.9,110.18)}, rotate = 405.43] [fill={rgb, 255:red, 0; green, 0; blue, 0 }  ,fill opacity=1 ][line width=0.75]  [draw opacity=0] (8.93,-4.29) -- (0,0) -- (8.93,4.29) -- cycle    ;

%Straight Lines [id:da7634339969245241] 
\draw [color={rgb, 255:red, 0; green, 0; blue, 0 }  ,draw opacity=1 ]   (327.13,101.27) -- (316.25,90.23) ;
\draw [shift={(314.85,88.8)}, rotate = 405.43] [fill={rgb, 255:red, 0; green, 0; blue, 0 }  ,fill opacity=1 ][line width=0.75]  [draw opacity=0] (8.93,-4.29) -- (0,0) -- (8.93,4.29) -- cycle    ;

%Shape: Circle [id:dp22897188795601275] 
\draw  [color={rgb, 255:red, 0; green, 0; blue, 0 }  ,draw opacity=1 ] (391.25,131.8) .. controls (389.02,129.54) and (389.05,125.9) .. (391.31,123.67) .. controls (393.57,121.44) and (397.21,121.47) .. (399.44,123.73) .. controls (401.67,126) and (401.64,129.64) .. (399.38,131.87) .. controls (397.12,134.09) and (393.48,134.07) .. (391.25,131.8) -- cycle ;
%Shape: Circle [id:dp30173633511867726] 
\draw   (370.9,111.15) .. controls (368.67,108.88) and (368.7,105.24) .. (370.96,103.01) .. controls (373.22,100.79) and (376.86,100.81) .. (379.09,103.08) .. controls (381.32,105.34) and (381.29,108.98) .. (379.03,111.21) .. controls (376.77,113.44) and (373.13,113.41) .. (370.9,111.15) -- cycle ;
%Shape: Circle [id:dp3583877292040176] 
\draw   (349.14,89.06) .. controls (346.91,86.8) and (346.94,83.16) .. (349.2,80.93) .. controls (351.47,78.7) and (355.11,78.73) .. (357.34,80.99) .. controls (359.56,83.25) and (359.54,86.89) .. (357.27,89.12) .. controls (355.01,91.35) and (351.37,91.32) .. (349.14,89.06) -- cycle ;
%Straight Lines [id:da630984205412318] 
\draw [color={rgb, 255:red, 0; green, 0; blue, 0 }  ,draw opacity=1 ]   (391.31,123.67) -- (380.43,112.63) ;
\draw [shift={(379.03,111.21)}, rotate = 405.43] [fill={rgb, 255:red, 0; green, 0; blue, 0 }  ,fill opacity=1 ][line width=0.75]  [draw opacity=0] (8.93,-4.29) -- (0,0) -- (8.93,4.29) -- cycle    ;

%Straight Lines [id:da007257030222384886] 
\draw [color={rgb, 255:red, 0; green, 0; blue, 0 }  ,draw opacity=1 ]   (370.96,103.01) -- (358.68,90.55) ;
\draw [shift={(357.27,89.12)}, rotate = 405.43] [fill={rgb, 255:red, 0; green, 0; blue, 0 }  ,fill opacity=1 ][line width=0.75]  [draw opacity=0] (8.93,-4.29) -- (0,0) -- (8.93,4.29) -- cycle    ;

%Shape: Circle [id:dp3040734959801379] 
\draw  [color={rgb, 255:red, 0; green, 0; blue, 0 }  ,draw opacity=1 ] (412.62,110.75) .. controls (410.39,108.49) and (410.42,104.85) .. (412.68,102.62) .. controls (414.94,100.39) and (418.59,100.42) .. (420.81,102.68) .. controls (423.04,104.94) and (423.02,108.58) .. (420.75,110.81) .. controls (418.49,113.04) and (414.85,113.01) .. (412.62,110.75) -- cycle ;
%Shape: Circle [id:dp913444075201896] 
\draw   (392.27,90.09) .. controls (390.04,87.83) and (390.07,84.19) .. (392.33,81.96) .. controls (394.59,79.73) and (398.23,79.76) .. (400.46,82.02) .. controls (402.69,84.28) and (402.66,87.92) .. (400.4,90.15) .. controls (398.14,92.38) and (394.5,92.35) .. (392.27,90.09) -- cycle ;
%Straight Lines [id:da39168040825014816] 
\draw [color={rgb, 255:red, 0; green, 0; blue, 0 }  ,draw opacity=1 ]   (412.68,102.62) -- (401.81,91.58) ;
\draw [shift={(400.4,90.15)}, rotate = 405.43] [fill={rgb, 255:red, 0; green, 0; blue, 0 }  ,fill opacity=1 ][line width=0.75]  [draw opacity=0] (8.93,-4.29) -- (0,0) -- (8.93,4.29) -- cycle    ;

%Straight Lines [id:da7473620357418091] 
\draw [color={rgb, 255:red, 0; green, 0; blue, 0 }  ,draw opacity=1 ]   (164.84,102.17) -- (177.3,89.89) ;
\draw [shift={(178.73,88.49)}, rotate = 495.43] [fill={rgb, 255:red, 0; green, 0; blue, 0 }  ,fill opacity=1 ][line width=0.75]  [draw opacity=0] (8.93,-4.29) -- (0,0) -- (8.93,4.29) -- cycle    ;

%Straight Lines [id:da5840941936042909] 
\draw [color={rgb, 255:red, 0; green, 0; blue, 0 }  ,draw opacity=1 ]   (186.6,122.84) -- (199.07,110.56) ;
\draw [shift={(200.5,109.16)}, rotate = 495.43] [fill={rgb, 255:red, 0; green, 0; blue, 0 }  ,fill opacity=1 ][line width=0.75]  [draw opacity=0] (8.93,-4.29) -- (0,0) -- (8.93,4.29) -- cycle    ;

%Straight Lines [id:da41783619997227417] 
\draw [color={rgb, 255:red, 0; green, 0; blue, 0 }  ,draw opacity=1 ]   (208.69,101.09) -- (220.44,89.51) ;
\draw [shift={(221.87,88.11)}, rotate = 495.43] [fill={rgb, 255:red, 0; green, 0; blue, 0 }  ,fill opacity=1 ][line width=0.75]  [draw opacity=0] (8.93,-4.29) -- (0,0) -- (8.93,4.29) -- cycle    ;

%Straight Lines [id:da9271822822102496] 
\draw [color={rgb, 255:red, 0; green, 0; blue, 0 }  ,draw opacity=1 ]   (206.25,142.79) -- (218.72,130.51) ;
\draw [shift={(220.15,129.11)}, rotate = 495.43] [fill={rgb, 255:red, 0; green, 0; blue, 0 }  ,fill opacity=1 ][line width=0.75]  [draw opacity=0] (8.93,-4.29) -- (0,0) -- (8.93,4.29) -- cycle    ;

%Straight Lines [id:da5047439490426826] 
\draw [color={rgb, 255:red, 0; green, 0; blue, 0 }  ,draw opacity=1 ]   (228.34,121.04) -- (240.09,109.46) ;
\draw [shift={(241.52,108.05)}, rotate = 495.43] [fill={rgb, 255:red, 0; green, 0; blue, 0 }  ,fill opacity=1 ][line width=0.75]  [draw opacity=0] (8.93,-4.29) -- (0,0) -- (8.93,4.29) -- cycle    ;

%Straight Lines [id:da8140805643244295] 
\draw    (249.71,99.98) -- (261.46,88.4) ;
\draw [shift={(262.89,87)}, rotate = 495.43] [fill={rgb, 255:red, 0; green, 0; blue, 0 }  ][line width=0.75]  [draw opacity=0] (8.93,-4.29) -- (0,0) -- (8.93,4.29) -- cycle    ;

%Straight Lines [id:da1750528489941079] 
\draw [color={rgb, 255:red, 0; green, 0; blue, 0 }  ,draw opacity=1 ]   (228.02,163.46) -- (240.49,151.18) ;
\draw [shift={(241.91,149.78)}, rotate = 495.43] [fill={rgb, 255:red, 0; green, 0; blue, 0 }  ,fill opacity=1 ][line width=0.75]  [draw opacity=0] (8.93,-4.29) -- (0,0) -- (8.93,4.29) -- cycle    ;

%Straight Lines [id:da11575394484211077] 
\draw [color={rgb, 255:red, 0; green, 0; blue, 0 }  ,draw opacity=1 ]   (250.1,141.71) -- (261.86,130.13) ;
\draw [shift={(263.28,128.72)}, rotate = 495.43] [fill={rgb, 255:red, 0; green, 0; blue, 0 }  ,fill opacity=1 ][line width=0.75]  [draw opacity=0] (8.93,-4.29) -- (0,0) -- (8.93,4.29) -- cycle    ;

%Straight Lines [id:da7107829472652736] 
\draw    (271.48,120.65) -- (283.23,109.07) ;
\draw [shift={(284.65,107.67)}, rotate = 495.43] [fill={rgb, 255:red, 0; green, 0; blue, 0 }  ][line width=0.75]  [draw opacity=0] (8.93,-4.29) -- (0,0) -- (8.93,4.29) -- cycle    ;

%Straight Lines [id:da283754811134032] 
\draw [color={rgb, 255:red, 0; green, 0; blue, 0 }  ,draw opacity=1 ][fill={rgb, 255:red, 74; green, 144; blue, 226 }  ,fill opacity=1 ]   (292.83,102.43) -- (305.29,90.15) ;
\draw [shift={(306.72,88.74)}, rotate = 495.43] [fill={rgb, 255:red, 0; green, 0; blue, 0 }  ,fill opacity=1 ][line width=0.75]  [draw opacity=0] (8.93,-4.29) -- (0,0) -- (8.93,4.29) -- cycle    ;

%Straight Lines [id:da7281120199786721] 
\draw [color={rgb, 255:red, 0; green, 0; blue, 0 }  ,draw opacity=1 ]   (249.05,187.66) -- (261.52,175.38) ;
\draw [shift={(262.94,173.98)}, rotate = 495.43] [fill={rgb, 255:red, 0; green, 0; blue, 0 }  ,fill opacity=1 ][line width=0.75]  [draw opacity=0] (8.93,-4.29) -- (0,0) -- (8.93,4.29) -- cycle    ;

%Straight Lines [id:da5731981663629391] 
\draw [color={rgb, 255:red, 0; green, 0; blue, 0 }  ,draw opacity=1 ]   (271.14,165.91) -- (282.89,154.33) ;
\draw [shift={(284.32,152.92)}, rotate = 495.43] [fill={rgb, 255:red, 0; green, 0; blue, 0 }  ,fill opacity=1 ][line width=0.75]  [draw opacity=0] (8.93,-4.29) -- (0,0) -- (8.93,4.29) -- cycle    ;

%Straight Lines [id:da8764862915834009] 
\draw [color={rgb, 255:red, 0; green, 0; blue, 0 }  ,draw opacity=1 ]   (292.51,144.85) -- (304.26,133.27) ;
\draw [shift={(305.69,131.87)}, rotate = 495.43] [fill={rgb, 255:red, 0; green, 0; blue, 0 }  ,fill opacity=1 ][line width=0.75]  [draw opacity=0] (8.93,-4.29) -- (0,0) -- (8.93,4.29) -- cycle    ;

%Straight Lines [id:da8153304617830819] 
\draw [color={rgb, 255:red, 0; green, 0; blue, 0 }  ,draw opacity=1 ]   (313.88,123.8) -- (326.35,111.52) ;
\draw [shift={(327.77,110.11)}, rotate = 495.43] [fill={rgb, 255:red, 0; green, 0; blue, 0 }  ,fill opacity=1 ][line width=0.75]  [draw opacity=0] (8.93,-4.29) -- (0,0) -- (8.93,4.29) -- cycle    ;

%Straight Lines [id:da06741680279488471] 
\draw [color={rgb, 255:red, 0; green, 0; blue, 0 }  ,draw opacity=1 ]   (335.96,102.04) -- (347.72,90.46) ;
\draw [shift={(349.14,89.06)}, rotate = 495.43] [fill={rgb, 255:red, 0; green, 0; blue, 0 }  ,fill opacity=1 ][line width=0.75]  [draw opacity=0] (8.93,-4.29) -- (0,0) -- (8.93,4.29) -- cycle    ;

%Straight Lines [id:da690488444222989] 
\draw [color={rgb, 255:red, 208; green, 2; blue, 27 }  ,draw opacity=1 ][fill={rgb, 255:red, 208; green, 2; blue, 27 }  ,fill opacity=1 ]   (270.82,208.33) -- (283.28,196.05) ;
\draw [shift={(284.71,194.65)}, rotate = 495.43] [fill={rgb, 255:red, 208; green, 2; blue, 27 }  ,fill opacity=1 ][line width=0.75]  [draw opacity=0] (8.93,-4.29) -- (0,0) -- (8.93,4.29) -- cycle    ;

%Straight Lines [id:da32374759052143887] 
\draw [color={rgb, 255:red, 208; green, 2; blue, 27 }  ,draw opacity=1 ][fill={rgb, 255:red, 208; green, 2; blue, 27 }  ,fill opacity=1 ]   (292.9,186.58) -- (304.66,175) ;
\draw [shift={(306.08,173.59)}, rotate = 495.43] [fill={rgb, 255:red, 208; green, 2; blue, 27 }  ,fill opacity=1 ][line width=0.75]  [draw opacity=0] (8.93,-4.29) -- (0,0) -- (8.93,4.29) -- cycle    ;

%Straight Lines [id:da4875183845519382] 
\draw [color={rgb, 255:red, 208; green, 2; blue, 27 }  ,draw opacity=1 ][fill={rgb, 255:red, 208; green, 2; blue, 27 }  ,fill opacity=1 ]   (314.27,164.52) -- (326.03,152.94) ;
\draw [shift={(327.45,151.54)}, rotate = 495.43] [fill={rgb, 255:red, 208; green, 2; blue, 27 }  ,fill opacity=1 ][line width=0.75]  [draw opacity=0] (8.93,-4.29) -- (0,0) -- (8.93,4.29) -- cycle    ;

%Straight Lines [id:da1057743081965945] 
\draw [color={rgb, 255:red, 0; green, 0; blue, 0 }  ,draw opacity=1 ][fill={rgb, 255:red, 74; green, 144; blue, 226 }  ,fill opacity=1 ]   (335.65,144.47) -- (348.11,132.19) ;
\draw [shift={(349.54,130.78)}, rotate = 495.43] [fill={rgb, 255:red, 0; green, 0; blue, 0 }  ,fill opacity=1 ][line width=0.75]  [draw opacity=0] (8.93,-4.29) -- (0,0) -- (8.93,4.29) -- cycle    ;

%Straight Lines [id:da007013487661584161] 
\draw [color={rgb, 255:red, 0; green, 0; blue, 0 }  ,draw opacity=1 ][fill={rgb, 255:red, 74; green, 144; blue, 226 }  ,fill opacity=1 ]   (379.1,101.66) -- (390.86,90.08) ;
\draw [shift={(392.28,88.68)}, rotate = 495.43] [fill={rgb, 255:red, 0; green, 0; blue, 0 }  ,fill opacity=1 ][line width=0.75]  [draw opacity=0] (8.93,-4.29) -- (0,0) -- (8.93,4.29) -- cycle    ;

%Straight Lines [id:da21723569388553754] 
\draw [color={rgb, 255:red, 0; green, 0; blue, 0 }  ,draw opacity=1 ][fill={rgb, 255:red, 74; green, 144; blue, 226 }  ,fill opacity=1 ]   (357.72,124.13) -- (369.47,112.55) ;
\draw [shift={(370.9,111.15)}, rotate = 495.43] [fill={rgb, 255:red, 0; green, 0; blue, 0 }  ,fill opacity=1 ][line width=0.75]  [draw opacity=0] (8.93,-4.29) -- (0,0) -- (8.93,4.29) -- cycle    ;

%Straight Lines [id:da9750918071899783] 
\draw [color={rgb, 255:red, 0; green, 0; blue, 0 }  ,draw opacity=1 ]   (290.46,229.69) -- (302.92,217.41) ;
\draw [shift={(304.35,216.01)}, rotate = 495.43] [fill={rgb, 255:red, 0; green, 0; blue, 0 }  ,fill opacity=1 ][line width=0.75]  [draw opacity=0] (8.93,-4.29) -- (0,0) -- (8.93,4.29) -- cycle    ;

%Straight Lines [id:da5998700804093786] 
\draw [color={rgb, 255:red, 0; green, 0; blue, 0 }  ,draw opacity=1 ]   (312.54,207.94) -- (324.3,196.36) ;
\draw [shift={(325.72,194.95)}, rotate = 495.43] [fill={rgb, 255:red, 0; green, 0; blue, 0 }  ,fill opacity=1 ][line width=0.75]  [draw opacity=0] (8.93,-4.29) -- (0,0) -- (8.93,4.29) -- cycle    ;

%Straight Lines [id:da11640387669124341] 
\draw [color={rgb, 255:red, 0; green, 0; blue, 0 }  ,draw opacity=1 ]   (333.91,186.88) -- (345.67,175.3) ;
\draw [shift={(347.09,173.9)}, rotate = 495.43] [fill={rgb, 255:red, 0; green, 0; blue, 0 }  ,fill opacity=1 ][line width=0.75]  [draw opacity=0] (8.93,-4.29) -- (0,0) -- (8.93,4.29) -- cycle    ;

%Straight Lines [id:da03952319094197865] 
\draw [color={rgb, 255:red, 0; green, 0; blue, 0 }  ,draw opacity=1 ]   (355.28,165.83) -- (367.75,153.55) ;
\draw [shift={(369.18,152.15)}, rotate = 495.43] [fill={rgb, 255:red, 0; green, 0; blue, 0 }  ,fill opacity=1 ][line width=0.75]  [draw opacity=0] (8.93,-4.29) -- (0,0) -- (8.93,4.29) -- cycle    ;

%Straight Lines [id:da9853213932012455] 
\draw [color={rgb, 255:red, 0; green, 0; blue, 0 }  ,draw opacity=1 ]   (377.37,144.08) -- (389.12,132.5) ;
\draw [shift={(390.55,131.09)}, rotate = 495.43] [fill={rgb, 255:red, 0; green, 0; blue, 0 }  ,fill opacity=1 ][line width=0.75]  [draw opacity=0] (8.93,-4.29) -- (0,0) -- (8.93,4.29) -- cycle    ;

%Straight Lines [id:da16372202326943142] 
\draw [color={rgb, 255:red, 0; green, 0; blue, 0 }  ,draw opacity=1 ]   (398.74,123.02) -- (410.5,111.44) ;
\draw [shift={(411.92,110.04)}, rotate = 495.43] [fill={rgb, 255:red, 0; green, 0; blue, 0 }  ,fill opacity=1 ][line width=0.75]  [draw opacity=0] (8.93,-4.29) -- (0,0) -- (8.93,4.29) -- cycle    ;

%Shape: Circle [id:dp99154496170896] 
\draw   (137,90.3) .. controls (134.77,88.04) and (134.79,84.39) .. (137.06,82.17) .. controls (139.32,79.94) and (142.96,79.96) .. (145.19,82.23) .. controls (147.42,84.49) and (147.39,88.13) .. (145.13,90.36) .. controls (142.86,92.59) and (139.22,92.56) .. (137,90.3) -- cycle ;
%Straight Lines [id:da7649729985203606] 
\draw [color={rgb, 255:red, 0; green, 0; blue, 0 }  ,draw opacity=1 ]   (156.71,102.11) -- (146.53,91.78) ;
\draw [shift={(145.13,90.36)}, rotate = 405.43] [fill={rgb, 255:red, 0; green, 0; blue, 0 }  ,fill opacity=1 ][line width=0.75]  [draw opacity=0] (8.93,-4.29) -- (0,0) -- (8.93,4.29) -- cycle    ;

%Shape: Circle [id:dp09206038238617165] 
\draw   (433.99,89.7) .. controls (431.76,87.44) and (431.79,83.8) .. (434.05,81.57) .. controls (436.32,79.34) and (439.96,79.37) .. (442.19,81.63) .. controls (444.41,83.89) and (444.39,87.53) .. (442.12,89.76) .. controls (439.86,91.99) and (436.22,91.96) .. (433.99,89.7) -- cycle ;
%Straight Lines [id:da22995938344364752] 
\draw [color={rgb, 255:red, 0; green, 0; blue, 0 }  ,draw opacity=1 ]   (420.11,101.97) -- (431.87,90.39) ;
\draw [shift={(433.29,88.99)}, rotate = 495.43] [fill={rgb, 255:red, 0; green, 0; blue, 0 }  ,fill opacity=1 ][line width=0.75]  [draw opacity=0] (8.93,-4.29) -- (0,0) -- (8.93,4.29) -- cycle    ;

%Rounded Rect [id:dp8096780248638284] 
\draw  [color={rgb, 255:red, 2; green, 20; blue, 252 }  ,draw opacity=1 ][dash pattern={on 4.5pt off 4.5pt}] (138.94,75.51) .. controls (140.31,74.16) and (142.52,74.19) .. (143.86,75.57) -- (262.93,197.89) .. controls (264.27,199.26) and (264.24,201.47) .. (262.86,202.81) -- (255.38,210.1) .. controls (254,211.44) and (251.8,211.41) .. (250.46,210.03) -- (131.39,87.72) .. controls (130.04,86.34) and (130.07,84.13) .. (131.45,82.79) -- cycle ;
%Rounded Rect [id:dp8956391076932613] 
\draw  [color={rgb, 255:red, 2; green, 20; blue, 252 }  ,draw opacity=1 ][dash pattern={on 4.5pt off 4.5pt}] (178.26,73.18) .. controls (179.5,71.98) and (181.48,72) .. (182.68,73.24) -- (283.53,176.84) .. controls (284.74,178.08) and (284.71,180.06) .. (283.47,181.27) -- (276.74,187.82) .. controls (275.5,189.02) and (273.52,189) .. (272.32,187.76) -- (171.47,84.16) .. controls (170.26,82.92) and (170.29,80.94) .. (171.53,79.73) -- cycle ;
%Rounded Rect [id:dp40634587502842967] 
\draw  [color={rgb, 255:red, 2; green, 20; blue, 252 }  ,draw opacity=1 ][dash pattern={on 4.5pt off 4.5pt}] (222.2,74.24) .. controls (223.47,73) and (225.5,73.03) .. (226.74,74.3) -- (305.77,155.49) .. controls (307.01,156.76) and (306.98,158.79) .. (305.71,160.03) -- (298.8,166.76) .. controls (297.53,168) and (295.5,167.97) .. (294.26,166.7) -- (215.23,85.51) .. controls (213.99,84.24) and (214.02,82.21) .. (215.29,80.97) -- cycle ;
%Rounded Rect [id:dp7628571387465659] 
\draw  [color={rgb, 255:red, 2; green, 20; blue, 252 }  ,draw opacity=1 ][dash pattern={on 4.5pt off 4.5pt}] (446.93,83.92) .. controls (448.36,85.34) and (448.37,87.65) .. (446.95,89.07) -- (311.32,225.65) .. controls (309.9,227.08) and (307.59,227.09) .. (306.16,225.67) -- (298.4,217.96) .. controls (296.97,216.54) and (296.96,214.23) .. (298.38,212.8) -- (434,76.22) .. controls (435.42,74.79) and (437.73,74.79) .. (439.16,76.21) -- cycle ;
%Rounded Rect [id:dp004388936145785172] 
\draw  [color={rgb, 255:red, 2; green, 20; blue, 252 }  ,draw opacity=1 ][dash pattern={on 4.5pt off 4.5pt}] (406.29,83.38) .. controls (407.78,84.86) and (407.79,87.28) .. (406.3,88.78) -- (349.19,146.29) .. controls (347.7,147.79) and (345.29,147.8) .. (343.79,146.31) -- (335.66,138.24) .. controls (334.17,136.75) and (334.16,134.34) .. (335.65,132.84) -- (392.76,75.33) .. controls (394.25,73.83) and (396.66,73.82) .. (398.16,75.31) -- cycle ;
%Rounded Rect [id:dp08309422412730427] 
\draw  [color={rgb, 255:red, 2; green, 20; blue, 252 }  ,draw opacity=1 ][dash pattern={on 4.5pt off 4.5pt}] (364.58,82.01) .. controls (365.98,83.4) and (365.98,85.65) .. (364.6,87.04) -- (327.17,124.73) .. controls (325.79,126.12) and (323.54,126.13) .. (322.15,124.75) -- (314.58,117.24) .. controls (313.19,115.86) and (313.18,113.61) .. (314.57,112.21) -- (352,74.52) .. controls (353.38,73.13) and (355.63,73.12) .. (357.02,74.5) -- cycle ;
%Rounded Rect [id:dp13717654253656208] 
\draw  [color={rgb, 255:red, 2; green, 20; blue, 252 }  ,draw opacity=1 ][dash pattern={on 4.5pt off 4.5pt}] (322.05,82.47) .. controls (323.5,83.92) and (323.51,86.27) .. (322.06,87.73) -- (306.29,103.62) .. controls (304.84,105.07) and (302.49,105.08) .. (301.03,103.63) -- (293.12,95.78) .. controls (291.67,94.33) and (291.66,91.98) .. (293.1,90.53) -- (308.88,74.64) .. controls (310.33,73.18) and (312.68,73.17) .. (314.14,74.62) -- cycle ;
%Rounded Rect [id:dp9274689551196154] 
\draw  [color={rgb, 255:red, 2; green, 20; blue, 252 }  ,draw opacity=1 ][dash pattern={on 4.5pt off 4.5pt}] (263.27,73.22) .. controls (264.6,71.93) and (266.72,71.96) .. (268.01,73.28) -- (285.17,90.91) .. controls (286.46,92.23) and (286.43,94.36) .. (285.11,95.65) -- (277.9,102.66) .. controls (276.58,103.95) and (274.45,103.92) .. (273.16,102.6) -- (256.01,84.97) .. controls (254.71,83.65) and (254.74,81.53) .. (256.07,80.23) -- cycle ;

\end{tikzpicture}

%% file: fig3.tex
\tikzset{every picture/.style={line width=0.75pt}} %set default line width to 0.75pt        

\begin{tikzpicture}[x=0.75pt,y=0.75pt,yscale=-1,xscale=1]
%uncomment if require: \path (0,347); %set diagram left start at 0, and has height of 347

%Shape: Ellipse [id:dp09232514513701906] 
\draw   (13,160.25) .. controls (13,157.9) and (14.79,156) .. (17,156) .. controls (19.21,156) and (21,157.9) .. (21,160.25) .. controls (21,162.6) and (19.21,164.5) .. (17,164.5) .. controls (14.79,164.5) and (13,162.6) .. (13,160.25) -- cycle ;
%Shape: Ellipse [id:dp7621358394740794] 
\draw   (33,160.25) .. controls (33,157.9) and (34.79,156) .. (37,156) .. controls (39.21,156) and (41,157.9) .. (41,160.25) .. controls (41,162.6) and (39.21,164.5) .. (37,164.5) .. controls (34.79,164.5) and (33,162.6) .. (33,160.25) -- cycle ;
%Shape: Ellipse [id:dp2662872749256606] 
\draw   (54,160.25) .. controls (54,157.9) and (55.79,156) .. (58,156) .. controls (60.21,156) and (62,157.9) .. (62,160.25) .. controls (62,162.6) and (60.21,164.5) .. (58,164.5) .. controls (55.79,164.5) and (54,162.6) .. (54,160.25) -- cycle ;
%Shape: Ellipse [id:dp939672406643125] 
\draw   (74,160.25) .. controls (74,157.9) and (75.79,156) .. (78,156) .. controls (80.21,156) and (82,157.9) .. (82,160.25) .. controls (82,162.6) and (80.21,164.5) .. (78,164.5) .. controls (75.79,164.5) and (74,162.6) .. (74,160.25) -- cycle ;
%Shape: Ellipse [id:dp39481983413777755] 
\draw   (94,160.25) .. controls (94,157.9) and (95.79,156) .. (98,156) .. controls (100.21,156) and (102,157.9) .. (102,160.25) .. controls (102,162.6) and (100.21,164.5) .. (98,164.5) .. controls (95.79,164.5) and (94,162.6) .. (94,160.25) -- cycle ;
%Shape: Ellipse [id:dp2878095936923595] 
\draw   (114,160.25) .. controls (114,157.9) and (115.79,156) .. (118,156) .. controls (120.21,156) and (122,157.9) .. (122,160.25) .. controls (122,162.6) and (120.21,164.5) .. (118,164.5) .. controls (115.79,164.5) and (114,162.6) .. (114,160.25) -- cycle ;
%Shape: Ellipse [id:dp942224605726991] 
\draw   (135,160.25) .. controls (135,157.9) and (136.79,156) .. (139,156) .. controls (141.21,156) and (143,157.9) .. (143,160.25) .. controls (143,162.6) and (141.21,164.5) .. (139,164.5) .. controls (136.79,164.5) and (135,162.6) .. (135,160.25) -- cycle ;
%Shape: Ellipse [id:dp1065991687683776] 
\draw   (155,160.25) .. controls (155,157.9) and (156.79,156) .. (159,156) .. controls (161.21,156) and (163,157.9) .. (163,160.25) .. controls (163,162.6) and (161.21,164.5) .. (159,164.5) .. controls (156.79,164.5) and (155,162.6) .. (155,160.25) -- cycle ;
%Shape: Ellipse [id:dp500294019665525] 
\draw   (173,160.25) .. controls (173,157.9) and (174.79,156) .. (177,156) .. controls (179.21,156) and (181,157.9) .. (181,160.25) .. controls (181,162.6) and (179.21,164.5) .. (177,164.5) .. controls (174.79,164.5) and (173,162.6) .. (173,160.25) -- cycle ;
%Shape: Ellipse [id:dp2107351639635997] 
\draw   (193,160.25) .. controls (193,157.9) and (194.79,156) .. (197,156) .. controls (199.21,156) and (201,157.9) .. (201,160.25) .. controls (201,162.6) and (199.21,164.5) .. (197,164.5) .. controls (194.79,164.5) and (193,162.6) .. (193,160.25) -- cycle ;
%Shape: Ellipse [id:dp46758350838023643] 
\draw   (214,160.25) .. controls (214,157.9) and (215.79,156) .. (218,156) .. controls (220.21,156) and (222,157.9) .. (222,160.25) .. controls (222,162.6) and (220.21,164.5) .. (218,164.5) .. controls (215.79,164.5) and (214,162.6) .. (214,160.25) -- cycle ;
%Shape: Ellipse [id:dp1394476154596811] 
\draw   (234,160.25) .. controls (234,157.9) and (235.79,156) .. (238,156) .. controls (240.21,156) and (242,157.9) .. (242,160.25) .. controls (242,162.6) and (240.21,164.5) .. (238,164.5) .. controls (235.79,164.5) and (234,162.6) .. (234,160.25) -- cycle ;
%Shape: Ellipse [id:dp5887309027427874] 
\draw   (254,160.25) .. controls (254,157.9) and (255.79,156) .. (258,156) .. controls (260.21,156) and (262,157.9) .. (262,160.25) .. controls (262,162.6) and (260.21,164.5) .. (258,164.5) .. controls (255.79,164.5) and (254,162.6) .. (254,160.25) -- cycle ;
%Shape: Ellipse [id:dp5866375133200394] 
\draw   (274,160.25) .. controls (274,157.9) and (275.79,156) .. (278,156) .. controls (280.21,156) and (282,157.9) .. (282,160.25) .. controls (282,162.6) and (280.21,164.5) .. (278,164.5) .. controls (275.79,164.5) and (274,162.6) .. (274,160.25) -- cycle ;
%Shape: Ellipse [id:dp11735881798412295] 
\draw   (295,160.25) .. controls (295,157.9) and (296.79,156) .. (299,156) .. controls (301.21,156) and (303,157.9) .. (303,160.25) .. controls (303,162.6) and (301.21,164.5) .. (299,164.5) .. controls (296.79,164.5) and (295,162.6) .. (295,160.25) -- cycle ;
%Shape: Ellipse [id:dp9504864790752279] 
\draw   (315,160.25) .. controls (315,157.9) and (316.79,156) .. (319,156) .. controls (321.21,156) and (323,157.9) .. (323,160.25) .. controls (323,162.6) and (321.21,164.5) .. (319,164.5) .. controls (316.79,164.5) and (315,162.6) .. (315,160.25) -- cycle ;
%Shape: Ellipse [id:dp4316318645855215] 
\draw   (24,131.25) .. controls (24,128.9) and (25.79,127) .. (28,127) .. controls (30.21,127) and (32,128.9) .. (32,131.25) .. controls (32,133.6) and (30.21,135.5) .. (28,135.5) .. controls (25.79,135.5) and (24,133.6) .. (24,131.25) -- cycle ;
%Shape: Ellipse [id:dp710677761938862] 
\draw   (104,130.25) .. controls (104,127.9) and (105.79,126) .. (108,126) .. controls (110.21,126) and (112,127.9) .. (112,130.25) .. controls (112,132.6) and (110.21,134.5) .. (108,134.5) .. controls (105.79,134.5) and (104,132.6) .. (104,130.25) -- cycle ;
%Shape: Ellipse [id:dp33936602503180535] 
\draw   (65,131.25) .. controls (65,128.9) and (66.79,127) .. (69,127) .. controls (71.21,127) and (73,128.9) .. (73,131.25) .. controls (73,133.6) and (71.21,135.5) .. (69,135.5) .. controls (66.79,135.5) and (65,133.6) .. (65,131.25) -- cycle ;
%Shape: Ellipse [id:dp890825036453228] 
\draw   (142,129.25) .. controls (142,126.9) and (143.79,125) .. (146,125) .. controls (148.21,125) and (150,126.9) .. (150,129.25) .. controls (150,131.6) and (148.21,133.5) .. (146,133.5) .. controls (143.79,133.5) and (142,131.6) .. (142,129.25) -- cycle ;
%Shape: Ellipse [id:dp5136821500152011] 
\draw   (184,129.25) .. controls (184,126.9) and (185.79,125) .. (188,125) .. controls (190.21,125) and (192,126.9) .. (192,129.25) .. controls (192,131.6) and (190.21,133.5) .. (188,133.5) .. controls (185.79,133.5) and (184,131.6) .. (184,129.25) -- cycle ;
%Shape: Ellipse [id:dp2987867798379529] 
\draw   (263,128.25) .. controls (263,125.9) and (264.79,124) .. (267,124) .. controls (269.21,124) and (271,125.9) .. (271,128.25) .. controls (271,130.6) and (269.21,132.5) .. (267,132.5) .. controls (264.79,132.5) and (263,130.6) .. (263,128.25) -- cycle ;
%Shape: Ellipse [id:dp4064874239012295] 
\draw   (225,129.25) .. controls (225,126.9) and (226.79,125) .. (229,125) .. controls (231.21,125) and (233,126.9) .. (233,129.25) .. controls (233,131.6) and (231.21,133.5) .. (229,133.5) .. controls (226.79,133.5) and (225,131.6) .. (225,129.25) -- cycle ;
%Shape: Ellipse [id:dp5216435948254377] 
\draw   (302,128.25) .. controls (302,125.9) and (303.79,124) .. (306,124) .. controls (308.21,124) and (310,125.9) .. (310,128.25) .. controls (310,130.6) and (308.21,132.5) .. (306,132.5) .. controls (303.79,132.5) and (302,130.6) .. (302,128.25) -- cycle ;
%Shape: Ellipse [id:dp930585266904479] 
\draw   (43,100.25) .. controls (43,97.9) and (44.79,96) .. (47,96) .. controls (49.21,96) and (51,97.9) .. (51,100.25) .. controls (51,102.6) and (49.21,104.5) .. (47,104.5) .. controls (44.79,104.5) and (43,102.6) .. (43,100.25) -- cycle ;
%Shape: Ellipse [id:dp9713329761758267] 
\draw   (123,99.25) .. controls (123,96.9) and (124.79,95) .. (127,95) .. controls (129.21,95) and (131,96.9) .. (131,99.25) .. controls (131,101.6) and (129.21,103.5) .. (127,103.5) .. controls (124.79,103.5) and (123,101.6) .. (123,99.25) -- cycle ;
%Shape: Ellipse [id:dp18047216191364246] 
\draw   (205,99.25) .. controls (205,96.9) and (206.79,95) .. (209,95) .. controls (211.21,95) and (213,96.9) .. (213,99.25) .. controls (213,101.6) and (211.21,103.5) .. (209,103.5) .. controls (206.79,103.5) and (205,101.6) .. (205,99.25) -- cycle ;
%Shape: Ellipse [id:dp9452450458808552] 
\draw   (282,97.25) .. controls (282,94.9) and (283.79,93) .. (286,93) .. controls (288.21,93) and (290,94.9) .. (290,97.25) .. controls (290,99.6) and (288.21,101.5) .. (286,101.5) .. controls (283.79,101.5) and (282,99.6) .. (282,97.25) -- cycle ;
%Shape: Ellipse [id:dp6344947207956657] 
\draw   (83,71.25) .. controls (83,68.9) and (84.79,67) .. (87,67) .. controls (89.21,67) and (91,68.9) .. (91,71.25) .. controls (91,73.6) and (89.21,75.5) .. (87,75.5) .. controls (84.79,75.5) and (83,73.6) .. (83,71.25) -- cycle ;
%Shape: Ellipse [id:dp29936863445424966] 
\draw   (243,69.25) .. controls (243,66.9) and (244.79,65) .. (247,65) .. controls (249.21,65) and (251,66.9) .. (251,69.25) .. controls (251,71.6) and (249.21,73.5) .. (247,73.5) .. controls (244.79,73.5) and (243,71.6) .. (243,69.25) -- cycle ;
%Shape: Ellipse [id:dp8616773124591064] 
\draw   (165,30.25) .. controls (165,27.9) and (166.79,26) .. (169,26) .. controls (171.21,26) and (173,27.9) .. (173,30.25) .. controls (173,32.6) and (171.21,34.5) .. (169,34.5) .. controls (166.79,34.5) and (165,32.6) .. (165,30.25) -- cycle ;
%Shape: Ellipse [id:dp1899518522357737] 
\draw   (13,219.25) .. controls (13,216.9) and (14.79,215) .. (17,215) .. controls (19.21,215) and (21,216.9) .. (21,219.25) .. controls (21,221.6) and (19.21,223.5) .. (17,223.5) .. controls (14.79,223.5) and (13,221.6) .. (13,219.25) -- cycle ;
%Shape: Ellipse [id:dp6460044725512804] 
\draw   (33,219.25) .. controls (33,216.9) and (34.79,215) .. (37,215) .. controls (39.21,215) and (41,216.9) .. (41,219.25) .. controls (41,221.6) and (39.21,223.5) .. (37,223.5) .. controls (34.79,223.5) and (33,221.6) .. (33,219.25) -- cycle ;
%Shape: Ellipse [id:dp7808323289005565] 
\draw   (54,219.25) .. controls (54,216.9) and (55.79,215) .. (58,215) .. controls (60.21,215) and (62,216.9) .. (62,219.25) .. controls (62,221.6) and (60.21,223.5) .. (58,223.5) .. controls (55.79,223.5) and (54,221.6) .. (54,219.25) -- cycle ;
%Shape: Ellipse [id:dp40411985197075695] 
\draw   (74,219.25) .. controls (74,216.9) and (75.79,215) .. (78,215) .. controls (80.21,215) and (82,216.9) .. (82,219.25) .. controls (82,221.6) and (80.21,223.5) .. (78,223.5) .. controls (75.79,223.5) and (74,221.6) .. (74,219.25) -- cycle ;
%Shape: Ellipse [id:dp00381226856781125] 
\draw   (94,219.25) .. controls (94,216.9) and (95.79,215) .. (98,215) .. controls (100.21,215) and (102,216.9) .. (102,219.25) .. controls (102,221.6) and (100.21,223.5) .. (98,223.5) .. controls (95.79,223.5) and (94,221.6) .. (94,219.25) -- cycle ;
%Shape: Ellipse [id:dp5810834110241345] 
\draw   (114,219.25) .. controls (114,216.9) and (115.79,215) .. (118,215) .. controls (120.21,215) and (122,216.9) .. (122,219.25) .. controls (122,221.6) and (120.21,223.5) .. (118,223.5) .. controls (115.79,223.5) and (114,221.6) .. (114,219.25) -- cycle ;
%Shape: Ellipse [id:dp27980838138678066] 
\draw   (135,219.25) .. controls (135,216.9) and (136.79,215) .. (139,215) .. controls (141.21,215) and (143,216.9) .. (143,219.25) .. controls (143,221.6) and (141.21,223.5) .. (139,223.5) .. controls (136.79,223.5) and (135,221.6) .. (135,219.25) -- cycle ;
%Shape: Ellipse [id:dp3440608435563699] 
\draw   (155,219.25) .. controls (155,216.9) and (156.79,215) .. (159,215) .. controls (161.21,215) and (163,216.9) .. (163,219.25) .. controls (163,221.6) and (161.21,223.5) .. (159,223.5) .. controls (156.79,223.5) and (155,221.6) .. (155,219.25) -- cycle ;
%Shape: Ellipse [id:dp06914771256287633] 
\draw   (173,219.25) .. controls (173,216.9) and (174.79,215) .. (177,215) .. controls (179.21,215) and (181,216.9) .. (181,219.25) .. controls (181,221.6) and (179.21,223.5) .. (177,223.5) .. controls (174.79,223.5) and (173,221.6) .. (173,219.25) -- cycle ;
%Shape: Ellipse [id:dp37924947249427654] 
\draw   (193,219.25) .. controls (193,216.9) and (194.79,215) .. (197,215) .. controls (199.21,215) and (201,216.9) .. (201,219.25) .. controls (201,221.6) and (199.21,223.5) .. (197,223.5) .. controls (194.79,223.5) and (193,221.6) .. (193,219.25) -- cycle ;
%Shape: Ellipse [id:dp10656659858454143] 
\draw   (214,219.25) .. controls (214,216.9) and (215.79,215) .. (218,215) .. controls (220.21,215) and (222,216.9) .. (222,219.25) .. controls (222,221.6) and (220.21,223.5) .. (218,223.5) .. controls (215.79,223.5) and (214,221.6) .. (214,219.25) -- cycle ;
%Shape: Ellipse [id:dp22165039318432744] 
\draw   (234,219.25) .. controls (234,216.9) and (235.79,215) .. (238,215) .. controls (240.21,215) and (242,216.9) .. (242,219.25) .. controls (242,221.6) and (240.21,223.5) .. (238,223.5) .. controls (235.79,223.5) and (234,221.6) .. (234,219.25) -- cycle ;
%Shape: Ellipse [id:dp3616241809987837] 
\draw   (254,219.25) .. controls (254,216.9) and (255.79,215) .. (258,215) .. controls (260.21,215) and (262,216.9) .. (262,219.25) .. controls (262,221.6) and (260.21,223.5) .. (258,223.5) .. controls (255.79,223.5) and (254,221.6) .. (254,219.25) -- cycle ;
%Shape: Ellipse [id:dp3706554448254311] 
\draw   (274,219.25) .. controls (274,216.9) and (275.79,215) .. (278,215) .. controls (280.21,215) and (282,216.9) .. (282,219.25) .. controls (282,221.6) and (280.21,223.5) .. (278,223.5) .. controls (275.79,223.5) and (274,221.6) .. (274,219.25) -- cycle ;
%Shape: Ellipse [id:dp515513838051608] 
\draw   (295,219.25) .. controls (295,216.9) and (296.79,215) .. (299,215) .. controls (301.21,215) and (303,216.9) .. (303,219.25) .. controls (303,221.6) and (301.21,223.5) .. (299,223.5) .. controls (296.79,223.5) and (295,221.6) .. (295,219.25) -- cycle ;
%Shape: Ellipse [id:dp3871106794143815] 
\draw   (315,219.25) .. controls (315,216.9) and (316.79,215) .. (319,215) .. controls (321.21,215) and (323,216.9) .. (323,219.25) .. controls (323,221.6) and (321.21,223.5) .. (319,223.5) .. controls (316.79,223.5) and (315,221.6) .. (315,219.25) -- cycle ;
%Shape: Ellipse [id:dp9387739974519476] 
\draw   (334,219.25) .. controls (334,216.9) and (335.79,215) .. (338,215) .. controls (340.21,215) and (342,216.9) .. (342,219.25) .. controls (342,221.6) and (340.21,223.5) .. (338,223.5) .. controls (335.79,223.5) and (334,221.6) .. (334,219.25) -- cycle ;
%Shape: Ellipse [id:dp8613796605677091] 
\draw   (354,219.25) .. controls (354,216.9) and (355.79,215) .. (358,215) .. controls (360.21,215) and (362,216.9) .. (362,219.25) .. controls (362,221.6) and (360.21,223.5) .. (358,223.5) .. controls (355.79,223.5) and (354,221.6) .. (354,219.25) -- cycle ;
%Shape: Ellipse [id:dp3495208572898749] 
\draw   (396,219.25) .. controls (396,216.9) and (397.79,215) .. (400,215) .. controls (402.21,215) and (404,216.9) .. (404,219.25) .. controls (404,221.6) and (402.21,223.5) .. (400,223.5) .. controls (397.79,223.5) and (396,221.6) .. (396,219.25) -- cycle ;
%Shape: Ellipse [id:dp05850524223049214] 
\draw   (416,219.25) .. controls (416,216.9) and (417.79,215) .. (420,215) .. controls (422.21,215) and (424,216.9) .. (424,219.25) .. controls (424,221.6) and (422.21,223.5) .. (420,223.5) .. controls (417.79,223.5) and (416,221.6) .. (416,219.25) -- cycle ;
%Shape: Ellipse [id:dp4057728277413186] 
\draw   (436,219.25) .. controls (436,216.9) and (437.79,215) .. (440,215) .. controls (442.21,215) and (444,216.9) .. (444,219.25) .. controls (444,221.6) and (442.21,223.5) .. (440,223.5) .. controls (437.79,223.5) and (436,221.6) .. (436,219.25) -- cycle ;
%Shape: Ellipse [id:dp5977574512873787] 
\draw   (456,219.25) .. controls (456,216.9) and (457.79,215) .. (460,215) .. controls (462.21,215) and (464,216.9) .. (464,219.25) .. controls (464,221.6) and (462.21,223.5) .. (460,223.5) .. controls (457.79,223.5) and (456,221.6) .. (456,219.25) -- cycle ;
%Shape: Ellipse [id:dp08318006045951876] 
\draw   (477,219.25) .. controls (477,216.9) and (478.79,215) .. (481,215) .. controls (483.21,215) and (485,216.9) .. (485,219.25) .. controls (485,221.6) and (483.21,223.5) .. (481,223.5) .. controls (478.79,223.5) and (477,221.6) .. (477,219.25) -- cycle ;
%Shape: Ellipse [id:dp40946700603030384] 
\draw   (497,219.25) .. controls (497,216.9) and (498.79,215) .. (501,215) .. controls (503.21,215) and (505,216.9) .. (505,219.25) .. controls (505,221.6) and (503.21,223.5) .. (501,223.5) .. controls (498.79,223.5) and (497,221.6) .. (497,219.25) -- cycle ;
%Shape: Ellipse [id:dp6577323143797846] 
\draw   (515,219.25) .. controls (515,216.9) and (516.79,215) .. (519,215) .. controls (521.21,215) and (523,216.9) .. (523,219.25) .. controls (523,221.6) and (521.21,223.5) .. (519,223.5) .. controls (516.79,223.5) and (515,221.6) .. (515,219.25) -- cycle ;
%Shape: Ellipse [id:dp40236695469153183] 
\draw   (535,219.25) .. controls (535,216.9) and (536.79,215) .. (539,215) .. controls (541.21,215) and (543,216.9) .. (543,219.25) .. controls (543,221.6) and (541.21,223.5) .. (539,223.5) .. controls (536.79,223.5) and (535,221.6) .. (535,219.25) -- cycle ;
%Shape: Ellipse [id:dp07823703969109119] 
\draw   (556,219.25) .. controls (556,216.9) and (557.79,215) .. (560,215) .. controls (562.21,215) and (564,216.9) .. (564,219.25) .. controls (564,221.6) and (562.21,223.5) .. (560,223.5) .. controls (557.79,223.5) and (556,221.6) .. (556,219.25) -- cycle ;
%Shape: Ellipse [id:dp28881230713421524] 
\draw   (576,219.25) .. controls (576,216.9) and (577.79,215) .. (580,215) .. controls (582.21,215) and (584,216.9) .. (584,219.25) .. controls (584,221.6) and (582.21,223.5) .. (580,223.5) .. controls (577.79,223.5) and (576,221.6) .. (576,219.25) -- cycle ;
%Shape: Ellipse [id:dp4147677681300712] 
\draw   (596,219.25) .. controls (596,216.9) and (597.79,215) .. (600,215) .. controls (602.21,215) and (604,216.9) .. (604,219.25) .. controls (604,221.6) and (602.21,223.5) .. (600,223.5) .. controls (597.79,223.5) and (596,221.6) .. (596,219.25) -- cycle ;
%Shape: Ellipse [id:dp45632020713098087] 
\draw   (616,219.25) .. controls (616,216.9) and (617.79,215) .. (620,215) .. controls (622.21,215) and (624,216.9) .. (624,219.25) .. controls (624,221.6) and (622.21,223.5) .. (620,223.5) .. controls (617.79,223.5) and (616,221.6) .. (616,219.25) -- cycle ;
%Shape: Ellipse [id:dp7570321407175278] 
\draw   (637,219.25) .. controls (637,216.9) and (638.79,215) .. (641,215) .. controls (643.21,215) and (645,216.9) .. (645,219.25) .. controls (645,221.6) and (643.21,223.5) .. (641,223.5) .. controls (638.79,223.5) and (637,221.6) .. (637,219.25) -- cycle ;
%Straight Lines [id:da6598935266723598] 
\draw    (17,167.5) -- (17,215) ;
\draw [shift={(17,164.5)}, rotate = 90] [fill={rgb, 255:red, 0; green, 0; blue, 0 }  ][line width=0.08]  [draw opacity=0] (10.72,-5.15) -- (0,0) -- (10.72,5.15) -- (7.12,0) -- cycle    ;
%Straight Lines [id:da39404035266115844] 
\draw    (37,167.5) -- (37,215) ;
\draw [shift={(37,164.5)}, rotate = 90] [fill={rgb, 255:red, 0; green, 0; blue, 0 }  ][line width=0.08]  [draw opacity=0] (10.72,-5.15) -- (0,0) -- (10.72,5.15) -- (7.12,0) -- cycle    ;
%Straight Lines [id:da9671566764135258] 
\draw    (58,167.5) -- (58,215) ;
\draw [shift={(58,164.5)}, rotate = 90] [fill={rgb, 255:red, 0; green, 0; blue, 0 }  ][line width=0.08]  [draw opacity=0] (10.72,-5.15) -- (0,0) -- (10.72,5.15) -- (7.12,0) -- cycle    ;
%Straight Lines [id:da24108668370140163] 
\draw    (78,167.5) -- (78,215) ;
\draw [shift={(78,164.5)}, rotate = 90] [fill={rgb, 255:red, 0; green, 0; blue, 0 }  ][line width=0.08]  [draw opacity=0] (10.72,-5.15) -- (0,0) -- (10.72,5.15) -- (7.12,0) -- cycle    ;
%Straight Lines [id:da6384449495545459] 
\draw    (99,167.5) -- (99,215) ;
\draw [shift={(99,164.5)}, rotate = 90] [fill={rgb, 255:red, 0; green, 0; blue, 0 }  ][line width=0.08]  [draw opacity=0] (10.72,-5.15) -- (0,0) -- (10.72,5.15) -- (7.12,0) -- cycle    ;
%Straight Lines [id:da3010593447043943] 
\draw    (119,167.5) -- (119,215) ;
\draw [shift={(119,164.5)}, rotate = 90] [fill={rgb, 255:red, 0; green, 0; blue, 0 }  ][line width=0.08]  [draw opacity=0] (10.72,-5.15) -- (0,0) -- (10.72,5.15) -- (7.12,0) -- cycle    ;
%Straight Lines [id:da3866879402750476] 
\draw    (140,167.5) -- (140,215) ;
\draw [shift={(140,164.5)}, rotate = 90] [fill={rgb, 255:red, 0; green, 0; blue, 0 }  ][line width=0.08]  [draw opacity=0] (10.72,-5.15) -- (0,0) -- (10.72,5.15) -- (7.12,0) -- cycle    ;
%Straight Lines [id:da09563675130712013] 
\draw    (160,167.5) -- (160,215) ;
\draw [shift={(160,164.5)}, rotate = 90] [fill={rgb, 255:red, 0; green, 0; blue, 0 }  ][line width=0.08]  [draw opacity=0] (10.72,-5.15) -- (0,0) -- (10.72,5.15) -- (7.12,0) -- cycle    ;
%Straight Lines [id:da3538931757773187] 
\draw    (176,168.5) -- (176,216) ;
\draw [shift={(176,165.5)}, rotate = 90] [fill={rgb, 255:red, 0; green, 0; blue, 0 }  ][line width=0.08]  [draw opacity=0] (10.72,-5.15) -- (0,0) -- (10.72,5.15) -- (7.12,0) -- cycle    ;
%Straight Lines [id:da6327334507839775] 
\draw    (196,168.5) -- (196,216) ;
\draw [shift={(196,165.5)}, rotate = 90] [fill={rgb, 255:red, 0; green, 0; blue, 0 }  ][line width=0.08]  [draw opacity=0] (10.72,-5.15) -- (0,0) -- (10.72,5.15) -- (7.12,0) -- cycle    ;
%Straight Lines [id:da47281605117390124] 
\draw    (217,168.5) -- (217,216) ;
\draw [shift={(217,165.5)}, rotate = 90] [fill={rgb, 255:red, 0; green, 0; blue, 0 }  ][line width=0.08]  [draw opacity=0] (10.72,-5.15) -- (0,0) -- (10.72,5.15) -- (7.12,0) -- cycle    ;
%Straight Lines [id:da2832012849669745] 
\draw    (237,168.5) -- (237,216) ;
\draw [shift={(237,165.5)}, rotate = 90] [fill={rgb, 255:red, 0; green, 0; blue, 0 }  ][line width=0.08]  [draw opacity=0] (10.72,-5.15) -- (0,0) -- (10.72,5.15) -- (7.12,0) -- cycle    ;
%Straight Lines [id:da5349551111127224] 
\draw    (258,168.5) -- (258,216) ;
\draw [shift={(258,165.5)}, rotate = 90] [fill={rgb, 255:red, 0; green, 0; blue, 0 }  ][line width=0.08]  [draw opacity=0] (10.72,-5.15) -- (0,0) -- (10.72,5.15) -- (7.12,0) -- cycle    ;
%Straight Lines [id:da1312019681120371] 
\draw    (278,168.5) -- (278,216) ;
\draw [shift={(278,165.5)}, rotate = 90] [fill={rgb, 255:red, 0; green, 0; blue, 0 }  ][line width=0.08]  [draw opacity=0] (10.72,-5.15) -- (0,0) -- (10.72,5.15) -- (7.12,0) -- cycle    ;
%Straight Lines [id:da9425931387399851] 
\draw    (299,168.5) -- (299,216) ;
\draw [shift={(299,165.5)}, rotate = 90] [fill={rgb, 255:red, 0; green, 0; blue, 0 }  ][line width=0.08]  [draw opacity=0] (10.72,-5.15) -- (0,0) -- (10.72,5.15) -- (7.12,0) -- cycle    ;
%Straight Lines [id:da08608624656776453] 
\draw    (319,168.5) -- (319,216) ;
\draw [shift={(319,165.5)}, rotate = 90] [fill={rgb, 255:red, 0; green, 0; blue, 0 }  ][line width=0.08]  [draw opacity=0] (10.72,-5.15) -- (0,0) -- (10.72,5.15) -- (7.12,0) -- cycle    ;
%Curve Lines [id:da08964204122960573] 
\draw    (641,215) .. controls (561.62,32.72) and (196.04,-17.02) .. (170.16,23.41) ;
\draw [shift={(169,26)}, rotate = 284.93] [fill={rgb, 255:red, 0; green, 0; blue, 0 }  ][line width=0.08]  [draw opacity=0] (10.72,-5.15) -- (0,0) -- (10.72,5.15) -- (7.12,0) -- cycle    ;
%Curve Lines [id:da2779127694051524] 
\draw    (620,215) .. controls (488.68,15.08) and (287.25,38.45) .. (249.13,63.47) ;
\draw [shift={(247,65)}, rotate = 321.45] [fill={rgb, 255:red, 0; green, 0; blue, 0 }  ][line width=0.08]  [draw opacity=0] (10.72,-5.15) -- (0,0) -- (10.72,5.15) -- (7.12,0) -- cycle    ;
%Curve Lines [id:da10622171075190323] 
\draw    (600,215) .. controls (516.28,-25.34) and (144.39,24.66) .. (89.29,65.16) ;
\draw [shift={(87,67)}, rotate = 318.46] [fill={rgb, 255:red, 0; green, 0; blue, 0 }  ][line width=0.08]  [draw opacity=0] (10.72,-5.15) -- (0,0) -- (10.72,5.15) -- (7.12,0) -- cycle    ;
%Curve Lines [id:da7786920793812522] 
\draw    (580,215) .. controls (595.92,182.66) and (423.74,-14.02) .. (288.04,91.39) ;
\draw [shift={(286,93)}, rotate = 321.29] [fill={rgb, 255:red, 0; green, 0; blue, 0 }  ][line width=0.08]  [draw opacity=0] (10.72,-5.15) -- (0,0) -- (10.72,5.15) -- (7.12,0) -- cycle    ;
%Curve Lines [id:da04434714801617212] 
\draw    (560,215) .. controls (461.99,-10.72) and (278.71,69.11) .. (211.02,94.26) ;
\draw [shift={(209,95)}, rotate = 339.83] [fill={rgb, 255:red, 0; green, 0; blue, 0 }  ][line width=0.08]  [draw opacity=0] (10.72,-5.15) -- (0,0) -- (10.72,5.15) -- (7.12,0) -- cycle    ;
%Curve Lines [id:da17460046740770951] 
\draw    (539,215) .. controls (440.49,-8.88) and (242,79.36) .. (128.7,94.77) ;
\draw [shift={(127,95)}, rotate = 352.56] [fill={rgb, 255:red, 0; green, 0; blue, 0 }  ][line width=0.08]  [draw opacity=0] (10.72,-5.15) -- (0,0) -- (10.72,5.15) -- (7.12,0) -- cycle    ;
%Curve Lines [id:da5706196678747444] 
\draw    (519,215) .. controls (381.69,25.95) and (162.21,80.69) .. (48.7,95.78) ;
\draw [shift={(47,96)}, rotate = 352.56] [fill={rgb, 255:red, 0; green, 0; blue, 0 }  ][line width=0.08]  [draw opacity=0] (10.72,-5.15) -- (0,0) -- (10.72,5.15) -- (7.12,0) -- cycle    ;
%Curve Lines [id:da23366314427354107] 
\draw    (481,215) .. controls (520.6,185.3) and (366.14,79.15) .. (307.74,122.64) ;
\draw [shift={(306,124)}, rotate = 320.49] [fill={rgb, 255:red, 0; green, 0; blue, 0 }  ][line width=0.08]  [draw opacity=0] (10.72,-5.15) -- (0,0) -- (10.72,5.15) -- (7.12,0) -- cycle    ;
%Curve Lines [id:da7787764587424726] 
\draw    (460,215) .. controls (499.6,185.3) and (330.44,66.41) .. (268.83,122.25) ;
\draw [shift={(267,124)}, rotate = 315] [fill={rgb, 255:red, 0; green, 0; blue, 0 }  ][line width=0.08]  [draw opacity=0] (10.72,-5.15) -- (0,0) -- (10.72,5.15) -- (7.12,0) -- cycle    ;
%Curve Lines [id:da18556437783608715] 
\draw    (440,215) .. controls (479.6,185.3) and (264.37,58.57) .. (233.87,127.11) ;
\draw [shift={(233,129.25)}, rotate = 290.23] [fill={rgb, 255:red, 0; green, 0; blue, 0 }  ][line width=0.08]  [draw opacity=0] (10.72,-5.15) -- (0,0) -- (10.72,5.15) -- (7.12,0) -- cycle    ;
%Curve Lines [id:da7694926467407672] 
\draw    (420,215) .. controls (459.8,185.15) and (261.99,39.47) .. (193.03,127.9) ;
\draw [shift={(192,129.25)}, rotate = 306.69] [fill={rgb, 255:red, 0; green, 0; blue, 0 }  ][line width=0.08]  [draw opacity=0] (10.72,-5.15) -- (0,0) -- (10.72,5.15) -- (7.12,0) -- cycle    ;
%Straight Lines [id:da8455801247834533] 
\draw    (30,219.25) -- (21,219.25) ;
\draw [shift={(33,219.25)}, rotate = 180] [fill={rgb, 255:red, 0; green, 0; blue, 0 }  ][line width=0.08]  [draw opacity=0] (10.72,-5.15) -- (0,0) -- (10.72,5.15) -- (7.12,0) -- cycle    ;
%Straight Lines [id:da21062459488524676] 
\draw    (51,219.25) -- (42,219.25) ;
\draw [shift={(54,219.25)}, rotate = 180] [fill={rgb, 255:red, 0; green, 0; blue, 0 }  ][line width=0.08]  [draw opacity=0] (10.72,-5.15) -- (0,0) -- (10.72,5.15) -- (7.12,0) -- cycle    ;
%Straight Lines [id:da5520598147108031] 
\draw    (72,219.25) -- (63,219.25) ;
\draw [shift={(75,219.25)}, rotate = 180] [fill={rgb, 255:red, 0; green, 0; blue, 0 }  ][line width=0.08]  [draw opacity=0] (10.72,-5.15) -- (0,0) -- (10.72,5.15) -- (7.12,0) -- cycle    ;
%Straight Lines [id:da9765529936549371] 
\draw    (93,219.25) -- (84,219.25) ;
\draw [shift={(96,219.25)}, rotate = 180] [fill={rgb, 255:red, 0; green, 0; blue, 0 }  ][line width=0.08]  [draw opacity=0] (10.72,-5.15) -- (0,0) -- (10.72,5.15) -- (7.12,0) -- cycle    ;
%Shape: Ellipse [id:dp13133696708948506] 
\draw   (156,219.25) .. controls (156,216.9) and (157.79,215) .. (160,215) .. controls (162.21,215) and (164,216.9) .. (164,219.25) .. controls (164,221.6) and (162.21,223.5) .. (160,223.5) .. controls (157.79,223.5) and (156,221.6) .. (156,219.25) -- cycle ;
%Straight Lines [id:da47129018640812115] 
\draw    (112,219.25) -- (103,219.25) ;
\draw [shift={(115,219.25)}, rotate = 180] [fill={rgb, 255:red, 0; green, 0; blue, 0 }  ][line width=0.08]  [draw opacity=0] (10.72,-5.15) -- (0,0) -- (10.72,5.15) -- (7.12,0) -- cycle    ;
%Straight Lines [id:da7203928089342118] 
\draw    (133,219.25) -- (124,219.25) ;
\draw [shift={(136,219.25)}, rotate = 180] [fill={rgb, 255:red, 0; green, 0; blue, 0 }  ][line width=0.08]  [draw opacity=0] (10.72,-5.15) -- (0,0) -- (10.72,5.15) -- (7.12,0) -- cycle    ;
%Straight Lines [id:da7227477402344549] 
\draw    (173,219.25) -- (164,219.25) ;
\draw [shift={(176,219.25)}, rotate = 180] [fill={rgb, 255:red, 0; green, 0; blue, 0 }  ][line width=0.08]  [draw opacity=0] (10.72,-5.15) -- (0,0) -- (10.72,5.15) -- (7.12,0) -- cycle    ;
%Straight Lines [id:da47013073843148834] 
\draw    (152,219.25) -- (143,219.25) ;
\draw [shift={(155,219.25)}, rotate = 180] [fill={rgb, 255:red, 0; green, 0; blue, 0 }  ][line width=0.08]  [draw opacity=0] (10.72,-5.15) -- (0,0) -- (10.72,5.15) -- (7.12,0) -- cycle    ;
%Shape: Ellipse [id:dp8933317837918862] 
\draw   (254,219.25) .. controls (254,216.9) and (255.79,215) .. (258,215) .. controls (260.21,215) and (262,216.9) .. (262,219.25) .. controls (262,221.6) and (260.21,223.5) .. (258,223.5) .. controls (255.79,223.5) and (254,221.6) .. (254,219.25) -- cycle ;
%Straight Lines [id:da8299891094529555] 
\draw    (190,219.25) -- (181,219.25) ;
\draw [shift={(193,219.25)}, rotate = 180] [fill={rgb, 255:red, 0; green, 0; blue, 0 }  ][line width=0.08]  [draw opacity=0] (10.72,-5.15) -- (0,0) -- (10.72,5.15) -- (7.12,0) -- cycle    ;
%Straight Lines [id:da8528182406900935] 
\draw    (211,219.25) -- (202,219.25) ;
\draw [shift={(214,219.25)}, rotate = 180] [fill={rgb, 255:red, 0; green, 0; blue, 0 }  ][line width=0.08]  [draw opacity=0] (10.72,-5.15) -- (0,0) -- (10.72,5.15) -- (7.12,0) -- cycle    ;
%Straight Lines [id:da7948126519635388] 
\draw    (232,219.25) -- (223,219.25) ;
\draw [shift={(235,219.25)}, rotate = 180] [fill={rgb, 255:red, 0; green, 0; blue, 0 }  ][line width=0.08]  [draw opacity=0] (10.72,-5.15) -- (0,0) -- (10.72,5.15) -- (7.12,0) -- cycle    ;
%Straight Lines [id:da3001420935552819] 
\draw    (272,219.25) -- (263,219.25) ;
\draw [shift={(275,219.25)}, rotate = 180] [fill={rgb, 255:red, 0; green, 0; blue, 0 }  ][line width=0.08]  [draw opacity=0] (10.72,-5.15) -- (0,0) -- (10.72,5.15) -- (7.12,0) -- cycle    ;
%Straight Lines [id:da2666024882028597] 
\draw    (293,219.25) -- (284,219.25) ;
\draw [shift={(296,219.25)}, rotate = 180] [fill={rgb, 255:red, 0; green, 0; blue, 0 }  ][line width=0.08]  [draw opacity=0] (10.72,-5.15) -- (0,0) -- (10.72,5.15) -- (7.12,0) -- cycle    ;
%Straight Lines [id:da11139679889782372] 
\draw    (333,219.25) -- (324,219.25) ;
\draw [shift={(336,219.25)}, rotate = 180] [fill={rgb, 255:red, 0; green, 0; blue, 0 }  ][line width=0.08]  [draw opacity=0] (10.72,-5.15) -- (0,0) -- (10.72,5.15) -- (7.12,0) -- cycle    ;
%Straight Lines [id:da9391032050729293] 
\draw    (312,219.25) -- (303,219.25) ;
\draw [shift={(315,219.25)}, rotate = 180] [fill={rgb, 255:red, 0; green, 0; blue, 0 }  ][line width=0.08]  [draw opacity=0] (10.72,-5.15) -- (0,0) -- (10.72,5.15) -- (7.12,0) -- cycle    ;
%Shape: Ellipse [id:dp969328651568582] 
\draw   (436,219.25) .. controls (436,216.9) and (437.79,215) .. (440,215) .. controls (442.21,215) and (444,216.9) .. (444,219.25) .. controls (444,221.6) and (442.21,223.5) .. (440,223.5) .. controls (437.79,223.5) and (436,221.6) .. (436,219.25) -- cycle ;
%Straight Lines [id:da6785834608515142] 
\draw    (351,219.25) -- (342,219.25) ;
\draw [shift={(354,219.25)}, rotate = 180] [fill={rgb, 255:red, 0; green, 0; blue, 0 }  ][line width=0.08]  [draw opacity=0] (10.72,-5.15) -- (0,0) -- (10.72,5.15) -- (7.12,0) -- cycle    ;
%Straight Lines [id:da3702297263462495] 
\draw    (393,219.25) -- (384,219.25) ;
\draw [shift={(396,219.25)}, rotate = 180] [fill={rgb, 255:red, 0; green, 0; blue, 0 }  ][line width=0.08]  [draw opacity=0] (10.72,-5.15) -- (0,0) -- (10.72,5.15) -- (7.12,0) -- cycle    ;
%Straight Lines [id:da014834635302684251] 
\draw    (414,219.25) -- (405,219.25) ;
\draw [shift={(417,219.25)}, rotate = 180] [fill={rgb, 255:red, 0; green, 0; blue, 0 }  ][line width=0.08]  [draw opacity=0] (10.72,-5.15) -- (0,0) -- (10.72,5.15) -- (7.12,0) -- cycle    ;
%Straight Lines [id:da7936312659295932] 
\draw    (454,219.25) -- (445,219.25) ;
\draw [shift={(457,219.25)}, rotate = 180] [fill={rgb, 255:red, 0; green, 0; blue, 0 }  ][line width=0.08]  [draw opacity=0] (10.72,-5.15) -- (0,0) -- (10.72,5.15) -- (7.12,0) -- cycle    ;
%Straight Lines [id:da0904910414489506] 
\draw    (475,219.25) -- (466,219.25) ;
\draw [shift={(478,219.25)}, rotate = 180] [fill={rgb, 255:red, 0; green, 0; blue, 0 }  ][line width=0.08]  [draw opacity=0] (10.72,-5.15) -- (0,0) -- (10.72,5.15) -- (7.12,0) -- cycle    ;
%Straight Lines [id:da30794949216335743] 
\draw    (515,219.25) -- (506,219.25) ;
\draw [shift={(518,219.25)}, rotate = 180] [fill={rgb, 255:red, 0; green, 0; blue, 0 }  ][line width=0.08]  [draw opacity=0] (10.72,-5.15) -- (0,0) -- (10.72,5.15) -- (7.12,0) -- cycle    ;
%Straight Lines [id:da5944230266970574] 
\draw    (494,219.25) -- (485,219.25) ;
\draw [shift={(497,219.25)}, rotate = 180] [fill={rgb, 255:red, 0; green, 0; blue, 0 }  ][line width=0.08]  [draw opacity=0] (10.72,-5.15) -- (0,0) -- (10.72,5.15) -- (7.12,0) -- cycle    ;
%Shape: Ellipse [id:dp15638697098130194] 
\draw   (596,219.25) .. controls (596,216.9) and (597.79,215) .. (600,215) .. controls (602.21,215) and (604,216.9) .. (604,219.25) .. controls (604,221.6) and (602.21,223.5) .. (600,223.5) .. controls (597.79,223.5) and (596,221.6) .. (596,219.25) -- cycle ;
%Straight Lines [id:da14767319984334848] 
\draw    (532,219.25) -- (523,219.25) ;
\draw [shift={(535,219.25)}, rotate = 180] [fill={rgb, 255:red, 0; green, 0; blue, 0 }  ][line width=0.08]  [draw opacity=0] (10.72,-5.15) -- (0,0) -- (10.72,5.15) -- (7.12,0) -- cycle    ;
%Straight Lines [id:da6208186834863845] 
\draw    (553,219.25) -- (544,219.25) ;
\draw [shift={(556,219.25)}, rotate = 180] [fill={rgb, 255:red, 0; green, 0; blue, 0 }  ][line width=0.08]  [draw opacity=0] (10.72,-5.15) -- (0,0) -- (10.72,5.15) -- (7.12,0) -- cycle    ;
%Straight Lines [id:da6122074311605128] 
\draw    (574,219.25) -- (565,219.25) ;
\draw [shift={(577,219.25)}, rotate = 180] [fill={rgb, 255:red, 0; green, 0; blue, 0 }  ][line width=0.08]  [draw opacity=0] (10.72,-5.15) -- (0,0) -- (10.72,5.15) -- (7.12,0) -- cycle    ;
%Straight Lines [id:da4059603671108871] 
\draw    (614,219.25) -- (605,219.25) ;
\draw [shift={(617,219.25)}, rotate = 180] [fill={rgb, 255:red, 0; green, 0; blue, 0 }  ][line width=0.08]  [draw opacity=0] (10.72,-5.15) -- (0,0) -- (10.72,5.15) -- (7.12,0) -- cycle    ;
%Straight Lines [id:da68375146025402] 
\draw    (635,219.25) -- (626,219.25) ;
\draw [shift={(638,219.25)}, rotate = 180] [fill={rgb, 255:red, 0; green, 0; blue, 0 }  ][line width=0.08]  [draw opacity=0] (10.72,-5.15) -- (0,0) -- (10.72,5.15) -- (7.12,0) -- cycle    ;
%Straight Lines [id:da38420458637857435] 
\draw    (593,219.25) -- (584,219.25) ;
\draw [shift={(596,219.25)}, rotate = 180] [fill={rgb, 255:red, 0; green, 0; blue, 0 }  ][line width=0.08]  [draw opacity=0] (10.72,-5.15) -- (0,0) -- (10.72,5.15) -- (7.12,0) -- cycle    ;
%Straight Lines [id:da5941683451135498] 
\draw    (433,219.25) -- (424,219.25) ;
\draw [shift={(436,219.25)}, rotate = 180] [fill={rgb, 255:red, 0; green, 0; blue, 0 }  ][line width=0.08]  [draw opacity=0] (10.72,-5.15) -- (0,0) -- (10.72,5.15) -- (7.12,0) -- cycle    ;
%Straight Lines [id:da3161136706712544] 
\draw    (251,219.25) -- (242,219.25) ;
\draw [shift={(254,219.25)}, rotate = 180] [fill={rgb, 255:red, 0; green, 0; blue, 0 }  ][line width=0.08]  [draw opacity=0] (10.72,-5.15) -- (0,0) -- (10.72,5.15) -- (7.12,0) -- cycle    ;
%Straight Lines [id:da7467383704459072] 
\draw    (18.42,153.36) -- (28,135.5) ;
\draw [shift={(17,156)}, rotate = 298.22] [fill={rgb, 255:red, 0; green, 0; blue, 0 }  ][line width=0.08]  [draw opacity=0] (10.72,-5.15) -- (0,0) -- (10.72,5.15) -- (7.12,0) -- cycle    ;
%Straight Lines [id:da6220262576744529] 
\draw    (35.79,153.25) -- (28,135.5) ;
\draw [shift={(37,156)}, rotate = 246.3] [fill={rgb, 255:red, 0; green, 0; blue, 0 }  ][line width=0.08]  [draw opacity=0] (10.72,-5.15) -- (0,0) -- (10.72,5.15) -- (7.12,0) -- cycle    ;
%Straight Lines [id:da28244762511405197] 
\draw    (59.42,152.36) -- (69,134.5) ;
\draw [shift={(58,155)}, rotate = 298.22] [fill={rgb, 255:red, 0; green, 0; blue, 0 }  ][line width=0.08]  [draw opacity=0] (10.72,-5.15) -- (0,0) -- (10.72,5.15) -- (7.12,0) -- cycle    ;
%Straight Lines [id:da7498308906940656] 
\draw    (76.79,152.25) -- (69,134.5) ;
\draw [shift={(78,155)}, rotate = 246.3] [fill={rgb, 255:red, 0; green, 0; blue, 0 }  ][line width=0.08]  [draw opacity=0] (10.72,-5.15) -- (0,0) -- (10.72,5.15) -- (7.12,0) -- cycle    ;
%Straight Lines [id:da09901260532579781] 
\draw    (98.42,153.36) -- (108,135.5) ;
\draw [shift={(97,156)}, rotate = 298.22] [fill={rgb, 255:red, 0; green, 0; blue, 0 }  ][line width=0.08]  [draw opacity=0] (10.72,-5.15) -- (0,0) -- (10.72,5.15) -- (7.12,0) -- cycle    ;
%Straight Lines [id:da5689429301026236] 
\draw    (115.79,153.25) -- (108,135.5) ;
\draw [shift={(117,156)}, rotate = 246.3] [fill={rgb, 255:red, 0; green, 0; blue, 0 }  ][line width=0.08]  [draw opacity=0] (10.72,-5.15) -- (0,0) -- (10.72,5.15) -- (7.12,0) -- cycle    ;
%Straight Lines [id:da542900920854833] 
\draw    (139.05,152.19) -- (146,133.5) ;
\draw [shift={(138,155)}, rotate = 290.41] [fill={rgb, 255:red, 0; green, 0; blue, 0 }  ][line width=0.08]  [draw opacity=0] (10.72,-5.15) -- (0,0) -- (10.72,5.15) -- (7.12,0) -- cycle    ;
%Straight Lines [id:da15944328324436796] 
\draw    (157.5,153.4) -- (146,133.5) ;
\draw [shift={(159,156)}, rotate = 239.98] [fill={rgb, 255:red, 0; green, 0; blue, 0 }  ][line width=0.08]  [draw opacity=0] (10.72,-5.15) -- (0,0) -- (10.72,5.15) -- (7.12,0) -- cycle    ;
%Straight Lines [id:da7568971384000329] 
\draw    (178.42,152.36) -- (188,134.5) ;
\draw [shift={(177,155)}, rotate = 298.22] [fill={rgb, 255:red, 0; green, 0; blue, 0 }  ][line width=0.08]  [draw opacity=0] (10.72,-5.15) -- (0,0) -- (10.72,5.15) -- (7.12,0) -- cycle    ;
%Straight Lines [id:da5651716560717468] 
\draw    (195.79,152.25) -- (188,134.5) ;
\draw [shift={(197,155)}, rotate = 246.3] [fill={rgb, 255:red, 0; green, 0; blue, 0 }  ][line width=0.08]  [draw opacity=0] (10.72,-5.15) -- (0,0) -- (10.72,5.15) -- (7.12,0) -- cycle    ;
%Straight Lines [id:da3992964286944043] 
\draw    (219.42,151.36) -- (229,133.5) ;
\draw [shift={(218,154)}, rotate = 298.22] [fill={rgb, 255:red, 0; green, 0; blue, 0 }  ][line width=0.08]  [draw opacity=0] (10.72,-5.15) -- (0,0) -- (10.72,5.15) -- (7.12,0) -- cycle    ;
%Straight Lines [id:da37111743241014916] 
\draw    (258.42,152.36) -- (268,134.5) ;
\draw [shift={(257,155)}, rotate = 298.22] [fill={rgb, 255:red, 0; green, 0; blue, 0 }  ][line width=0.08]  [draw opacity=0] (10.72,-5.15) -- (0,0) -- (10.72,5.15) -- (7.12,0) -- cycle    ;
%Straight Lines [id:da7820451190013977] 
\draw    (275.79,152.25) -- (268,134.5) ;
\draw [shift={(277,155)}, rotate = 246.3] [fill={rgb, 255:red, 0; green, 0; blue, 0 }  ][line width=0.08]  [draw opacity=0] (10.72,-5.15) -- (0,0) -- (10.72,5.15) -- (7.12,0) -- cycle    ;
%Straight Lines [id:da9823447970251191] 
\draw    (299.05,151.19) -- (306,132.5) ;
\draw [shift={(298,154)}, rotate = 290.41] [fill={rgb, 255:red, 0; green, 0; blue, 0 }  ][line width=0.08]  [draw opacity=0] (10.72,-5.15) -- (0,0) -- (10.72,5.15) -- (7.12,0) -- cycle    ;
%Straight Lines [id:da8820138957914381] 
\draw    (313.79,150.25) -- (306,132.5) ;
\draw [shift={(315,153)}, rotate = 246.3] [fill={rgb, 255:red, 0; green, 0; blue, 0 }  ][line width=0.08]  [draw opacity=0] (10.72,-5.15) -- (0,0) -- (10.72,5.15) -- (7.12,0) -- cycle    ;
%Straight Lines [id:da26695363172305964] 
\draw    (236.79,151.25) -- (229,133.5) ;
\draw [shift={(238,154)}, rotate = 246.3] [fill={rgb, 255:red, 0; green, 0; blue, 0 }  ][line width=0.08]  [draw opacity=0] (10.72,-5.15) -- (0,0) -- (10.72,5.15) -- (7.12,0) -- cycle    ;
%Straight Lines [id:da28961419301195757] 
\draw    (29.94,124.71) -- (47,104.5) ;
\draw [shift={(28,127)}, rotate = 310.18] [fill={rgb, 255:red, 0; green, 0; blue, 0 }  ][line width=0.08]  [draw opacity=0] (10.72,-5.15) -- (0,0) -- (10.72,5.15) -- (7.12,0) -- cycle    ;
%Straight Lines [id:da2455777156521617] 
\draw    (66.9,124.85) -- (47,104.5) ;
\draw [shift={(69,127)}, rotate = 225.64] [fill={rgb, 255:red, 0; green, 0; blue, 0 }  ][line width=0.08]  [draw opacity=0] (10.72,-5.15) -- (0,0) -- (10.72,5.15) -- (7.12,0) -- cycle    ;
%Straight Lines [id:da984445039646096] 
\draw    (108.94,121.71) -- (126,101.5) ;
\draw [shift={(107,124)}, rotate = 310.18] [fill={rgb, 255:red, 0; green, 0; blue, 0 }  ][line width=0.08]  [draw opacity=0] (10.72,-5.15) -- (0,0) -- (10.72,5.15) -- (7.12,0) -- cycle    ;
%Straight Lines [id:da5046307164987143] 
\draw    (145.9,121.85) -- (126,101.5) ;
\draw [shift={(148,124)}, rotate = 225.64] [fill={rgb, 255:red, 0; green, 0; blue, 0 }  ][line width=0.08]  [draw opacity=0] (10.72,-5.15) -- (0,0) -- (10.72,5.15) -- (7.12,0) -- cycle    ;
%Straight Lines [id:da05535810580501099] 
\draw    (190.94,121.71) -- (208,101.5) ;
\draw [shift={(189,124)}, rotate = 310.18] [fill={rgb, 255:red, 0; green, 0; blue, 0 }  ][line width=0.08]  [draw opacity=0] (10.72,-5.15) -- (0,0) -- (10.72,5.15) -- (7.12,0) -- cycle    ;
%Straight Lines [id:da8825079719113462] 
\draw    (227.9,121.85) -- (208,101.5) ;
\draw [shift={(230,124)}, rotate = 225.64] [fill={rgb, 255:red, 0; green, 0; blue, 0 }  ][line width=0.08]  [draw opacity=0] (10.72,-5.15) -- (0,0) -- (10.72,5.15) -- (7.12,0) -- cycle    ;
%Straight Lines [id:da7645656077604437] 
\draw    (268.94,121.71) -- (286,101.5) ;
\draw [shift={(267,124)}, rotate = 310.18] [fill={rgb, 255:red, 0; green, 0; blue, 0 }  ][line width=0.08]  [draw opacity=0] (10.72,-5.15) -- (0,0) -- (10.72,5.15) -- (7.12,0) -- cycle    ;
%Straight Lines [id:da35160666347006075] 
\draw    (304.01,121.76) -- (286,101.5) ;
\draw [shift={(306,124)}, rotate = 228.37] [fill={rgb, 255:red, 0; green, 0; blue, 0 }  ][line width=0.08]  [draw opacity=0] (10.72,-5.15) -- (0,0) -- (10.72,5.15) -- (7.12,0) -- cycle    ;
%Straight Lines [id:da8336037755347332] 
\draw    (49.67,94.63) -- (87,75.5) ;
\draw [shift={(47,96)}, rotate = 332.86] [fill={rgb, 255:red, 0; green, 0; blue, 0 }  ][line width=0.08]  [draw opacity=0] (10.72,-5.15) -- (0,0) -- (10.72,5.15) -- (7.12,0) -- cycle    ;
%Straight Lines [id:da5319584642052453] 
\draw    (124.3,93.69) -- (87,75.5) ;
\draw [shift={(127,95)}, rotate = 205.99] [fill={rgb, 255:red, 0; green, 0; blue, 0 }  ][line width=0.08]  [draw opacity=0] (10.72,-5.15) -- (0,0) -- (10.72,5.15) -- (7.12,0) -- cycle    ;
%Straight Lines [id:da03300003206094515] 
\draw    (211.61,93.52) -- (247,73.5) ;
\draw [shift={(209,95)}, rotate = 330.5] [fill={rgb, 255:red, 0; green, 0; blue, 0 }  ][line width=0.08]  [draw opacity=0] (10.72,-5.15) -- (0,0) -- (10.72,5.15) -- (7.12,0) -- cycle    ;
%Straight Lines [id:da45773457989441635] 
\draw    (283.41,91.48) -- (251,72.5) ;
\draw [shift={(286,93)}, rotate = 210.36] [fill={rgb, 255:red, 0; green, 0; blue, 0 }  ][line width=0.08]  [draw opacity=0] (10.72,-5.15) -- (0,0) -- (10.72,5.15) -- (7.12,0) -- cycle    ;
%Straight Lines [id:da3919725534593237] 
\draw    (89.71,65.72) -- (165,30.25) ;
\draw [shift={(87,67)}, rotate = 334.77] [fill={rgb, 255:red, 0; green, 0; blue, 0 }  ][line width=0.08]  [draw opacity=0] (10.72,-5.15) -- (0,0) -- (10.72,5.15) -- (7.12,0) -- cycle    ;
%Straight Lines [id:da646363938253359] 
\draw    (244.28,63.72) -- (173,30.25) ;
\draw [shift={(247,65)}, rotate = 205.15] [fill={rgb, 255:red, 0; green, 0; blue, 0 }  ][line width=0.08]  [draw opacity=0] (10.72,-5.15) -- (0,0) -- (10.72,5.15) -- (7.12,0) -- cycle    ;
%Shape: Ellipse [id:dp01748356472289747] 
\draw   (376,219.25) .. controls (376,216.9) and (377.79,215) .. (380,215) .. controls (382.21,215) and (384,216.9) .. (384,219.25) .. controls (384,221.6) and (382.21,223.5) .. (380,223.5) .. controls (377.79,223.5) and (376,221.6) .. (376,219.25) -- cycle ;
%Straight Lines [id:da10841497534524835] 
\draw    (373,219.25) -- (362,219.25) ;
\draw [shift={(376,219.25)}, rotate = 180] [fill={rgb, 255:red, 0; green, 0; blue, 0 }  ][line width=0.08]  [draw opacity=0] (10.72,-5.15) -- (0,0) -- (10.72,5.15) -- (7.12,0) -- cycle    ;
%Curve Lines [id:da8148923884802934] 
\draw    (400,223.5) .. controls (439.8,193.65) and (182.59,238.55) .. (148.5,125.71) ;
\draw [shift={(148,124)}, rotate = 74.45] [fill={rgb, 255:red, 0; green, 0; blue, 0 }  ][line width=0.08]  [draw opacity=0] (10.72,-5.15) -- (0,0) -- (10.72,5.15) -- (7.12,0) -- cycle    ;
%Curve Lines [id:da10539922940149049] 
\draw    (380,215) .. controls (419.8,185.15) and (146.75,244.65) .. (112.5,131.96) ;
\draw [shift={(112,130.25)}, rotate = 74.45] [fill={rgb, 255:red, 0; green, 0; blue, 0 }  ][line width=0.08]  [draw opacity=0] (10.72,-5.15) -- (0,0) -- (10.72,5.15) -- (7.12,0) -- cycle    ;
%Curve Lines [id:da4141907246953045] 
\draw    (358,215) .. controls (397.8,185.15) and (107.92,245.64) .. (73.5,132.96) ;
\draw [shift={(73,131.25)}, rotate = 74.45] [fill={rgb, 255:red, 0; green, 0; blue, 0 }  ][line width=0.08]  [draw opacity=0] (10.72,-5.15) -- (0,0) -- (10.72,5.15) -- (7.12,0) -- cycle    ;
%Curve Lines [id:da09839539684502308] 
\draw    (338,215) .. controls (377.8,185.15) and (67.13,245.64) .. (32.5,132.96) ;
\draw [shift={(32,131.25)}, rotate = 74.45] [fill={rgb, 255:red, 0; green, 0; blue, 0 }  ][line width=0.08]  [draw opacity=0] (10.72,-5.15) -- (0,0) -- (10.72,5.15) -- (7.12,0) -- cycle    ;

\end{tikzpicture}

%% file: diamond.tex
\subsection{\io{} of Diamond DAG}
In this section, we present an example application of our new lower bound technique to obtain asymptotically tight \io{} lower bounds for  ``\emph{Diamond DAGs}'' which can be obtained by taking a $b\times b$ mesh (\myie{} a two-dimensional array), by directing all the edges towards the upper right corner. The graph obtained as such has $n=b^2$ vertices, a single input vertex (\myie{} in the bottom left corner), and a single output vertex (\myie{} in the upper right corner). An example for $b=9$ is given in Figure~\ref{fig:brokend}. 

Besides its independent interest, this example is meant to showcase the advantage of our technique compared to that of Hong and Kung and the power of the introduced generalization. Let $G=\left(V,E\right)$ be a $b$-side Diamond DAG. According to the definition of $S$-partition, the family $\{V\}$ is indeed a 1-partition of $V$ as $V$ has a dominator of cardinality $1$ (composed by the single input vertex)
and empty minimum set. Hence, the tightest lower bound that can be claimed by using Hong and Kung's method in Theorem~\ref{thm:hongkung} (\cite[Theorem3.1]{jia1981complexity}) is the trivial one according to which  the \io{} complexity of $G$ is greater or equal to zero. This is due to the fact that the $M$-partition technique does not capture the fact that the computation of the DAG itself  requires the use of a certain amount of memory locations, and, if the available cache memory is of finite size, data has to be moved from cache to slow memory and vice versa.
In contrast, our technique correctly accounts for this phenomenon thanks to the characterization of the internal and external boundaries in the definition of a segment partition.

% \begin{figure}[ht]
%   \center
%   \input{3-diamond.tex}
%   \caption{Example of a $3$-diamond DAG of side $6$. The DAG is obtained as a combination of a $3$-pyramid of height $6$ (highlighted in red) and of a reverse $3$-pyramid of height $6$ (highlighted in blue) by merging the input vertices of the former with the output vertices of the latter.}
%   \label{fig:3dia}
% \end{figure}

Alternatively, to the previous definition, it is possible to think of a Diamond DAG as obtained by ``\emph{merging}'' a $2$-pyramid of height $b$ and a reverse $2$-pyramid of height $b$ by fusing the $b$ input vertices of the pyramid with the $b$ output vertices of the reverse pyramid. This approach allows generalizing the family of Diamond DAGs to those which can be constructed by merging a $q$-pyramid of height $b$.
% An example of such construction for a $3$-diamond is presented in Figure~\ref{fig:3dia}. 

Various complexity measures of diamond DAGs have been studied in several models of computation
since such DAGs model a number of interesting computations. For 
example, 3-diamonds are sub-DAGs of the computation DAG of a linear array of processing elements, each endowed with one word of memory.

Before delving into the proof of the main result, we introduce the following lemma:
\begin{lemma}\label{lem:vdispaths}
\sloppy Given a $q$-diamond of side $b$ DAG $G = \left(V,E\right)$, let $\Pi= \{v_1,v_2,v_3,\ldots,v_{b-1},v_b\}$ denote the set of vertices in a directed path from the input vertex $v_1$ of $G$ to a vertex $v_b$ in the diagonal. Further let $\Pi'\subseteq \Pi\setminus \{v_b\}$, with $|\Pi'|\neq 0$. 
There exist $|\Pi'|\left(r-1\right)+1$ paths connecting vertices in $\Pi'$ to the vertices of the diagonal of $G$ which share only vertices in $\Pi'$. 
\end{lemma}
\begin{proof}
Let $G'=\left(V',E'\right)$ denote  the sub-DAG of $G$ such that $V'$ includes the vertices of $V$ besides the successors of the vertices on the diagonal of $G$, and $E'=\left(V'\times V'\right)\cap E$. That is, $G'$ is the reverse $q$-pyramid which includes the input of $G$ up to the vertices of the diagonal. 
In the following, we consider the reverse DAG of $G'$, denoted as $G'_R$. If $v_{b-1}\in\Pi'$ the statement follows directly from Lemma~\ref{lem:piramid}. If $v_{b-1}\notin\Pi'$, let $v_j$ denote the vertex of $\Pi'$ on the lowest layer of $G'_R$ (according to Definition~\ref{def:pyramid}). The statement then follows from Lemma~\ref{lem:piramid} and by noting that the path composed by the last part of $\Pi$ starting from $v_j$, that is $\{v_j,v_{y+1},\ldots,v_b\}$ does not share vertices with the paths constructed according to the proof of Lemma~\ref{lem:piramid} besides $v_j$.
\end{proof}

Lemma~\ref{lem:vdispaths} captures an important structural property of diamond (and pyramid) DAGs which will be of crucial importance for the analysis of  their \io{} complexity.

\subsubsection{\io{} lower bound}
\begin{theorem}\label{thm:diamio}
    Let $G=\left(V,E\right)$ be a $q$-diamond DAG of side $b$. The \io-complexity of $G$ when run on a machine equipped with a cache memory of size $M$ and where for each \io{} operation it is possible to move up to $L$ memory words stored in consecutive memory locations from the cache to slow memory or vice versa, is:
\begin{equation*}
    IO\left(G,M\right)\geq \left\lfloor b\bigg/\left\lceil \frac{2M+r-2}{r-1}\right\rceil\right\rfloor^2 \frac{M}{L}.
\end{equation*}
\end{theorem}
\begin{proof}
\input{brokendiamond.tex}
We present proof for  the  case $L=  1$. The result then trivially generalizes for other values of $L$.
In order to simplify the presentation, in the proof, we focus on the case for $r=2$. The proof for the general case follows simple, albeit tedious, modifications. Note that for $r=2$ we have:
\begin{equation*}
    \left\lceil\frac{2M+r-2}{r-1}\right\rceil = 2M.
\end{equation*}
Further, we assume that $b$ is a multiple of $2M$. If that is not the case, the proof proceeds as if the DAG begin considered is the $\left\lceil\frac{b}{2M}\right\rceil2M$-sided Diamond sub-DAG whose input vertex corresponds to the input of $G$.

Let $G_R= \left(V,E_R\right)$ denote the reverse DAG of $G$. $G_R$ is itself a $b$-sided Diamond DAG.
We introduce  the visit rule $\visrule^*$ of $G_R$ using Figure~\ref{fig:brokend} as reference. We divide the DAG $G_R$ into $\left(b/2M\right)^2$ Diamond sub-DAGs, each of side $2M$, starting from the bottom left corner of $G_R$ (\myie{} the input vertex of $G_R$), as depicted in Figure~\ref{fig:brokend}. Let $G_j= \left(V_j,E_j\right)$ denote the $j$-th Diamond sub-DAG of $G_R$ for $j=1,2,\ldots,\left(b/2M\right)^2$. By construction, $\{V_1,V_2,\ldots V_{\left(b/2M\right)^2}\}$ (resp., $\{E_1,E_2,\ldots E_{\left(b/2M\right)^2}\}$) partition $V$ (resp., $E_R$).

$\visrule^*$ is defined in such a way that each vertex $v\in V$ is enabled by the family of singleton sets each containing a predecessor of $V$, if any, provided that the edge connected said predecessor to $v$ is not crossed by a red dashed line. That is, $\visrule^*$  behaves like the standard singleton visit rule \emph{internally} to the Diamond sub-DAGs and which \emph{disables} the edges connecting vertices of different sub-DAGs. Clearly $\visrule^*\in\visset{G_R}$.

In the following, we show that for any $\visrule{}^*$-visit $\psi\in\visitset{\visrule{}^*}{G_R}$ there exists a segment partition $\left(i_1,i_2,\ldots,i_{(b/2M)^2+1}\right)$ such that 
\begin{equation*}
\sum_{j=1}^{(b/2M)^2}\max\{0, |B^{(ent)}_{\visrule^*}\left(\psi[i_j..i_{j+1})\right)|-M\}\geq \left( \frac{b}{2M}\right)^2 M
\end{equation*}
from whence, by Theorem~\ref{thm:finiolwb}, the statement follows.
% From Theorem~\ref{thm:mainiothm} we have that the number of \io{} operations executed by any schedule for $G$ is bounded by $(Q-1)M$ where $Q$ denotes the minimum number of segments of a $2M$-segment partition of any $\visrule^*$-visit of $G_R$.
% In the remainder of the proof we shall show that for any $\visrule^*$-vist $\psi\in\visitset{\visrule^*}{G_R}$, any $2M$-segment partition of $\psi$ has at least $\left(b/4M\right)^2+1$ segments.\\
For all $j\in\{1,\ldots,(b/2M)^2$, let $i_j$ denote the index of the first step of the visit $\psi$ corresponding to which a vertex of the diagonal of $G_{j-1}$ is visited in $\psi$, and let us assume that $i_1<i_2<\ldots<i_{(b/2M)^2}$. This assumption is without loss of generality as it is possible to assign indices to the Diamond sub-DAG to be consistent with it. We set $i_0=1$

For each $G_j$, we denote the ``\emph{vertices on the diagonal}'' the vertices on the diagonal from the top-left corner to the bottom right corner. The vertices on the diagonal in the bottom left Diamond sub-DAG depicted in Figure~\ref{fig:brokend} are encircled by a blue dashed line.
%We refer to the set of vertices of the diagonal as $\Delta_j$. By construction $\Delta_{j'}\cap\Delta_{j''}=\emptyset$ for any $j'\neq j''$. Further we define $\Delta = \cup_{j=1}^{(b/4M)^2} \Delta_j$.  
For  each of the $(b/2M)^2$ Diamond sub-DAGs, we denote as the ``\emph{input vertex}'' the vertex in the bottom-left corner of the Diamond sub-DAG. 
For all $j\in\{1,\ldots,(b/2M)^2\}$, there exists a set $\Pi_j\in\infix{\psi}{1}{i_j}$ such that the vertices in $\Pi_j$ form a path directed from the input vertex of $G_j$ to the vertex of its diagonal visited in $\psi[i_j]$(excluded). This follows from the properties of $\visrule^*$: For vertex $\psi[i_j]$ to be visited at step $i_j$, at least one of its predecessors must have been previously visited during $\psi[1..i_j)$. For such vertex to have been visited, one of its predecessors must have been visited previously. The same reasoning can be applied iteratively until the input vertex of $G_j$, which, by the construction of $\visrule^*$, is enabled by the empty set. 

For each $i_j$, with $j\in\{1,\ldots,(b/2M)^2\}$,  by Lemma~\ref{lem:vdispaths} there exist $|\Pi_j|+1=2M$ vertex disjoints paths connecting vertices of $\Pi_j$, which are visited in $\psi[1..i_j)$, to vertices in the diagonal of $G_j$, which, by construction, are all visited in  $\infix{\psi}{i_j}{n}$. By construction of the visit rule $\visrule^*$ there must therefore be at least one vertex for each of such paths in $B_{\visrule^*}\left(\psi[1..i_j)\right)$ and, thus, in
\begin{equation*}
    \cup_{j=1}^{(b/2M)^2} B_{\visrule^*}\left(\psi[1,i_j)\right) \cap \psi[i_j..i_{j+1}) = \cup_{j=1}^{(b/4M)^2} B_{r^*}^{(ent)}(\psi[i_j..i_{j+1}))
\end{equation*}
As, by construction, the intervals are disjoint, a vertex in $B_{\visrule^*}\left(\psi[1,i_j)\right)$ may appear in exactly one of the entering boundaries $B_{r^*}^{(ent)}(\psi[i_k..i_{k+1}))$ for $j\leq k \leq (b/4M)^2$. 

The same reasoning holds for all $i_j$'s. As, by construction the $G_j$ are vertex disjoint, vertex disjoint paths in each of them will also be vertex disjoint among each other. Thus, 

\begin{align*}
     \sum_{j=1}^{(b/2M)^2} \max\{0, |B_{r^*}^{(ent)}(\psi[i_j..i_{j+1}))|-M\} &\geq  \sum_{j=1}^{(b/2M)^2} |B_{r^*}^{(ent)}(\psi[i_j..i_{j+1}))|-  \left(\frac{b}{2M}\right)^2M\\
     &\geq \sum_{j=2}^{(b/2M)^2+1} |\Pi_j|+1 - \left(\frac{b}{2M}\right)^2M\\
     &\geq \left(\frac{b}{2M}\right)^2 2M - \left(\frac{b}{2M}\right)^2M.
\end{align*}
\end{proof}

For the special case of $2$-diamond DAGs, it is possible to easily restate the \io{} lower bound given by Theorem~\ref{thm:diamio} in terms of its number of vertices $|V|=n = b^2$ as $IO_M\left(G \right)\geq \BOme{n/M}$. In general, for  $q$-diamond DAGs we have $|V|=n = \left(b-1\right)\left(2+\left(r-1\right)\left(b-1\right)\right)+1$, and thus $IO_M\left(G \right)\geq \BOme{nrM/\left(M+r\right)^2}$.

\subsubsection{Upper bound} 
In order to verify the tightness of the lower bound in Theorem~\ref{thm:diamio}, we present an algorithm $\mathcal{A}^*$ which allows to compute a given a $q$-diamond of side $b$ DAG  $G\left(r,b\right)$ using $\BO{\left(br\right)^2/\left(M+r\right)}$.

\begin{figure}[t]
\begin{subfigure}{.5\linewidth}
\centering
\input{uppb1.tex}
\caption{Example for $2$-diamond DAG} 
\end{subfigure}%
\begin{subfigure}{.5\linewidth}
\centering
\input{uppb2.tex}
\caption{Example for $3$-diamond DAG} 
\end{subfigure}
\caption{Examples of sub-DAGs $G_j$ constructed in the computation of algorithm $\mathcal{A}^*$. The sub-DAG $G_j$ are delimited by red dashed lines, the vertices highlighted in green (resp., blue) are vertices in $G_j$ (resp., not in $G_j$) which are immediate predecessors of vertices in not in $G_j$ (resp., in $G_j$) which are written (resp., read) from cache to slow memory (resp., from slow to cache memory) during the computation of $G_j$ by algorithm $\mathcal{A}^*$.}
\label{fig:uppbound}
\end{figure}
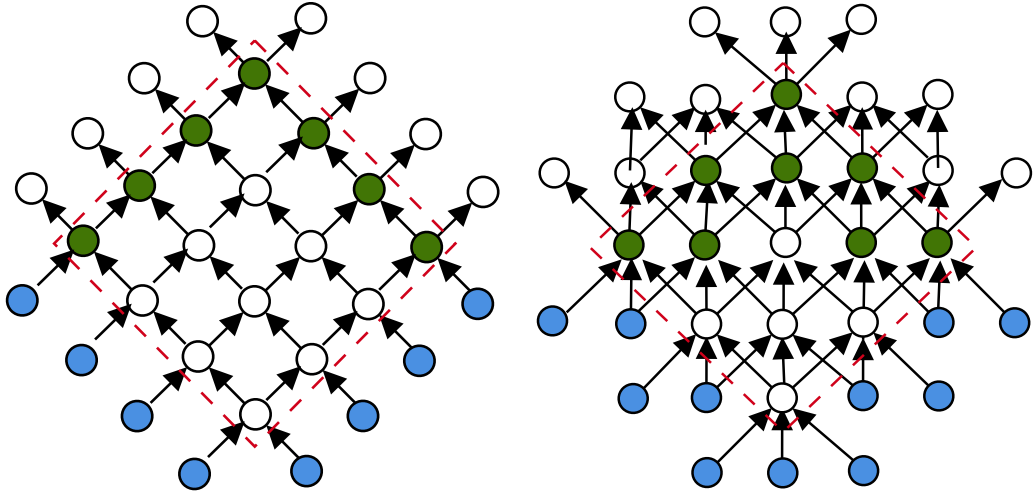

Consider a partitioning of the given DAG into  $q$-diamond sub-DAGs of side $b^*$ according to the subdivision scheme discussed in the proof of Theorem~\ref{thm:diamio} and represented in Figure~\ref{fig:brokend}, where $b^*$ is chosen as the maximum integer value such that 
$$\left(b^*-1\right)\left(r-1\right)+1 \leq M.$$
That is,
$$b^*=\lfloor \left(M-1\right)\left(r-1\right)+1\rfloor.$$
Clearly there are at most $\left\lceil b\bigg/\left\lfloor\left(M-1\right)\left(r-1\right)^{-1}+1\right\rfloor\right\rceil^2$ such sub-DAGs.
Consider one of such sub-DAG $G_j$: there are at most $2\left(b^*-1\right)\left(r-1\right)+1<2M$ vertices of $G$ which are immediate predecessors (resp., successors) of vertices in $G_j$ while not included in $G_j$. We present an example of such sub-DAGs for 2 and 3-diamond DAGs in Figure~\ref{fig:uppbound}.
Algorithm  $\mathcal{A}^*$ proceeds to evaluate $G$ starting from the sub-DAG denoted as $G_1$, which has a single input corresponding to the single input of $G$ itself. We assume this input value to be initially stored in the slow memory. Note that $G_1$ can be evaluated entirely in memory without additional \io{} operations. This can be achieved by evaluating the vertices in a \emph{breath-first} manner according to the layers of the sub-DAG. Whenever a vertex that has at least one success outside of $G_1$ is evaluated, such value is written to the slow memory using a \texttt{write} \io{} operation. 

Algorithm $\mathcal{A}^*$ then proceeds by evaluating in a similar manner those sub-DAGs $G_j$ for which all vertices of $G$ which have immediate successors in  $G_j$ were already computed and written in the slow memory. In particular when evaluating  $G_j$ the algorithm process by evaluating its vertices one layer at a time, loading the at most $2M$ predecessors of vertices in $G_j$ which are not in  $G_j$ when necessary from the slow memory using a \texttt{read} \io{} operation. Once used to compute their successors in $G_j$, such vertices are removed from the cache. Whenever a vertex of $G_j$ which is the output vertex of $G$ or  which has a successor outside $G_j$ itself  is evaluated, the value associated with such vertex is written to the slow memory using a \texttt{write} \io{} operation. By construction, there must be at most $2M$ such vertices for each $G_j$.  

Algorithm $\mathcal{A}^*$ proceeds until all of the sub-DAGs have been evaluated. This ensures that the entire DAG $G$ is indeed evaluated. These considerations straightforwardly lead to the following result:

\begin{theorem}\label{thm:diamioupp}
The  number  of  I/O  operations  executed  by $\mathcal{A}^*$ when  evaluating a $q$-diamond DAG of side $b$, $G$ when run using a machine equipped with a cache memory of size $M$ is:
\begin{equation*}
    IO_M\left(G\right)\leq \BO{\frac{b^2r^2}{M+r}}.
\end{equation*}
\end{theorem}
\begin{proof}
Following from the previous description of the execution of  algorithm $\mathcal{A}^*$ we have that at most $4M$ \io{} operations are executed while evaluating each of the $G_j$ sub-DAGs. As, by construction, the $\left\lceil b\bigg/\left\lfloor\left(M-1\right)\left(r-1\right)^{-1}+1\right\rfloor\right\rceil^2$ sub-DAGs are vertex disjoint, we have that the algorithm executes at most
$$\left\lceil \frac{b}{\left\lfloor\frac{M-1}{r-1}+1\right\rfloor}\right\rceil^24M\leq \BO{\frac{b^2r^2}{M+r}}$$
\io{} operations. The theorem follows.
\end{proof}

Thus, we can conclude that the lower bound in Theorem~\ref{thm:diamio} is asymptotically tight and that algorithm $\mathcal{A}^*$ is asymptotically optimal.

%% file: uppb1.tex
\resizebox{\textwidth}{!}{

\tikzset{every picture/.style={line width=0.75pt}} %set default line width to 0.75pt        

\begin{tikzpicture}[x=0.75pt,y=0.75pt,yscale=-1,xscale=1]
%uncomment if require: \path (0,210.3333282470703); %set diagram left start at 0, and has height of 210.3333282470703

%Shape: Circle [id:dp20210006473129538] 
\draw  [fill={rgb, 255:red, 74; green, 144; blue, 226 }  ,fill opacity=1 ] (77.59,192.81) .. controls (75.36,190.55) and (75.39,186.91) .. (77.65,184.68) .. controls (79.91,182.45) and (83.55,182.48) .. (85.78,184.74) .. controls (88.01,187) and (87.98,190.65) .. (85.72,192.87) .. controls (83.46,195.1) and (79.82,195.07) .. (77.59,192.81) -- cycle ;
%Shape: Circle [id:dp3230818484130964] 
\draw  [fill={rgb, 255:red, 74; green, 144; blue, 226 }  ,fill opacity=1 ] (55.84,170.72) .. controls (53.61,168.46) and (53.64,164.82) .. (55.9,162.59) .. controls (58.16,160.36) and (61.8,160.39) .. (64.03,162.65) .. controls (66.26,164.92) and (66.23,168.56) .. (63.97,170.79) .. controls (61.71,173.01) and (58.06,172.99) .. (55.84,170.72) -- cycle ;
%Shape: Circle [id:dp813623866828254] 
\draw  [fill={rgb, 255:red, 74; green, 144; blue, 226 }  ,fill opacity=1 ] (34.79,149.35) .. controls (32.56,147.09) and (32.59,143.44) .. (34.85,141.22) .. controls (37.11,138.99) and (40.75,139.02) .. (42.98,141.28) .. controls (45.21,143.54) and (45.18,147.18) .. (42.92,149.41) .. controls (40.66,151.64) and (37.01,151.61) .. (34.79,149.35) -- cycle ;
%Shape: Circle [id:dp5814591326927978] 
\draw  [fill={rgb, 255:red, 74; green, 144; blue, 226 }  ,fill opacity=1 ] (12.33,126.55) .. controls (10.11,124.28) and (10.13,120.64) .. (12.4,118.42) .. controls (14.66,116.19) and (18.3,116.22) .. (20.53,118.48) .. controls (22.76,120.74) and (22.73,124.38) .. (20.47,126.61) .. controls (18.2,128.84) and (14.56,128.81) .. (12.33,126.55) -- cycle ;
%Shape: Circle [id:dp4521969597706659] 
\draw  [fill={rgb, 255:red, 74; green, 144; blue, 226 }  ,fill opacity=1 ] (120.02,191.72) .. controls (117.8,189.46) and (117.82,185.82) .. (120.09,183.59) .. controls (122.35,181.36) and (125.99,181.39) .. (128.22,183.65) .. controls (130.45,185.92) and (130.42,189.56) .. (128.16,191.79) .. controls (125.89,194.01) and (122.25,193.99) .. (120.02,191.72) -- cycle ;
%Shape: Circle [id:dp997574545597677] 
\draw   (100.68,170.06) .. controls (98.45,167.8) and (98.48,164.16) .. (100.74,161.93) .. controls (103,159.7) and (106.64,159.73) .. (108.87,161.99) .. controls (111.1,164.25) and (111.07,167.89) .. (108.81,170.12) .. controls (106.54,172.35) and (102.9,172.32) .. (100.68,170.06) -- cycle ;
%Shape: Circle [id:dp07762260711286384] 
\draw   (78.92,147.97) .. controls (76.7,145.71) and (76.72,142.07) .. (78.99,139.84) .. controls (81.25,137.61) and (84.89,137.64) .. (87.12,139.9) .. controls (89.35,142.17) and (89.32,145.81) .. (87.06,148.03) .. controls (84.79,150.26) and (81.15,150.23) .. (78.92,147.97) -- cycle ;
%Shape: Circle [id:dp15380558924814602] 
\draw   (57.87,126.6) .. controls (55.65,124.33) and (55.67,120.69) .. (57.94,118.47) .. controls (60.2,116.24) and (63.84,116.26) .. (66.07,118.53) .. controls (68.3,120.79) and (68.27,124.43) .. (66.01,126.66) .. controls (63.74,128.89) and (60.1,128.86) .. (57.87,126.6) -- cycle ;
%Shape: Circle [id:dp3340315078454674] 
\draw  [fill={rgb, 255:red, 65; green, 117; blue, 5 }  ,fill opacity=1 ] (35.42,103.8) .. controls (33.19,101.53) and (33.22,97.89) .. (35.48,95.66) .. controls (37.75,93.44) and (41.39,93.46) .. (43.62,95.73) .. controls (45.84,97.99) and (45.82,101.63) .. (43.55,103.86) .. controls (41.29,106.09) and (37.65,106.06) .. (35.42,103.8) -- cycle ;
%Straight Lines [id:da4193798422358561] 
\draw    (121.09,182.59) -- (110.21,171.55) ;
\draw [shift={(108.81,170.12)}, rotate = 405.44] [fill={rgb, 255:red, 0; green, 0; blue, 0 }  ][line width=0.75]  [draw opacity=0] (8.93,-4.29) -- (0,0) -- (8.93,4.29) -- cycle    ;

%Straight Lines [id:da8026976527815206] 
\draw    (100.74,161.93) -- (88.46,149.46) ;
\draw [shift={(87.06,148.03)}, rotate = 405.44] [fill={rgb, 255:red, 0; green, 0; blue, 0 }  ][line width=0.75]  [draw opacity=0] (8.93,-4.29) -- (0,0) -- (8.93,4.29) -- cycle    ;

%Straight Lines [id:da391960949852002] 
\draw    (78.29,139.13) -- (67.41,128.08) ;
\draw [shift={(66.01,126.66)}, rotate = 405.44] [fill={rgb, 255:red, 0; green, 0; blue, 0 }  ][line width=0.75]  [draw opacity=0] (8.93,-4.29) -- (0,0) -- (8.93,4.29) -- cycle    ;

%Straight Lines [id:da49140665908383063] 
\draw    (57.94,118.47) -- (44.96,105.28) ;
\draw [shift={(43.55,103.86)}, rotate = 405.44] [fill={rgb, 255:red, 0; green, 0; blue, 0 }  ][line width=0.75]  [draw opacity=0] (8.93,-4.29) -- (0,0) -- (8.93,4.29) -- cycle    ;

%Shape: Circle [id:dp7232604775286668] 
\draw  [fill={rgb, 255:red, 255; green, 255; blue, 255 }  ,fill opacity=1 ] (15.78,83.85) .. controls (13.55,81.58) and (13.58,77.94) .. (15.84,75.71) .. controls (18.1,73.49) and (21.74,73.51) .. (23.97,75.78) .. controls (26.2,78.04) and (26.17,81.68) .. (23.91,83.91) .. controls (21.64,86.14) and (18,86.11) .. (15.78,83.85) -- cycle ;
%Straight Lines [id:da4902185035627853] 
\draw    (35.48,95.66) -- (25.31,85.33) ;
\draw [shift={(23.91,83.91)}, rotate = 405.44] [fill={rgb, 255:red, 0; green, 0; blue, 0 }  ][line width=0.75]  [draw opacity=0] (8.93,-4.29) -- (0,0) -- (8.93,4.29) -- cycle    ;

%Shape: Circle [id:dp752373829089108] 
\draw  [fill={rgb, 255:red, 74; green, 144; blue, 226 }  ,fill opacity=1 ] (141.4,170.67) .. controls (139.17,168.41) and (139.2,164.77) .. (141.46,162.54) .. controls (143.72,160.31) and (147.36,160.34) .. (149.59,162.6) .. controls (151.82,164.87) and (151.79,168.51) .. (149.53,170.74) .. controls (147.27,172.96) and (143.63,172.94) .. (141.4,170.67) -- cycle ;
%Shape: Circle [id:dp8903709069996251] 
\draw   (122.05,149.01) .. controls (119.82,146.75) and (119.85,143.11) .. (122.11,140.88) .. controls (124.38,138.65) and (128.02,138.68) .. (130.25,140.94) .. controls (132.47,143.2) and (132.45,146.84) .. (130.18,149.07) .. controls (127.92,151.3) and (124.28,151.27) .. (122.05,149.01) -- cycle ;
%Shape: Circle [id:dp08276655816901846] 
\draw   (100.3,126.92) .. controls (98.07,124.66) and (98.1,121.02) .. (100.36,118.79) .. controls (102.62,116.56) and (106.27,116.59) .. (108.49,118.85) .. controls (110.72,121.12) and (110.69,124.76) .. (108.43,126.98) .. controls (106.17,129.21) and (102.53,129.19) .. (100.3,126.92) -- cycle ;
%Shape: Circle [id:dp7602725273430508] 
\draw   (79.25,105.55) .. controls (77.02,103.28) and (77.05,99.64) .. (79.31,97.42) .. controls (81.58,95.19) and (85.22,95.22) .. (87.44,97.48) .. controls (89.67,99.74) and (89.64,103.38) .. (87.38,105.61) .. controls (85.12,107.84) and (81.48,107.81) .. (79.25,105.55) -- cycle ;
%Shape: Circle [id:dp4916951257307052] 
\draw  [fill={rgb, 255:red, 65; green, 117; blue, 5 }  ,fill opacity=1 ] (56.8,82.75) .. controls (54.57,80.48) and (54.6,76.84) .. (56.86,74.61) .. controls (59.12,72.39) and (62.76,72.41) .. (64.99,74.68) .. controls (67.22,76.94) and (67.19,80.58) .. (64.93,82.81) .. controls (62.67,85.04) and (59.03,85.01) .. (56.8,82.75) -- cycle ;
%Straight Lines [id:da6691435669529446] 
\draw    (142.46,161.54) -- (131.59,150.5) ;
\draw [shift={(130.18,149.07)}, rotate = 405.44] [fill={rgb, 255:red, 0; green, 0; blue, 0 }  ][line width=0.75]  [draw opacity=0] (8.93,-4.29) -- (0,0) -- (8.93,4.29) -- cycle    ;

%Straight Lines [id:da757900962880528] 
\draw    (122.11,140.88) -- (109.83,128.41) ;
\draw [shift={(108.43,126.98)}, rotate = 405.44] [fill={rgb, 255:red, 0; green, 0; blue, 0 }  ][line width=0.75]  [draw opacity=0] (8.93,-4.29) -- (0,0) -- (8.93,4.29) -- cycle    ;

%Straight Lines [id:da03636328877130035] 
\draw    (99.66,118.08) -- (88.78,107.03) ;
\draw [shift={(87.38,105.61)}, rotate = 405.44] [fill={rgb, 255:red, 0; green, 0; blue, 0 }  ][line width=0.75]  [draw opacity=0] (8.93,-4.29) -- (0,0) -- (8.93,4.29) -- cycle    ;

%Straight Lines [id:da20617663839458578] 
\draw    (79.31,97.42) -- (66.33,84.23) ;
\draw [shift={(64.93,82.81)}, rotate = 405.44] [fill={rgb, 255:red, 0; green, 0; blue, 0 }  ][line width=0.75]  [draw opacity=0] (8.93,-4.29) -- (0,0) -- (8.93,4.29) -- cycle    ;

%Shape: Circle [id:dp6254258889586082] 
\draw  [fill={rgb, 255:red, 255; green, 255; blue, 255 }  ,fill opacity=1 ] (37.15,62.8) .. controls (34.92,60.53) and (34.95,56.89) .. (37.21,54.66) .. controls (39.48,52.44) and (43.12,52.46) .. (45.34,54.73) .. controls (47.57,56.99) and (47.54,60.63) .. (45.28,62.86) .. controls (43.02,65.09) and (39.38,65.06) .. (37.15,62.8) -- cycle ;
%Straight Lines [id:da9218509723574893] 
\draw    (56.86,74.61) -- (46.69,64.28) ;
\draw [shift={(45.28,62.86)}, rotate = 405.44] [fill={rgb, 255:red, 0; green, 0; blue, 0 }  ][line width=0.75]  [draw opacity=0] (8.93,-4.29) -- (0,0) -- (8.93,4.29) -- cycle    ;

%Shape: Circle [id:dp2707132693449785] 
\draw  [fill={rgb, 255:red, 74; green, 144; blue, 226 }  ,fill opacity=1 ] (162.77,149.62) .. controls (160.55,147.36) and (160.57,143.72) .. (162.84,141.49) .. controls (165.1,139.26) and (168.74,139.29) .. (170.97,141.55) .. controls (173.2,143.82) and (173.17,147.46) .. (170.91,149.69) .. controls (168.64,151.91) and (165,151.89) .. (162.77,149.62) -- cycle ;
%Shape: Circle [id:dp2008397937239026] 
\draw   (143.43,127.96) .. controls (141.2,125.7) and (141.23,122.06) .. (143.49,119.83) .. controls (145.75,117.6) and (149.39,117.63) .. (151.62,119.89) .. controls (153.85,122.15) and (153.82,125.79) .. (151.56,128.02) .. controls (149.3,130.25) and (145.65,130.22) .. (143.43,127.96) -- cycle ;
%Shape: Circle [id:dp873274286120407] 
\draw   (121.68,105.87) .. controls (119.45,103.61) and (119.47,99.97) .. (121.74,97.74) .. controls (124,95.51) and (127.64,95.54) .. (129.87,97.8) .. controls (132.1,100.07) and (132.07,103.71) .. (129.81,105.94) .. controls (127.54,108.16) and (123.9,108.14) .. (121.68,105.87) -- cycle ;
%Shape: Circle [id:dp696949764416567] 
\draw   (100.63,84.5) .. controls (98.4,82.23) and (98.43,78.59) .. (100.69,76.37) .. controls (102.95,74.14) and (106.59,74.17) .. (108.82,76.43) .. controls (111.05,78.69) and (111.02,82.33) .. (108.76,84.56) .. controls (106.49,86.79) and (102.85,86.76) .. (100.63,84.5) -- cycle ;
%Shape: Circle [id:dp6953096222611246] 
\draw  [fill={rgb, 255:red, 65; green, 117; blue, 5 }  ,fill opacity=1 ] (78.17,61.7) .. controls (75.94,59.43) and (75.97,55.79) .. (78.24,53.57) .. controls (80.5,51.34) and (84.14,51.36) .. (86.37,53.63) .. controls (88.59,55.89) and (88.57,59.53) .. (86.3,61.76) .. controls (84.04,63.99) and (80.4,63.96) .. (78.17,61.7) -- cycle ;
%Straight Lines [id:da21155836700235775] 
\draw    (163.84,140.49) -- (152.96,129.45) ;
\draw [shift={(151.56,128.02)}, rotate = 405.44] [fill={rgb, 255:red, 0; green, 0; blue, 0 }  ][line width=0.75]  [draw opacity=0] (8.93,-4.29) -- (0,0) -- (8.93,4.29) -- cycle    ;

%Straight Lines [id:da19120451395789462] 
\draw    (143.49,119.83) -- (131.21,107.36) ;
\draw [shift={(129.81,105.94)}, rotate = 405.44] [fill={rgb, 255:red, 0; green, 0; blue, 0 }  ][line width=0.75]  [draw opacity=0] (8.93,-4.29) -- (0,0) -- (8.93,4.29) -- cycle    ;

%Straight Lines [id:da1766545976921079] 
\draw    (121.04,97.03) -- (110.16,85.98) ;
\draw [shift={(108.76,84.56)}, rotate = 405.44] [fill={rgb, 255:red, 0; green, 0; blue, 0 }  ][line width=0.75]  [draw opacity=0] (8.93,-4.29) -- (0,0) -- (8.93,4.29) -- cycle    ;

%Straight Lines [id:da37651568542174485] 
\draw    (100.69,76.37) -- (87.71,63.18) ;
\draw [shift={(86.3,61.76)}, rotate = 405.44] [fill={rgb, 255:red, 0; green, 0; blue, 0 }  ][line width=0.75]  [draw opacity=0] (8.93,-4.29) -- (0,0) -- (8.93,4.29) -- cycle    ;

%Shape: Circle [id:dp7069022626051684] 
\draw  [fill={rgb, 255:red, 255; green, 255; blue, 255 }  ,fill opacity=1 ] (58.53,41.75) .. controls (56.3,39.48) and (56.33,35.84) .. (58.59,33.61) .. controls (60.85,31.39) and (64.49,31.41) .. (66.72,33.68) .. controls (68.95,35.94) and (68.92,39.58) .. (66.66,41.81) .. controls (64.4,44.04) and (60.75,44.01) .. (58.53,41.75) -- cycle ;
%Straight Lines [id:da8876618170830703] 
\draw    (78.24,53.57) -- (68.06,43.23) ;
\draw [shift={(66.66,41.81)}, rotate = 405.44] [fill={rgb, 255:red, 0; green, 0; blue, 0 }  ][line width=0.75]  [draw opacity=0] (8.93,-4.29) -- (0,0) -- (8.93,4.29) -- cycle    ;

%Shape: Circle [id:dp9601381885931519] 
\draw  [fill={rgb, 255:red, 74; green, 144; blue, 226 }  ,fill opacity=1 ] (184.86,127.87) .. controls (182.63,125.61) and (182.66,121.97) .. (184.93,119.74) .. controls (187.19,117.51) and (190.83,117.54) .. (193.06,119.8) .. controls (195.28,122.07) and (195.26,125.71) .. (192.99,127.93) .. controls (190.73,130.16) and (187.09,130.13) .. (184.86,127.87) -- cycle ;
%Shape: Circle [id:dp3742404427048862] 
\draw  [fill={rgb, 255:red, 65; green, 117; blue, 5 }  ,fill opacity=1 ] (165.51,106.21) .. controls (163.29,103.95) and (163.31,100.31) .. (165.58,98.08) .. controls (167.84,95.85) and (171.48,95.88) .. (173.71,98.14) .. controls (175.94,100.4) and (175.91,104.04) .. (173.65,106.27) .. controls (171.38,108.5) and (167.74,108.47) .. (165.51,106.21) -- cycle ;
%Shape: Circle [id:dp4280787443629057] 
\draw  [fill={rgb, 255:red, 65; green, 117; blue, 5 }  ,fill opacity=1 ] (143.76,84.12) .. controls (141.54,81.86) and (141.56,78.22) .. (143.83,75.99) .. controls (146.09,73.76) and (149.73,73.79) .. (151.96,76.05) .. controls (154.19,78.31) and (154.16,81.96) .. (151.89,84.18) .. controls (149.63,86.41) and (145.99,86.38) .. (143.76,84.12) -- cycle ;
%Shape: Circle [id:dp805599568592626] 
\draw  [fill={rgb, 255:red, 65; green, 117; blue, 5 }  ,fill opacity=1 ] (122.71,62.75) .. controls (120.49,60.48) and (120.51,56.84) .. (122.78,54.61) .. controls (125.04,52.39) and (128.68,52.41) .. (130.91,54.68) .. controls (133.14,56.94) and (133.11,60.58) .. (130.85,62.81) .. controls (128.58,65.04) and (124.94,65.01) .. (122.71,62.75) -- cycle ;
%Shape: Circle [id:dp4003876176771548] 
\draw  [fill={rgb, 255:red, 65; green, 117; blue, 5 }  ,fill opacity=1 ] (100.26,39.95) .. controls (98.03,37.68) and (98.06,34.04) .. (100.32,31.81) .. controls (102.59,29.59) and (106.23,29.61) .. (108.45,31.88) .. controls (110.68,34.14) and (110.65,37.78) .. (108.39,40.01) .. controls (106.13,42.24) and (102.49,42.21) .. (100.26,39.95) -- cycle ;
%Straight Lines [id:da7389189104680325] 
\draw    (185.93,118.74) -- (175.05,107.7) ;
\draw [shift={(173.65,106.27)}, rotate = 405.44] [fill={rgb, 255:red, 0; green, 0; blue, 0 }  ][line width=0.75]  [draw opacity=0] (8.93,-4.29) -- (0,0) -- (8.93,4.29) -- cycle    ;

%Straight Lines [id:da9643633578506188] 
\draw    (165.58,98.08) -- (153.3,85.61) ;
\draw [shift={(151.89,84.18)}, rotate = 405.44] [fill={rgb, 255:red, 0; green, 0; blue, 0 }  ][line width=0.75]  [draw opacity=0] (8.93,-4.29) -- (0,0) -- (8.93,4.29) -- cycle    ;

%Straight Lines [id:da7634339969245241] 
\draw    (143.12,75.28) -- (132.25,64.23) ;
\draw [shift={(130.85,62.81)}, rotate = 405.44] [fill={rgb, 255:red, 0; green, 0; blue, 0 }  ][line width=0.75]  [draw opacity=0] (8.93,-4.29) -- (0,0) -- (8.93,4.29) -- cycle    ;

%Straight Lines [id:da2853123529836046] 
\draw    (122.78,54.61) -- (109.8,41.43) ;
\draw [shift={(108.39,40.01)}, rotate = 405.44] [fill={rgb, 255:red, 0; green, 0; blue, 0 }  ][line width=0.75]  [draw opacity=0] (8.93,-4.29) -- (0,0) -- (8.93,4.29) -- cycle    ;

%Shape: Circle [id:dp3691777982612152] 
\draw  [fill={rgb, 255:red, 255; green, 255; blue, 255 }  ,fill opacity=1 ] (80.61,19.99) .. controls (78.39,17.73) and (78.41,14.09) .. (80.68,11.86) .. controls (82.94,9.64) and (86.58,9.66) .. (88.81,11.93) .. controls (91.04,14.19) and (91.01,17.83) .. (88.75,20.06) .. controls (86.48,22.29) and (82.84,22.26) .. (80.61,19.99) -- cycle ;
%Straight Lines [id:da9440694416141753] 
\draw    (100.32,31.81) -- (90.15,21.48) ;
\draw [shift={(88.75,20.06)}, rotate = 405.44] [fill={rgb, 255:red, 0; green, 0; blue, 0 }  ][line width=0.75]  [draw opacity=0] (8.93,-4.29) -- (0,0) -- (8.93,4.29) -- cycle    ;

%Shape: Circle [id:dp30173633511867726] 
\draw  [fill={rgb, 255:red, 255; green, 255; blue, 255 }  ,fill opacity=1 ] (186.89,85.16) .. controls (184.66,82.9) and (184.69,79.26) .. (186.95,77.03) .. controls (189.22,74.8) and (192.86,74.83) .. (195.08,77.09) .. controls (197.31,79.35) and (197.28,82.99) .. (195.02,85.22) .. controls (192.76,87.45) and (189.12,87.42) .. (186.89,85.16) -- cycle ;
%Shape: Circle [id:dp3583877292040176] 
\draw  [fill={rgb, 255:red, 255; green, 255; blue, 255 }  ,fill opacity=1 ] (165.14,63.07) .. controls (162.91,60.81) and (162.94,57.17) .. (165.2,54.94) .. controls (167.46,52.71) and (171.1,52.74) .. (173.33,55) .. controls (175.56,57.27) and (175.53,60.91) .. (173.27,63.13) .. controls (171.01,65.36) and (167.37,65.33) .. (165.14,63.07) -- cycle ;
%Shape: Circle [id:dp27788297163083087] 
\draw  [fill={rgb, 255:red, 255; green, 255; blue, 255 }  ,fill opacity=1 ] (144.09,41.7) .. controls (141.86,39.43) and (141.89,35.79) .. (144.15,33.56) .. controls (146.41,31.34) and (150.05,31.36) .. (152.28,33.63) .. controls (154.51,35.89) and (154.48,39.53) .. (152.22,41.76) .. controls (149.96,43.99) and (146.32,43.96) .. (144.09,41.7) -- cycle ;
%Shape: Circle [id:dp7575659591016339] 
\draw  [fill={rgb, 255:red, 255; green, 255; blue, 255 }  ,fill opacity=1 ] (121.64,18.9) .. controls (119.41,16.63) and (119.44,12.99) .. (121.7,10.76) .. controls (123.96,8.54) and (127.6,8.56) .. (129.83,10.83) .. controls (132.06,13.09) and (132.03,16.73) .. (129.77,18.96) .. controls (127.5,21.19) and (123.86,21.16) .. (121.64,18.9) -- cycle ;
%Straight Lines [id:da9271822822102496] 
\draw    (22.24,116.78) -- (34.71,104.5) ;
\draw [shift={(36.13,103.09)}, rotate = 495.44] [fill={rgb, 255:red, 0; green, 0; blue, 0 }  ][line width=0.75]  [draw opacity=0] (8.93,-4.29) -- (0,0) -- (8.93,4.29) -- cycle    ;

%Straight Lines [id:da5047439490426826] 
\draw    (44.33,95.03) -- (56.08,83.45) ;
\draw [shift={(57.51,82.04)}, rotate = 495.44] [fill={rgb, 255:red, 0; green, 0; blue, 0 }  ][line width=0.75]  [draw opacity=0] (8.93,-4.29) -- (0,0) -- (8.93,4.29) -- cycle    ;

%Straight Lines [id:da8140805643244295] 
\draw    (65.7,73.98) -- (77.46,62.4) ;
\draw [shift={(78.89,60.99)}, rotate = 495.44] [fill={rgb, 255:red, 0; green, 0; blue, 0 }  ][line width=0.75]  [draw opacity=0] (8.93,-4.29) -- (0,0) -- (8.93,4.29) -- cycle    ;

%Straight Lines [id:da21330766742550655] 
\draw    (87.08,52.93) -- (99.55,40.65) ;
\draw [shift={(100.97,39.24)}, rotate = 495.44] [fill={rgb, 255:red, 0; green, 0; blue, 0 }  ][line width=0.75]  [draw opacity=0] (8.93,-4.29) -- (0,0) -- (8.93,4.29) -- cycle    ;

%Straight Lines [id:da8159057154034646] 
\draw    (109.17,31.17) -- (120.92,19.6) ;
\draw [shift={(122.35,18.19)}, rotate = 495.44] [fill={rgb, 255:red, 0; green, 0; blue, 0 }  ][line width=0.75]  [draw opacity=0] (8.93,-4.29) -- (0,0) -- (8.93,4.29) -- cycle    ;

%Straight Lines [id:da1750528489941079] 
\draw    (44,137.45) -- (56.47,125.17) ;
\draw [shift={(57.9,123.77)}, rotate = 495.44] [fill={rgb, 255:red, 0; green, 0; blue, 0 }  ][line width=0.75]  [draw opacity=0] (8.93,-4.29) -- (0,0) -- (8.93,4.29) -- cycle    ;

%Straight Lines [id:da11575394484211077] 
\draw    (66.09,115.7) -- (77.85,104.12) ;
\draw [shift={(79.27,102.72)}, rotate = 495.44] [fill={rgb, 255:red, 0; green, 0; blue, 0 }  ][line width=0.75]  [draw opacity=0] (8.93,-4.29) -- (0,0) -- (8.93,4.29) -- cycle    ;

%Straight Lines [id:da7107829472652736] 
\draw    (87.47,94.65) -- (99.22,83.07) ;
\draw [shift={(100.65,81.67)}, rotate = 495.44] [fill={rgb, 255:red, 0; green, 0; blue, 0 }  ][line width=0.75]  [draw opacity=0] (8.93,-4.29) -- (0,0) -- (8.93,4.29) -- cycle    ;

%Straight Lines [id:da283754811134032] 
\draw    (108.84,73.6) -- (121.31,61.32) ;
\draw [shift={(122.74,59.92)}, rotate = 495.44] [fill={rgb, 255:red, 0; green, 0; blue, 0 }  ][line width=0.75]  [draw opacity=0] (8.93,-4.29) -- (0,0) -- (8.93,4.29) -- cycle    ;

%Straight Lines [id:da7281120199786721] 
\draw    (65.03,161.65) -- (77.5,149.38) ;
\draw [shift={(78.92,147.97)}, rotate = 495.44] [fill={rgb, 255:red, 0; green, 0; blue, 0 }  ][line width=0.75]  [draw opacity=0] (8.93,-4.29) -- (0,0) -- (8.93,4.29) -- cycle    ;

%Straight Lines [id:da5731981663629391] 
\draw    (87.12,139.9) -- (98.87,128.33) ;
\draw [shift={(100.3,126.92)}, rotate = 495.44] [fill={rgb, 255:red, 0; green, 0; blue, 0 }  ][line width=0.75]  [draw opacity=0] (8.93,-4.29) -- (0,0) -- (8.93,4.29) -- cycle    ;

%Straight Lines [id:da8764862915834009] 
\draw    (108.49,118.85) -- (120.25,107.28) ;
\draw [shift={(121.68,105.87)}, rotate = 495.44] [fill={rgb, 255:red, 0; green, 0; blue, 0 }  ][line width=0.75]  [draw opacity=0] (8.93,-4.29) -- (0,0) -- (8.93,4.29) -- cycle    ;

%Straight Lines [id:da8153304617830819] 
\draw    (129.87,97.8) -- (142.34,85.52) ;
\draw [shift={(143.76,84.12)}, rotate = 495.44] [fill={rgb, 255:red, 0; green, 0; blue, 0 }  ][line width=0.75]  [draw opacity=0] (8.93,-4.29) -- (0,0) -- (8.93,4.29) -- cycle    ;

%Straight Lines [id:da06741680279488471] 
\draw    (151.96,76.05) -- (163.71,64.47) ;
\draw [shift={(165.14,63.07)}, rotate = 495.44] [fill={rgb, 255:red, 0; green, 0; blue, 0 }  ][line width=0.75]  [draw opacity=0] (8.93,-4.29) -- (0,0) -- (8.93,4.29) -- cycle    ;

%Straight Lines [id:da690488444222989] 
\draw    (86.79,182.33) -- (99.26,170.05) ;
\draw [shift={(100.69,168.65)}, rotate = 495.44] [fill={rgb, 255:red, 0; green, 0; blue, 0 }  ][line width=0.75]  [draw opacity=0] (8.93,-4.29) -- (0,0) -- (8.93,4.29) -- cycle    ;

%Straight Lines [id:da32374759052143887] 
\draw    (108.88,160.58) -- (120.64,149) ;
\draw [shift={(122.06,147.6)}, rotate = 495.44] [fill={rgb, 255:red, 0; green, 0; blue, 0 }  ][line width=0.75]  [draw opacity=0] (8.93,-4.29) -- (0,0) -- (8.93,4.29) -- cycle    ;

%Straight Lines [id:da4875183845519382] 
\draw    (130.26,139.53) -- (142.01,127.95) ;
\draw [shift={(143.44,126.55)}, rotate = 495.44] [fill={rgb, 255:red, 0; green, 0; blue, 0 }  ][line width=0.75]  [draw opacity=0] (8.93,-4.29) -- (0,0) -- (8.93,4.29) -- cycle    ;

%Straight Lines [id:da1057743081965945] 
\draw    (151.63,118.48) -- (164.1,106.2) ;
\draw [shift={(165.53,104.8)}, rotate = 495.44] [fill={rgb, 255:red, 0; green, 0; blue, 0 }  ][line width=0.75]  [draw opacity=0] (8.93,-4.29) -- (0,0) -- (8.93,4.29) -- cycle    ;

%Straight Lines [id:da232573861932992] 
\draw    (130.91,54.68) -- (142.66,43.1) ;
\draw [shift={(144.09,41.7)}, rotate = 495.44] [fill={rgb, 255:red, 0; green, 0; blue, 0 }  ][line width=0.75]  [draw opacity=0] (8.93,-4.29) -- (0,0) -- (8.93,4.29) -- cycle    ;

%Straight Lines [id:da21723569388553754] 
\draw    (173.71,98.14) -- (185.47,86.56) ;
\draw [shift={(186.89,85.16)}, rotate = 495.44] [fill={rgb, 255:red, 0; green, 0; blue, 0 }  ][line width=0.75]  [draw opacity=0] (8.93,-4.29) -- (0,0) -- (8.93,4.29) -- cycle    ;

%Flowchart: Decision [id:dp37728239412858966] 
\draw  [color={rgb, 255:red, 208; green, 2; blue, 27 }  ,draw opacity=1 ][dash pattern={on 4.5pt off 4.5pt}] (104.5,23.33) -- (180.5,100.83) -- (104.5,178.33) -- (28.5,100.83) -- cycle ;

\end{tikzpicture}
}

%% file: uppb2.tex
\resizebox{\textwidth}{!}{

\tikzset{every picture/.style={line width=0.75pt}} %set default line width to 0.75pt        

\begin{tikzpicture}[x=0.75pt,y=0.75pt,yscale=-1,xscale=1]
%uncomment if require: \path (0,211); %set diagram left start at 0, and has height of 211

%Straight Lines [id:da6980996438321216] 
\draw    (53.9,112.23) -- (53.63,136.55) ;

\draw [shift={(53.92,110.23)}, rotate = 90.65] [fill={rgb, 255:red, 0; green, 0; blue, 0 }  ][line width=0.75]  [draw opacity=0] (8.93,-4.29) -- (0,0) -- (8.93,4.29) -- cycle    ;
%Straight Lines [id:da4602450017891331] 
\draw    (114.17,114.06) -- (114.12,135.15) ;

\draw [shift={(114.17,112.06)}, rotate = 90.13] [fill={rgb, 255:red, 0; green, 0; blue, 0 }  ][line width=0.75]  [draw opacity=0] (8.93,-4.29) -- (0,0) -- (8.93,4.29) -- cycle    ;
%Straight Lines [id:da6920700215556668] 
\draw    (170.63,111.72) -- (149.74,131.89) ;

\draw [shift={(172.07,110.33)}, rotate = 136] [fill={rgb, 255:red, 0; green, 0; blue, 0 }  ][line width=0.75]  [draw opacity=0] (8.93,-4.29) -- (0,0) -- (8.93,4.29) -- cycle    ;
%Straight Lines [id:da5632020932411057] 
\draw    (140.6,111.57) -- (117.73,132.77) ;

\draw [shift={(142.07,110.21)}, rotate = 137.17] [fill={rgb, 255:red, 0; green, 0; blue, 0 }  ][line width=0.75]  [draw opacity=0] (8.93,-4.29) -- (0,0) -- (8.93,4.29) -- cycle    ;
%Straight Lines [id:da8641260844648164] 
\draw    (108.62,112.47) -- (87.73,132.65) ;

\draw [shift={(110.06,111.08)}, rotate = 136] [fill={rgb, 255:red, 0; green, 0; blue, 0 }  ][line width=0.75]  [draw opacity=0] (8.93,-4.29) -- (0,0) -- (8.93,4.29) -- cycle    ;
%Straight Lines [id:da738429538936199] 
\draw    (77.63,111.35) -- (56.74,131.52) ;

\draw [shift={(79.07,109.96)}, rotate = 136] [fill={rgb, 255:red, 0; green, 0; blue, 0 }  ][line width=0.75]  [draw opacity=0] (8.93,-4.29) -- (0,0) -- (8.93,4.29) -- cycle    ;
%Straight Lines [id:da9038223658539914] 
\draw    (48.63,112.23) -- (28.24,132.02) ;

\draw [shift={(50.06,110.84)}, rotate = 135.86] [fill={rgb, 255:red, 0; green, 0; blue, 0 }  ][line width=0.75]  [draw opacity=0] (8.93,-4.29) -- (0,0) -- (8.93,4.29) -- cycle    ;
%Straight Lines [id:da9544283004350946] 
\draw    (83.51,115) -- (83.63,136.67) ;

\draw [shift={(83.5,113)}, rotate = 89.7] [fill={rgb, 255:red, 0; green, 0; blue, 0 }  ][line width=0.75]  [draw opacity=0] (8.93,-4.29) -- (0,0) -- (8.93,4.29) -- cycle    ;
%Straight Lines [id:da5178886073042022] 
\draw    (146.14,113.19) -- (145.63,142.92) ;

\draw [shift={(146.17,111.19)}, rotate = 90.98] [fill={rgb, 255:red, 0; green, 0; blue, 0 }  ][line width=0.75]  [draw opacity=0] (8.93,-4.29) -- (0,0) -- (8.93,4.29) -- cycle    ;
%Straight Lines [id:da6402104884143478] 
\draw    (176.09,113.31) -- (175.11,136.39) ;

\draw [shift={(176.17,111.31)}, rotate = 92.42] [fill={rgb, 255:red, 0; green, 0; blue, 0 }  ][line width=0.75]  [draw opacity=0] (8.93,-4.29) -- (0,0) -- (8.93,4.29) -- cycle    ;
%Straight Lines [id:da1003101913559894] 
\draw    (59.6,112.35) -- (79.6,132.56) ;

\draw [shift={(58.19,110.92)}, rotate = 45.31] [fill={rgb, 255:red, 0; green, 0; blue, 0 }  ][line width=0.75]  [draw opacity=0] (8.93,-4.29) -- (0,0) -- (8.93,4.29) -- cycle    ;
%Straight Lines [id:da5620891161360622] 
\draw    (89.6,112.47) -- (109.6,132.68) ;

\draw [shift={(88.19,111.05)}, rotate = 45.31] [fill={rgb, 255:red, 0; green, 0; blue, 0 }  ][line width=0.75]  [draw opacity=0] (8.93,-4.29) -- (0,0) -- (8.93,4.29) -- cycle    ;
%Straight Lines [id:da15666982453552158] 
\draw    (119.69,112.49) -- (141.61,131.81) ;

\draw [shift={(118.19,111.17)}, rotate = 41.41] [fill={rgb, 255:red, 0; green, 0; blue, 0 }  ][line width=0.75]  [draw opacity=0] (8.93,-4.29) -- (0,0) -- (8.93,4.29) -- cycle    ;
%Straight Lines [id:da7498808018989551] 
\draw    (181.6,111.84) -- (201.61,132.05) ;

\draw [shift={(180.2,110.41)}, rotate = 45.31] [fill={rgb, 255:red, 0; green, 0; blue, 0 }  ][line width=0.75]  [draw opacity=0] (8.93,-4.29) -- (0,0) -- (8.93,4.29) -- cycle    ;
%Straight Lines [id:da21192355530257467] 
\draw    (151.6,111.72) -- (171.61,131.93) ;

\draw [shift={(150.2,110.29)}, rotate = 45.31] [fill={rgb, 255:red, 0; green, 0; blue, 0 }  ][line width=0.75]  [draw opacity=0] (8.93,-4.29) -- (0,0) -- (8.93,4.29) -- cycle    ;
%Straight Lines [id:da9734252411833684] 
\draw    (113.71,143.79) -- (114.57,165.05) ;

\draw [shift={(113.63,141.79)}, rotate = 87.67] [fill={rgb, 255:red, 0; green, 0; blue, 0 }  ][line width=0.75]  [draw opacity=0] (8.93,-4.29) -- (0,0) -- (8.93,4.29) -- cycle    ;
%Straight Lines [id:da06815393580375662] 
\draw    (140.06,141.3) -- (117.69,162.11) ;

\draw [shift={(141.52,139.94)}, rotate = 137.06] [fill={rgb, 255:red, 0; green, 0; blue, 0 }  ][line width=0.75]  [draw opacity=0] (8.93,-4.29) -- (0,0) -- (8.93,4.29) -- cycle    ;
%Straight Lines [id:da15482139397044792] 
\draw    (108.08,142.21) -- (87.7,161.99) ;

\draw [shift={(109.52,140.81)}, rotate = 135.86] [fill={rgb, 255:red, 0; green, 0; blue, 0 }  ][line width=0.75]  [draw opacity=0] (8.93,-4.29) -- (0,0) -- (8.93,4.29) -- cycle    ;
%Straight Lines [id:da09769618738111152] 
\draw    (83.62,143.67) -- (83.57,165.92) ;

\draw [shift={(83.63,141.67)}, rotate = 90.14] [fill={rgb, 255:red, 0; green, 0; blue, 0 }  ][line width=0.75]  [draw opacity=0] (8.93,-4.29) -- (0,0) -- (8.93,4.29) -- cycle    ;
%Straight Lines [id:da6096774833493699] 
\draw    (145.62,144.17) -- (145.59,165.26) ;

\draw [shift={(145.63,142.17)}, rotate = 90.09] [fill={rgb, 255:red, 0; green, 0; blue, 0 }  ][line width=0.75]  [draw opacity=0] (8.93,-4.29) -- (0,0) -- (8.93,4.29) -- cycle    ;
%Straight Lines [id:da6264550473502231] 
\draw    (89.09,142.17) -- (109.56,162.03) ;

\draw [shift={(87.65,140.78)}, rotate = 44.12] [fill={rgb, 255:red, 0; green, 0; blue, 0 }  ][line width=0.75]  [draw opacity=0] (8.93,-4.29) -- (0,0) -- (8.93,4.29) -- cycle    ;
%Straight Lines [id:da3135662647337487] 
\draw    (119.18,142.19) -- (141.57,161.16) ;

\draw [shift={(117.65,140.9)}, rotate = 40.27] [fill={rgb, 255:red, 0; green, 0; blue, 0 }  ][line width=0.75]  [draw opacity=0] (8.93,-4.29) -- (0,0) -- (8.93,4.29) -- cycle    ;
%Straight Lines [id:da5050609639464458] 
\draw    (151.09,141.42) -- (171.57,161.28) ;

\draw [shift={(149.65,140.03)}, rotate = 44.12] [fill={rgb, 255:red, 0; green, 0; blue, 0 }  ][line width=0.75]  [draw opacity=0] (8.93,-4.29) -- (0,0) -- (8.93,4.29) -- cycle    ;
%Straight Lines [id:da35162158609292393] 
\draw    (113.58,172.14) -- (113.53,195.23) ;

\draw [shift={(113.59,170.14)}, rotate = 90.14] [fill={rgb, 255:red, 0; green, 0; blue, 0 }  ][line width=0.75]  [draw opacity=0] (8.93,-4.29) -- (0,0) -- (8.93,4.29) -- cycle    ;
%Straight Lines [id:da8281544095173596] 
\draw    (119.13,171.55) -- (141.79,191.08) ;

\draw [shift={(117.61,170.24)}, rotate = 40.75] [fill={rgb, 255:red, 0; green, 0; blue, 0 }  ][line width=0.75]  [draw opacity=0] (8.93,-4.29) -- (0,0) -- (8.93,4.29) -- cycle    ;
%Straight Lines [id:da5653097277707899] 
\draw [fill={rgb, 255:red, 255; green, 255; blue, 255 }  ,fill opacity=1 ]   (108.07,171.58) -- (87.92,191.91) ;

\draw [shift={(109.48,170.16)}, rotate = 134.74] [fill={rgb, 255:red, 0; green, 0; blue, 0 }  ][line width=0.75]  [draw opacity=0] (8.93,-4.29) -- (0,0) -- (8.93,4.29) -- cycle    ;
%Shape: Circle [id:dp03619387197274637] 
\draw  [fill={rgb, 255:red, 255; green, 255; blue, 255 }  ,fill opacity=1 ] (117.73,132.77) .. controls (119.96,135.03) and (119.92,138.67) .. (117.65,140.9) .. controls (115.38,143.12) and (111.74,143.08) .. (109.52,140.81) .. controls (107.3,138.54) and (107.33,134.9) .. (109.6,132.68) .. controls (111.87,130.46) and (115.51,130.5) .. (117.73,132.77) -- cycle ;
%Shape: Circle [id:dp8467214836049302] 
\draw  [fill={rgb, 255:red, 255; green, 255; blue, 255 }  ,fill opacity=1 ] (117.69,162.11) .. controls (119.92,164.38) and (119.88,168.02) .. (117.61,170.24) .. controls (115.34,172.47) and (111.7,172.43) .. (109.48,170.16) .. controls (107.26,167.89) and (107.3,164.25) .. (109.56,162.03) .. controls (111.83,159.81) and (115.47,159.84) .. (117.69,162.11) -- cycle ;
%Shape: Circle [id:dp3535610238059079] 
\draw  [fill={rgb, 255:red, 255; green, 255; blue, 255 }  ,fill opacity=1 ] (87.73,132.65) .. controls (89.96,134.91) and (89.92,138.55) .. (87.65,140.78) .. controls (85.38,143) and (81.74,142.96) .. (79.52,140.69) .. controls (77.3,138.42) and (77.33,134.78) .. (79.6,132.56) .. controls (81.87,130.34) and (85.51,130.38) .. (87.73,132.65) -- cycle ;
%Shape: Circle [id:dp20356414921879895] 
\draw  [fill={rgb, 255:red, 74; green, 144; blue, 226 }  ,fill opacity=1 ] (87.7,161.99) .. controls (89.92,164.26) and (89.88,167.9) .. (87.61,170.12) .. controls (85.34,172.35) and (81.7,172.31) .. (79.48,170.04) .. controls (77.26,167.77) and (77.3,164.13) .. (79.56,161.91) .. controls (81.83,159.69) and (85.47,159.72) .. (87.7,161.99) -- cycle ;
%Shape: Circle [id:dp9995476381724882] 
\draw  [fill={rgb, 255:red, 255; green, 255; blue, 255 }  ,fill opacity=1 ] (149.74,131.89) .. controls (151.96,134.16) and (151.92,137.8) .. (149.65,140.03) .. controls (147.38,142.25) and (143.74,142.21) .. (141.52,139.94) .. controls (139.3,137.67) and (139.34,134.03) .. (141.61,131.81) .. controls (143.87,129.59) and (147.51,129.63) .. (149.74,131.89) -- cycle ;
%Shape: Circle [id:dp7076442444436186] 
\draw  [fill={rgb, 255:red, 74; green, 144; blue, 226 }  ,fill opacity=1 ] (149.7,161.24) .. controls (151.92,163.51) and (151.88,167.15) .. (149.61,169.37) .. controls (147.35,171.59) and (143.71,171.56) .. (141.48,169.29) .. controls (139.26,167.02) and (139.3,163.38) .. (141.57,161.16) .. controls (143.84,158.93) and (147.48,158.97) .. (149.7,161.24) -- cycle ;
%Shape: Circle [id:dp5026693185756004] 
\draw  [fill={rgb, 255:red, 74; green, 144; blue, 226 }  ,fill opacity=1 ] (57.73,132.53) .. controls (59.96,134.79) and (59.92,138.43) .. (57.65,140.66) .. controls (55.38,142.88) and (51.74,142.84) .. (49.52,140.57) .. controls (47.3,138.3) and (47.33,134.66) .. (49.6,132.44) .. controls (51.87,130.22) and (55.51,130.26) .. (57.73,132.53) -- cycle ;
%Shape: Circle [id:dp7092022897699362] 
\draw  [fill={rgb, 255:red, 74; green, 144; blue, 226 }  ,fill opacity=1 ] (179.74,132.01) .. controls (181.96,134.28) and (181.92,137.92) .. (179.65,140.15) .. controls (177.38,142.37) and (173.74,142.33) .. (171.52,140.06) .. controls (169.3,137.79) and (169.34,134.15) .. (171.61,131.93) .. controls (173.87,129.71) and (177.51,129.75) .. (179.74,132.01) -- cycle ;
%Shape: Circle [id:dp035231123521940466] 
\draw  [fill={rgb, 255:red, 74; green, 144; blue, 226 }  ,fill opacity=1 ] (179.7,161.36) .. controls (181.92,163.63) and (181.88,167.27) .. (179.61,169.49) .. controls (177.35,171.71) and (173.71,171.68) .. (171.48,169.41) .. controls (169.26,167.14) and (169.3,163.5) .. (171.57,161.28) .. controls (173.84,159.05) and (177.48,159.09) .. (179.7,161.36) -- cycle ;
%Shape: Circle [id:dp9354210241377361] 
\draw  [fill={rgb, 255:red, 74; green, 144; blue, 226 }  ,fill opacity=1 ] (26.74,131.4) .. controls (28.96,133.67) and (28.92,137.31) .. (26.65,139.53) .. controls (24.39,141.75) and (20.75,141.72) .. (18.52,139.45) .. controls (16.3,137.18) and (16.34,133.54) .. (18.61,131.32) .. controls (20.88,129.09) and (24.52,129.13) .. (26.74,131.4) -- cycle ;
%Shape: Circle [id:dp15371615185516063] 
\draw  [fill={rgb, 255:red, 74; green, 144; blue, 226 }  ,fill opacity=1 ] (209.74,132.14) .. controls (211.96,134.4) and (211.92,138.04) .. (209.65,140.27) .. controls (207.38,142.49) and (203.74,142.45) .. (201.52,140.18) .. controls (199.3,137.91) and (199.34,134.27) .. (201.61,132.05) .. controls (203.87,129.83) and (207.51,129.87) .. (209.74,132.14) -- cycle ;
%Shape: Circle [id:dp9521510994432152] 
\draw  [fill={rgb, 255:red, 74; green, 144; blue, 226 }  ,fill opacity=1 ] (117.92,192.03) .. controls (120.14,194.3) and (120.11,197.94) .. (117.84,200.16) .. controls (115.57,202.39) and (111.93,202.35) .. (109.71,200.08) .. controls (107.48,197.81) and (107.52,194.17) .. (109.79,191.95) .. controls (112.06,189.73) and (115.7,189.77) .. (117.92,192.03) -- cycle ;
%Shape: Circle [id:dp7487567352011979] 
\draw  [fill={rgb, 255:red, 74; green, 144; blue, 226 }  ,fill opacity=1 ] (87.92,191.91) .. controls (90.14,194.18) and (90.11,197.82) .. (87.84,200.04) .. controls (85.57,202.27) and (81.93,202.23) .. (79.71,199.96) .. controls (77.48,197.69) and (77.52,194.05) .. (79.79,191.83) .. controls (82.06,189.61) and (85.7,189.64) .. (87.92,191.91) -- cycle ;
%Shape: Circle [id:dp2395589061513077] 
\draw  [fill={rgb, 255:red, 74; green, 144; blue, 226 }  ,fill opacity=1 ] (149.92,191.16) .. controls (152.15,193.43) and (152.11,197.07) .. (149.84,199.29) .. controls (147.57,201.52) and (143.93,201.48) .. (141.71,199.21) .. controls (139.49,196.94) and (139.52,193.3) .. (141.79,191.08) .. controls (144.06,188.86) and (147.7,188.89) .. (149.92,191.16) -- cycle ;
%Straight Lines [id:da9718610920066355] 
\draw    (114.88,104.14) -- (114.85,83.04) ;
\draw [shift={(114.84,81.04)}, rotate = 449.9] [fill={rgb, 255:red, 0; green, 0; blue, 0 }  ][line width=0.75]  [draw opacity=0] (8.93,-4.29) -- (0,0) -- (8.93,4.29) -- cycle    ;

%Straight Lines [id:da20245335930437336] 
\draw    (56.98,101.1) -- (77.79,80.84) ;
\draw [shift={(79.22,79.44)}, rotate = 495.77] [fill={rgb, 255:red, 0; green, 0; blue, 0 }  ][line width=0.75]  [draw opacity=0] (8.93,-4.29) -- (0,0) -- (8.93,4.29) -- cycle    ;

%Straight Lines [id:da15067127965795546] 
\draw    (86.98,101.1) -- (109.76,79.81) ;
\draw [shift={(111.22,78.44)}, rotate = 496.94] [fill={rgb, 255:red, 0; green, 0; blue, 0 }  ][line width=0.75]  [draw opacity=0] (8.93,-4.29) -- (0,0) -- (8.93,4.29) -- cycle    ;

%Straight Lines [id:da5922604216611129] 
\draw    (118.98,100.1) -- (139.79,79.84) ;
\draw [shift={(141.22,78.44)}, rotate = 495.77] [fill={rgb, 255:red, 0; green, 0; blue, 0 }  ][line width=0.75]  [draw opacity=0] (8.93,-4.29) -- (0,0) -- (8.93,4.29) -- cycle    ;

%Straight Lines [id:da07886965816681224] 
\draw    (149.98,101.1) -- (170.79,80.84) ;
\draw [shift={(172.22,79.44)}, rotate = 495.77] [fill={rgb, 255:red, 0; green, 0; blue, 0 }  ][line width=0.75]  [draw opacity=0] (8.93,-4.29) -- (0,0) -- (8.93,4.29) -- cycle    ;

%Straight Lines [id:da31507060817712484] 
\draw    (178.98,100.1) -- (199.28,80.23) ;
\draw [shift={(200.71,78.83)}, rotate = 495.63] [fill={rgb, 255:red, 0; green, 0; blue, 0 }  ][line width=0.75]  [draw opacity=0] (8.93,-4.29) -- (0,0) -- (8.93,4.29) -- cycle    ;

%Straight Lines [id:da5969045200821681] 
\draw    (144.88,104.14) -- (144.85,82.04) ;
\draw [shift={(144.84,80.04)}, rotate = 449.9] [fill={rgb, 255:red, 0; green, 0; blue, 0 }  ][line width=0.75]  [draw opacity=0] (8.93,-4.29) -- (0,0) -- (8.93,4.29) -- cycle    ;

%Straight Lines [id:da3007925850754374] 
\draw    (82.88,105.14) -- (84.24,81.49) ;
\draw [shift={(84.35,79.49)}, rotate = 453.27] [fill={rgb, 255:red, 0; green, 0; blue, 0 }  ][line width=0.75]  [draw opacity=0] (8.93,-4.29) -- (0,0) -- (8.93,4.29) -- cycle    ;

%Straight Lines [id:da3752499473594755] 
\draw    (52.88,105.14) -- (53.77,82.04) ;
\draw [shift={(53.84,80.04)}, rotate = 452.19] [fill={rgb, 255:red, 0; green, 0; blue, 0 }  ][line width=0.75]  [draw opacity=0] (8.93,-4.29) -- (0,0) -- (8.93,4.29) -- cycle    ;

%Straight Lines [id:da6889328497003284] 
\draw    (174.88,104.14) -- (175.73,83.04) ;
\draw [shift={(175.81,81.04)}, rotate = 452.3] [fill={rgb, 255:red, 0; green, 0; blue, 0 }  ][line width=0.75]  [draw opacity=0] (8.93,-4.29) -- (0,0) -- (8.93,4.29) -- cycle    ;

%Straight Lines [id:da9776367762550342] 
\draw    (170.84,100.04) -- (150.76,79.91) ;
\draw [shift={(149.35,78.49)}, rotate = 405.08000000000004] [fill={rgb, 255:red, 0; green, 0; blue, 0 }  ][line width=0.75]  [draw opacity=0] (8.93,-4.29) -- (0,0) -- (8.93,4.29) -- cycle    ;

%Straight Lines [id:da2920600842952987] 
\draw    (140.84,100.04) -- (120.76,79.91) ;
\draw [shift={(119.35,78.49)}, rotate = 405.08000000000004] [fill={rgb, 255:red, 0; green, 0; blue, 0 }  ][line width=0.75]  [draw opacity=0] (8.93,-4.29) -- (0,0) -- (8.93,4.29) -- cycle    ;

%Straight Lines [id:da9887370232601322] 
\draw    (110.84,100.04) -- (88.86,80.81) ;
\draw [shift={(87.35,79.49)}, rotate = 401.18] [fill={rgb, 255:red, 0; green, 0; blue, 0 }  ][line width=0.75]  [draw opacity=0] (8.93,-4.29) -- (0,0) -- (8.93,4.29) -- cycle    ;

%Straight Lines [id:da18079371450864823] 
\draw    (48.84,101.04) -- (28.76,80.91) ;
\draw [shift={(27.35,79.49)}, rotate = 405.08000000000004] [fill={rgb, 255:red, 0; green, 0; blue, 0 }  ][line width=0.75]  [draw opacity=0] (8.93,-4.29) -- (0,0) -- (8.93,4.29) -- cycle    ;

%Straight Lines [id:da9946841215539046] 
\draw    (78.84,101.04) -- (58.76,80.91) ;
\draw [shift={(57.35,79.49)}, rotate = 405.08000000000004] [fill={rgb, 255:red, 0; green, 0; blue, 0 }  ][line width=0.75]  [draw opacity=0] (8.93,-4.29) -- (0,0) -- (8.93,4.29) -- cycle    ;

%Straight Lines [id:da6546288964331937] 
\draw    (115.31,74.4) -- (114.36,53.15) ;
\draw [shift={(114.27,51.15)}, rotate = 447.44] [fill={rgb, 255:red, 0; green, 0; blue, 0 }  ][line width=0.75]  [draw opacity=0] (8.93,-4.29) -- (0,0) -- (8.93,4.29) -- cycle    ;

%Straight Lines [id:da8856242977515623] 
\draw    (57.4,71.36) -- (77.71,51.5) ;
\draw [shift={(79.14,50.1)}, rotate = 495.63] [fill={rgb, 255:red, 0; green, 0; blue, 0 }  ][line width=0.75]  [draw opacity=0] (8.93,-4.29) -- (0,0) -- (8.93,4.29) -- cycle    ;

%Straight Lines [id:da9273236571430648] 
\draw    (87.4,71.36) -- (109.68,50.47) ;
\draw [shift={(111.14,49.1)}, rotate = 496.83] [fill={rgb, 255:red, 0; green, 0; blue, 0 }  ][line width=0.75]  [draw opacity=0] (8.93,-4.29) -- (0,0) -- (8.93,4.29) -- cycle    ;

%Straight Lines [id:da4648165074820716] 
\draw    (119.4,70.36) -- (139.71,50.5) ;
\draw [shift={(141.14,49.1)}, rotate = 495.63] [fill={rgb, 255:red, 0; green, 0; blue, 0 }  ][line width=0.75]  [draw opacity=0] (8.93,-4.29) -- (0,0) -- (8.93,4.29) -- cycle    ;

%Straight Lines [id:da03170544818334009] 
\draw    (149.4,70.36) -- (169.71,50.5) ;
\draw [shift={(171.14,49.1)}, rotate = 495.63] [fill={rgb, 255:red, 0; green, 0; blue, 0 }  ][line width=0.75]  [draw opacity=0] (8.93,-4.29) -- (0,0) -- (8.93,4.29) -- cycle    ;

%Straight Lines [id:da5982723807380737] 
\draw    (145.31,74.4) -- (145.27,52.15) ;
\draw [shift={(145.27,50.15)}, rotate = 449.91] [fill={rgb, 255:red, 0; green, 0; blue, 0 }  ][line width=0.75]  [draw opacity=0] (8.93,-4.29) -- (0,0) -- (8.93,4.29) -- cycle    ;

%Straight Lines [id:da4816848086147494] 
\draw    (83.27,65.15) -- (83.27,53.15) ;
\draw [shift={(83.27,51.15)}, rotate = 450] [fill={rgb, 255:red, 0; green, 0; blue, 0 }  ][line width=0.75]  [draw opacity=0] (8.93,-4.29) -- (0,0) -- (8.93,4.29) -- cycle    ;

%Straight Lines [id:da5432589545173918] 
\draw    (53.31,75.4) -- (54.19,52.15) ;
\draw [shift={(54.27,50.15)}, rotate = 452.18] [fill={rgb, 255:red, 0; green, 0; blue, 0 }  ][line width=0.75]  [draw opacity=0] (8.93,-4.29) -- (0,0) -- (8.93,4.29) -- cycle    ;

%Straight Lines [id:da7214664560429493] 
\draw    (175.31,74.4) -- (174.99,53.15) ;
\draw [shift={(174.96,51.15)}, rotate = 449.13] [fill={rgb, 255:red, 0; green, 0; blue, 0 }  ][line width=0.75]  [draw opacity=0] (8.93,-4.29) -- (0,0) -- (8.93,4.29) -- cycle    ;

%Straight Lines [id:da47879747398311734] 
\draw    (171.27,70.31) -- (150.71,50.53) ;
\draw [shift={(149.27,49.15)}, rotate = 403.89] [fill={rgb, 255:red, 0; green, 0; blue, 0 }  ][line width=0.75]  [draw opacity=0] (8.93,-4.29) -- (0,0) -- (8.93,4.29) -- cycle    ;

%Straight Lines [id:da020448942660924674] 
\draw    (141.27,70.31) -- (120.71,50.53) ;
\draw [shift={(119.27,49.15)}, rotate = 403.89] [fill={rgb, 255:red, 0; green, 0; blue, 0 }  ][line width=0.75]  [draw opacity=0] (8.93,-4.29) -- (0,0) -- (8.93,4.29) -- cycle    ;

%Straight Lines [id:da44950867292461694] 
\draw    (111.27,70.31) -- (88.8,51.43) ;
\draw [shift={(87.27,50.15)}, rotate = 400.03999999999996] [fill={rgb, 255:red, 0; green, 0; blue, 0 }  ][line width=0.75]  [draw opacity=0] (8.93,-4.29) -- (0,0) -- (8.93,4.29) -- cycle    ;

%Straight Lines [id:da14746431444501606] 
\draw    (79.27,71.31) -- (58.71,51.53) ;
\draw [shift={(57.27,50.15)}, rotate = 403.89] [fill={rgb, 255:red, 0; green, 0; blue, 0 }  ][line width=0.75]  [draw opacity=0] (8.93,-4.29) -- (0,0) -- (8.93,4.29) -- cycle    ;

%Straight Lines [id:da8067389112178966] 
\draw    (115.23,45.06) -- (115.19,21.97) ;
\draw [shift={(115.19,19.97)}, rotate = 449.91] [fill={rgb, 255:red, 0; green, 0; blue, 0 }  ][line width=0.75]  [draw opacity=0] (8.93,-4.29) -- (0,0) -- (8.93,4.29) -- cycle    ;

%Straight Lines [id:da7049691487410963] 
\draw    (111.19,40.97) -- (88.45,21.53) ;
\draw [shift={(86.92,20.23)}, rotate = 400.52] [fill={rgb, 255:red, 0; green, 0; blue, 0 }  ][line width=0.75]  [draw opacity=0] (8.93,-4.29) -- (0,0) -- (8.93,4.29) -- cycle    ;

%Straight Lines [id:da13583095838966353] 
\draw [fill={rgb, 255:red, 255; green, 255; blue, 255 }  ,fill opacity=1 ]   (119.32,41.02) -- (139.39,20.6) ;
\draw [shift={(140.79,19.18)}, rotate = 494.51] [fill={rgb, 255:red, 0; green, 0; blue, 0 }  ][line width=0.75]  [draw opacity=0] (8.93,-4.29) -- (0,0) -- (8.93,4.29) -- cycle    ;

%Shape: Circle [id:dp3372900698061978] 
\draw  [fill={rgb, 255:red, 255; green, 255; blue, 255 }  ,fill opacity=1 ] (110.79,108.18) .. controls (108.56,105.92) and (108.58,102.28) .. (110.84,100.04) .. controls (113.1,97.81) and (116.74,97.84) .. (118.98,100.1) .. controls (121.21,102.36) and (121.18,106) .. (118.92,108.23) .. controls (116.67,110.46) and (113.02,110.44) .. (110.79,108.18) -- cycle ;
%Shape: Circle [id:dp013544924331004404] 
\draw  [fill={rgb, 255:red, 65; green, 117; blue, 5 }  ,fill opacity=1 ] (111.22,78.44) .. controls (108.99,76.18) and (109.01,72.54) .. (111.27,70.31) .. controls (113.53,68.08) and (117.17,68.1) .. (119.4,70.36) .. controls (121.63,72.62) and (121.61,76.26) .. (119.35,78.49) .. controls (117.09,80.73) and (113.45,80.7) .. (111.22,78.44) -- cycle ;
%Shape: Circle [id:dp20699505592801914] 
\draw  [fill={rgb, 255:red, 65; green, 117; blue, 5 }  ,fill opacity=1 ] (111.14,49.1) .. controls (108.91,46.84) and (108.93,43.2) .. (111.19,40.97) .. controls (113.45,38.73) and (117.09,38.76) .. (119.32,41.02) .. controls (121.55,43.28) and (121.53,46.92) .. (119.27,49.15) .. controls (117.01,51.38) and (113.37,51.36) .. (111.14,49.1) -- cycle ;
%Shape: Circle [id:dp4167990916818498] 
\draw  [fill={rgb, 255:red, 65; green, 117; blue, 5 }  ,fill opacity=1 ] (140.79,108.18) .. controls (138.56,105.92) and (138.58,102.28) .. (140.84,100.04) .. controls (143.1,97.81) and (146.74,97.84) .. (148.98,100.1) .. controls (151.21,102.36) and (151.18,106) .. (148.92,108.23) .. controls (146.67,110.46) and (143.02,110.44) .. (140.79,108.18) -- cycle ;
%Shape: Circle [id:dp4160900635656233] 
\draw  [fill={rgb, 255:red, 65; green, 117; blue, 5 }  ,fill opacity=1 ] (141.22,78.44) .. controls (138.99,76.18) and (139.01,72.54) .. (141.27,70.31) .. controls (143.53,68.08) and (147.17,68.1) .. (149.4,70.36) .. controls (151.63,72.62) and (151.61,76.26) .. (149.35,78.49) .. controls (147.09,80.73) and (143.45,80.7) .. (141.22,78.44) -- cycle ;
%Shape: Circle [id:dp4986975787271337] 
\draw   (141.14,50.1) .. controls (138.91,47.84) and (138.93,44.2) .. (141.19,41.97) .. controls (143.45,39.73) and (147.09,39.76) .. (149.32,42.02) .. controls (151.55,44.28) and (151.53,47.92) .. (149.27,50.15) .. controls (147.01,52.38) and (143.37,52.36) .. (141.14,50.1) -- cycle ;
%Shape: Circle [id:dp7052603587790223] 
\draw  [fill={rgb, 255:red, 65; green, 117; blue, 5 }  ,fill opacity=1 ] (78.79,109.18) .. controls (76.56,106.92) and (76.58,103.28) .. (78.84,101.04) .. controls (81.1,98.81) and (84.74,98.84) .. (86.98,101.1) .. controls (89.21,103.36) and (89.18,107) .. (86.92,109.23) .. controls (84.67,111.46) and (81.02,111.44) .. (78.79,109.18) -- cycle ;
%Shape: Circle [id:dp24410580929719305] 
\draw  [fill={rgb, 255:red, 65; green, 117; blue, 5 }  ,fill opacity=1 ] (79.22,79.44) .. controls (76.99,77.18) and (77.01,73.54) .. (79.27,71.31) .. controls (81.53,69.08) and (85.17,69.1) .. (87.4,71.36) .. controls (89.63,73.62) and (89.61,77.26) .. (87.35,79.49) .. controls (85.09,81.73) and (81.45,81.7) .. (79.22,79.44) -- cycle ;
%Shape: Circle [id:dp7134894413488662] 
\draw   (79.14,51.1) .. controls (76.91,48.84) and (76.93,45.2) .. (79.19,42.97) .. controls (81.45,40.73) and (85.09,40.76) .. (87.32,43.02) .. controls (89.55,45.28) and (89.53,48.92) .. (87.27,51.15) .. controls (85.01,53.38) and (81.37,53.36) .. (79.14,51.1) -- cycle ;
%Shape: Circle [id:dp9334692238719176] 
\draw  [fill={rgb, 255:red, 65; green, 117; blue, 5 }  ,fill opacity=1 ] (170.79,108.18) .. controls (168.56,105.92) and (168.58,102.28) .. (170.84,100.04) .. controls (173.1,97.81) and (176.74,97.84) .. (178.98,100.1) .. controls (181.21,102.36) and (181.18,106) .. (178.92,108.23) .. controls (176.67,110.46) and (173.02,110.44) .. (170.79,108.18) -- cycle ;
%Shape: Circle [id:dp8122816791900898] 
\draw   (171.22,79.44) .. controls (168.99,77.18) and (169.01,73.54) .. (171.27,71.31) .. controls (173.53,69.08) and (177.17,69.1) .. (179.4,71.36) .. controls (181.63,73.62) and (181.61,77.26) .. (179.35,79.49) .. controls (177.09,81.73) and (173.45,81.7) .. (171.22,79.44) -- cycle ;
%Shape: Circle [id:dp6799199055118115] 
\draw   (171.14,49.1) .. controls (168.91,46.84) and (168.93,43.2) .. (171.19,40.97) .. controls (173.45,38.73) and (177.09,38.76) .. (179.32,41.02) .. controls (181.55,43.28) and (181.53,46.92) .. (179.27,49.15) .. controls (177.01,51.38) and (173.37,51.36) .. (171.14,49.1) -- cycle ;
%Shape: Circle [id:dp9446619228327204] 
\draw  [fill={rgb, 255:red, 65; green, 117; blue, 5 }  ,fill opacity=1 ] (48.79,109.18) .. controls (46.56,106.92) and (46.58,103.28) .. (48.84,101.04) .. controls (51.1,98.81) and (54.74,98.84) .. (56.98,101.1) .. controls (59.21,103.36) and (59.18,107) .. (56.92,109.23) .. controls (54.67,111.46) and (51.02,111.44) .. (48.79,109.18) -- cycle ;
%Shape: Circle [id:dp3623343720605299] 
\draw   (49.22,80.44) .. controls (46.99,78.18) and (47.01,74.54) .. (49.27,72.31) .. controls (51.53,70.08) and (55.17,70.1) .. (57.4,72.36) .. controls (59.63,74.62) and (59.61,78.26) .. (57.35,80.49) .. controls (55.09,82.73) and (51.45,82.7) .. (49.22,80.44) -- cycle ;
%Shape: Circle [id:dp13720638213826408] 
\draw   (49.14,50.1) .. controls (46.91,47.84) and (46.93,44.2) .. (49.19,41.97) .. controls (51.45,39.73) and (55.09,39.76) .. (57.32,42.02) .. controls (59.55,44.28) and (59.53,47.92) .. (57.27,50.15) .. controls (55.01,52.38) and (51.37,52.36) .. (49.14,50.1) -- cycle ;
%Shape: Circle [id:dp09741502411769654] 
\draw   (202.22,80.44) .. controls (199.99,78.18) and (200.01,74.54) .. (202.27,72.31) .. controls (204.53,70.08) and (208.17,70.1) .. (210.4,72.36) .. controls (212.63,74.62) and (212.61,78.26) .. (210.35,80.49) .. controls (208.09,82.73) and (204.45,82.7) .. (202.22,80.44) -- cycle ;
%Shape: Circle [id:dp5045442257500612] 
\draw   (19.22,80.44) .. controls (16.99,78.18) and (17.01,74.54) .. (19.27,72.31) .. controls (21.53,70.08) and (25.17,70.1) .. (27.4,72.36) .. controls (29.63,74.62) and (29.61,78.26) .. (27.35,80.49) .. controls (25.09,82.73) and (21.45,82.7) .. (19.22,80.44) -- cycle ;
%Shape: Circle [id:dp6145791638025704] 
\draw   (110.79,20.18) .. controls (108.56,17.92) and (108.58,14.28) .. (110.84,12.04) .. controls (113.1,9.81) and (116.74,9.84) .. (118.98,12.1) .. controls (121.21,14.36) and (121.18,18) .. (118.92,20.23) .. controls (116.67,22.46) and (113.02,22.44) .. (110.79,20.18) -- cycle ;
%Shape: Circle [id:dp3105092519557038] 
\draw   (140.79,19.18) .. controls (138.56,16.92) and (138.58,13.28) .. (140.84,11.04) .. controls (143.1,8.81) and (146.74,8.84) .. (148.98,11.1) .. controls (151.21,13.36) and (151.18,17) .. (148.92,19.23) .. controls (146.67,21.46) and (143.02,21.44) .. (140.79,19.18) -- cycle ;
%Shape: Circle [id:dp5895969488660826] 
\draw   (78.79,20.18) .. controls (76.56,17.92) and (76.58,14.28) .. (78.84,12.04) .. controls (81.1,9.81) and (84.74,9.84) .. (86.98,12.1) .. controls (89.21,14.36) and (89.18,18) .. (86.92,20.23) .. controls (84.67,22.46) and (81.02,22.44) .. (78.79,20.18) -- cycle ;
%Shape: Circle [id:dp44060252569177405] 
\draw  [fill={rgb, 255:red, 74; green, 144; blue, 226 }  ,fill opacity=1 ] (58.7,162.36) .. controls (60.92,164.63) and (60.88,168.27) .. (58.61,170.49) .. controls (56.35,172.71) and (52.71,172.68) .. (50.48,170.41) .. controls (48.26,168.14) and (48.3,164.5) .. (50.57,162.28) .. controls (52.84,160.05) and (56.48,160.09) .. (58.7,162.36) -- cycle ;
%Straight Lines [id:da13981716714306192] 
\draw    (78.13,142.14) -- (58.7,162.36) ;

\draw [shift={(79.52,140.69)}, rotate = 133.86] [fill={rgb, 255:red, 0; green, 0; blue, 0 }  ][line width=0.75]  [draw opacity=0] (8.93,-4.29) -- (0,0) -- (8.93,4.29) -- cycle    ;
%Shape: Diamond [id:dp31258328831101667] 
\draw  [color={rgb, 255:red, 208; green, 2; blue, 27 }  ,draw opacity=1 ][dash pattern={on 4.5pt off 4.5pt}] (114,32.51) -- (190.5,106.1) -- (114,179.68) -- (37.5,106.1) -- cycle ;

\end{tikzpicture}

}